\documentclass[a4paper, 12pt]{article}
\usepackage{amsmath,amssymb, amscd}
\usepackage{pdfsync}
\usepackage{graphicx,subfigure}
\usepackage{color}

\newcommand{\sign}{\mbox{sign}\hskip 0.5truemm}
\numberwithin{equation}{section}

\newtheorem{theorem}{Theorem}[section]
\newtheorem{prop}[theorem]{Proposition}
\newtheorem{lem}[theorem]{Lemma}

\newtheorem{conj}[theorem]{Conjecture}

\newtheorem{rem}[theorem]{Remark}
\newenvironment{proof}{\noindent {\it Proof }}{\hfill $\square$}

\newcommand{\eqa}{\begin{eqnarray}}
\newcommand{\eeqa}{\end{eqnarray}}
\newcommand{\beq}{\begin{equation}}
\newcommand{\eeq}{\end{equation}}
\newcommand{\nn}{\nonumber}

\newcommand{\pal}{\partial}
\newcommand{\pf}{\noindent{\it Proof \ }}

\newcommand{\epf}{$\quad$\hfill
\raisebox{0.11truecm}{\fbox{}}\par\vskip0.4truecm}

\begin{document}

\title{On critical behaviour in systems of Hamiltonian partial 
differential equations}

\author{B.~Dubrovin, T. Grava, C. Klein, A. Moro}


\date{}
\maketitle

\begin{abstract} 
We study the critical behaviour of solutions to weakly dispersive Hamiltonian systems considered as perturbations  of elliptic and
hyperbolic systems of hydrodynamic type with two components. We argue 
that near the critical point of gradient catastrophe of the 
dispersionless system, the solutions to a suitable initial
value problem for the perturbed equations are  approximately described by 
particular solutions to the Painlev\'e-I (P$_I$) equation or its fourth order analogue P$_I^2$. As concrete 
examples we discuss nonlinear Schr\"odinger equations in the 
semiclassical limit. A numerical 
study of these cases provides strong evidence in support of the 
conjecture. 
\end{abstract}


\section{Introduction}
Critical phenomena in the solutions of partial differential equations (PDEs) are important from various theoretical and applied points of view since such phenomena  generally indicate 
 the appearence of new  behaviours  as the onset of rapid oscillations, the appearence of multiple scales, or a loss of regularity in the solutions.
Some of the most powerful techniques in the asymptotic description of such phenomena are due to the theory of completely integrable systems which were so far restricted to integrable PDEs.
 In \cite{du2} this restriction was overcome by introducing the  concept of integrability up to a finite order of some small parameter $\epsilon$. This has allowed to apply techniques from the theory of integrable systems to a large class of non-integrable equations and to obtain asymptotic descriptions of solutions to such equations in the vicinity of critical points of these PDEs. The scalar case was studied along these lines in \cite{du2}. Basically it was shown that solutions to  dispersive regularizations of a nonlinear transport equation near a point of gradient catastrophe (for the transport equation itself) behave like solutions of the celebrated Korteweg--de Vries equations, which at such point can be asymptotically expressed in terms of a particular solution to a fourth order ordinary differential equation from the Painlev\'e-I family. In \cite{DGK} this concept was generalized to the  study of  the semiclassical limit of the  integrable focusing cubic nonlinear Schr\"odinger equation (NLS) which can be seen as a perturbation of   $2\times 2$  elliptic system and in \cite{dub3} to  a certain class of  integrable Hamiltonian perturbation of $2\times 2$ elliptic and hyperbolic systems.
 The idea that integrable behaviour persists in certain nonintegrable cases has been already developed in the study of long time behaviour of solutions to several nonintegrable equations, like the pertuberd NLS equation \cite{DZ}, see also \cite{tao} for a general overview about the soliton resolution conjecture.
 
  In this paper we consider general  two-component Hamiltonian systems  which  contain a small dispersion parameter $\epsilon$.  When  $\epsilon=0$  the Hamiltonian system reduces to a $2\times2$ quasilinear system of elliptic or hyperbolic type, so that the Hamiltonian system can be considered as a perturbation of the elliptic or hyperbolic systems. We study the behaviour of solutions  to  such Hamiltonian systems when the parameter $\epsilon $ tends to zero. The fundamental question we address is how does a solution to  Hamiltonian equations behave near the point where  the solution of the unperturbed elliptic or hyperbolic system breaks up.

We consider  Hamiltonian PDEs obtained as perturbations of systems of hydrodynamic type of the form
\eqa\label{cano00}
&&
u_t=\pal_x\frac{\delta H_0}{\delta v(x)}\equiv \pal_x \frac{\pal h}{\pal v}
\nn\\
&&
\\
&&
v_t=\pal_x\frac{\delta H_0}{\delta u(x)}\equiv \pal_x \frac{\pal h}{\pal u},
\nn
\eeqa
with  
 $u=u(x,t)$, $v=v(x,t)$, scalar functions,  $x\in\mathbb{R}$ and
$$
H_0=\int h(u,v)\, dx,
$$
where $ h=h(u,v)$ is a smooth function of $u$ and $v$. Such perturbations can be written in the form
\eqa\label{cano10}
&&
u_t=\pal_x\frac{\delta H}{\delta v(x)}
\nn\\
&&
\\
&&
v_t=\pal_x\frac{\delta H}{\delta u(x)}.
\nn
\eeqa
where $H$ is the perturbed Hamiltonian, $H=H_0+\epsilon \, H_1 +\epsilon^2 H_2 +\dots$. By definition the $k$-th order term of the perturbative expansion must have the form
$$
H_k=\int h_k\left(u, v, u_x, v_x,u_{xx}, v_{xx}, \dots, u^{(k)}, v^{(k)}\right)\, dx
$$
where $h_k\left(u, v, u_x, v_x,u_{xx}, v_{xx}, \dots, u^{(k)}, v^{(k)}\right)$ is a graded homogeneous polynomial of degree $k$ in the variables $u_x$, $v_x$, \dots, $u^{(k)}$, $v^{(k)}$, i.e., it satisfies the identity
$$
h_k\left(u, v, \lambda\, u_x, \lambda\, v_x,\lambda^2 u_{xx}, \lambda^2 v_{xx}, \dots, \lambda^k u^{(k)}, \lambda^k v^{(k)}\right)=\lambda^k h_k\left(u, v, u_x, v_x,u_{xx}, v_{xx}, \dots, u^{(k)}, v^{(k)}\right)
$$
for an arbitrary $\lambda$. The Hamiltonian system \eqref{cano10} can be considered as a weakly dispersive perturbation of the 1st order quasilinear system \eqref{cano00}. 

After certain simplification of the system \eqref{cano10} by a suitable class of $\epsilon$-dependent canonical transformations 
\eqa
&&
u(x)\mapsto \tilde u(x)=u(x)+\epsilon\,\{ u(x), F\} +{\mathcal O}(\epsilon^2)
\nn\\
&&
v(x)\mapsto \tilde v(x)=v(x)+\epsilon\,\{ v(x), F\} +{\mathcal O}(\epsilon^2)
\nn
\eeqa
generated by a Hamiltonian $F$ (see Section~\ref{section1})
the system can be spelled out as follows
\eqa\label{cano100}
&&
u_t =\pal_x \frac{\delta H}{\delta v(x)} = h_{uv}u_x +h_{vv} v_x 
\nn\\
&&
+\epsilon^2 \left[b\, u_{xxx} + c\, v_{xxx} +( -a_v + 3 b_u) u_{xx} u_x +(b_v + c_u) u_{xx} v_x + 2 c_u v_{xx} u_x + 2 c_v v_{xx} v_x
\right.
\nn\\
&&\left.
+\left( -\frac12 a_{uv} +b_{uu}\right) u_x^3 +\left( -\frac12 a_{vv} + b_{uv} +c_{uu}\right) u_x^2 v_x +\frac32 c_{uv} u_x v_x^2 +\frac12 c_{vv} v_x^3\right]
\nn\\
&&
\\
&&
v_t=\pal_x \frac{\delta H}{\delta u(x)} = h_{uu} u_x + h_{uv} v_x
\nn\\
&&
+\epsilon^2 \left[ a\, u_{xxx} +b\, v_{xxx}+2 a_u u_{xx} u_x + 2 a_v u_{xx} v_x +(a_v + b_u) v_{xx} u_x +(3 b_v -c_u) v_{xx} v_x
\right.
\nn\\
&&\left.
+\frac12 a_{uu} u_x^3 +\frac32 a_{uv} u_x^2 v_x +\left( a_{vv} +b_{uv} -\frac12 c_{uu}\right) u_x v_x^2 +\left( b_{vv}-\frac12 c_{uv}\right) v_x^3\right]
\nn
\eeqa
up to terms of order $\epsilon^3$.
Here $a=a(u,v)$, $b=b(u,v)$ and $c=c(u,v)$ are arbitrary smooth functions of $u$ and $v$ at least in the domain where the solution of the unperturbed equation (\ref{cano00}) takes values. The corresponding perturbed Hamiltonian reads
\beq\label{pert200}
H=H_0 +\epsilon^2 H_2=\int \left[ h -\frac{\epsilon^2}2 \left( a\, u_x^2 + 2 b\, u_x v_x + c\, v_x^2\right)\right]\, dx
\eeq

The  family of equations of the form \eqref{cano100} contains important examples such as the 
generalised nonlinear Schr\"odinger (NLS) equations (also in a 
nonlocal version),
the long wave limit of lattice equations like the Fermi--Pasta--Ulam or Toda lattice equation, Boussinesq equation, two-component Camassa--Holm equation \cite{falqui}  and many others.
For certain choices of the functions $h(u,v)$, $a(u,v)$, 
$b(u,v)$ and $c(u,v)$, the system of equations (\ref{cano100}) is integrable up to the order $\epsilon^3$ \cite{dub3}.  However the complete  classification 
of integrable cases in the class of equations of the form 
(\ref{cano100}) remains open, see \cite{degasperis}, \cite{du1}, 
\cite{km} for the current state of the art in this context.

The study of scalar weakly dispersive equations
\eqa
&&
u_t=\pal_x \frac{\delta H}{\delta u(x)}=\pal_x\left( h'(u)+\epsilon^2 \left[ a(u) u_{xx} +\frac12 a'(u) u_x^2\right]+{\mathcal O}\left(\epsilon^3\right)\right)
\nn\\
&&
\nn\\
&&
 H=\int \left[ h(u) -\frac{\epsilon^2}2 a(u) u_x^2 +\dots\right]\, dx 
\nn
\eeqa
of the form similar to (\ref{cano10}), (\ref{pert200}) in the limit $\epsilon\rightarrow 0$  in the strongly nonlinear regime 
was initiated by the seminal paper  by 
Gurevich and Pitaevsky \cite{GP} about ``collisionless shock waves" described by KdV equation (see also the book \cite{novik} and references therein). Rigorous mathematical results in this direction were obtained by Lax, Levermore and Venakides \cite{ll}, \cite{llv}, \cite{venakides}  and Deift, 
Venakides and Zhou 
\cite{dvz}  (see also \cite{gk1} and \cite{physicad} 
for numerical comparison). For two-component systems \eqref{cano100} an analogous line of research was 
started with the works \cite{ca}, \cite{ge}, \cite{gr}  on the 
semiclassical limit of  generalized  defocusing nonlinear 
Schr\"odinger equation in  several  space dimensions for times less 
than the critical time  $t_0$ of the cusp catastrophe. It was 
studied in more details for arbitrary times  for  the integrable case \cite{za},  
namely for the spatially one-dimensional cubic defocusing NLS  in   
\cite{jin}, \cite{ji}, \cite{difranco}.  Another system that is 
included in the class (\ref{cano10}) is the long wave limit of the 
Toda lattice equation that has been studied in detail  for arbitrary 
times in \cite{deiftmc} and, in the context of  Hermitian random matrix 
models with exponential weights  by many authors, see the book \cite{Deift} and references therein. Interesting results, in the spirit of the original Gurevich and Pitaevsky setting,  have been obtained for certain nonintegrable cases in \cite{el}, \cite{hoefer}. Possible relations between integrable and non-integrable behaviour have been  also analyzed in the framework of the long wave limit of the Fermi--Pasta--Ulam system by Zabusky and Kruskal \cite{za}   and, more recently, in \cite{bambusi}, \cite{paleari}, \cite{benettin}.

The study of  solutions to Hamiltonian systems of the form (\ref{cano100}) in the limit $\epsilon\rightarrow 0 $ whith the leading term (\ref{cano00}) of elliptic type was initiated  by the analysis of the  semiclassical limit of the focusing  cubic nonlinear Schr\"odinger equation \cite{kam2}, see also  \cite{bronski}, \cite{ce} \cite{miller}, \cite{to1}, \cite{to2}.
Other interesting  Hamiltonian systems not included in the class 
(\ref{cano100})  have been considered in the limit $\epsilon\rightarrow 0$ in \cite{miller0}, \cite{bu}. 

Our study can be considered as a continuation of the programme 
initiated in \cite{du2}, \cite{dub3} and \cite{DGK} aimed at studying critical behaviour of 
Hamiltonian perturbations of quasilinear hyperbolic and elliptic  PDEs. 
The most important of the concepts developped in these papers  is the idea of {\it universality} of the 
critical behaviour. We borrow this notion from the theory of random matrices where various universality types of critical   behaviour appear 
in  the study of phase transitions in random matrix ensembles,  see for example  
\cite{bleher}, \cite{bertola}, \cite{dkmvz1},  \cite{dkmvz2}, 
\cite{duits}, \cite{cv1} for mathematically oriented references. 
The description of the  critical behaviour for generalized Burgers equation with small viscosity was found by Il'in \cite{ilin}; for more general weakly dissipative equations see \cite{dubel}, \cite{arise}.

\medskip

In the present paper solutions $u(x,t; \epsilon)$, $v(x,t; \epsilon)$ to the Cauchy problem
\beq\label{kosh00}
u(x,0; \epsilon)=u_0(x), \quad v(x,0;\epsilon)=v_0(x)
\eeq
for the system \eqref{cano100} with $\epsilon$-independent smooth 
initial data in a suitable functional class will be under 
consideration. The contribution of higher order terms is believed to 
be negligible as long  as the solution $\left(u(x,t; \epsilon), 
v(x,t; \epsilon)\right)$ remains a slowly varying function of $x$ and 
$t$, that is, it  changes by ${\cal O}(1)$ on the space- and 
time-scale of order ${\cal O}\left( \epsilon^{-1}\right)$. A rigorous 
proof of such a statement would justify existence, for sufficiently 
small values of the parameter $\epsilon$, of the solution to the 
Cauchy problem \eqref{cano10}, \eqref{kosh00} on a finite time 
interval $0\leq t \leq t_0$  depending on the initial condition but not on $\epsilon$. This was proven by Lax, Levermore and Venakides \cite{ll}, \cite{llv} for the particular case of the Korteweg - de Vries (KdV) equation with rapidly decreasing initial data. In a more general setting of a certain class of generalized KdV equations with no integrability assumption, the statement was proven more recently in \cite{maoero}.

Actually we expect validity of a more bold statement that, in particular, gives an efficient upper bound for the life span of a solution to \eqref{cano10} with given initial data \eqref{kosh00}. Namely, we start with considering the solution $\left( u(x,t), v(x,t)\right)$ to the Cauchy problem for the unperturbed system \eqref{cano00} with the same initial data\footnote{Analyticity of the initial data will be assumed in case the quasilinear system \eqref{cano00} is of elliptic type. The precise formulation of our Main Conjecture has to be refined in the non analytic case}
$$
u(x,0)=u_0(x), \quad v(x,0)=v_0(x).
$$
Such a solution exists for times below the time $t_0$ of \emph{gradient catastrophe}. We expect that the life span of the perturbed solution $\left(u(x,t; \epsilon), v(x,t; \epsilon)\right)$ for sufficiently small $\epsilon$ is at least the interval $[0,t_0]$. More precisely, we have the following

\noindent
{\bf Main Conjecture}. \\
\\
Part 1.  {\it There exists a positive constant $\Delta t(\epsilon)>0$ depending on the initial condition \eqref{kosh00} such that the solution to the Cauchy problem \eqref{cano10}, \eqref{kosh00} exists for $0\leq t < t_0+\Delta t(\epsilon)$ for sufficiently small $\epsilon$.
}

\noindent
Part 2.  {\it When $\epsilon\to 0$ the perturbed solution $\left(u(x,t; \epsilon), v(x,t; \epsilon)\right)$ converges to the unperturbed one $\left(u(x,t), v(x,t)\right)$ uniformly on compacts $x_1\leq x\leq x_2$, $0\leq t \leq t_1$ for any $t_1< t_0$ and arbitrary $x_1$ and $x_2$.
}

In the Main Conjecture we do not specify the class of boundary 
conditions for the smooth (or even analytic, in the elliptic case) 
initial data $u_0(x)$, $v_0(x)$. We believe that the statement is applicable to a wide class of boundary conditions like rapidly decreasing, step-like, periodic etc. Moreover, the shape of the universal critical behavior at the point of catastrophe (see below) should be  independent of the choice of boundary conditions.

The last statement of the Main Conjecture refers to the behavior of a generic perturbed solution near the point of gradient catastrophe of the unperturbed one. Our main goal is to find an asymptotic description for the dispersive 
regularisation of the elliptic umbilic
 singularity or the cusp catastrophe  when the dispersive terms are added, i.e., 
we want to describe the leading term of the asymptotic behaviour for $\epsilon\to 0$ 
of the solution to \eqref{cano100} near the critical point, say $(x_0, t_0)$, of a generic solution to 
\eqref{cano00}.

At the point of catastrophe the solutions $u(x,t)$, $v(x,t)$ to the 
Cauchy problem \eqref{cano00}, \eqref{kosh00} remain continuous, but their derivatives blow up. The  generic singularities are classified as follows \cite{du2}, \cite{dub3}.
\begin{itemize}
\item If the system (\ref{cano00}) is elliptic, $h_{uu}h_{vv}<0$, 
then the generic singularity is   a point of elliptic umbilic catastrophe.
This codimension 2 singularity is one of the real forms labeled by the root system of the $D_4$ type in the terminology of Arnold {\sl et al.} \cite{ar}.
\item  If the system (\ref{cano00}) is hyperbolic, $h_{uu}h_{vv}>0$, 
then the generic singularity is   a point of cusp  catastrophe or more precisely the Whitney W3 [4] singularity.
\end{itemize}

Elliptic umbilic singularities appear  in experimental and 
theoretical studies of diffraction in more than one spatial dimension 
\cite{berry}, in plasma physics \cite{sl}, \cite{si}, in the Hele--Shaw problem \cite{mm1},  and also in random matrices, \cite{fokas}, \cite{bertola}.
Formation of singularities for general quasi-linear hyperbolic systems  in many spatial dimensions has been considered in  \cite{al} (see \cite{manakov} for  an  explicit example).
For the particular case of $2\times 2$ systems we are mainly dealing 
with, the derivation
of the cusp catastrophe was obtained for $\mathcal{C}^4$ initial data  in \cite{kong}, see also \cite{dub3} for an alternative derivation.

Let us return to the Cauchy problem for the perturbed system \eqref{cano100} with the same initial data \eqref{kosh00}. 
The fundamental idea of \emph{universality} first formulated in \cite{du2} for scalar Hamiltonian PDEs suggests that, at the leading order of asymptotic approximation, such behavior does depend neither on the choice of generic initial data nor on the choice of generic Hamiltonian perturbation. One of the goals of the present paper is to give a precise formulation of the universality conjecture for a quite general class of systems of Hamiltonian PDEs of order two (for certain particular subclasses of such PDEs the universality conjecture has already been formulated in \cite{dub3}).


 The general formulation of  universality 
introduced in \cite{du2} for the case of Hamiltonian perturbations of 
the scalar nonlinear transport equation and in \cite{dub3} for Hamiltonian perturbation of the nonlinear wave equation  says that the leading term of 
the multiscale asymptotics of the generic solution near the critical point does not 
depend on the choice of the solution, modulo Galilean transformations and rescalings. 
This leading term was identified  via  a particular solution to the fourth order 
analogue of the Painlev\'e-I (P$_{I}$) equation (the so-called P$_{I}^2$ 
equation). Earlier the particular solution to the P$_{I}^2$ equation proved to be important in the theory of random matrices \cite{moore}, \cite{bmp}; in the context of the so-called Gurevich--Pitaevsky solution to the KdV equation it was derived in \cite{ks}.
The existence of the needed smooth solution to P$_{I}^2$ has been rigorously established in \cite{cv2}. 
Moreover, it was argued in \cite{du2} and \cite{dub3} that the shape of the leading term describing the critical  behaviour is 
essentially independent on the particular form of the Hamiltonian perturbation. 
Some of these universality conjectures have been supported 
by numerical experiments carried out in \cite{gr}, \cite{DGK11}.  The 
rigourous analytical proof of this conjecture has been obtained for 
the KdV equation in \cite{cg1}.

In \cite{DGK} the  universality conjecture for the critical behaviour of solutions 
to the focusing cubic NLS has been formulated, and in \cite{dub3} the universality conjecture has been extended to other integrable Hamiltonian perturbations of elliptic systems.
The universality conjecture in this case  suggests that the description of the leading term 
in the asymptotic expansion of the solution to the focusing NLS equation in the semiclassical limit,
 near the point of elliptic umbilic catastrophe, is given  via a particular solution to the classical 
Painlev\'e-I equation (P$_I$), namely the tritronqu\'ee solution 
first introduced by Boutroux \cite{bo} one hundred years ago; see 
\cite{jk}, \cite{ka} regarding some important properties of the tritronqu\'ee solution and its characterization in the framework of the theory of isomonodromy deformations.  The  
smoothness of the tritronqu\'ee solution  in a sector of 
the complex $z$-plane   of angle $|\arg z|<4\pi/5$ conjectured in \cite{DGK} has only 
recently been  proved in \cite{co1}. Other arguments supporting  the universality conjecture for the focusing case were found in \cite{bt}.

In this paper we extend these ideas to the  more general class of 
systems  of the form (\ref{cano100}).
More precisely, our main goal is a precise formulation of the following conjectural statement.

{\bf Main Conjecture}, Part 3.
{\it
\begin{itemize}
\item The solution of the generic system (\ref{cano100}) with generic $\epsilon$-independent analytic initial data near a point of elliptic umbilic catastrophe 
of the unperturbed elliptic  system (\ref{cano00}) in the limit 
$\epsilon\rightarrow 0$ is described by the tritronqu\'ee solution to the P$_I$ equation; 
\item  the solution of the generic system (\ref{cano100}) with generic $\epsilon$-independent smooth initial data near a point of cusp  catastrophe 
of the unperturbed hyperbolic system (\ref{cano00}) is described in the limit 
$\epsilon\rightarrow 0$ by a particular solution to the P$_{I}^2$ equation.
\end{itemize}
}
An important aspect of the above conjectures is the existence of the solution of the perturbed
 Hamiltonian systems (\ref{cano100}) for times $t$ up to and  slightly beyond the critical time $t_0$ for
the solution of the unperturbed system (\ref{cano00}). The  study of the local or global  well-posedness of the Cauchy problem for the full class of  equations (\ref{cano100}) remains open even though a 
large class of equations has been studied, see for example \cite{gwp1} or \cite{Tao}, \cite{ponce}, \cite{bour} for a survey of the state-of-the-art.
For finite $\epsilon$ it is known that the solution of the Cauchy 
problem of certain classes of equations of the form (\ref{cano100}) develops blow-up in finite time, see for example 
\cite{sulem}, \cite{merle}.
For the class of equations  of the form (\ref{cano10}) and initial data such that the solution develops 
a blow up  in finite time $t_B$, we consequentially  conjecture that, 
for sufficiently small $\epsilon$, the blow-up time $t_B$ is always 
larger than the critical time $t_0$ of the dispersionless system. The 
blow-up behaviour of solutions to certain class of equations, like 
the focusing NLS equation has been studied in detail in
\cite{MM}, however the issue of the determination of the blow-up time remains open. 
For the particular case of the quintic focusing NLS equation, we claim  that the blow-up time $t_B$ which depends on $\epsilon$ is close in the limit $\epsilon \rightarrow 0$ to the time of elliptic umbilic catastrophe, more precisely   the ratio $(t_B-t_0)/\epsilon^{\frac{4}{5}}$ is asymptotically equal to a constant that depends 
on the location of the  first pole   of the Painlev\'e-I tritronqu\'ee solution on the negative real axis.

This paper is organized as follows. In section~\ref{section1} we single out the class of Hamiltonian
systems (\ref{cano10}) and we recall the procedure of obtaining solutions of the system (\ref{cano00})
 by a suitable form of the method of characteristics. In section~\ref{section2} we study the generic singularity of the  solutions to (\ref{cano00}) and describe the conjectural behavior for the generic solution of a Hamiltonian perturbation \eqref{cano100} of the hyperbolic system (\ref{cano00}) in the neighbourhhod of such singularity. The same program is realized in 
section~\ref{section3}  for Hamiltonian perturbations of an elliptic system of the form (\ref{cano00}).
In section~\ref{section4} we consider in more details the above results 
for the generalized nonlinear Schr\"odinger equation (NLS) and the 
nonlocal NLS equation, and in section~\ref{section5} we study 
analytically some particular solutions of the system (\ref{cano00}) 
up to the critical time $t_0$  for the generalized 
NLS equation.
In sections~\ref{section6} to \ref{section8} we present 
numerical evidences supporting the validity of the above conjectures. 

\section{Hamiltonian systems}\label{section1}
In this section we identify the class of Hamiltonian equations we are interested in.
Let us consider the class of systems of Hamiltonian PDEs of the form
\eqa\label{sys1}
&&
u^i_t = A^i_j(u) u^j_x +\epsilon\left[ B^i_j(u) u^j_{xx} +\frac12 L^i_{jk}(u) u^j_x u^k_x\right]
\\
&&
\quad\quad +\epsilon^2 \left[ C^i_j(u) u^j_{xxx} +M^i_{jk}(u) u^j_{xx} u^k_x + \frac16 N^i_{jkm}(u) u^j_x u^k_x u^m_x\right] +O\left( \epsilon^3\right),
\nn\\
&&
\nn\\
&&
\quad\quad\quad i=1, \dots, n,
\nn
\eeqa
where we are taking the sum over repeated indices.
The system of coordinates on the space of dependent variables can be chosen in such a way that the Poisson bracket takes the standard form,  \cite{DN}
\beq\label{sys2}
\{ u^i(x), u^j(y)\} =\eta^{ij} \delta'(x-y), \quad i, j=1, \dots, n
\eeq
where $\left(\eta^{ij}\right)$ is a constant symmetric nondegenerate matrix. Choosing a Hamiltonian in the form
\eqa\label{ham}
&&
H=\int h(u; u_x, \dots; \epsilon)\, dx
\\
&&
 h(u; u_x, \dots; \epsilon)= h^{[0]}(u) +\epsilon\, p_i(u) u^i_x +\frac12\epsilon^2 q_{ij}(u) u^i_x u^j_x +O\left( \epsilon^3\right)
 \nn
 \eeqa
 one obtains the following representation of the system \eqref{sys1} 
 \beq\label{sys3}
 u^i_t =\pal_x \eta^{ij}\frac{\delta H}{\delta u^j(x)}, \quad i=1, \dots, n.
 \eeq
 This yields, in particular, that
 \eqa\label{sys4}
 &&
 A^i_j(u) =\eta^{ik}\frac{\pal^2 h^{[0]}(u)}{\pal u^k\pal u^j}
 \nn\\
 &&
 B^i_j(u) =-\eta^{ik} \left[ p_{k,j}(u) - p_{j,k}(u)\right], \quad \mbox{where} \quad p_{i,j}(u) := \frac{\pal p_i(u)}{\pal u^j}
\\
&&
C^i_j(u) =- \eta^{ik} q_{kj}(u).
\nn
\eeqa
Let us observe that a nonlinear change of dependent variables
\beq\label{cha1}
\tilde u^i = \tilde u^i (u), \quad i=1, \dots, n
\eeq
brings the Poisson bracket \eqref{sys2} to the form \cite{DN}
\beq\label{sys5}
\{ \tilde u^i(x), \tilde u^j(y)\} =\tilde g^{ij}(\tilde u(x))\delta'(x-y) +\tilde\Gamma^{ij}_k(\tilde u) \tilde u^k_x \delta(x-y)
\eeq 
where the symmetric tensor
$$
\tilde g^{ij}(\tilde u) = \frac{\pal \tilde u^i}{\pal u^k}  \frac{\pal \tilde u^j}{\pal u^l}\eta^{kl}
$$
is a (contravariant) metric of zero curvature (not necessarily positive definite) and
$$
\tilde\Gamma^{ij}_k(\tilde u)=  \frac{\pal\tilde u^i}{\pal u^l}\eta^{lm}\frac{\pal^2 \tilde u^j}{\pal u^m \pal u^k}
$$
is expressed via the Christoffel coefficients of the Levi-Civita connection for the metric
$$
\tilde\Gamma^{ij}_k(\tilde u)=-\tilde g^{is}(\tilde u)\tilde\Gamma^{j}_{sk}(\tilde u).
$$
Any Hamiltonian system with $n$  dependent variables can be 
locally reduced to the standard form \eqref{sys3}, \eqref{ham} by the action of the group of generalized Miura transformations  \cite{dsm}, \cite{getzler}  changing the dependent variables as follows
\eqa\label{miura}
&&
u^i\mapsto \tilde u^i=F^i(u) +\epsilon\, P^i_j(u)u^j_x +\epsilon^2\left[ Q^i_j(u) u^j_{xx} +\frac12 R^i_{jk}(u) u^j_x u^k_x\right] +O\left( \epsilon^3\right)
\nn\\
&&
\\
&&
\det\left( \frac{\pal F^i(u)}{\pal u^j}\right) \neq 0.
\nn
\eeqa

We will now concentrate on the case of a second order Hamiltonian system, $n=2$. It will be assumed that the metric $\left(\eta^{ij}\right)$ in the coordinates $(u,v)$ has the canonical antidiagonal form
\beq\label{anti}
\left( \eta^{ij}\right) = \left(\begin{array}{cc} 0 & 1 \\ 1 & 0\end{array}\right).
\eeq
Thus the Hamiltonian system with a Hamiltonian $H=H[u,v]$ reads
\eqa\label{cano}
&&
u_t=\pal_x\frac{\delta H}{\delta v(x)}
\nn\\
&&
\\
&&
v_t=\pal_x\frac{\delta H}{\delta u(x)}.
\nn
\eeqa
A general perturbation of degree 2 of the Hamiltonian $H_{0}$
takes the form
\beq\label{pert2}
H=H_0 +\epsilon \, H_1 +\epsilon^2 H_2=\int \left[ h+\epsilon \left( p\, u_x + q\, v_x\right) -\frac{\epsilon^2}2 \left( a\, u_x^2 + 2 b\, u_x v_x + c\, v_x^2\right)\right]\, dx
\eeq
where $p=p(u,v)$, $q=q(u,v)$, $a=a(u,v)$, $b=b(u,v)$, $c=c(u,v)$
are some smooth functions. A simple calculation yields the following explicit form of the Hamiltonian flow
\eqa\label{pert3}
&&
u_t =\pal_x \frac{\delta H}{\delta v(x)} = h_{uv}u_x +h_{vv} v_x +\epsilon\left[ \omega\, u_{xx} +\omega_u u_x^2 +\omega_v u_x v_x\right]
\nn\\
&&
+\epsilon^2 \left[b\, u_{xxx} + c\, v_{xxx} +( -a_v + 3 b_u) u_{xx} u_x +(b_v + c_u) u_{xx} v_x + 2 c_u v_{xx} u_x + 2 c_v v_{xx} v_x
\right.
\nn\\
&&\left.
+\left( -\frac12 a_{uv} +b_{uu}\right) u_x^3 +\left( -\frac12 a_{vv} + b_{uv} +c_{uu}\right) u_x^2 v_x +\frac32 c_{uv} u_x v_x^2 +\frac12 c_{vv} v_x^3\right]
\nn\\
&&
\\
&&
v_t=\pal_x \frac{\delta H}{\delta u(x)} = h_{uu} u_x + h_{uv} v_x -\epsilon\left[ \omega\, v_{xx} +\omega_u u_x v_x +\omega_v v_x^2\right]
\nn\\
&&
+\epsilon^2 \left[ a\, u_{xxx} +b\, v_{xxx}+2 a_u u_{xx} u_x + 2 a_v u_{xx} v_x +(a_v + b_u) v_{xx} u_x +(3 b_v -c_u) v_{xx} v_x
\right.
\nn\\
&&\left.
+\frac12 a_{uu} u_x^3 +\frac32 a_{uv} u_x^2 v_x +\left( a_{vv} +b_{uv} -\frac12 c_{uu}\right) u_x v_x^2 +\left( b_{vv}-\frac12 c_{uv}\right) v_x^3\right]
\nn
\eeqa
where
\beq\label{pert4}
\omega=p_v -q_u.
\eeq
The linear terms in $\epsilon$ can be eliminated from equations \eqref{pert3} by a canonical transformation, as it follows from 

\begin{lem} The canonical transformation
\eqa
&&
u(x)\mapsto \tilde u(x)=u(x)+\epsilon\,\{ u(x), F\} +{\mathcal O}(\epsilon^2)
\nn\\
&&
v(x)\mapsto \tilde v(x)=v(x)+\epsilon\,\{ v(x), F\} +{\mathcal O}(\epsilon^2)
\nn
\eeqa
generated by the Hamiltonian
\beq\label{cano1}
F=\int f(u,v)\, dx
\eeq
transforms the Hamiltonian $H$  in (\ref{pert2}) as follows
\eqa
&&
H\mapsto \tilde H =H -\epsilon\, \{ H, F\} +{\mathcal O}(\epsilon^2)
\nn\\
&&
\tilde H =\int \left[ h(\tilde u , \tilde v) +\epsilon (\tilde p\, \tilde u_x +\tilde p\, \tilde v_x)+\dots \right]\, dx
\nn
\eeqa
with
$$
\tilde \omega = \tilde p_{\tilde v} -\tilde q_{\tilde u} = \omega +(h_{uu} f_{vv} -h_{vv} f_{uu}).
$$
\end{lem}
Therefore any  Hamiltonian of the form (\ref{pert2}) can be reduced to the form
\begin{equation}
\label{pert20}
H=H_0 +\epsilon^2 H_2=\int \left[ h -\frac{\epsilon^2}2 \left( a\, u_x^2 + 2 b\, u_x v_x + c\, v_x^2\right)\right]\, dx
\eeq
where the terms of order $\epsilon$ have been eliminated by a canonical transformation.
The system of the form (\ref{pert3}) can then be reduced to the form \eqref{cano100} (see above).

Let us compute the general solution to the leading term of \eqref{cano100} obtained by setting $\epsilon =0$, i.e.
\beq\label{pert0}
\left(\begin{array}{c} u_t \\ v_t\end{array}\right) =\left(\begin{array}{cc}h_{uv} & h_{vv}\\ h_{uu} & h_{uv}\end{array}\right) \left(\begin{array}{c} u_x \\ v_x\end{array}\right).
\eeq
We will consider systems for which the eigenvalues $h_{uv}\pm\sqrt{h_{uu}h_{vv}}$ of the above matrix  are distinct, namely 
\[
h_{uu}h_{vv}\neq 0.
\]
We will deal with smooth initial data only.
A solution $u=u(x,t), \, v=v(x,t)$ is called {\it nondegenerate} on a domain $D\subset \mathbb R^2$ of the $(x,t)$-plane if the Jacobian
\beq\label{jac1}
\det\left( \begin{array}{cc} u_x & u_t \\ v_x & v_t\end{array}\right)=h_{uu} u_x^2 -h_{vv}v_x^2\neq 0
\eeq
does not vanish $\forall \, (x,t)\in D$.
The following version of the classical \emph{hodograph transform} 
will be used for the local description of nondegenerate solutions.

\begin{lem} Let $u=u(x,t), \, v=v(x,t)$ be a solution to \eqref{pert0} nondegenerate on a neighborhood of a point $(x_0, t_0)$. Denote $u_0=u(x_0, t_0)$, $v_0=v(x_0, t_0)$. Then there exists a function $f=f(u,v)$ defined on a neighborhood of the point $(u_0, v_0)$ and satisfying the linear PDE
\beq\label{lineq}
h_{uu} f_{vv} = h_{vv} f_{uu}
\eeq
such that on a sufficiently small neighborhood of this point the 
following two equations hold identically true,
\eqa\label{hodog1}
&&
x+ t\, h_{uv}\left(u(x,t), v(x,t)\right) = f_{uv}\left(u(x,t), v(x,t)\right)
\nn\\
&&
\\
&&
\quad\quad t\, h_{vv}\left(u(x,t), v(x,t)\right) = f_{vv}\left(u(x,t), v(x,t)\right).
\nn
\eeqa
Conversely, given any solution $f=f(u,v)$ to the linear PDE \eqref{lineq} defined on a neighborhood of the point $(u_0, v_0)$, then the functions $u=u(x,t), \, v=v(x,t)$ locally defined by the system 
\eqa\label{hodog11}
&&
x+ t\, h_{uv}\left(u, v\right) = f_{uv}\left(u, v\right)
\nn\\
&&
\\
&&
\quad\quad t\, h_{vv}\left(u, v\right) = f_{vv}\left(u, v\right).
\nn
\eeqa
satisfy \eqref{pert0} provided the assumption
\eqa\label{hodog2}
&&
\det\left(\begin{array}{cc} t\, h_{uuv}-f_{uuv} & t\, h_{uvv}-f_{uvv} \\ t\, h_{uvv}-f_{uvv} & t\, h_{vvv}- f_{vvv}\end{array}\right) 
\\
&&
= \frac1{h_{vv}} \left[ h_{uu} (t\, h_{vvv}- f_{vvv})^2 - h_{vv} ( t\, h_{uvv}- f_{uvv})^2\right]\neq 0
\nn
\eeqa
 of the implicit function theorem holds true at the point $(u_0, v_0)$ such that
\eqa\label{hodog3}
&&
x_0+ t_0 h_{uv}\left(u_0, v_0\right) = f_{uv}\left(u_0, v_0\right)
\nn\\
&&
\\
&&
\quad\quad t_0 h_{vv}\left(u_0, v_0\right) = f_{vv}\left(u_0, v_0\right).
\nn
\eeqa
\end{lem}

\pf For the inverse functions $x=x(u,v)$, $t=t(u,v)$ one obtains from \eqref{pert0}
\eqa\label{sys12}
&&
x_u = ~~~ h_{uu} t_v -h_{uv} t_u
\nn\\
&&
\\
&&
x_v= - h_{uv} t_v +h_{vv} t_u.
\nn
\eeqa
This system can be recast into the form
\eqa
&&
\frac{\pal}{\pal u} \left( x+t\, h_{uv}\right) = \frac{\pal}{\pal v} (t\, h_{uu})
\nn\\
&&
\frac{\pal}{\pal v} \left( x+t\, h_{uv}\right) = \frac{\pal}{\pal u} (t\, h_{vv}).
\nn
\eeqa
Hence there locally exists a pair of functions $\phi=\phi(u,v)$, $\psi=\psi(u,v)$ such that
\eqa
&&
x+t\, h_{uv} = \phi_v, \quad t\, h_{uu}= \phi_u
\nn\\
&&
x+t\, h_{uv} = \psi_u, \quad t\, h_{vv}= \psi_v.
\nn
\eeqa
This implies
$$
\phi_v=\psi_u.
$$
Therefore a function $f=f(u,v)$ locally exists such that
$$
\phi=f_u, \quad \psi=f_v.
$$
Thus
$$
t\, h_{uu}=f_{uu}, \quad t\, h_{vv}=f_{vv}.
$$ 
The linear PDE \eqref{lineq} as well as the implicit function equations \eqref{hodog1} readily follow.  The proof of the converse statement can be obtained by a straightforward computation using the expressions derived with the help of the implicit function theorem
\eqa
&&
u_x=\frac{f_{vvv} -t\, h_{vvv}}{D}, \quad v_x =-\frac{f_{uvv}-t\, h_{uvv}}{D}
\nn\\
&&
u_t=~~~h_{uv}\frac{f_{vvv}-t\, h_{vvv}}{D} -h_{vv}\frac{f_{uvv}-t\, h_{uvv}}{D}
\nn\\
&& 
v_t=-h_{uv} \frac{f_{uvv}-t\, h_{uvv}}{D} +h_{vv}\frac{f_{uuv}-t\, h_{uuv}}{D}
\nn
\eeqa
(here $D$ is the determinant \eqref{hodog2}). \epf

\begin{rem} Observe invariance of the implicit function equations \eqref{hodog1} with respect to transformations of the dependent variables $(u,v)$ preserving the antidiagonal form \eqref{anti} of the metric $\eta$.
\end{rem}

\section{Hyperbolic case}\label{section2}
In this section we study solutions to the system (\ref{cano100}) when the unperturbed systems (\ref{pert0}) is hyperbolic. We will restrict our analysis to smooth  initial data.
We first derive the generic singularity of the solution to  the hyperbolic  systems of the form (\ref{pert0}) and then we  study the local behaviour of the solution of 
the system (\ref{cano100})  with $\epsilon>0$ near such a singularity.  Our first observation is that,  in a suitable system of  dependent and independent coordinates, the system  of equations (\ref{cano100})   decouples in a double scaling limit near the singularity   into two equations: one ODE and one PDEs equivalent to the Korteweg de Vries equation.
We then   argue that the local behaviour of the solution of (\ref{cano100}) near the singularity of the solution to (\ref{pert0}),  in such  a double scaling limit is described by a particular solution to  the P$_{I}^2$  equation.

The system  \eqref{pert0} is hyperbolic if  the eigenvalues of the coefficient matrix
\beq\label{eigen1}
\lambda_\pm = h_{uv} \pm \sqrt{h_{uu} h_{vv}}
\eeq
are real and distinct, i.e.,
\beq\label{hyper1}
h _{uu} h _{vv}>0.
\eeq 
The proof of the following statement is straightforward.

\begin{lem}
The hodograph equations \eqref{hodog11} can be rewritten in the form
\eqa\label{hodog4}
&&
x+\lambda_+(u,v)\, t=\mu_+(u,v)
\nn\\
&&
\\
&&
x+\lambda_-(u,v)\, t=\mu_-(u,v)
\nn
\eeqa
where
\beq\label{hodog5}
\lambda_\pm=h_{uv} \pm \sqrt{h_{uu} h_{vv}}, \quad \mu_\pm =f_{uv}\pm \sqrt{\frac{h_{uu}}{h_{vv}}}\,f_{vv}.
\eeq
\end{lem} 
Denoting by $r_\pm$ the Riemann invariants of the system, we get for 
their differentials 
\beq\label{eigen2}
dr_\pm =\kappa_\pm \left( \pm \sqrt{h_{uu} }\,du + \sqrt{h_{vv}}\,dv \right)
\end{equation}
where $\kappa_\pm =\kappa_\pm(u,v)$ are integrating factors. 
The leading order system \eqref{pert0} becomes diagonal in the coordinates $r_+$, $r_-$, i.e.
\eqa\label{lead}
&&
\pal_t r_+ = \lambda_+ (r) \pal _x r_+
\nn\\
&&
\\
&&
\pal_t r_- = \lambda_- (r) \pal _x r_-.
\nn
\eeqa
It is convenient to write the hodograph equations \eqref{hodog4} in 
terms of the Riemann invariants $r=(r_+, r_-)$
\eqa\label{impl}
&&
x+\lambda_+(r) \, t = \mu_+(r)
\nn\\
&&
\\
&&
x+\lambda_-(r) \, t = \mu_-(r)
\nn
\eeqa
where the functions $\mu_\pm=\mu_\pm(r)$ must satisfy the linear system
\beq\label{hod}
\frac{\pal \mu_+}{\pal r_-} =\frac{\mu_+ -\mu_-}{\lambda_+-\lambda_-}\, \frac{\pal \lambda_+}{\pal r_-} , \quad \frac{\pal \mu_-}{\pal r_+} = \frac{\mu_+ -\mu_-}{\lambda_+-\lambda_-}\,\frac{\pal \lambda_-}{\pal r_+} 
\eeq
equivalent to \eqref{lineq}. The functions $\mu_+(r)$, $\mu_-(r)$ have to be determined from the system \eqref{hod} along with the conditions at $t=0$
\beq\label{init0}
r_+(x):=r_+(u(x,0), v(x,0)),\quad  r_-(x)=:r_-(u(x,0), v(x,0))
\eeq
and 
\[
x=\mu_{\pm}(r_+(x),r_-(x))
\]
for given Cauchy data $u(x,0)$, $v(x,0)$ for the system \eqref{pert0}.
It is easy to see that such a solution is determined uniquely and it is smooth on any interval of monotonicity of both initial Riemann invariants $r_+(x)$, $r_-(x)$
provided the values of the characteristic velocities $\lambda_\pm (r(x)):=\lambda_\pm(r_+(x), r_-(x))$ on the initial curve are distinct
$$
\lambda_+(r(x))\neq \lambda_-(r(x)).
$$

Our first goal is to derive a normal form of the system \eqref{lead} near a point of  gradient catastrophe of the leading term \eqref{lead}.
 The limiting values of the solutions $r_\pm(x,t)$ to \eqref{lead} at the point of gradient catastrophe  $(x_0, t_0)$  will be denoted
$$
r^0_\pm:= r^0_\pm (x_0, t_0).
$$
Let us also introduce the shifted dependent variables denoted as
\beq\label{shift}
\bar r_\pm =r_\pm -r^0_\pm
\eeq
and the notation
$$
\lambda^0_\pm =\lambda_\pm (r^0_+, r^0_-)
$$
etc. for the values of the coefficients and their derivatives at the point of catastrophe.

In the generic situation the $x$-derivative of only one of the Riemann invariants becomes infinite at the point of catastrophe. To be more specific let us assume
\beq\label{cata1}
\begin{array}{l}\pal_x r_-(x,t) \to \infty\\
\pal_x r_+(x,t) \to \mbox{const}\end{array} \quad \mbox{for}\quad x\to x_0, \quad t\to t_0.
\eeq
We say that the point of catastrophe \eqref{cata1} is {\it generic} if
\beq\label{gen1}
\lambda_{-,-}^0:= \frac{\pal \lambda_-(r)}{\pal r_-}|_{r=r^0}\neq 0
\eeq
and, moreover, the graph of the function $r_-(x,t_0)$ has a nondegenerate inflection point at $x=x_0$.

Introduce characteristic variables
\beq\label{char}
x_\pm =x-x_0 +\lambda^0_\pm (t-t_0)
\eeq
at the point of catastrophe. One can represent the functions $r_\pm=r_\pm(x,t)$ as functions of $(x_+, x_-)$. Let us redenote $\bar r_\pm=r_\pm(x_+, x_-)-r_\pm^0$ the resulting transformed functions. It will be convenient to normalize\footnote{Sometimes a different normalization of Riemann invariants is more convenient - see, e.g., \eqref{nls01} below.} the Riemann invariants in such a way that
\beq\label{kappa0}
\kappa_+^0=\kappa_-^0=\left\{ \begin{array}{rl} 1, & h_{uu}^0 ~\mbox{and}~ h_{vv}^0 >0\\
\sqrt{-1}, & h_{uu}^0 ~\mbox{and}~ h_{vv}^0 <0.\end{array}\right.
\eeq

\begin{lem} For a generic solution to the system \eqref{lead} there exist the limits
\eqa\label{lim1}
&&
R_+(X_+, X_-)=\lim_{k\to 0} k^{-2/3}\bar r_+ \left( k^{2/3} X_+,k\,X_-\right)
\nn\\
&&
\\
&& 
R_-(X_+, X_-)=\lim_{k\to 0} k^{-1/3} \bar r_- \left( k^{2/3} X_+,k\,X_-\right)
\nn
\eeqa
for sufficiently small $|X_\pm |$ satisfying
$$
\frac{X_+-X_-}{\lambda_+^0-\lambda_-^0}<0.
$$
The limiting functions satisfy the system
\eqa\label{whit}
&&
X_+ = \alpha\,R_+
\nn\\
&&
\\
&&
X_- = \beta\, X_+ R_- - \frac16\gamma\, R_-^3
\nn
\eeqa
with
\eqa\label{abc}
&&
\alpha=\mu_{+,+}^0 -t_0 \lambda_{+,+}^0
\\
&&
\beta=-\frac{\lambda_{-,-}^0}{\lambda_+^0-\lambda_-^0}=-\frac1{8 \kappa_-^0} \left[ 3 h_{uvv}^0 \sqrt{h_{uu}^0} -3 h_{uvv}^0 \sqrt{h_{vv}^0} +\frac{h_{uuu}^0 h_{vv}^0}{\sqrt{h_{uu}^0}} -\frac{h_{vvv}^0 h_{uu}^0}{\sqrt{h_{vv}^0}}\right]
\nn\\
&&
\gamma=-\mu_{-,---}^0 +t_0 \lambda_{-,---}^0.
\nn
\eeqa
\end{lem}

\pf A generic solution to \eqref{lead} for $t<t_0$ is determined from the implicit function equations \eqref{impl}.
At the point of catastrophe of the Riemann invariant $r_-$ one has
\beq\label{cata}
\mu_{-,-}^0 -t_0\lambda_{-,-}^0=0, \quad \mu_{-,--}^0 -t_0\lambda_{-,--}^0=0
\eeq
where we use notations similar to those in \eqref{gen1}
$$
\mu_{-,-}^0=\left(\frac{\pal \mu_-}{\pal r_-}\right)_{r=r^0}, \quad \lambda_{-,--}^0=\left(\frac{\pal^2 \lambda_-}{\pal r_-^2}\right)_{r=r^0}
$$
etc. 
The point is generic if, along with the condition \eqref{gen1} one also has
\beq\label{gen2}
\mu_{+,+}^0 -t_0\lambda_{+,+}^0\neq 0, \quad \mu_{-,---}^0-t_0\lambda_{-,---}^0\neq 0.
\eeq
Expanding equations \eqref{impl} in Taylor series near the point $(r_+^0, r_-^0)$ and using \eqref{hod} one obtains, after the rescaling
\beq\label{scal}
\begin{array}{cc} x_+ = k^{2/3}X_+ , & x_-= k\, X_-\\
\\
\bar r_+ = k^{2/3} R_+, & \quad \bar r_- = k^{1/3} R_-
\end{array}
\eeq
the relations
\eqa
&&
X_+ = \left( \mu_{+,+}^0 -t_0\lambda_{+,+}^0\right) R_+ + {\mathcal O}\left( k^{1/3}\right)
\nn\\
&&
X_-= -\frac{\lambda_{-,-}^0}{\lambda_+^0-\lambda_-^0} \,X_+ R_- + \frac16 \left( \mu_{-,---}^0-t_0\lambda_{-,---}^0\right) R_-^3+ {\mathcal O}\left( k^{1/3}\right).
\nn
\eeqa\epf

Applying a similar procedure directly to the system \eqref{lead} one obtains the following

\begin{lem}\label{lem15} The limiting functions \eqref{lim1} satisfy the following system of PDEs
\eqa\label{lead1}
&&
\frac{\pal R_+}{\pal X_-}=0
\nn\\
&&
\\
&&
\frac{\pal R_-}{\pal X_+}= -\beta\, R_- \frac{\pal R_-}{\pal X_-} 
\nn
\eeqa
where the constant $\beta$ is defined in \eqref{abc}.
\end{lem}

\pf Using
\eqa\label{part1}
&&
\frac{\pal}{\pal x_+} =~\,\, \frac1{\lambda_+^0-\lambda_-^0} \left[ \frac{\pal}{\pal t} -\lambda_-^0 \frac{\pal}{\pal x}\right]
\nn\\
&&
\\
&&
\frac{\pal}{\pal x_-} =-\frac1{\lambda_+^0-\lambda_-^0} \left[ \frac{\pal}{\pal t} -\lambda_+^0 \frac{\pal}{\pal x}\right]
\nn
\eeqa
we obtain from (\ref{lead})
\eqa
&&
\frac{\pal r_+}{\pal x_-} = - \frac{\lambda_+(r)-\lambda_+^0}{\lambda_+^0-\lambda_-^0} \frac{\pal r_+}{\pal x} = -\frac1{\lambda_+^0-\lambda_-^0}\left[ \lambda_{+,+}^0 \bar r_+ +\lambda_{+,-}^0 \bar r_-+O(|\bar r|^2)\right] \frac{\pal r_+}{\pal x}
\nn\\
&&
\nn\\
&&
\frac{\pal r_-}{\pal x_+} = ~\,\, \frac{\lambda_-(r)-\lambda_-^0}{\lambda_+^0-\lambda_-^0} \frac{\pal r_-}{\pal x} =~\,\,\frac1{\lambda_+^0-\lambda_-^0}\left[ \lambda_{-,+}^0 \bar r_+ +\lambda_{-,-}^0 \bar r_-+O(|\bar r|^2)\right] \frac{\pal r_-}{\pal x}
\nn
\eeqa
Substituting
\eqa
&&
\frac{\pal r_+}{\pal x} = \frac{\pal r_+}{\pal x_+}+\frac{\pal r_+}{\pal x_-}=\quad\quad \frac{\pal R_+}{\pal X_+} +k^{-1/3} \frac{\pal R_+}{\pal X_-}
\nn\\
&&
\frac{\pal r_-}{\pal x} = \frac{\pal r_-}{\pal x_+}+\frac{\pal r_-}{\pal x_-}=k^{-1/3}\frac{\pal R_-}{\pal X_+} +k^{-2/3} \frac{\pal R_-}{\pal X_-}
\nn
\eeqa
in (\ref{lead}),
we obtain, in the limit $k\to 0$, the equation~(\ref{lead1}). \epf

Let us proceed to the study of solutions to the perturbed system (\ref{cano100}).
Choosing the Riemann invariants $r_\pm=r_\pm (u,v)$ of the leading 
term as a system of coordinates on the space of dependent variables, we obtain the system \eqref{cano100} in the form
\eqa\label{pde1}
&&
\pal_t r_+ = \lambda_+ (r) \pal _x r_+  +\epsilon^2 \left[C^+_+(r)  \pal_x^3 r_+ +C^+_-(r) \pal_x^3 r_-+\dots\right] +O\left( \epsilon^3\right)
\nn\\
&&
\\
&&
\pal_t r_- = \lambda_- (r) \pal _x r_-+\epsilon^2 \left[C^-_+(r)  \pal_x^3 r_+ +C^-_-(r) \pal_x^3 r_-+\dots\right] +O\left( \epsilon^3\right)
\nn
\eeqa
with 
\eqa\label{cc}
&&
C_+^+=\frac{a\, h_{vv} + 2 b\, \sqrt{h_{uu} h_{vv}} +c\, h_{uu}}{2\sqrt{h_{uu}h_{vv}}}, \quad C_-^+= \frac{\kappa_+}{\kappa_-}\frac{c\, h_{uu} -a\, h_{vv}}{2\sqrt{h_{uu} h_{vv}}},
\nn\\
&&
\\
&&
C_+^- =\frac{\kappa_-}{\kappa_+} \frac{a\, h_{vv}-c\, h_{uu}}{2\sqrt{h_{uu} h_{vv}}}, \quad C_-^-=-\frac{a\, h_{vv} - 2 b\, \sqrt{h_{uu} h_{vv}} +c\, h_{uu}}{2\sqrt{h_{uu} h_{vv}}}.
\nn
\eeqa

We are now ready to prove the first result of this Section.

\begin{theorem} Let $r_\pm=r_\pm(x,t;\epsilon)$ be a solution to the system \eqref{pde1} defined for $|x-x_0|<\xi$, $0\leq t < \tau$ such that

\noindent $\bullet$ there exists a time $t_0$ satisfying $0<t_0 <\tau$ such that for any $0\leq t<t_0$ and sufficiently small $|x-x_0|$ the limits
$$
{\rm r}_\pm(x,t)=\lim_{\epsilon\to 0} r(x,t; \epsilon)
$$
exist and satisfy  the system \eqref{lead}. 

Let us consider the solution ${\rm r}_\pm(x,t)$  represented in the hodograph form \eqref{impl} and,
assume that

\noindent $\bullet$  it has a gradient catastrophe at the point $(x_0, t_0)$ of the form described in Lemma \ref{lem15};

\noindent $\bullet$ there exist the limits 
\eqa\label{lim2}
&&
R_+(X_+, X_-;\varepsilon)=\lim_{k\to 0} k^{-2/3}\bar r_+ \left( k^{2/3} X_+,k\,X_-; k^{7/6} \varepsilon\right)
\nn\\
&&
\\
&& 
R_-(X_+, X_-;\varepsilon)=\lim_{k\to 0} k^{-1/3} \bar r_- \left( k^{2/3} X_+,k\,X_-; k^{7/6} \varepsilon\right);
\nn
\eeqa

\noindent $\bullet$ the constants $\alpha$, $\beta$, $\gamma$ in \eqref{abc} do not vanish and $\beta\, \gamma >0$;

\noindent $\bullet$ the constant
\beq\label{rho}
\rho = -\frac{C_-^-(r^0)}{2\sqrt{h_{uu}^0 h_{vv}^0}}=\frac{a_0 h_{vv}^0 -2 b_0 \sqrt{h_{uu}^0 h_{vv}^0} +c_0 h_{uu}^0}{4 h_{uu} ^0h_{vv}^0} \neq 0.
\eeq
Then the limiting function $R_-=R_-(X_+, X_-; \varepsilon)$ satisfies the KdV equation
\beq\label{kdv1}
\frac{\pal R_-}{\pal X_+} +\beta\, R_- \frac{\pal R_-}{\pal X_-} +\varepsilon^2 \rho\, \frac{\pal^3 R_-}{\pal X_-^3}=0.
\eeq
The limiting function $R_+=R_+(X_+, X_-; \epsilon)$ is given by the formula
\beq\label{kdv2}
R_+= \alpha^{-1} X_+ -\varepsilon^2 \sigma\, \frac{\pal^2 R_-}{\pal X_-^2}
\eeq
where
\beq\label{sigma}
\sigma=\frac{C_-^+(r_0)}{2\sqrt{h_{uu}^0 h_{vv}^0}}= \frac{c_0 h_{uu}^0-a_0 h_{vv}^0}{4 h_{uu}^0 h_{vv}^0}.
\eeq
\end{theorem}

A solution $r_\pm (x, t; \epsilon)$ to the system \eqref{pde1} with a hyperbolic leading term satisfying the assumption \eqref{gen1} along with
\beq\label{gen00}
\alpha\neq 0, \quad \beta\neq 0, \quad \gamma\neq 0, \quad \beta\, \gamma>0, \quad C_-^-\left(r^0\right)\neq 0
\eeq
will be called \emph{generic}.

\begin{conj} A generic solution to the $\epsilon$-independent Cauchy 
    problem for the generic Hamiltonian perturbation of a hyperbolic system \eqref{pert0}  containing no ${\mathcal O}(\epsilon)$ terms near a generic point of break-up of the second Riemann invariant admits the following asymptotic representation
\eqa\label{conj1}
&&
r_+(x,t,\epsilon)-r_+^0= \epsilon^{4/7}\left[ \alpha^{-1} x_+ -\dfrac{\sigma\nu_+\nu_-}{\beta}U_{XX}\left( \frac{\nu_- x_-}{\epsilon^{6/7}}, \frac{\nu_+ x_+}{\epsilon^{4/7}}\right)\right] +{\mathcal O}\left( \epsilon^{6/7}\right)
\nn\\
&&
r_-(x,t,\epsilon)-r_-^0= \dfrac{\nu_+  \epsilon^{2/7}}{\beta\nu_-} U\left( \frac{\nu_- x_-}{\epsilon^{6/7}}, \frac{\nu_+ x_+}{\epsilon^{4/7}}\right) +{\mathcal O}\left( \epsilon^{4/7}\right)
\nn\\
&&
\nn\\
&&
x_\pm =(x-x_0) +\lambda_\pm^0 (t-t_0)
\\
&&
\nn\\
&&
\nu_-= \left(\frac{\beta^3 }{12^3 \rho^3\gamma}\right)^{1/7}, \quad \nu_+ =\left(\frac{\beta^9 }{12^2 \rho^2\gamma^3}\right)^{1/7}
\nn
\eeqa
with $\alpha$, $\beta$, $\gamma$ and $\rho$ defined in (\ref{abc}) and (\ref{rho}) respectively and 
where $U=U(X,T)$ is the  smooth solution to the P$_{I}^2$ equation
\beq\label{p12}
X=U\, T -\left[ \frac16 U^3 +\frac1{24} \left( U_X^2 +2 U\, U_{XX}\right) +\frac1{240} U_{XXXX}\right]
\eeq
uniquely determined by  the  asymptotic behavior 
     \begin{equation}
\label{PI2asym}
        U(X,T)=\mp (6|X|)^{1/3}\mp \frac{1}{3}6^{2/3}T|X|^{-1/3}
            +O(|X|^{-1}),
            \qquad\mbox{as $X\to\pm\infty$,}
\end{equation}
for each fixed $T\in\mathbb{R}$. 
\end{conj}
The existence of such solution to the P$_{I}^2$ equation
 has been conjectured in \cite{du2} (for $T=0$ such a conjecture gas already been formulated in \cite{bmp}) and proved in \cite{cv2}. See Figure~\ref{figPI2}  below for a plot of such solution in the $(x,t)$ plane.
 \begin{figure}[htb!]
    \includegraphics[width=0.7\textwidth]{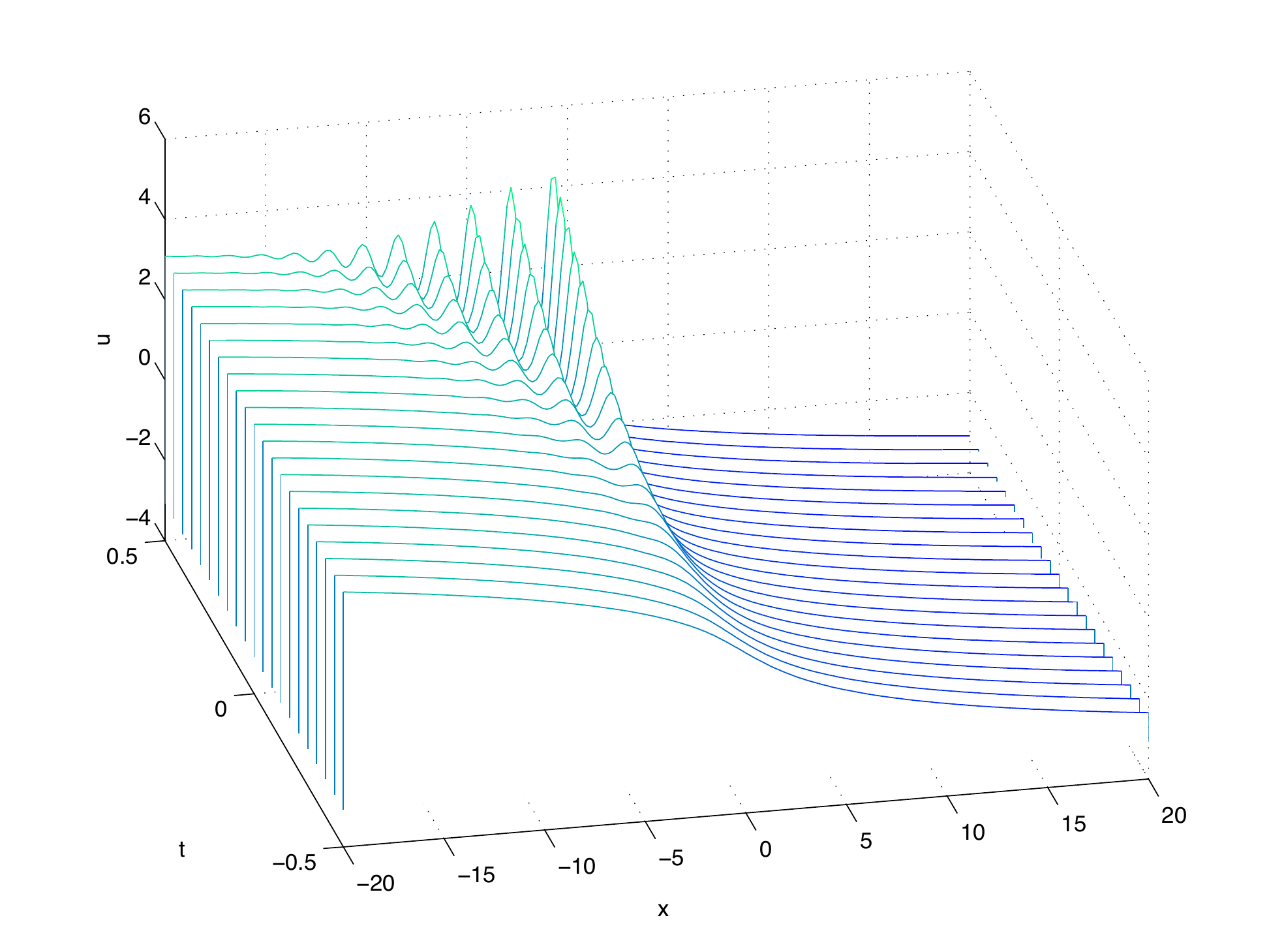}
 \caption{The special solution to the P$_{I}^{2}$ equation for 
 several values of $t$.}
 \label{figPI2}
\end{figure}

Recalling that the function $U(X,T)$ satisfies the KdV equation
\beq\label{kdv12}
U_T +U\, U_X +\frac1{12} U_{XXX}=0,
\eeq
the asymptotic formulae \eqref{conj1} meet the following two conditions:

\noindent $\bullet$ for $t< t_0$ the solution \eqref{conj1} tends to the hodograph solution \eqref{whit} as $\epsilon \to 0$;

\noindent $\bullet$ near the point of break-up the rescaled Riemann invariant $r_-$ approximately satisfies the KdV equation \eqref{kdv1} while the rescaled Riemann invariant $r_+$ admits an approximate representation \eqref{kdv2}. Indeed, choosing
$$
k=\epsilon^{6/7}
$$
one obtains
$$
\varepsilon=1.
$$
So, after rescaling of the characteristic variables
$$
X_+=\left(\frac{12^2 \rho^2 \gamma^3}{\beta^9}\right)^{1/7} \hat X+, \quad X_-=\left( \frac{12^3 \rho^3 \gamma}{\beta^3}\right)^{1/7}\hat X_-
$$
one derives from \eqref{kdv1} that the rescaled function
$$
\hat R_-= \left(\frac{\beta\,\gamma^2}{12 \rho}\right)^{1/7} R_-
$$
satisfies the normalized KdV equation \eqref{kdv12},
$$
\frac{\pal \hat R_-}{\pal \hat X_+}+ \hat R_- \frac{\pal \hat R_-}{\pal \hat X_-} +\frac1{12} \frac{\pal^3\hat R_-}{\pal \hat X_-^3}=0.
$$
Moreover, for large $\hat X_-$ and negative $\hat X_+$ it behaves like the root of the cubic equation
$$
\hat X_-=\hat X_+ \hat R_- -\frac16 \hat R_-^3.
$$
The function
$$
\hat R_-= U\left( \hat X_-, \hat X_+\right)
$$
is a solution to KdV satisfying these properties. Returning to the original variables $\bar r_-$, $\bar x_\pm$ one arrives at the formula \eqref{conj1}.

\noindent

\medskip

\section{Elliptic case}\label{section3}
In this section we study solutions to the system (\ref{cano100}) when the unperturbed systems (\ref{pert0}) is elliptic. We will restrict our analysis to analytic initial data.
We first derive the generic singularity of the solution to  the elliptic systems of the form (\ref{pert0}) and then we  study the local behaviour of the solution of the system (\ref{cano100})  with $\epsilon>0 $
near such a singularity. We  argue that such behaviour in a double scaling limit is described by the tritronqu\'ee solution to the Painlev\'e I equation.

Let us now proceed to considering the elliptic case for the system (\ref{pert0}), namely
\beq\label{ell1}
h_{uu} h_{vv}<0.
\eeq
The initial data $u(x,0)$ and $v(x,0)$ are analytic functions.
The Riemann invariants
\beq\label{rme}
dr_\pm =\kappa_\pm \left( \sqrt{|h_{vv}|}\, dv \pm i \sqrt{|h_{uu}|}\, du\right), \quad \kappa_-=\kappa_+^*
\eeq
and the characteristic speeds
\beq\label{rme1}
\lambda_\pm=h_{uv}\pm i\, \sign (h_{vv}) \sqrt{|h_{uu} h_{vv}|}.
\eeq 
are complex conjugate (the asterisk will be used for the complex conjugation),
$$
r_- =r_+^* \qquad \lambda_-=\lambda_+^*
$$
At the point of elliptic break-up of a solution, written in the form \eqref{impl}, the following two complex conjugated equations hold
\eqa\label{ell2}
&&
\mu_{+,+}^0 = \lambda_{+,+}^0 t_0
\nn\\
&&
\\
&&
\mu_{-,-}^0 = \lambda_{-,-}^0 t_0.
\nn
\eeqa
The characteristic variables at the point of catastrophe are defined as
\beq\label{ell3}
x_\pm =(x-x_0) +\lambda_\pm^0 (t-t_0)
\eeq
and are also complex conjugate.
One can represent the functions $r_\pm=r_\pm(x,t)$ as functions of $(x_+, x_-)$. Let us redenote $\bar r_\pm=r_\pm(x_+, x_-)-r_\pm^0$ the resulting shifted and transformed  Riemann invariants.

\begin{lem} For a generic solution to the system \eqref{lead} near a point of elliptic break-up the limits 
\beq\label{ell4}
R_\pm (X_\pm) =\lim_{k\to 0} k^{-1/2} \bar r_\pm (k\, X_+, k\, X_-)
\eeq
exist and satisfy the quadratic equation
\beq\label{ell5}
X_\pm = \frac12 a_\pm R_\pm^2
\eeq
with
\beq\label{ell6}
x_{\pm} = k X_{\pm}, \qquad a_\pm = \mu_{\pm, \pm\pm}^0 - t_0 \lambda_{\pm, \pm\pm}^0.
\eeq
\end{lem}

In the sequel it will be assumed that
\beq\label{ell61}
a_\pm \neq 0
\eeq
(this condition will be added to the genericity assumptions). 

\pf  Differentiating the hodograph relations \eqref{impl} one obtains
$$
\mu_{+,-} - t\, \lambda_{+, -}\equiv 0, \quad \mu_{-,+} - t\, \lambda_{-, +}\equiv 0.
$$
Moreover, differentiating \eqref{hod} one finds that
\eqa
&&
\mu_{+, + -} -t\, \lambda_{+, + -} =\lambda_{+, -} \frac{\mu_{+,+}-t\, \lambda_{+, +}}{\lambda_+-\lambda_-}
\nn\\
&&
\mu_{-, + -} -t\, \lambda_{-, + -} =-\lambda_{-, +} \frac{\mu_{-,-}-t\, \lambda_{-, -}}{\lambda_+-\lambda_-}
\nn\\
&&
\mu_{+,--} -t\, \lambda_{+,--}=-\lambda_{+,-}\frac{\mu_{-,-}-t\, \lambda_{-,-}}{\lambda_+ -\lambda_-}
\nn\\
&&
\mu_{-,++} -t\, \lambda_{-,++}=\lambda_{-,+}\frac{\mu_{+,+}-t\, \lambda_{+,+}}{\lambda_+ -\lambda_-}.
\nn
\eeqa
Hence, due to \eqref{ell2}, all these combinations of the second 
derivatives vanish at the break-up point. Expanding the hodograph equations \eqref{impl} in Taylor series near the point of catastrophe, one easily arrives at \eqref{ell5}. \epf

Choosing Riemann invariants $r_\pm=r_\pm (u,v)$ of the leading term as a system of coordinates on the space of dependent variables, and $x_{\pm}$ as independent variables,
 the system \eqref{cano100} takes the form

\eqa\label{pde1e}
&&
\pal_t r_+ = \lambda_+ (r) \pal _x r_+  +\epsilon^2 \left[C^+_+(r)  \pal_x^3 r_+ +C^+_-(r) \pal_x^3 r_-+\dots\right] +O\left( \epsilon^3\right)
\nn\\
&&
\\
&&
\pal_t r_- = \lambda_- (r) \pal _x r_-+\epsilon^2 \left[C^-_+(r)  \pal_x^3 r_+ +C^-_-(r) \pal_x^3 r_-+\dots\right] +O\left( \epsilon^3\right)
\nn
\eeqa
with 
\eqa\label{cce}
&&
C_+^+=\frac{a\, h_{vv} + 2i b\, \sqrt{|h_{uu} h_{vv}|} +c\, h_{uu}}{2i\sqrt{|h_{uu}h_{vv}}|}, \quad C_-^+= \frac{\kappa_+}{\kappa_-}\frac{c\, h_{uu} -a\, h_{vv}}{2i\sqrt{|h_{uu} h_{vv}}|}
\nn\\
&&
\\
&&
C_+^- =\frac{\kappa_-}{\kappa_+} \frac{a\, h_{vv}-c\, h_{uu}}{2i\sqrt{|h_{uu} h_{vv}}|}, \quad C_-^-=-\frac{a\, h_{vv} - 2i b\, \sqrt{|h_{uu} h_{vv}|} +c\, h_{uu}}{2i\sqrt{|h_{uu} h_{vv}}|}.
\nn
\eeqa

As above we will denote $\bar r_\pm =\bar r_\pm ( x_+,  x_-; \epsilon)$ a shifted generic solution to the system \eqref{pde1e} with $\epsilon$-independent initial data written as functions of the complex conjugated linearized characteristic variables \eqref{ell3}. Like above we will be interested in the multiscale expansion of these complex conjugated functions
\eqa\label{musc}
&&
\bar r_\pm (\bar x_+, \bar x_-; \epsilon) =k^{1/2} R_\pm \left( X_+, X_-; \varepsilon\right) +k\, \Delta R_\pm \left( X_+, X_-; \varepsilon\right)+{\mathcal O}\left( k^{3/2}\right)
\\
&&
x_\pm = k\, X_\pm, \quad \epsilon =k^{5/4}\varepsilon, \quad k\to 0.
\nonumber
\eeqa
We will now show that existence of such expansions implies that the leading term  is a holomorphic/antiholomorphic function
$$
\dfrac{\partial R_{\pm}}{\partial X_{\mp}}=0
$$
satisfying an ODE.

\begin{theorem}\label{theoexpansion}
Let $r_{\pm}(x,t,;\epsilon)$ be a solution of the system (\ref{pde1e}) such that there exist the limits 
\begin{equation}
\begin{split}
\label{R}
R_{\pm}(X_+,X_-;\varepsilon)&=\lim_{k\rightarrow 0}k^{-\frac{1}{2}} \bar{r}_{\pm}(kX_+,kX_-;k^{\frac{5}{4}}\varepsilon)\\
\Delta R_{\pm}(X_{+}, X_-;\varepsilon)&=\lim_{k\rightarrow 0}\dfrac{ \bar{r}_{\pm}(kX_+,kX_-;k^{\frac{5}{4}}\varepsilon)-k^{\frac{1}{2}}R_{\pm}(X_+, X_-;\varepsilon)}{k}
\end{split}
\end{equation}
Then the function $R_+=R_{+}(X_{+}, X_-;\varepsilon)$  satisfies the Cauchy--Riemann equation 
\beq\label{aa}
\dfrac{\partial R_{+}(X_{+}, X_-;\varepsilon)}{\partial X_-}=0
\eeq  and also the equation
\begin{equation}
\label{eqR1}
\lambda_{+,+}^0R_+\dfrac{\partial R_+}{\partial X_+}+\varepsilon^2C_+^+(r^0)\dfrac{\partial^3 R_+}{\partial X^3_+}=c_+
\end{equation}
where $c_+$ is a holomorphic function of $X_+$ such that
\begin{equation}
\label{cp}
c_{+}(X_+)=\dfrac{\lambda^0_{+,+}}{a_+}+O(1/X_+^{\delta}),\;\;\mbox{\rm as} \quad X_+\rightarrow \infty\quad \mbox{\rm and}\quad  
\delta>0.
\end{equation} 
Here $C^+_+$ has been defined in (\ref{cce}). The function $R_{-}=R_-( X_-;\varepsilon)$ is antiholomorphic and satisfies the complex conjugate of \eqref{eqR1}.
The function $\Delta R_{+}(X_{+}, X_-;\varepsilon)$ satisfies the equation
\begin{equation}\label{eqR22}
(\lambda_-^0-\lambda_+^0)\dfrac{\partial }{\partial X_-}\Delta R_+=\lambda_{+,-}^0R_-\dfrac{\partial R_+}{\partial X_+}+\varepsilon^2C_-^+(r^0)\dfrac{\partial^3 R_-}{\partial X_-^3}+c_+,
\end{equation}
where $C^+_-$ has been defined in (\ref{cce}). The function $\Delta R_{-}(X_{+}, X_-;\varepsilon)$ satisfies the complex conjugate of the above equation.
\end{theorem}
\begin{proof}
In order to prove the theorem it is sufficient to plug the expansion \eqref{musc} and (\ref{R}) into equation (\ref{pde1e}) giving the following expansions
\begin{equation}
\label{Mexpansion}
\begin{split}
&k^{-1/2} \left( \lambda_+^0-\lambda_-^0\right) \frac{\pal R_+}{\pal X_-} +\left( \lambda_+^0-\lambda_-^0\right) \frac{\pal \Delta R_+}{\pal X_-} +\left( \lambda_{+,+}^0 R_+ +\lambda_{+,-}^0 R_-\right)\left( \frac{\pal}{\pal X_+} +\frac{\pal}{\pal X_-}\right) R_+
\\
&
 +\varepsilon^2 C_+^+(r^0) \left( \frac{\pal}{\pal X_+} +\frac{\pal}{\pal X_-}\right)^3 R_+ +\varepsilon^2 C_+^-(r^0) \left( \frac{\pal}{\pal X_+} +\frac{\pal}{\pal X_-}\right)^3 R_-= {\mathcal O}\left( k^{1/2}\right)
\\
&
k^{-1/2} \left( \lambda_-^0-\lambda_+^0\right) \frac{\pal R_-}{\pal X_+} +\left( \lambda_-^0-\lambda_+^0\right) \frac{\pal \Delta R_-}{\pal X_+} +\left( \lambda_{-,+}^0 R_+ +\lambda_{-,-}^0 R_-\right)\left( \frac{\pal}{\pal X_+} +\frac{\pal}{\pal X_-}\right) R_-
\\
&
 +\varepsilon^2 C_-^+(r^0) \left( \frac{\pal}{\pal X_+} +\frac{\pal}{\pal X_-}\right)^3 R_++\varepsilon^2 C_-^-(r^0) \left( \frac{\pal}{\pal X_+} +\frac{\pal}{\pal X_-}\right)^3 R_- = {\mathcal O}\left( k^{1/2}\right)
\end{split}
\end{equation}
Since $\lambda_+^0\neq \lambda_-^0$, from the leading term it readily follows that
\begin{equation}
\label{holo}
\frac{\pal R_+}{\pal X_-}=0, \quad \frac{\pal R_-}{\pal X_+}=0.
\end{equation}
Separating holomorphic and antiholomorphic parts in the terms of order ${\mathcal O}(1)$ one arrives at equations 
\eqref{eqR1}, \eqref{eqR22} and their complex conjugates. 

Equation (\ref{eqR1}) must have a solution with asymptotic behaviour determined by (\ref{ell5}), namely 
\begin{equation}
\label{Rboundary}
R_+(X_+)\rightarrow\pm\sqrt{\dfrac{2X_+}{a_+}}, \quad \mbox{as}   \;\; X_+\rightarrow\infty.
\end{equation}
 This immediately gives  that $c_+$ is an analytic function of $X_+$ 
 with asymptotic behavior at infinity    
\begin{equation}
\label{cpp}
c_{+}(X_+)=\dfrac{\lambda^0_{+,+}}{a_+}+O(1/X_+^{\delta}),\;\;\delta>0\quad c_-(X_-)=\bar{c}_+(\bar{X}_+).
\end{equation} 
\end{proof}

Assuming $c_+={\rm const}$ we arrive at an ODE for the function $R_+=R_+(X_+)$ equivalent to the Painlev\'e-I equation,
\beq\label{ur1}
\varepsilon^2 C_+^+(r_0) R_+'' +\frac12 \lambda_{+,+}^0 R_+^2  =\frac{\lambda_{+,+}^0}{a_+},
\eeq
with asymptotic behaviour (\ref{Rboundary}).
The complex conjugate of the above equation gives the  corresponding   Painlev\'e-I equation for  $R_-=R_-(X_-)$.
 If we linearized  the  increments of the Riemann invariants   we obtain
 \beq\label{sdvigelliptic}
r_\pm-r^0_\pm =\kappa^0_\pm \left(   \sqrt{|h^0_{vv}|}\,(v-v_0)\pm i \sqrt{|h^0_{uu}| }\,(u-u_0) \right)+O(\epsilon^{\frac{4}{5}}).
\end{equation}
For simplicity we normalize the constant $\kappa^0_{\pm}$ to
\beq\label{kappa01}
\kappa_+^0=\kappa_-^0=\left\{ \begin{array}{rl} 1, & h_{uu}^0<0 ~\mbox{and}~ h_{vv}^0 >0\\
\sqrt{-1}, & h_{uu}^0>0 ~\mbox{and}~ h_{vv}^0 <0.\end{array}\right.
\eeq
From (\ref{ur1}) and (\ref{sdvigelliptic}) we arrive at the following
\begin{conj}
\label{conjelliptic}
The functions  $u(x,t,\epsilon)$ and $v(x,t,\epsilon)$  that solves  the system (\ref{cano100}) 
admits the  following asymptotic representation in the double scaling limit  $x\rightarrow x_0$, $t\rightarrow t_0$ and $\epsilon\rightarrow 0$  in such a way that 
\begin{equation}
\dfrac{x-x_0+\lambda^0_{\pm}(t-t_0)}{\epsilon^{\frac{4}{5}}}
\end{equation}
remains bounded
\begin{equation}
\label{rasymp}
    \sqrt{|h^0_{vv}|}\,(v(x,t,\epsilon)-v_0)+i\sqrt{|h^0_{uu}| }\,(u(x,t,\epsilon)-u_0) =-12 \left(\dfrac{\epsilon^2C_+^+(u_0,v_0)}{(12a_+)^2\lambda_{+,+}^0}\right)^{\frac{1}{5}}\Omega(\xi)+O(\epsilon^{\frac{4}{5}}),
\end{equation}
where
\begin{equation}
\label{xi}
\xi=\left(\dfrac{(\lambda_{+,+}^0)^2}{12a_+\epsilon^4(C_+^+(u_0,v_0))^2}\right)
^{\frac{1}{5}}(x-x_0+\lambda_+^0(t-t_0))
\end{equation}
where $a_+$ and  $C_+^+(u_0, v_0)$  have been defined in (\ref{ell6}) and (\ref{cce}) respectively
and $\Omega=\Omega(\xi)$ is the tritronqu\'ee solution to the Painlev\'e-I equation
\begin{equation}
\label{PItheo}
\Omega_{\xi\xi}=6\Omega^2-\xi,
\end{equation}
determined uniquely by the asymptotic conditions\footnote{Note that 
there are additional tritronqu\'ee solutions $\Omega_{n}$, 
$n=\pm1,\pm2$, related to $\Omega$ via $\Omega_{n}(\xi)=e^{\frac{4\pi 
i}{5}}\Omega(e^{\frac{2\pi i}{5}}\xi)$.}
\begin{equation}
\label{PIasymp}
\Omega(\xi)\simeq -\sqrt{\dfrac{\xi}{6}},\quad |\xi|\rightarrow \infty,\quad |\arg\xi|<\dfrac{4}{5}\pi.
\end{equation}
\end{conj}
 The   smoothness of the solution of (\ref{PItheo}) with asymptotic condition (\ref{PIasymp})  in a sector of 
the complex $z$-plane   of angle $|\arg z|<4\pi/5$ conjectured in \cite{DGK} has only 
recently been  proved in \cite{co1}. For a plot of such solution in the complex plane see \cite{DGK}.
\begin{rem}
Observe that the tritronqu\'ee solution to the Painlev\'e-I equation is invariant with respect to complex conjugation 
\beq\label{cpx}
\overline{\Omega\left(\bar\xi\right)}=\Omega(\xi).
\eeq
So the asymptotic representation of the  linearised Riemann invariant 
$  \sqrt{|h^0_{vv}|}\,(v(x,t,\epsilon)-v_0)-i\sqrt{|h^0_{uu}| }\,(u(x,t,\epsilon)-u_0)$  is given by the complex conjugate of \eqref{rasymp}.
\end{rem}
\begin{rem}
We write  the constant $a_+$ in the form
\[
\dfrac{1}{a_+}=i\left(\dfrac{C^+_+}{\lambda_{+,+}^0}\right)^2 qe^{i\psi},
\]
with $q>0$  and $\psi\in[-\pi,\pi]$. One  can check that when $\psi=0$ and 
$t=t_0$ the quantity  $\xi$ defined in
 (\ref{xi})  has to be purely imaginary, and this gives a rule for the selection of the fifth root, namely 
\[
\xi=i\left(\dfrac{qe^{i\psi}}{12\epsilon^4}\right)^{\frac{1}{5}}(x-x_0+\lambda_+^0(t-t_0)).
\]
Note that the angle of the line $\xi=\xi(x-x_0)$ for fixed  $t$,  is equal to 
$\dfrac{\pi}{2}+\dfrac{\psi}{5}$ with 
 $\psi\in[-\pi,\pi]$, thus the maximal value of  $\arg \xi$ is equal to 
$\dfrac{7}{10}\pi<\dfrac{4}{5}\pi$.
\end{rem}

In the next subsection we consider an alternative derivation of the P$_I$ equation for a subclass of Hamiltonian PDEs having the structure of a generalized nonlinear Schr\"odinger equation.
\subsection{Painlev\'e-I equation and almost integrable PDEs}
In this subsection we  give a different derivation of the conjecture~\ref{conjelliptic} for Hamiltonian equations  (\ref{pert0}) with  Hamiltonian $H_0=\int h(u,v)\, dx$ satisfying equation
\beq\label{nwave}
h_{uu}=P''(u) h_{vv}
\end{equation}
for some smooth function $P(u)$, with $P''(u)<0$.  We will see below that, in particular, a very general family of nonlinear Schr\"odinger equations belongs to this subclass. 
The condition
\begin{equation}
\label{P2}
\dfrac{h_{uu}}{h_{vv}}=P''(u)<0
\end{equation}
guarantees that the unperturbed quasi-linear system is elliptic.

A local solution of  the system (\ref{cano00})  with $h(u,v)$ satisfying (\ref{nwave})  for  given analytic  initial data $u_0(x)$, $v_0(x)$ takes the form
\begin{equation}
\begin{split}
x+&th_{uv}=f_u\\
&th_{vv}=f_v
\end{split}
\end{equation}
where the function $f=f(u,v)$ satisfies equation
\begin{equation}
\label{ellipticlinear}
f_{uu}=P''(u) f_{vv}
\end{equation}
and the condition
\begin{equation}
x=f_u(u_0(x),v_0(x)), \quad 0=f_v(u_0(x),u_0(x)).
\end{equation}
The equation for determining the  
point of elliptic umbilic catastrophe characterized by equation (\ref{ell2})
in the variables $u$ and $v$ takes the form
\begin{equation}
\begin{split}
h_{uvv}t_0-f_{uv}^0&=0\\
h^0_{vvv}t_0-f_{vv}^0&=0
\end{split}
\end{equation}
and the constants $a_{\pm}$ takes the form
\beq\label{ell6b}
a_\pm = f^0_{uvv}-t_0h_{uvvvv}^0\pm i\sqrt{|P''(u_0)|}(f_{vvv}^0-t_0h_{vvvv}^0).
\eeq

To study the critical behaviour of solutions of (\ref{pert3}) 
 we first restrict ourselves to the almost integrable cases and recall some results in \cite{dub3}. We describe Hamiltonian integrable perturbations of equation (\ref{pert0}) up to terms of
 order $O(\epsilon^3)$. 
 
 \begin{theorem}\cite{dub3}
 Any Hamiltonian perturbation integrable up to order $\epsilon^3$  of the system of equations
(\ref{pert0}) satisfying (\ref{nwave}) is given by equations
 \begin{equation}
\label{system1pert}
\begin{split}
u_t& =\pal_x \frac{\delta H_h}{\delta v(x)}  \\
v_t&=\pal_x \frac{\delta H_h}{\delta u(x)} 
\end{split}
\end{equation}
 with Hamiltonian $H_h=\int Dh\, dx$ and  Hamiltonian density  $Dh$ given by
 \begin{equation}
\label{DHamiltonian}
 \begin{split}
Dh& = h -\frac{\epsilon^2}2 \left\{ \left[P''( \rho_u h_{vvv} + \rho_v h_{uvv}) +\frac12 P''' \rho_v h_{vv}\right]\, u_x^2
\right.\\
&
+2\left( P''\rho_v h_{vvv} +\rho_u h_{uvv} +\frac{P'''}{4 P''} \rho_u h_{vv}\right)\, u_x v_x\\
&
\left.
+\left( \rho_u h_{vvv} +\rho_v h_{uvv}\right) \, v_x^2+ s_3 \,  \left( v_x^2 -P'' u_x^2\right)\, h_{vv}\right\}+{\mathcal O}\left( \epsilon^3\right),
\end{split}
\end{equation}
where   the function $\rho=\rho(u,v)$ satisfies the linear PDE
\begin{equation}
\label{rhoequation}
\rho_{uu} -P'' \rho_{vv} =\frac{P'''}{2P''} \,\rho_u
\end{equation}
and $s_3=s_3(u,v)$ is an arbitrary function.
For any function $f=f(u,v)$ that satisfies (\ref{ellipticlinear}) the corresponding Hamiltonian $H_f$  given  by  an equivalent expression to (\ref{DHamiltonian}),
Poisson commute with $H_h$ up to $\epsilon^3$, namely
\[
\{H_h,\,H_f\}=O(\epsilon^3).
\]
 \end{theorem}
 Furthermore, a class of solutions of the system (\ref{system1pert}) characterized by 
 an analogue of the \emph{string equation} is given by the following theorem.
 \begin{theorem} \cite{dub3} The solutions to the 
{\rm string equation}
\eqa\label{string}
&&
x+ t \,\frac{\delta H_{h'}}{\delta u(x)} =\frac{\delta H_{f}}{\delta u(x)}
\nn\\
&&
\\
&&
 \qquad t \,\frac{\delta H_{h'}}{\delta v(x)} =\frac{\delta H_{f}}{\delta v(x)}
\nn
\eeqa
also solve the Hamiltonian equations
\eqa\label{fl1}
&&
 u_t =\pal_x \frac{\delta H_{h}}{\delta v(x)} 
 \nn\\
 &&
 \\
 &&
 v_t =\pal_x \frac{\delta H_{h}}{\delta u(x)} 
\nn
\eeqa
 where $f=f(u,v)$ is another solution to $f_{uu}=P''(u) f_{vv}$, and 
 $$
 h':= \frac{\pal h}{\pal v}.
 $$
\end{theorem}
We remark that (\ref{string}) is  a system of couple ODEs for $u$ and $v$ having $t$ has a parameter.

We can apply to the system (\ref{string}) the rescaling (\ref{musc}). Let us first introduce the Riemann invariants for the  Hamiltonians $H_0$ satisfying (\ref{nwave})
\[
r_{\pm}=v\pm i Q(u),\quad Q'(u)=\sqrt{P''(u)}.
\]
Choosing the Riemann invariants $r_{\pm}=r_{\pm}(u,v)$ as a systems of coordinates on the space of dependent variables
one can write the string equation (\ref{string}) in the form
\begin{equation}
\label{stringriemann}
\begin{split}
x+\lambda_+t&=\mu_++\epsilon^2\left(\tilde{C}_+^+\dfrac{\partial^2}{\partial x^2}r_+
+\tilde{C}_-^+\dfrac{\partial^2}{\partial x^2}r_-+\dots\right)\\
x+\lambda_-t&=\mu_-+\epsilon^2\left(\tilde{C}_+^-\dfrac{\partial^2}{\partial x^2}r_++
\tilde{C}_-^-\dfrac{\partial^2}{\partial x^2}r_-+\dots\right)
\end{split}
\end{equation}
where the coefficients $\tilde{C}_{\pm}^{\pm}$ are as in (\ref{cce}) with  $a=a(u,v)$, $b=b(u,v)$ and $c=c(u,v)$  obtained by comparing 
the Hamiltonian $H_f-tH_{v'}$ to the general form 
(\ref{pert20}).
 \begin{prop}\label{stringPainleve}
 The string equation (\ref{string}) in the scaling (\ref{musc})  reduces to the Painlev\'e I equation 
\begin{equation}
\label{PI}
\begin{split}
&X_{+}=\dfrac{1}{2}a_{+}R^2_{+}+\epsilon^2a_{+}\dfrac{C_+^+(u_0,v_0)}{\lambda_{+,+}^0}\dfrac{\partial ^2}{\partial X_{+}^2} R_{+}\\
&X_{-}=\dfrac{1}{2}a_{-}R^2_{-}+\epsilon^2a_{-}\dfrac{C_-^-(u_0,v_0)}{\lambda_{-,-}^0}\dfrac{\partial ^2}{\partial X_{-}^2} R_{-}
\end{split}
\end{equation}
where  $C_+^+$  and $C^-_-$ have been defined in (\ref{cce})
with  $a=a(u,v)$, $b=b(u,v)$ and $c=c(u,v)$  obtained by comparing 
the Hamiltonian (\ref{DHamiltonian}) to the general form 
(\ref{pert20}),
 $a_{\pm}$ as in (\ref{ell6}), $\lambda_{\pm,\pm}^0=\dfrac{\partial}{\partial r_{\pm}}\lambda_{\pm}|_{r_{\pm}=r_{\pm}^0}$.
\end{prop}
 \pf
Using  the Riemann invariants as a system of dependent coordinates    the string equation (\ref{string}) takes the form  (\ref{stringriemann}).
Changing the independent coordinates $(x,t)$ to $(x_+,x_-)$ defined in (\ref{ell3}) and  performing  the scalings (\ref{musc}) one obtains  for $k\rightarrow 0$
\begin{equation}
\label{PI00}
X_{\pm}=\dfrac{1}{2}a_{\pm}R^2_{\pm}+\epsilon^2a_{\pm}(\rho_u^{0}\pm i\sqrt{|P_0''|}\rho_v^0)
\left(\dfrac{\partial}{\partial X_+}+\dfrac{\partial}{\partial X_-}\right)^2 R_{\pm},
\end{equation}
where $P_0=P(u_0,v_0)$.
Requiring the compatibility of the leading order expansion of the string equation 
with the leading order expansion of the system (\ref{pde1e}), we get that (\ref{holo}) has to be compatible with (\ref{PI00}), namely  
\begin{equation}
\label{PI0}
X_{\pm}=\dfrac{1}{2}a_{\pm}R^2_{\pm}+\epsilon^2a_{\pm}(\rho_u^0\pm i\sqrt{|P_0''|}\rho_v^0)\dfrac{\partial ^2}{\partial X_{\pm}^2} R_{\pm}
\end{equation}
which is equivalent to the Painlev\'e-I equation.
We observe that  the quantity $\rho_u^0+ i\sqrt{|P_0''|}\rho_v^0$ can be rewritten in the form
\begin{align}
&\rho_u^0+ i\sqrt{|P_0''|}\rho_v^0=\dfrac{C_+^+(u_0,v_0)}{\lambda_{+,+}^0}\\
&\lambda_{+,+}^0=h_{uvv}+ih_{vvv}\sqrt{|P''_0|}+\dfrac{P'''_0}{4P''_0}h_{vv}
\end{align}
with $C_+^+$ as in (\ref{cce}). In a similar way one can write the complex conjugate.
Therefore equations (\ref{PI0}) can be written in the form (\ref{PI}).
\epf
We finish this subsection by observing  that  for a subclass of Hamiltonian PDEs of the form 
(\ref{cano100}) with $h_{uu}=P''(u)h_{vv}$,  one can find solutions to 
quasi-integrable and non-integrable perturbations of the form (\ref{cano100}) that are close
at leading order up to the critical time $t_0$.
\begin{lem}
For any Hamiltonian system of the form (\ref{pert3})  with $h_{uu}=h_{vv}P''(u)$,  there exists an
almost  integrable system of the form 
(\ref{DHamiltonian}) such that the two systems of equations tend  in the multiple scale limit described in theorem~\ref{theoexpansion}
 to the same   equations (\ref{Mexpansion}).  \end{lem}
\pf 
It is sufficient  to show that for given $a=a(u,v)$, $b=b(u,v)$ and 
$c=c(u,v)$ one can find $\rho_u(u,v)$, $\rho_v(u,v)$ and $s_3(u,v)$ such that at the critical point $ (u_0,v_0)$ 
the following identities hold:
\begin{equation}
\label{rho1}
\begin{split}
a_0&=P_0''( \rho^0_u h^0_{vvv} + \rho^0_v h^0_{uvv}-s_3^0h_{vv}^0) +\frac12 P_0''' \rho^0_v h^0_{vv}\\
b_0&=P_0''\rho^0_v h^0_{vvv} +\rho^0_u h^0_{uvv} +\frac{P_0'''}{4 P_0''} \rho^0_u h^0_{vv}\\
c_0&=\rho^0_u h^0_{vvv} +\rho^0_v h^0_{uvv}+s_3^0h_{vv}^0.
\end{split}
\end{equation}
The constants  $\rho^0_u$ and $\rho_v^0$ can be chosen in an arbitrary way since
 they  solve the second order equation (\ref{rhoequation}) and $s_3(u,v)$ is an arbitrary function.
 The system (\ref{rho1}) is solvable for $\rho_u^0$, $\rho_v^0$ and $s_3^0$ as a function of $a_0,b_0, c_0$.
\epf
For a given inital datum the solutions of  two different Hamiltonian perturbations 
of the form (\ref{pert3})  with  the same unperturbed Hamiltonian density  $h(u,v)$ satisfying $h_{uu}=h_{vv}P''(u)$, 
have the same approximate solution  for $t<t_0$.   From our conjecture~\ref{conjelliptic} it follows that the  solutions near the critical point have the same leading asymptotic expansion if the
coefficients of the two systems  satisfy at the critical point the  relation (\ref{rho1}).
\medskip

\section{An example: generalized nonlinear Schr\"odinger equations}\label{section4}
 Let us now consider the example of generalized nonlinear 
 Schr\"odinger (NLS) equations
\beq\label{nls1}
i\, \epsilon\, \psi_t + \frac{\epsilon^2}2 \psi_{xx} \pm V\left(|\psi|^2\right)\, \psi=0,\quad \epsilon>0,
\eeq
where  $\psi=\psi(x,t)$ is a complex variable and $V$ is a  smooth function  monotone increasing on the positive real axis.  The case $V(u)=u$ is called cubic NLS, the case $V(u)=u^2/2 $ is called quintic NLS and so on.
The case with positive sign in front of the potential $V$ is the so-called focusing NLS, while the negative sign corresponds to the  
defocusing NLS.  For sufficiently regular $V$ the initial value 
problem of the defocusing NLS equation  is globally  well posed  in some suitable functional space , see \cite{gwp1,bour} and references therein, while  the solution of the initial value problem of the focusing case is globally well posed  when the nonlinearity $V\left(|\psi|^2\right)=|\psi|^2$, \cite{gwp1}

Equation (\ref{nls1})  can be rewritten in the standard Hamiltonian form \eqref{cano} with two real-valued dependent functions, the so called Madelung 
transform
$$
u=|\psi|^2, \quad v=\frac{\epsilon}{2 i } \left( \frac{\psi_x}{\psi}-\frac{\psi^*_x}{\psi^*}\right)
$$
(the star stands for the complex conjugation). Then equation  (\ref{nls1}) reduces to  the system of equations 
\eqa\label{nls2}
&&
u_t+(u\, v)_x=0
\nn\\
&&
\\
&&
v_t +\pal_x \left[ \frac{v^2}2\mp V(u)\right]= \frac{\epsilon^2}4 \pal_x \left( \frac{u_{xx}}{u} -\frac{u_x^2}{2 u^2}\right).
\nn
\eeqa
The  above system can be written in the  Hamiltonian 
form\footnote{Observe the change of sign in the definition of the 
Hamiltonian (cf. \eqref{cano}). The normalization used in the last two sections of the present paper is more widely accepted in the physics literature.} 
\eqa\label{cano0}
&&
u_t+\pal_x\frac{\delta H}{\delta v(x)}=0
\nn\\
&&
\\
&&
v_t+\pal_x\frac{\delta H}{\delta u(x)}=0
\nn
\eeqa
with the Hamiltonian 
\beq\label{nls21}
H=\int \left[ \frac12 u\, v^2 +W(u) +\frac{\epsilon^2}{8u} u_x^2\right]\, dx, \quad W'(u)=\mp V(u).
\eeq
The semiclassical  limit of this system
\eqa\label{nls0}
&&
u_t+(u\, v)_x=0
\nn\\
&&
\\
&&
v_t +\pal_x \left[ \frac{v^2}2\mp V(u)\right]= 0,
\nn
\eeqa
is  of elliptic or hyperbolic type respectively provided $V(u)$ is a monotonically increasing smooth function on the positive semiaxis.

\noindent
Another interesting NLS type model is given by the nonlocal NLS equation  \cite{conti},\cite{krol},\cite{trillo},
\begin{equation}
\label{full_NNLS_complex}
\begin{split}
&i \epsilon \psi_{t} + \frac{\epsilon^{2}}{2} \psi_{xx} \pm \theta \psi = 0 \\
& \theta - \epsilon^{2} \eta\; \theta_{xx} = |\psi|^{2},
\end{split}
\end{equation}
where $\eta$ is a  positive constant.
In the slow variables $u$, $v$
this nonlocal NLS model can be equivalently written as 
\begin{align}
\label{full_NNLS}
&u_{t} +(u v)_{x}=0 \\
&v_{t} +v v_{x} \mp \theta_x  + \frac{\epsilon^{2}}{4} \left( \frac{u_{x}^{2}}{2 u^{2}} - \frac{u_{xx}}{u} \right)_{x}=0 \\
& \theta - \epsilon^{2} \eta \theta_{xx} = u.
\end{align}
Writing $\theta$ from the last equation as the formal series
$$
\theta=u+\epsilon^2 \eta \, u_{xx}+\epsilon^4\eta^2 u_{xxxx}+\dots
$$
and
keeping terms up to order $\epsilon^{2}$ only, one arrives at a system of the above class
\begin{align*}
&u_{t} + (u v)_{x}=0 \\
&v_{t} +v v_{x} \mp u_{x} + \frac{\epsilon^{2}}{4} \left( \frac{u_{x}^{2}}{2 u^{2}} - \frac{u_{xx}}{u} \right)_{x}\mp \epsilon^{2}  \eta \; u_{xxx} = O(\epsilon^{4}).
\end{align*}

The nonlocal NLS can be written in the Hamiltonian form   (\ref{cano0})  with the Hamiltonian $ H =  \int h \; dx $ and 
\begin{equation}
h = \frac{1}{2} u v^{2} \mp \dfrac{u^2}{2}\pm \epsilon^2 \eta \dfrac{u_x^2}{2}+ \frac{\epsilon^{2}}{8} \frac{u_{x}^{2}}{u}+O(\epsilon^4)
\end{equation}
The above Hamiltonian coincides with the one of the cubic NLS when $\eta=0$. 

We are going to study the critical points of the solutions of the 
system (\ref{nls0}) for some initial data and then the solutions of 
equation (\ref{nls1})  or (\ref{full_NNLS_complex}) for the same data
near the critical points of the solution of (\ref{nls0}).
We first treat the hyperbolic case.
\subsection{Defocusing generalized NLS}
 The Riemann invariants and the characteristic velocities of equation (\ref{nls0}), in the hyperbolic case,  are
\beq\label{nls01}
r_\pm =v\pm Q(u), \quad Q'(u) =\sqrt{\frac{V'(u)}{u}}, \quad \lambda_\pm = v\pm \sqrt{u\, V'(u)}.
\eeq
The general solution to \eqref{nls0} can be represented in the implicit form
\beq\label{nls3}
\begin{array}{ccc} x & =& v\, t +f_u\\
\\
0 & = & u\, t +f_v\end{array}
\eeq
where the function $f=f(u,v)$ solves the linear PDE of the form \eqref{lineq}
\beq\label{nls4}
f_{uu}= \frac{V'(u)}{u} f_{vv}.
\eeq

The coordinates $(u_0, v_0)$ of the point of a generic break-up of the second Riemann invariant $r_-$ can be determined from the system
\eqa\label{nls5}
&&
f_{uv}^0 =\sqrt{\frac{V'_0}{u_0}} f_{vv}^0 +\dfrac{f_v^0}{u_0}
\nn\\
&&
\\
&&
f_{uvv}^0 =\sqrt{\frac{V'_0}{u_0}} f_{vvv}^0+\frac{V'_0-u_0 V''_0}{4 u_0 V'_0} f_{vv}^0.
\nn
\eeqa
In the Riemann invariants the system \eqref{nls2} reads
\beq\label{nls22}
\pal_t r_\pm +\left( v \pm \sqrt{u\, V'(u)}\right) \pal_x r_\pm 
=\pm\frac{\epsilon^2}{8 \sqrt{u\, V'(u)}} \left[ \pal_x^3 r_+ - \pal_x^3 r_-+\dots\right].
\eeq
The asymptotic representation of the shifted Riemann invariants
$$
r_\pm -r_\pm^0 \simeq \frac1{\sqrt{u_0}} \left[ \sqrt{u_0}(v-v_0)\pm \sqrt{V'_0} (u-u_0)\right]$$
is given as a function of the shifted characteristic variables
$$
x_\pm =(x-x_0) -\left(v_0 \pm \sqrt{ u_0 V'_0}\right) (t-t_0)
$$
in the form \eqref{conj1} with
\eqa\label{nls6}
&&
\alpha= 2 \frac{\sqrt{V'_0}}{\sqrt{u_0}} f_{vv}^0
\nn\\
&&
\beta=-\frac{u_0 V''_0 +3 V'_0}{8 \sqrt{u_0{V'_0}^{3}}}
\nn\\
&&
\\
&&
\gamma=-f_{uvvv}^0 + \sqrt{\frac{V'_0}{u_0}} f_{vvvv}^0 +\frac{V'_0 -u_0 V''_0}{4 u_0 V'_0} f_{vvv}^0+\frac{3{V'_0}^2 + 2 u_0 V'_0 V''_0 -5 u_0^2 {V''_0}^2 +4 u_0^2 V'_0 V'''_0}{32 u_0^{3/2} {V'_0}^{5/2}}f_{vv}^0
\nn\\
&&
\rho=\frac1{16 u_0 V'_0}
\nn\\
&&
\sigma=-\frac1{16 u_0 V'_0}.
\nn
\eeqa
In particular for the nonlocal defocusing NLS equation,  the shifted Riemann invariants
$$
r_\pm -r_\pm^0 \simeq \frac1{\sqrt{u_0}} \left[ \sqrt{u_0}(v-v_0)\pm \sqrt{V'_0} (u-u_0)\right]$$
as functions of the shifted characteristic variables
$$
x_\pm =(x-x_0) -\left(v_0 \pm \sqrt{ u_0 V'_0}\right) (t-t_0)
$$
behave in the vicinity of the point of gradient catastrophe as in  
\eqref{conj1}  with   $\alpha,\beta$ and  $\gamma$  as in (\ref{nls6}) with $V'_0=1$  and $\rho$ and $\sigma$ given by
\begin{equation}
\label{rhononlocal}
\rho=\dfrac{1-4\eta u_0}{16 u_0}=-\sigma.
\end{equation}

%
%
%
%
\subsection{Focusing generalized NLS}
The Riemann invariants and the characteristic velocities of system (\ref{nls0}) in the elliptic case are
\[
r_{\pm}=v\pm iQ(u),\;\;Q'(u)=\sqrt{\dfrac{V'(u)}{u}},\quad \lambda_{\pm}=-\left(v\pm i\sqrt{uV'(u)}\right).
\]
The general solution of  (\ref{nls0}) is obtained via the hodograph equations
\begin{equation}
\label{hodograph}
\begin{split}
&v\,t+f_u(u,v)=x\\
&u\,t+f_v(u,v)=0
\end{split}
\end{equation}
where the function $f(u,v)$ solves the linear equation
\begin{equation}
\label{eqf}
f_{uu}+\dfrac{ V'(u)}{u}f_{vv}=0.
\end{equation}
The point of elliptic umbilic catastrophe is determined by the equations (\ref{hodograph}) and the conditions
\begin{equation}
\label{umbilicatastrophe}
f_{uu}=0,\quad t+f_{uv}=0.
\end{equation}
The asymptotic formula (\ref{rasymp})  near the point of elliptic umbilic catastrophe takes the form
\begin{equation}
\label{F}
v-v_0+i\sqrt{\dfrac{V'_0}{u_0}}(u-u_0)=-12 \left(\dfrac{\epsilon^2C_+^+}{(12a_+)^2\lambda_{+,+}^0}\right)^{\frac{1}{5}}\Omega(\xi)+O(\epsilon^{\frac{4}{5}}),
\end{equation}
where
\[
\xi=\left(\dfrac{(\lambda_{++}^0)^2}{12a_+\epsilon^4(C_+^+)^2}\right)^{\frac{1}{5}}(x-x_0-(v_0+i\sqrt{u_0V'_0})(t-t_0))
\]
and 
\begin{equation}
\label{constantFNLS}
C_+^+=\dfrac{1}{8 i}\dfrac{1}{\sqrt{V'_0u_0}},\quad \lambda_{+,+}^0=-\dfrac{3}{4}-\dfrac{u_0V''_0}{4V'_0},\;\;a_+=f_{uvv}^0+iQ'_0f_{vvv}^0
\end{equation}
and $Q'(u)=\sqrt{\dfrac{V'(u)}{u}}$, $V'_0=V'(u_0)$, $Q'_0=Q'(u_0)$, $V''_0=V''(u_0)$.

\begin{rem}
In the formula (\ref{F})  the convention for choosing the fifth root 
is defined by the following condition: for symmetric initial data and $t=t_0$ the argument 
of the tritronqu\'ee solution has to be purely imaginary.
So, defining 
\[
a_+=-\dfrac{i}{r e^{i\psi}}
\]
one arrives at the formula
\begin{equation}
\label{F1}
v-v_0+i\sqrt{\dfrac{V'_0}{u_0}}(u-u_0)=6i \left(\dfrac{\epsilon^2r^2 e^{2i\psi} }{9\sqrt{\frac{u_0}{V_0'}}\left(3V_0'+u_0V''_0\right)}\right)^{\frac{1}{5}}\Omega(\xi)+O(\epsilon^{\frac{4}{5}}),
\end{equation}
where 
\begin{equation}
\label{xiNLS}
\xi=-i\left(\dfrac{u_0}{V_0'}\left(3V_0'+u_0V''_0\right)^2\dfrac{re^{i\psi}}{3\epsilon^4}\right)^{\frac{1}{5}}(x-x_0-(v_0+i\sqrt{u_0V'_0})(t-t_0)).
\end{equation}
\end{rem}

\begin{rem}
In the focusing nonlocal NLS model (\ref{full_NNLS_complex}) the behaviour of the solution near the point of elliptic umbilic catastrophe is given by the expression 
(\ref{F}) with $a_+$ and $\lambda_{+,+}^0$ as in (\ref{constantFNLS}) and 
\[
C_+^+=-\dfrac{1+4\eta u_0}{8 i\sqrt{u_0}}, 
\]
that is,
\begin{equation}
\label{Fnonlocal}
v-v_0+\dfrac{i}{\sqrt{u_0}}(u-u_0)=6i \left( \dfrac{\epsilon^2(1+4\eta u_0)r^2 e^{2i\psi} }{27\sqrt{u_0}}\right)^{\frac{1}{5}}\Omega(\xi)+O(\epsilon^{\frac{4}{5}}),
\end{equation}
where 
\[
\xi=-i\left(\dfrac{3u_0 r e^{i\psi}}{(1+4\eta u_0)^2\epsilon^4}\right)^{\frac{1}{5}}(x-x_0-(v_0+i\sqrt{u_0})(t-t_0))
\]
For $\eta=0$ such a formula was derived in an equivalent form in \cite{DGK}.
\end{rem}

\section{Studying particular solutions}\label{section5}
The present section is devoted to the comparison of solutions to the defocusing and focusing NLS equations with their unperturbed
counterparts near the critical points of solutions of the unperturbed system with, respectively, the asymptotic formula (\ref{conj1}) and (\ref{F}). We consider various examples of nonlinear potentials $V$ and initial data.

Let us consider the  Cauchy problem 
\begin{equation}
\label{cauchynls0}
\begin{split}
&\dfrac{\partial r_+}{\partial t}=\lambda_+ (r_+, r_-)\dfrac{\partial r_+}{\partial x},\quad \dfrac{\partial r_-}{\partial t}=\lambda_-(r_+, r_-) \dfrac{\partial r_-}{\partial x},\\
&r_+(x,t=0)=\varphi_+(x),\quad r_-(x,t=0)=\varphi_-(x).
\end{split}
\end{equation}
If the initial data $\varphi_{\pm}(x)$ are  bounded   analytic 
functions of $x$, then in virtue of the Cauchy--Kowalevskaya theorem (see \cite{Bressan})
$r_{\pm}(x,t) $ are analytic functions in the variable $x$  up to the time $t<t_0$ where $t_0$ is the time of gradient catastrophe.

The implicit solution of (\ref{cauchynls0}) is given by the hodograph equations as 
\begin{equation}
\label{hodographhyper}
x=-\lambda_{\pm}(r_+,r_-)t+\mu_{\pm}(r_+,r_-)
\end{equation}
where $\mu_{\pm}$ solves the system of linear PDEs equivalent to (\ref{nls4})
\begin{equation}
\label{overdetermined}
\dfrac{\partial \mu_+}{\partial r_-}=\dfrac{\mu_+-\mu_-}{\lambda_+-\lambda_-}\,\dfrac{\partial \lambda_+}{\partial r_-},\;\;\quad \dfrac{\partial \mu_-}{\partial r_+}=\dfrac{\mu_+-\mu_-}{\lambda_+-\lambda_-}\,\dfrac{\partial \lambda_-}{\partial r_+},
\end{equation}
with the constraint
\begin{equation}
\label{tzero}
x=\mu_+(\varphi_+(x),\varphi_-(x)),\quad x=\mu_-(\varphi_+(x),\varphi_-(x)).
\end{equation}

\subsection{Defocusing cubic NLS}

\noindent
The  cubic NLS equation written as 
$$
i\epsilon\psi_t +\frac{\epsilon^2}2\psi_{xx}- |\psi|^2\psi=0,
$$
corresponds to the case $V(u)=u$, and the Riemann invariants  and the characteristics velocities (\ref{nls01}) take the form
\[
r_{\pm}=v\pm 2\sqrt{ u},\quad  \lambda_{+}=-\dfrac{1}{4}(3r_++r_-),\quad \lambda_{-} = -\dfrac{1}{4}(r_++3r_-).
\]
Let us consider an initial datum rapidly going to a constant value at infinity
\[
r_{\pm}(x, t=0) = \varphi_{\pm}(x).
\]
The solution of the corresponding quasilinear system (\ref{lead})  is obtained as described below. 
Let us suppose that  the initial datum  $\varphi_+(x)$ has  a single positive hump at 
$x_M$ and that $\varphi_-(x)$ has  a single negative hump  at 
$x_m\leq x_M$, and   denote by $h_{L/R}^+(r_+)$, the inverse of the increasing and decreasing part of 
$\varphi_+(x)$ and by  $h_{L/R}^-(r_-)$, the inverse of the 
decreasing and increasing   part of $\varphi_-(x)$ respectively.
 Since  $\lambda_+>\lambda_-$, it follows that $x_M(t)\geq  x_m(t)$  for all $t\geq 0$.
In order to obtain the quantities $\mu_{\pm}(r_+,r_-)$  we use the 
formula by  Tian and Ye \cite{FRTYE}:

\noindent
$\bullet$  $x>x_M(t)$
\begin{equation}
\label{cNLSsol1}
\mu_{\pm} (r_{+},r_{-}) =h_{R}^{+}(r_+) - \frac{2}{\pi (r_{+}-r_{-})} \int\limits_{h_R^-(r_{-})}^{h_R^{+}(r_{+})} dx\int\limits^{\varphi_{-}(x)}_{r_{-}} 
\sqrt{\frac{\tau - r_{\mp}}{r_{\pm} - \tau}} \; \frac{\left(\tau - \frac{\varphi_{+}(x) + \varphi_{-}(x)}{2}\right)d\tau}{\sqrt{(\tau-\varphi_{+}(x)) (\tau - \varphi_{-}(x))}} 
\end{equation}
$\bullet$  $x_m(t)\leq x\leq x_M(t)$
\begin{equation}
\mu_{\pm} (r_{+},r_{-}) =h_{L}^{+}(r_+) - \frac{2}{\pi (r_{+}-r_{-})} \int\limits_{h_R^-(r_{-})}^{h_L^{+}(r_{+})}dx \int\limits^{\varphi_{-}(x)}_{r_{-}} 
\sqrt{\frac{\tau - r_{\mp}}{r_{\pm} - \tau}} \; \frac{\left(\tau - \frac{\varphi_{+}(x) + \varphi_{-}(x)}{2}\right)d\tau}{\sqrt{(\tau-\varphi_{+}(x)) (\tau - \varphi_{-}(x))}} 
\end{equation}
$\bullet$  $x<x_m(t)$\begin{equation}
\mu_{\pm} (r_{+},r_{-}) =h_{L}^{+}(r_{+}) - \frac{2}{\pi (r_{+}-r_{-})} \int\limits_{h_L^-(r_{-})}^{h_L^{+}(r_{+})} \int\limits^{\varphi_{-}(x)}_{r_{-}} 
\sqrt{\frac{\tau - r_{\mp}}{r_{\pm} - \tau}} \; \frac{\tau - \frac{\varphi_{+}(x) + \varphi_{-}(x)}{2}}{\sqrt{(\tau-\varphi_{+}(x)) (\tau - \varphi_{-}(x))}} \; d\tau \; dx
\end{equation}
For different choices of initial data, more complicated relations can be obtained.
Within the interval of monotinicity of the function $\varphi_{\pm}$ the solution (\ref{cNLSsol1}) can be written also in the equivalent  form
\cite{FRTYE}
\begin{equation}
\label{cNLSsol2}
\mu_{\pm} (r_+,r_-) = \frac{2}{\pi (r_+-r_-)}\left(\int\limits^{\varphi_-(\infty)}_{r_-}+\int\limits_{\varphi_+(\infty)}^{r_+}\right)\sqrt{
\dfrac{ \tau-r_{\mp} }{ r_{\pm}-\tau} }\theta'(\tau)d\tau,
\end{equation}  
with 
\begin{equation}
\label{theta}
\begin{split}
\theta'(\tau)&=\dfrac{\tau-\frac{\varphi_+(\infty)+\varphi_-(\infty)}{2}}{\sqrt{(\tau-\varphi_-(\infty))(\tau-\varphi_+(\infty))}}x(\tau)-\\
-&\int\limits^{\infty}_{x(\tau)}\left( \dfrac{\tau-\frac{\varphi_+(x)+\varphi_-(x)}{2}}{\sqrt{(\tau-\varphi_-(x))(\tau-\varphi_+(x))}}-\dfrac{\tau-\frac{\varphi_+(\infty)+\varphi_-(\infty)}{2}}{\sqrt{(\tau-\varphi_-(\infty))(\tau-\varphi_+(\infty))}}  \right)dx
\end{split} 
\end{equation}
where $x(\tau)$ is the inverse function of $\varphi_{\pm}(x)$ in the interval of monotonicity.


For the particular case $v(x,0)=0$, $u(x,0)=a^2\mbox{sech}^2 x$  one has
\[
\theta'(\tau)=\dfrac{1}{2}\log\dfrac{4a^2-\tau^2}{\tau^2}   
\]
and for $x>x_M(t)$
\begin{align}
\mu_{\pm}(r_=,r_-)&=-\log(\sqrt{2a+r_+}+\sqrt{2a+r_-})-\log(\sqrt{2a-r_+}+\sqrt{2a-r_-})+\log(r_+-r_-)\nn\\
&\pm\dfrac{1}{r_+-r_-}\left(\sqrt{(2a+r_+)(2a+r_-)}-\sqrt{(2a-r_+)(2a-r_-)})\right).
\end{align}
The critical point is obtained by the two equations (\ref{hodographhyper}) together with
\[
\dfrac{3}{4}t+\frac{\partial \mu_-}{\pal r_-}=0 \quad \dfrac{\partial^2\mu_-}{\partial r_-^2}=0,
\]
which give the equations
\eqa
\label{eqt}
&&
\dfrac{3}{4}t-\dfrac{1}{(r_+-r_-)^2}\left(\sqrt{\dfrac{(2a+r_+)^3}{2a+r_-}}-\sqrt{\dfrac{(2a-r_+)^3}{2a-r_-}}\right)=0
\nn\\
&&
\\
&&
\dfrac{\sqrt{(2a+r_+)^3}}{\sqrt{(2a+r_-)^3}}(8a+5r_--r_+)-\dfrac{\sqrt{(2a-r_+)^3}}{\sqrt{(2a-r_-)^3}}(8a-5r_-+r_+)=0.
\nn
\eeqa
Solving the above two equations together with (\ref{hodographhyper}) yields 
\[
r^0_+=\dfrac{a}{3}(6-\sqrt{33})\sqrt{2\sqrt{33}+6},\;\;r^0_-=-\dfrac{a}{3}\sqrt{2\sqrt{33}+6}, 
\]
\[
t_0=\dfrac{3\sqrt{2}}{32 a}\sqrt{69+11\sqrt{33}},\;\;x_0=-2.209395255.
\]
\begin{equation}
\begin{split}
\dfrac{\partial^3 \mu_-}{\partial r_-^3}&=\dfrac{1}{2 (r_--r_+)^4}\sqrt{\dfrac{(2a-r_+)^3}{(2a-r_-)^5}}(48a^2+3/2r_+^2+35/2 r_-^2-7r_-r_+-56ar_-+8ar_+)\\
&-\dfrac{1}{2(r_--r_+)^4}\sqrt{\dfrac{(2a+r_+)^3}{(2a+r_-)^5}}(48a^2+3/2r_+^2+35/2 r_-^2-7r_-r_++56ar_--8ar_+)
\end{split}
\end{equation}

The constants $b=12\dfrac{\rho}{\beta},\beta, \gamma$  and $\alpha$ defined in (\ref{nls6}) at the critical time are given by
\[
b=\dfrac{2}{\sqrt{u_0}}=\dfrac{8}{r^0_+-r^0_-}=\dfrac{3(7+\sqrt{33})}{2a\sqrt{6+2\sqrt{33}}}
\]
\[
\beta=-\dfrac{3}{8\sqrt{u_0}}=-\dfrac{3}{2(r_+^0-r_-^0)}=-\dfrac{9(7+\sqrt{33})}{32a\sqrt{6+2\sqrt{33}}}
\]
\[
\gamma=\left(-\dfrac{\partial^3 \mu_-}{\partial r_-^3}+t_0\dfrac{\partial^3 \lambda_-}{\partial r_-^3}\right)_{r_-=r_-^0, r_+=r_+^0} \simeq \frac{2.3269}{a^{3}}
\]
and 
\[
\alpha=\left.\left(\dfrac{\partial \mu_+}{\partial r_+}+\dfrac{3}{4}t_0\right)\right|_{r_+=r_+^0,r_-=r_-^0}=2.635171951.
\]



\subsection{Defocusing quintic  NLS}

Let us now proceed to the case $V(u)=u^2/2$.
The Riemann invariants of the quintic defocusing NLS
$$
i\epsilon\psi_t +\frac{\epsilon^2}2\psi_{xx}-\frac12 |\psi|^4\psi=0
$$  
are given by
\[
r_{\pm}=v\pm u.
\]
The equations (\ref{nls0}) reduce to the two decoupled Riemann wave equations
\[
\partial_tr_{\pm}+r_{\pm}\partial_x r_{\pm}=0,
\]
which can be solved by the method of characteristics. For the initial 
data $r_{\pm}(x,0)=\rho_{\pm}(x)$, one has the solution in implicit form
\begin{equation}
\label{hodograph11}
r_{\pm}(x,t)=\rho_{\pm}(\xi),\quad x=\rho_{\pm}(\xi)t+\xi.
\end{equation}
The point of gradient catastrophe is determined by the conditions
\[
F_{\pm}''(r)=0,\quad t+F_{\pm}(r)=0
\]
where $F_{\pm}$ is the inverse of the decreasing part of the initial data $\rho_{\pm}(x)$.
The constants $b=12\dfrac{\rho}{\beta},\, \beta$ and $\sigma$   defined in (\ref{nls6}) at the critical time are given by
\begin{equation}
\label{betasigma2}
b=\dfrac{3}{2u_0}=\dfrac{3}{r_+^0-r_-^0},\quad 
\beta=-\dfrac{1}{2u_0}=-\dfrac{1}{r_+^0-r_-^0},\quad \sigma=\dfrac{1}{16u_0^2}=\dfrac{1}{4(r_+^0-r_-^0)^2}.
\end{equation}
The constants $\alpha $ and $\gamma$ in (\ref{nls6}) depend on the initial data and are evaluated for several initial data below.

\noindent
{\bf Symmetric initial data.}
We consider the initial data 
\[
u(x,t=0)=A\,\mbox{sech}^2x,\quad v(x,t=0)=-B\tanh^2x,\quad B\leq A,
\]
with $A$ a positive constant.
For such initial data, both $r_{\pm}$ have a point of gradient catastrophe.
The evolution in time of the  decreasing part of  $r_+(x,t)$ gives
\begin{equation}
\label{Fp}
x=r_+t+F_+(r_+),\quad F_+(r_+)=\log\dfrac{\sqrt{B+A}+\sqrt{A-r_+}}{\sqrt{B+r_+}}.
\end{equation}
The point of gradient catastrophe is given by
\[
r^0_+=\dfrac{2A-B}{3},\;\;t^0_+=\dfrac{3\sqrt{3}}{4(A+B)},\quad x_+^0=\dfrac{\sqrt{3}}{4}\dfrac{2A-B}{A+B}+\log\dfrac{\sqrt{3}+1}{\sqrt{2}}.
\]
The second Riemann invariant $r_-(x^0_+,t^0_+)$ is determined from the equation  $x^0_+=r^0_-t^0_++F_-(r^0_-)$ with
\begin{equation}
\label{Fm}
 F_-(r_-)=\log\dfrac{\sqrt{A-B}+\sqrt{A+r_-}}{\sqrt{-B-r_-}}.
 \end{equation}
The constants $\gamma$ and $\alpha$ in (\ref{abc}) take the form 
\[
\gamma=-F'''_+(r^0_+)=\dfrac{81\sqrt{3}}{16(A+B)^3},\quad \alpha=F'_-(r^0_-)+t^0_+=-\dfrac{\sqrt{A-B}}{2\sqrt{A+r_-^0}(B+r_-^0)}+t^0_+ .
\]

The evolution in time of the  decreasing part of  $r_-(x,t)$ gives
\[
x=r_-t-F_-(r_-),
\]
with $F_-(r_-)$ as  in (\ref{Fm}).
The point of gradient catastrophe is given by
\[
r^0_-=-\dfrac{2A+B}{3},\;\;t^0_-=\dfrac{3\sqrt{3}}{4(A-B)},\quad x^0_-=-\dfrac{\sqrt{3}}{4}\dfrac{2A+B}{A-B}-\log\dfrac{\sqrt{3}+1}{\sqrt{2}}.
\]
The constants $\gamma$ and $\alpha$ in (\ref{abc}) take the form 
\[
\gamma=F'''_-(r^0_-)=\dfrac{81\sqrt{3}}{16(A-B)^3},\quad \alpha=-F'_+(r^0_+)+t^0_-=\dfrac{\sqrt{A+B}}{2\sqrt{A-r_+^0}(B+r_+^0)}+t^0_-
\]
where $r^0_+$ is determined from the equation $x^0_-=r^0_+t^0_--F_+(r^0_+)$ with $F_+(r_+)$ as in (\ref{Fp}).\\
 
\noindent
{\bf `Dark' initial data.}
We consider the initial data 
\[
u(x,0)=A\tanh^4\frac{x}{B},\qquad v(x,0)=0.
\]
In the evolution of this initial data  two points of gradient catastrophe occur, one at $x_0<0$ for the Riemann invariant $r_+$ and one at $x_0>0$ for the Riemann invariant $r_-$. 
For this initial data the Riemann invariant $r_+(x,t)$ for $x<x_m$, where 
$x_m$ is the point of  the minimum of $u$,  takes the form
\[
x=r_{+}t-F_+(r_{+}),\quad F_+(r_+)=\dfrac{1}{2B}\log\dfrac{1+\left(\dfrac{r_+}{A}\right)^{\frac{1}{4}}}{1-\left(\dfrac{r_+}{A}\right)^{\frac{1}{4}}}
\]
with critical point
\[
r_+^0=\dfrac{9A}{25},\quad t^0_+=\dfrac{25\sqrt{15}}{72AB},\quad x^0_+=\dfrac{\sqrt{15}}{8B}-\dfrac{1}{2B}\log(4+\sqrt{15}).
\]
The point $r_-^0(x^0_+,t^0_+)$ is determined from the condition $x^0_+=r_-^0t^0_+-F_-(r^0_-)$ with 
\[
F_{-}(r_-)=\dfrac{1}{2B}\log\dfrac{1+\left(-\dfrac{r_-}{A}\right)^{\frac{1}{4}}}{1-\left(-\dfrac{r_-}{A}\right)^{\frac{1}{4}}}.
\]
The constants $\gamma$ and $\alpha $ in (\ref{abc}) take the form
\[
\gamma=\dfrac{78125\sqrt{15}}{31104A^3B},\quad \alpha=t^0_+-\dfrac{1}{4AB \left(-\frac{r_-^0}{A}\right)^{\frac{3}{4}}(-1 +\sqrt{-\dfrac{r_-^0}{A}})}.
\]
The evolution of $r_-(x,t)$ for $x>x_m$, where $x_m$ is the point of  minimum,  is determined by the equation
\[
x=r_{-}t+F_{-}(r_{-}).
\]
The point of gradient catastrophe occurs at 
\[
r_-^0=-\dfrac{9A}{25},\quad t^0_-=\dfrac{25\sqrt{15}}{72AB},\quad x^0_-=-\dfrac{\sqrt{15}}{8B}+\dfrac{1}{2B}\log(4+\sqrt{15}).
\]
The point $r^0_+(x^0_-,t^0_-)$ is determined by the equation $x^0_-=r_+^0t^0_-+F_+(r^0_+)$.
The constants $\gamma$ and $\alpha $ in (\ref{abc}) take the form
\[
\gamma=\dfrac{78125\sqrt{15}}{31104A^3B},\quad \alpha=t^0_--\dfrac{1}{4AB \left(\frac{r_+^0}{A}\right)^{\frac{3}{4}}(-1 +\sqrt{\dfrac{r_+^0}{A}})}.
\]

\subsection{Focusing cubic NLS}
The case of the focusing cubic NLS equation
$$
i\epsilon \psi_t +\frac{\epsilon^2}2\psi_{xx}+|\psi|^2\psi=0
$$ 
was considered  extensively in \cite{DGK}.

For the initial data 
\begin{equation}
\label{initial}
\psi(x,t=0)=A_0\,\mbox{sech}\, x,\quad \mbox{or equivalently~~ $u=A_0^2\,$sech$^2x$, $v=0$,}
\end{equation}
 the solution of the equations
(\ref{nls0})   in the elliptic case is  given by 
\begin{align}
\label{eq1}
x&=vt+\Re\left[\mbox{arcsinh}\left(\dfrac{-\frac{1}{2}v+iA_0}{\sqrt{u}}\right)\right]\\
\label{eq2}
0&=tu-\Re\left[\sqrt{(-\frac{1}{2}v+iA_0)^2+u}\right]
\end{align}
and the function $f(u,v)$ takes the form
\beq\label{eqf1} 
f(u,v)=\Re \left[ \left( - \frac{v}{2} +i\, A_0\right) 
\sqrt{u+\left(- \frac{v}{2} +i\, A_0\right)^2} + u \, \log\frac{ 
(- \frac{v}{2} +i\, 
A_0)^{2}+\sqrt{(-\frac{v}{2}+iA_{0})^{2}+u}}{\sqrt{u}}\right].
\eeq
The point of elliptic umbilic catastrophe is given by 
\[
u_0=2A_0^2,\;\;v_0=0,\;\;x_0=0,\;\;\;t_0=\dfrac{1}{2A_0}
\]
and 
\[
f_{uvv}^0=0,\quad f_{vvv}^0=-\dfrac{u_0}{4A_0^3}.
\]
so that  $a_+$ in (\ref{constantFNLS}) becomes $a_+=-i\dfrac{\sqrt{u_0}}{4A_0^3}.$

\subsection{Focusing quintic NLS}
The Riemann invariants of the equation
$$
i\epsilon \psi_t +\frac{\epsilon^2}2\psi_{xx}+\frac12|\psi|^4\psi=0
$$  
are given by
\[
r_{+}=v+ iu,\quad r_-=r^{*}_+=v- iu.
\]
The equations (\ref{nls0})   reduce to two uncoupled Riemann  wave equations
\[
\partial_tr_{\pm}+r_{\pm}\partial_x r_{\pm}=0.
\]
For the initial datum 
\[
 r_{\pm}(x,t=0)=\pm i A_0^2\mbox{sech}^2x,\]
the solution is given by
\begin{equation}
\label{hodograph1}
x=r_{+}t+F(r_{+}), \quad x=r_-t+F^*(r_-)
\end{equation}
where $F$ is the inverse of the increasing  part of the initial data (\ref{initial}), namely
\[
F(r_+)=\log\dfrac{A_0-\sqrt{A_0^2+ir_+}}{\sqrt{-ir_+}}.
\]
An equivalent result can be obtained considering the decreasing part of the initial data.
Comparing (\ref{hodograph1}) with (\ref{hodograph}) one has 
\begin{equation}
\label{f1}
\Re(F)=f_u,\quad \Im(F)=f_v,
\end{equation}
and it easily follows that $f_{uu}+f_{vv}=0$.
The point of elliptic umbilic catastrophe is determined by the equations (\ref{hodograph1}) and the condition
\[
t+F'(r_+)=0
\]
or
\begin{equation}
\label{shock5nls}
\begin{split}
&x=vt+\Re(F)\\
&0=ut+\Im(F)\\
&v^2(A_0^2-3u)+u^3-A_0^2u^2=\dfrac{A_0^2}{4t^2}\\
&v(3u^2-2uA^2_0-v^2)=0
\end{split}
\end{equation}
The solution is given by
\[
x_0=0,\;\;v_0=0,\;\;  t_0=\dfrac{A_0}{2u_0\sqrt{u_0-A_0^2}},\;\;\dfrac{A_0}{\sqrt{u_0}}-\cos\dfrac{A_0}{2\sqrt{u_0-A^2_0}}=0.
\] 
The constants $r$ and $\psi$ in  (\ref{F1})  are  given by
\[
a_+=F''(r^0_+)=-\dfrac{i}{r e^{i\psi}}=i\dfrac{A_0}{4(u_0)^2}\dfrac{2A_0^2-3u_0}{(u_0-A_0^2)^{\frac{3}{2}}}.
\]

{\bf Asymmetric initial data.} \\
Let us first consider initial data $u=\mbox{sech}\,x$ and $ v=-\tanh x$. The solution defined by the hodograph transform takes the form
\begin{equation}
\label{hodograph2}
x=r_{\pm}t+F(r_{\pm})
\end{equation}
where $F$ is the inverse of the increasing part of the initial data (\ref{initial}), namely
\begin{equation}
\label{Fasym}
F(r_+)=\log\dfrac{i(1-r_+)}{r_++1}.
\end{equation}
The breaking condition
\[
t+F'(r_+)=t-\dfrac{1}{1-r_+}-\dfrac{1}{1+r_+}=0
\]
implies that the critical point is given by 
\[
v_0=0, \;\; t_0=\dfrac{2}{1+(u_0)^2}, \;\;x_0=0\quad 2u_0+((u_0)^2+1)\arctan\dfrac{1-(u_0)^2}{2u_0}=0.
\]
The constants $r$ and $\psi$ in  (\ref{F1})  are  given by
\[
a_+=F''(r^0_+)=-\dfrac{i}{r e^{i\psi}}=-\dfrac{4iu_0}{((u_0)^2+1)^2}
\]
The quantities in (\ref{constantFNLS}) take the form
\[
C_+^+=-\dfrac{1}{8 i u_0},\quad \lambda_{+,+}^0=-1,
\;\;a_+=f_{uvv}^0+iQ'_0f_{vvv}^0=-\dfrac{4iu_0}{((u_0)^2+1)^2}
\]

Asymmetric initial data can be obtained as a  solution of the hodograph equation in the form
\begin{equation}
\label{Hodasy}
x=r_+t+F(r_+)+\alpha\mathcal{F}(r_+),\quad \alpha\in\mathbb{R},
\end{equation}
where 
\[
 \mathcal{F}(r_+)=(r_+-1)\log(i(1-r_+))-(1+r_+)\log(1+r_+),
 \]
  that is  $\mathcal{F}'(r_+)=F(r_+)$ and $F$ is given in (\ref{Fasym}). Since $F$ and $\mathcal{F}$ are analytic functions, their  real and imaginary parts solve the Laplace equation. Therefore formula (\ref{Hodasy}) provides a solution to
the  elliptic system (\ref{nls0})  for $V(u)=\frac{u^2}2$.
In order to determine the point of elliptic umbilic catastrophe it is sufficient to 
consider the solution of (\ref{Hodasy}) together with the condition
\begin{equation}
\label{Hodasy1}
t+F'(r_+)+ \alpha F(r_+)=0.
\end{equation}
The real and imaginary parts of the equations (\ref{Hodasy}) and (\ref{Hodasy1}) give
\eqa\label{asym}
&&
x=vt+\dfrac{1}{2}(1+v\alpha)\log\dfrac{(1-v)^2+u^2}{(1+v)^2+u^2}-u\alpha\,\arctan\dfrac{1-u^2-v^2}{2u}
\nn\\
&&
\quad\quad -\dfrac{\alpha}{2}\log((1+u^2-v^2)^2+4u^2v^2)
\nn\\
&&
ut+(1+v\alpha)\arctan\dfrac{1-u^2-v^2}{2u}+\dfrac{u\,\alpha}{2}\log\dfrac{(1-v)^2+u^2}{(1+v)^2+u^2}-\alpha\arctan\dfrac{1+u^2-v^2}{2uv}=0
\nn\\
&&
t-\dfrac{1-v}{(1-v)^2+u^2}-\dfrac{1+v}{(1+v)^2+u^2}+\dfrac{\alpha}{2}\log\dfrac{(1-v)^2+u^2}{(1+v)^2+u^2}=0
\\
&&
\dfrac{-4uv}{[(1-v)^2+u^2][(1+v)^2+u^2]}+\alpha\arctan\dfrac{1-u^2-v^2}{2u}=0.
\nn
\eeqa
The solution of the above system determines the critical point $(x_0,t_0)$ and the values $v_0=v(x_0,t_0)$, $u_0=u(x_0,t_0)$.
The constants $r$ and $\psi$ in  (\ref{F1})  are  given by
\[
a_+=F''(r^0_+)+\alpha F'(r^0_+)=-\dfrac{i}{r e^{i\psi}}=\left.2\dfrac{\alpha(r_+^2-1)-2r_+}{(r_+^2-1)^2}\right|_{r_+=v_0+iu_0}.
\]
\noindent
{\bf Dark Soliton}
We consider the initial data $u(x,t=0)=\tanh^4x$  and $v=0$.
For such initial data the hodograph equations are
\[
x=rt+F(r_+),\quad F(r_+)=\dfrac{1}{2}\log\dfrac{1+(-ir_+)^{\frac{1}{4}} }{1-(-ir_+)^{\frac{1}{4}}}
\]
where $r_+=v+iu$.
The break up point is determined by the above  complex equation together with the condition
\[
t+F'(r_+)=0.
\]
As in this case it is not possible to obtain a simple analytic 
expression for the point of elliptic umbilic catatrophe $(x_0,t_0)$ 
and for $r_+^0, $ $r_-^0$, they are determined numerically.
The constants $r$ and $\psi$ that appear in (\ref{F1})  are  given by
\[
a_+=F''(r^0_+)=-\dfrac{i}{r e^{i\psi}}=-\dfrac{1}{16}\dfrac{(3-5\sqrt{-ir^0_+})(-ir^0_+)^{\frac{1}{4}}}{(r^0_+)^2(\sqrt{-ir^0_+}-1)^2}
\]

\section{Numerical Methods}\label{section6}
The numerical task in treating the semiclassical limit of the NLS 
equations consists in solving the NLS equations, the numerical 
evaluation of implicit solutions to certain ODEs and the direct 
solution of  ODEs of Painlev\'e 
type for a given asymptotic 
behavior. The present section provides a summary of 
how these different tasks are solved numerically, and how 
the numerical accuracy is controlled.

\subsection{NLS equations}
Critical phenomena are generally believed to be independent of 
the chosen boundary conditions. Thus we study a periodic 
setting in the following. This also includes rapidly decreasing functions which can be 
periodically continued as smooth functions within the finite numerical precision. 
This allows to approximate the spatial dependence 
via truncated Fourier series which leads for
the studied equations to large systems of ordinary differential 
equations (ODEs), see below.  
Fourier methods are convenient because of their excellent approximation 
properties for smooth functions  
(the numerical error in approximating smooth functions 
decreases faster than any power of the number $N$ of Fourier modes) 
and for minimizing the introduction of numerical 
dissipation which is important in the study of the purely dispersive 
effects considered here.
In Fourier space, equations (\ref{nls1})  have the form
\begin{equation}
    v_{t}=\mathbf{L}v+\mathbf{N}(v,t)
    \label{utrans},
\end{equation}
where $v$ denotes the (discrete) Fourier transform of $u$, 
and where $\mathbf{L}$ and $\mathbf{N}$ denote linear and nonlinear 
operators, respectively. The resulting system of ODEs consists in 
this case of \emph{stiff} equations.  
A stiff system is essentially a 
system for which explicit numerical schemes as explicit Runge-Kutta 
methods are inefficient, since prohibitively small time steps have to 
be chosen to control exponentially growing terms. 
The standard remedy for this is to use stable implicit schemes, which 
require, however, the iterative solution of a system of nonlinear equations at 
each time step which is computationally expensive. In addition the 
iteration often introduces numerical errors in the Fourier 
coefficients. 

The 
stiffness appears here in the linear part $\mathbf{L}$ (it is 
a consequence of the distribution of the eigenvalues of 
$\mathbf{L}$), whereas the 
nonlinear part is free of derivatives. In the semiclassical
limit, this stiffness is still present despite the small term 
$\epsilon$ in $\mathbf{L}$. This is due to the fact that the 
smaller $\epsilon$ is, the higher wavenumbers are needed to 
resolve the strong gradients. 
A possible way to deal with stiff systems are so-called 
implicit-explicit (IMEX) methods. 
The idea of IMEX  is the use of a 
stable implicit method for the linear part of the equation 
(\ref{utrans}) and an explicit scheme for the nonlinear part which is 
assumed to be non-stiff. In \cite{KassT} such schemes did not 
perform satisfactorily for dispersive PDEs which is why we  
consider a more sophisticated variant here. 
Driscoll's \cite{D} idea  was to split the linear part of the equation in
Fourier space into regimes of high and low wavenumbers. He
used the fourth order Runge-Kutta (RK) integrator for the low wavenumbers and the
lineary implicit RK method of order three for the high wavenumbers.
He showed that this method is in practice of fourth order over a wide range of step
sizes. 
In \cite{etna} we showed that this method performs best for the 
focusing case. We use it here also for the defocusing case where it 
was very efficient, but slighly outperformed by so-called 
time-splitting schemes as in \cite{baojin,baojin2}. For a discussion 
of exponential integrators in this context, see \cite{berska,BIS,etna}. 
Numerical approaches to the semiclassical limit of NLS can be also 
found in \cite{cern,ce}.

The accuracy of the numerical solution is controlled via the 
numerically computed conserved energy of the solution
\begin{equation}
    E[\psi] = 
    \int_{\mathbb{T}}^{}\left(\frac{\epsilon^{2}}{2}|\psi_{x}|^{2}-\frac{\rho}{s(s+1)}
    |\psi|^{2s+2}\right)dx
    \label{nlsE},
\end{equation}
which is an exactly conserved quantity for NLS equations. 
Numerically the energy $E$ will be a function of time due to 
unavoidable numerical errors. We define 
$\Delta E:=|(E(t)-E(0))/E(0)|$. It was shown in \cite{etna} that this quantity can 
be used as an indicator of the numerical accuracy if sufficient 
resolution in space is provided. The quantity $\Delta E$ typically 
overestimates the precision by two to three orders of magnitude. Since we are 
interested in an accuracy at least of order $\epsilon$, we will 
always ensure that the Fourier coefficients of the final state 
decrease well below $10^{-5}$, and that the quantity $\Delta E$ is 
smaller than $10^{-6}$ (in general it is of the order of machine 
precision, i.e.\ $10^{-14}$).

Focusing NLS equations have a \emph{modulational instability} due to 
the fact that they can be seen as a hyperbolic regularization of an 
elliptic semiclassical system for which initial value problems are 
ill-posed. In our context this instability shows up in the form of 
spurious growing modes for high wave numbers. To address this 
problem, we use a \emph{Krasny filter} \cite{krasny}, which means we 
put the Fourier coefficients with modulus below some treshold 
(typically $10^{-12}$) equal to zero. Thus the effect of rounding 
errors is reduced. In \cite{etna} it was pointed out that sufficient 
spatial resolution has to be provided to resolve the maximum of the 
solution close to the critical time to avoid instabilities. Thus we 
use $2^{14}$ to $2^{16}$ Fourier modes, and $10^{4}$ to $10^{5}$ time 
steps for the computations.

\subsection{Numerical solution of the semiclassical 
equations}
The solutions to the semiclassical equations are obtained in implicit 
form via hodograph techniques. 
These equations are of the form 
\begin{equation}
    S_{i}(\{y_{i}\},x,t)=0,\quad i=1,\ldots,M
    \label{Si},
\end{equation}
where the $S_{i}$ denote some given real  function of the $y_{i}$ and $x$, 
$t$. The task is to determine the $y_{i}$ in dependence of $x$ and 
$t$. To this end we determine the $y_{i}$ for given $x$ and $t$ as 
the zeros of the function $S:=\sum_{i=1}^{M}S_{i}^{2}$. This is
done numerically via a Newton iteration which is very efficient for a 
sufficiently good initial iterate. This iteration has 
the advantage that it can be done for all values of $x$ at the same 
simultaneously, i.e., in a vectorized way. Alternatively we use the 
algorithm \cite{optim} pointwise to solve (\ref{Si}). 
We calculate the zeros to the order of machine precision. The 
residual of the equations provides a check of the numerical 
acccuracy.

\subsection{Painlev\'e transcendents}
The asymptotic solutions near the break-up point  
are given by pole-free solutions with a given asymptotic behaviour for 
$x\to\pm\infty$ to the Painlev\'e-I and the P$^{2}_{I}$ equation. 
The standard way to solve 
these equations for large $|x|$ is to give a series solution to the respective equation 
with the imposed asymptotics that is generally divergent. These 
divergent series are truncated at finite values of $x$, $x_{l}<x_{r}$ at the first 
term that is of the order of machine precision. The sum of this 
truncated series at these points is then used as boundary data, and 
similarly for derivatives at these points. Thus the problem is 
translated to a boundary value problem on the finite interval 
$[x_{l},x_{r}]$. 

In \cite{GK2} we used for the ${\rm P}_I^2$ solution a collocation method with 
cubic splines distributed as \emph{bvp4} with Matlab, and the same 
approach in \cite{DGK} 
for the tritronqu\'ee solution of ${\rm P}_{I}$. Note that the 
tritronqu\'ee solutions are constructed on lines in the complex plane 
in the sector where the solution is conjectured (see \cite{DGK}) to 
have no poles. As in \cite{physicad} 
we use here a Chebyshev 
collocation method for both equations. The solution of the ODEs is sampled on Chebyshev collocation 
points $x_{j}$, $j=0,\ldots,N_{c}$ which can be related to an expansion 
of the solution in terms  of Chebyshev 
polynomials. The action of the 
derivative operator  is in this setting
equivalent to the action of a Chebyshev differentiation matrix on this space, 
see for instance 
\cite{trefethen}. The ODE is thus replaced by $N_{c}+1$ algebraic 
equations. The boundary data are included via a so-called 
$\tau$-method: The equations for $j=0$ and for $j=N_{c}$ (for the 
fourth order equation $j=0,1,N_{c}-1,N_{c}$) are replaced by the 
boundary conditions. The resulting system of algebraic equations is 
solved with a standard Newton method with relaxation which is 
necessary for the oscillatory ${\rm P}_I^2$ solution (there is no 
good initial iterate for the oscillatory solutions).  The convergence of the 
solutions is in general very fast. We always stop the Newton 
iteration when machine precision is reached. Again the highest  
Chebyshev coefficients are taken as an indication of sufficient 
resolution of the solutions (they have to reach machine precision). 
An efficient solution of the ODE is especially important in the  
P$^{2}_{I}$ case where the asymptotic solution to (\ref{p12}) has to be 
computed for many values of the parameter $t$. It can be seen in 
Fig.~\ref{figPI2}. For a more detailed discussion of this special  P$_{I}^{2}$ 
solution, also in the complex plane, see \cite{KKG}.

\section{Numerical study of defocusing  generalized  and nonlocal  NLS equations}\label{section7}
In this section we will study numerically solutions to defocusing 
NLS before and close to the break up of the corresponding 
semiclassical solutions. The solutions for NLS are compared to the 
corresponding semiclassical ones and for $t\sim t_0$ to an asymptotic 
description in terms of a special solution to the second equation in the Painlev\'e-I 
hierarchy. We will consider the cubic and the quintic version of 
these equations. The cubic NLS is the only completely integrable equation studied in 
this paper. Since the results for both cubic and quintic are very 
similar in this case, we present a more detailed investigation for 
the non-integrable quintic NLS. We also study a nonlocal variant of 
the cubic NLS equation. Unless otherwise noted, the 
considered critical point is always at the center of the figures. 

\subsection{sech$\,x$ initial data for the cubic defocusing NLS equation}

We will study the initial data $\psi_{0}(x)=\mbox{sech }x$ for several 
values of $\epsilon$. In 
this case there are two break-up points at $\pm x_{c}$ with $x_{c}= 
\sim 2.2093$ at the same time 
$t_0=\sim1.5244$. We will 
consider in the following always the break-up for negative values of 
$x$ where the Riemann invariant $r_{-}=v-2\sqrt{u}$ has a gradient 
catastrophe.

In Fig.~\ref{nlscubicc01}, the NLS solution, the semiclassical 
solution and the P$_{I}^2$ solution (\ref{conj1}) can be seen at the 
critical time close to the critical point of the semiclassical 
solution. 
\begin{figure}[thb!]
  \includegraphics[width=0.5\textwidth]{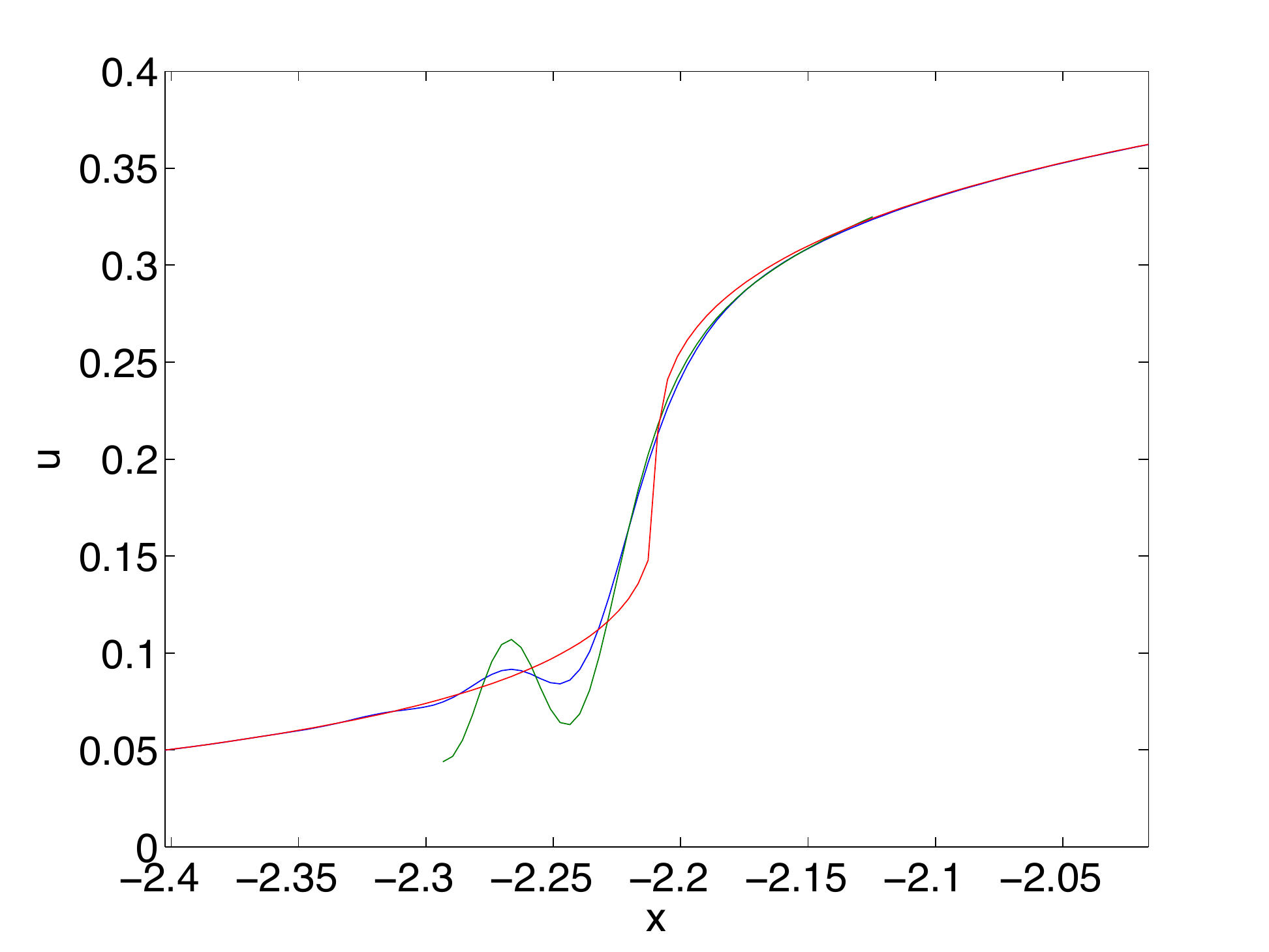}
  \includegraphics[width=0.5\textwidth]{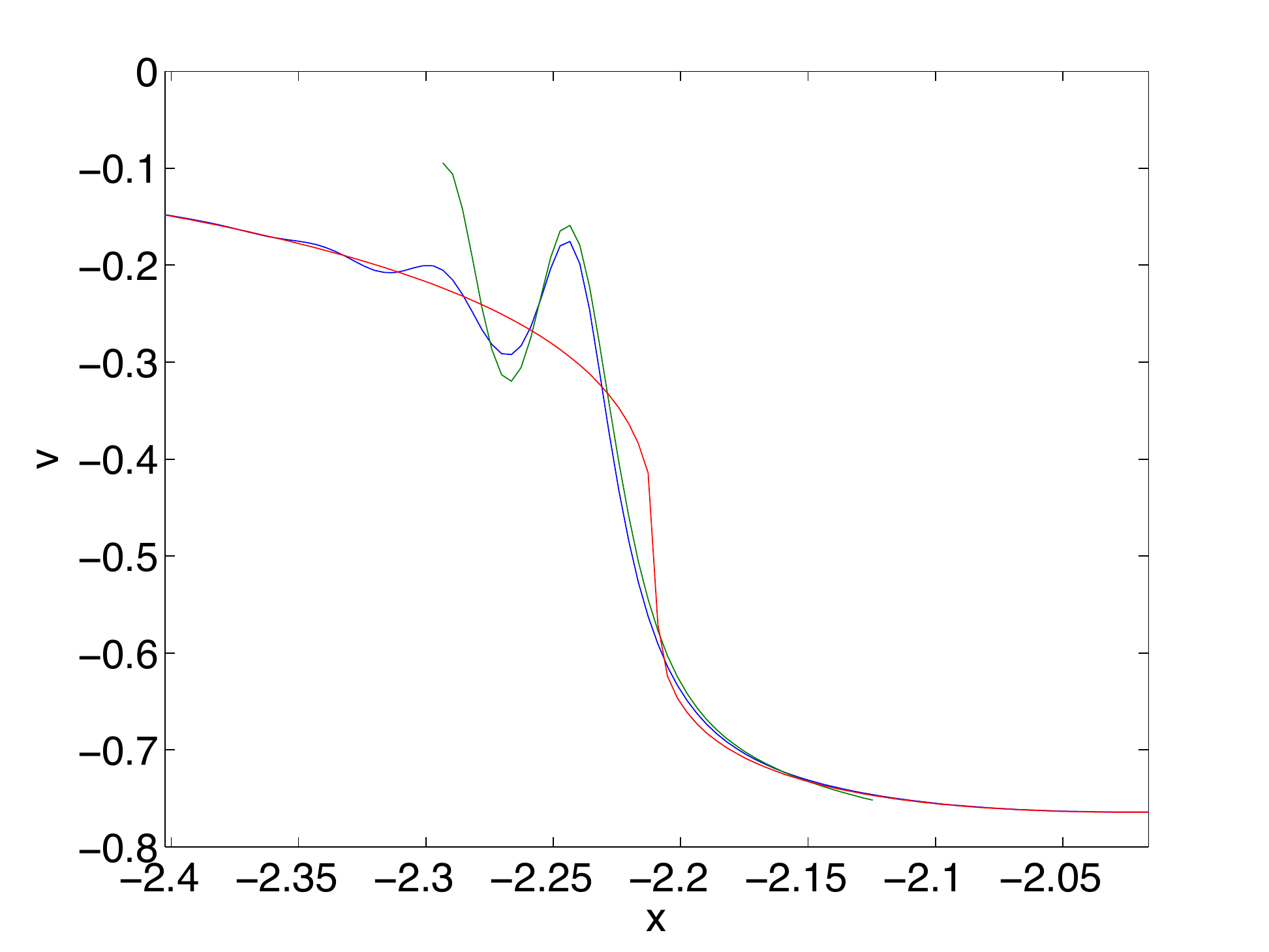}
 \caption{Solution to the defocusing cubic NLS equation for the initial data 
 $\psi_{0}(x)=\mbox{sech }x$ and $\epsilon=0.01$ at the critical time 
 $t_0$ in blue,  the corresponding semiclassical solution in red 
 and the P$_{I}^2$ solution (\ref{conj1}) in green; on the left the 
 function $u$, on the right the function $v$.}
 \label{nlscubicc01}
\end{figure}

The corresponding Riemann invariants can be seen in 
Fig.~\ref{nlscubicc01r}.

\begin{figure}[thb!]
\subfigure
{  \includegraphics[width=0.5\textwidth]{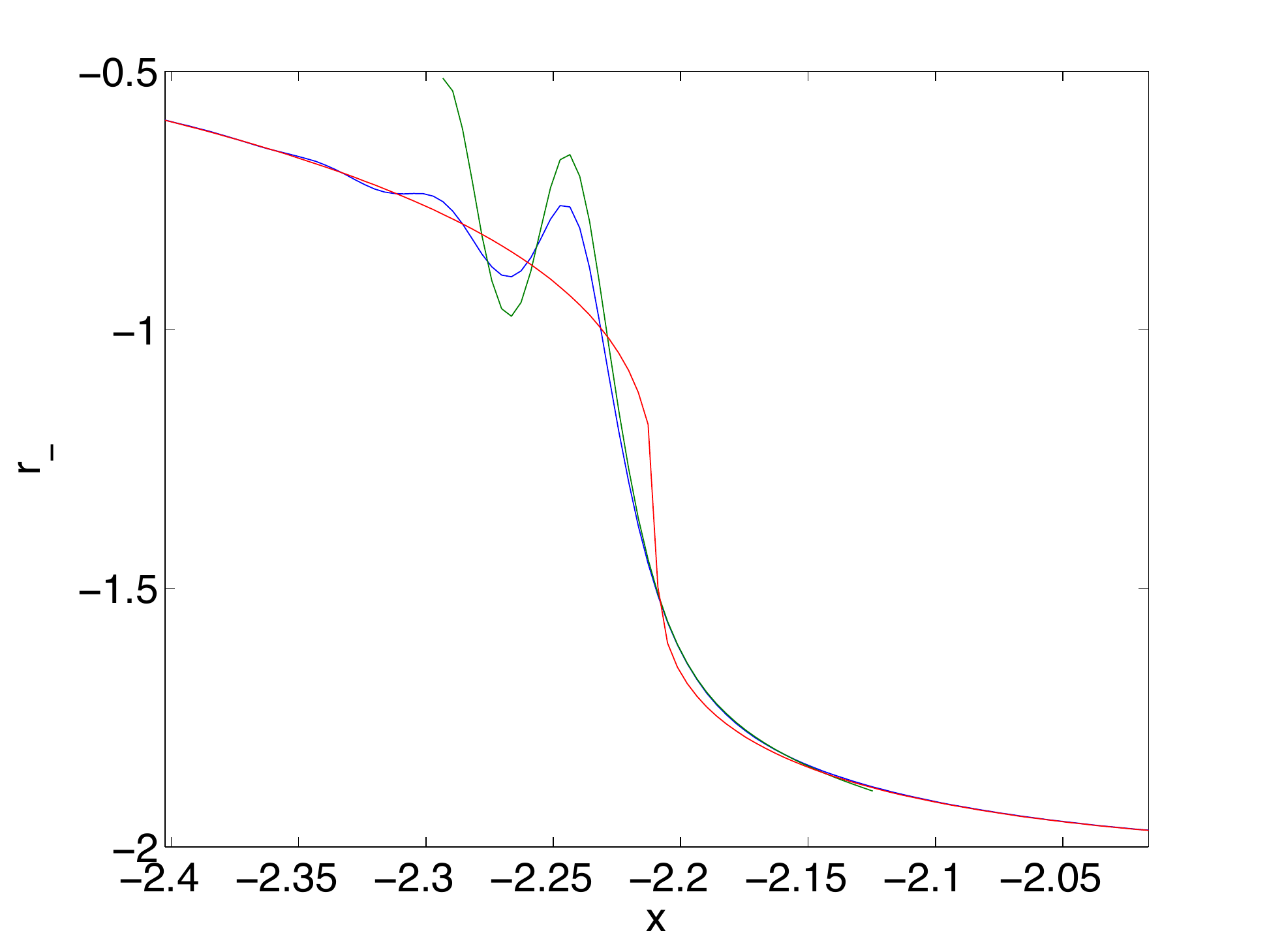}
  \includegraphics[width=0.5\textwidth]{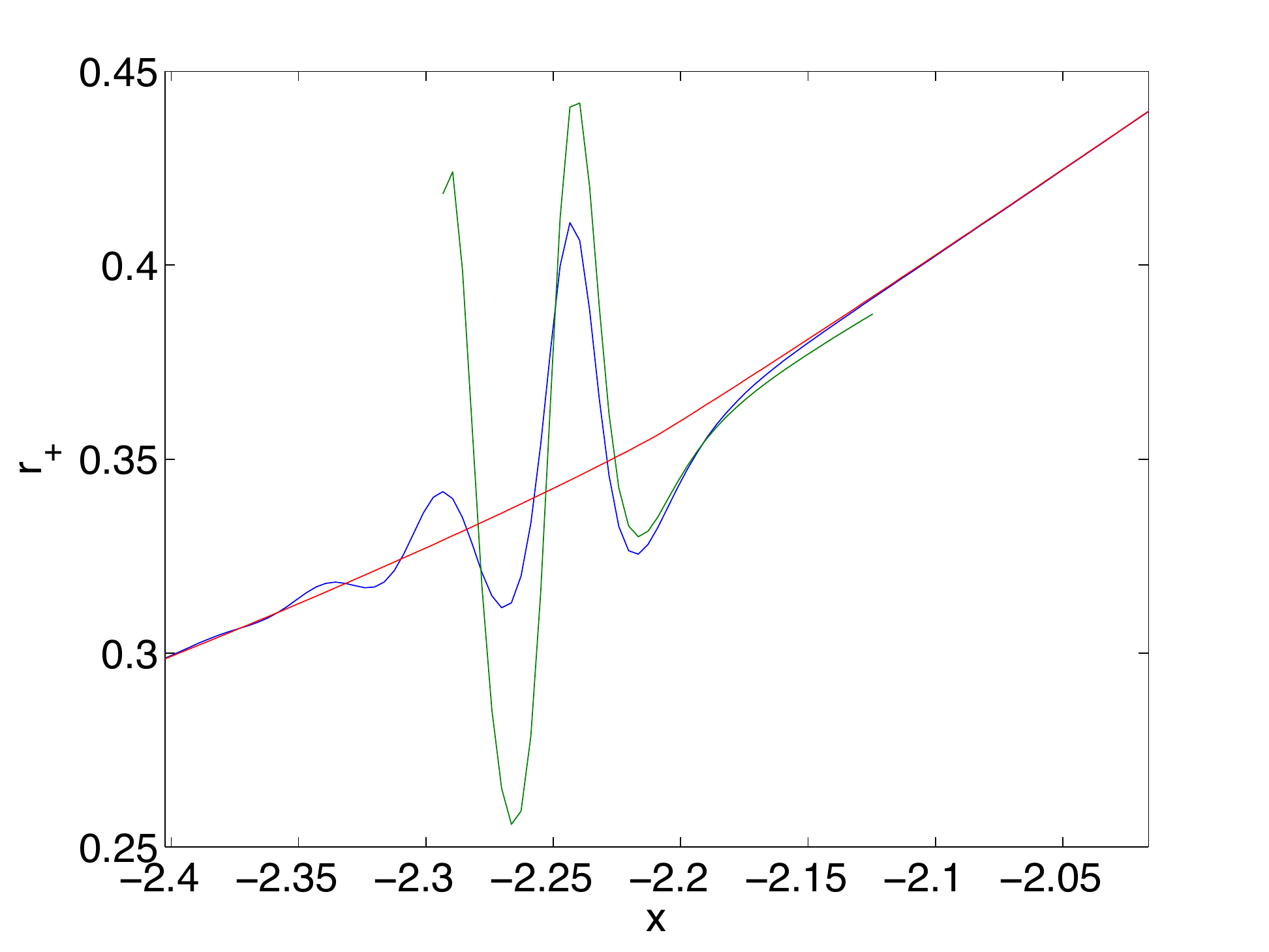}}
\vspace{5mm}
 \subfigure
 {\includegraphics[width=0.5\textwidth]{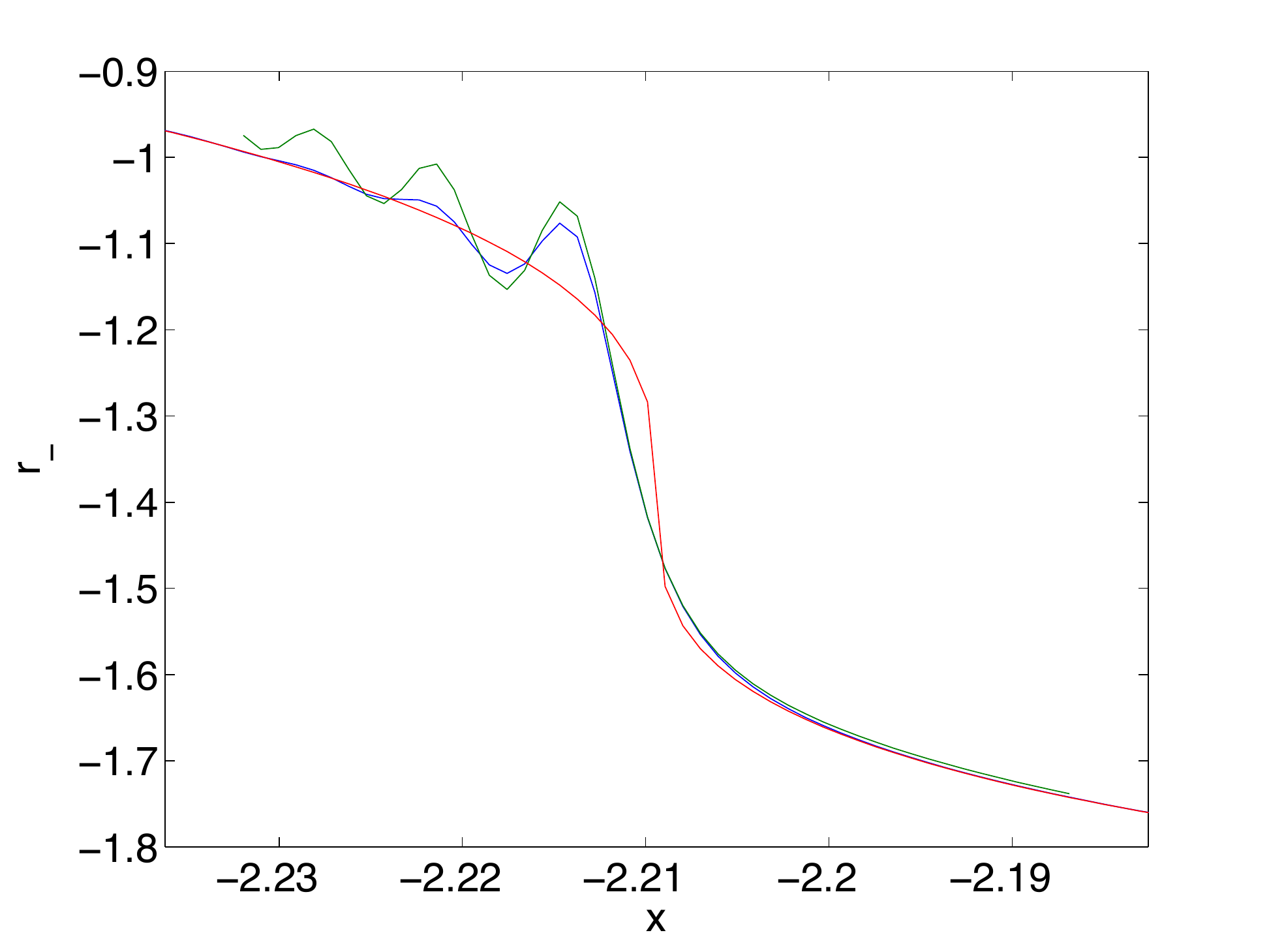}
  \includegraphics[width=0.5\textwidth]{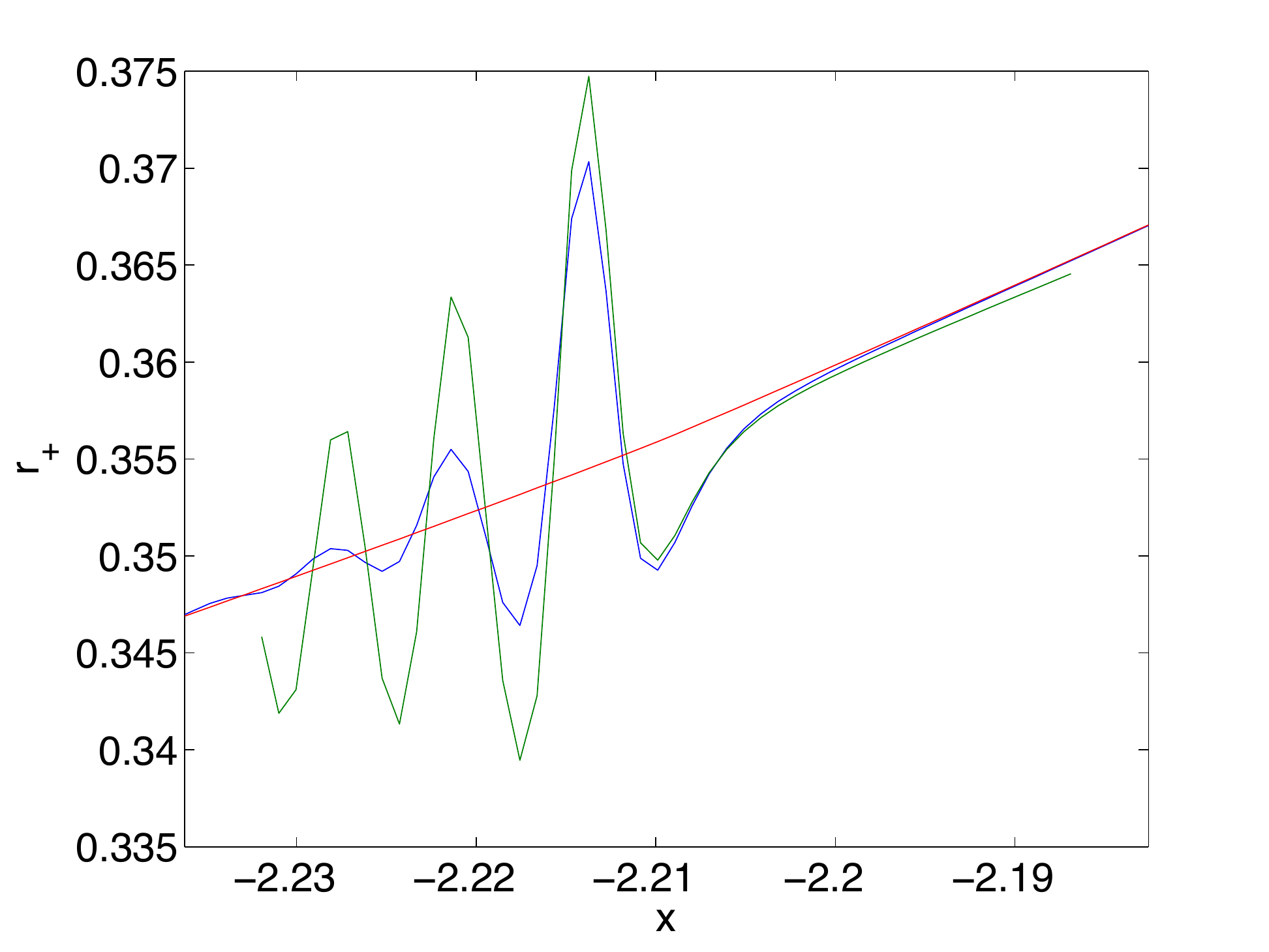}}
\caption{Solution to the defocusing cubic NLS equation for the initial data 
 $\psi_{0}(x)=\mbox{sech }x$  at the critical time 
 $t_0$ in blue,  the corresponding semiclassical solution in red 
 and the P$_{I}^2$ solution (\ref{conj1}) in green;  on the left the 
 Riemann invariant $r_{-}$, on the right the invariant $r_{+}$. The 
 upper figures are for  $\epsilon=0.01$, the lower ones for
   $\epsilon=0.001$.}
 \label{nlscubicc01r}
\end{figure}

For smaller $\epsilon$ the agreement of NLS and semiclassical 
solution becomes better. We show the Riemann invariants $r_{\pm}$ for 
$\epsilon=10^{-3}$ in Fig.~\ref{nlscubicc01r}. Notice that 
there are also oscillations in the invariant $r_{+}$ which stays smooth at 
this point in the semiclassical limit. 

\subsection{sech$\,x$ initial data for the defocusing quintic NLS equation}

We will first study the initial data $\psi_{0}(x)=\mbox{sech }x$ for 
values of $\epsilon$ of $0.1$, $0.09$,\ldots, $0.01$, $0.009$,\ldots, $0.001$. In 
this case there are two break-up points at $\pm x_{c}$ with $x_{c}= 
\ln((\sqrt{3}+1)/\sqrt{2}) + \sqrt{3}/2)\sim1.5245$ at the same time 
$t_0=3\sqrt{3}/4\sim1.2990$. The solution up to the critical time 
can be seen in Fig.~\ref{nlsquintd}, where the defocusing effect of 
the equation can be recognized. The critical value of the Riemann 
invariants at the respective break-up point is $\pm2/3$. We will 
consider in the following always the break-up for negative values of 
$x$ where the Riemann invariant $r_{-}=v-u$ has a gradient 
catastrophe. 
\begin{figure}[thb!]
  \includegraphics[width=0.7\textwidth]{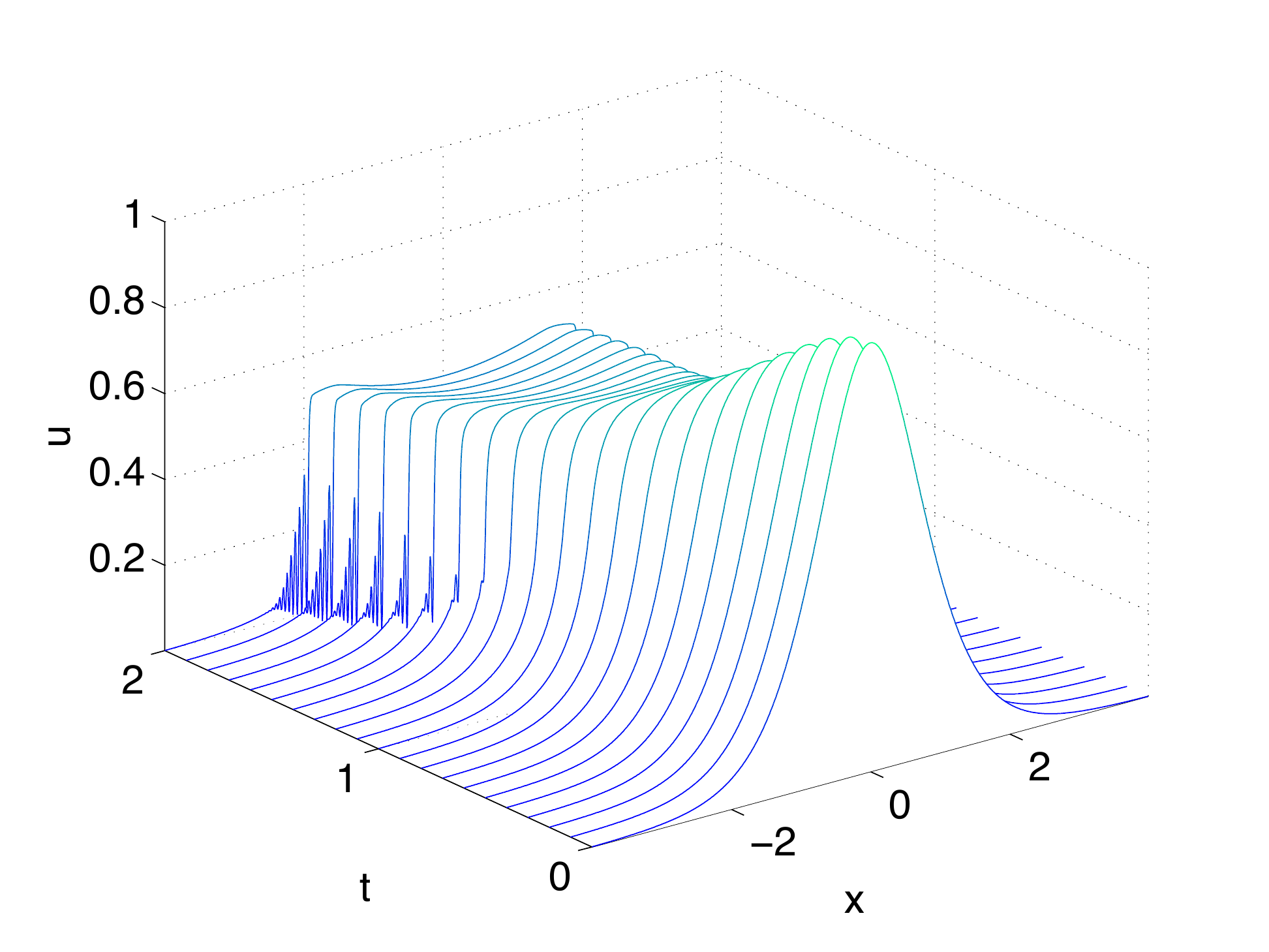}
 \caption{Solution to the defocusing quintic NLS equation for the initial data 
 $\psi_{0}(x)=\mbox{sech }x$ and $\epsilon=0.01$. The critical time is
 $t_0\sim1.2990$.}
 \label{nlsquintd}
\end{figure}

At the critical time the difference of the 
Riemann invariants $r_{-}$ between the 
semiclassical solution and the solution to the focusing quintic NLS 
scales roughly as $\epsilon^{2/7}$. More precisely we find via a linear
regression analysis for the logarithm of the difference $\Delta_{-}$
between NLS and semiclassical solution  a scaling
of the form $\Delta\propto \epsilon^{a}$ with $a=0.2952$ 
($2/7\sim0.2857$) with
standard deviation $\sigma_{a}=0.0017$ and correlation coefficient 
$r=0.9999$. At the same point the difference $\Delta_{+}$ between the 
Riemann invariants $r_{+}=v+u$ between the semiclassical and the NLS 
solution scales roughly as $\epsilon^{4/7}$ as predicted by the 
theory. A linear
regression analysis for the logarithm of the difference $\Delta_{+}$
gives  a scaling
of the form $\Delta\propto \epsilon^{a}$ with $a=0.5988$ 
($4/7\sim0.5714$) with
standard deviation $\sigma_{a}=0.0053$ and correlation coefficient 
$r=0.9998$. 

In Fig.~\ref{nlsquintdsechc01}, the NLS solution, the semiclassical 
solution and the P$_{I}^2$ solution (\ref{conj1}) can be seen at the 
critical time close to the critical point of the semiclassical 
solution. 
\begin{figure}[thb!]
  \includegraphics[width=0.5\textwidth]{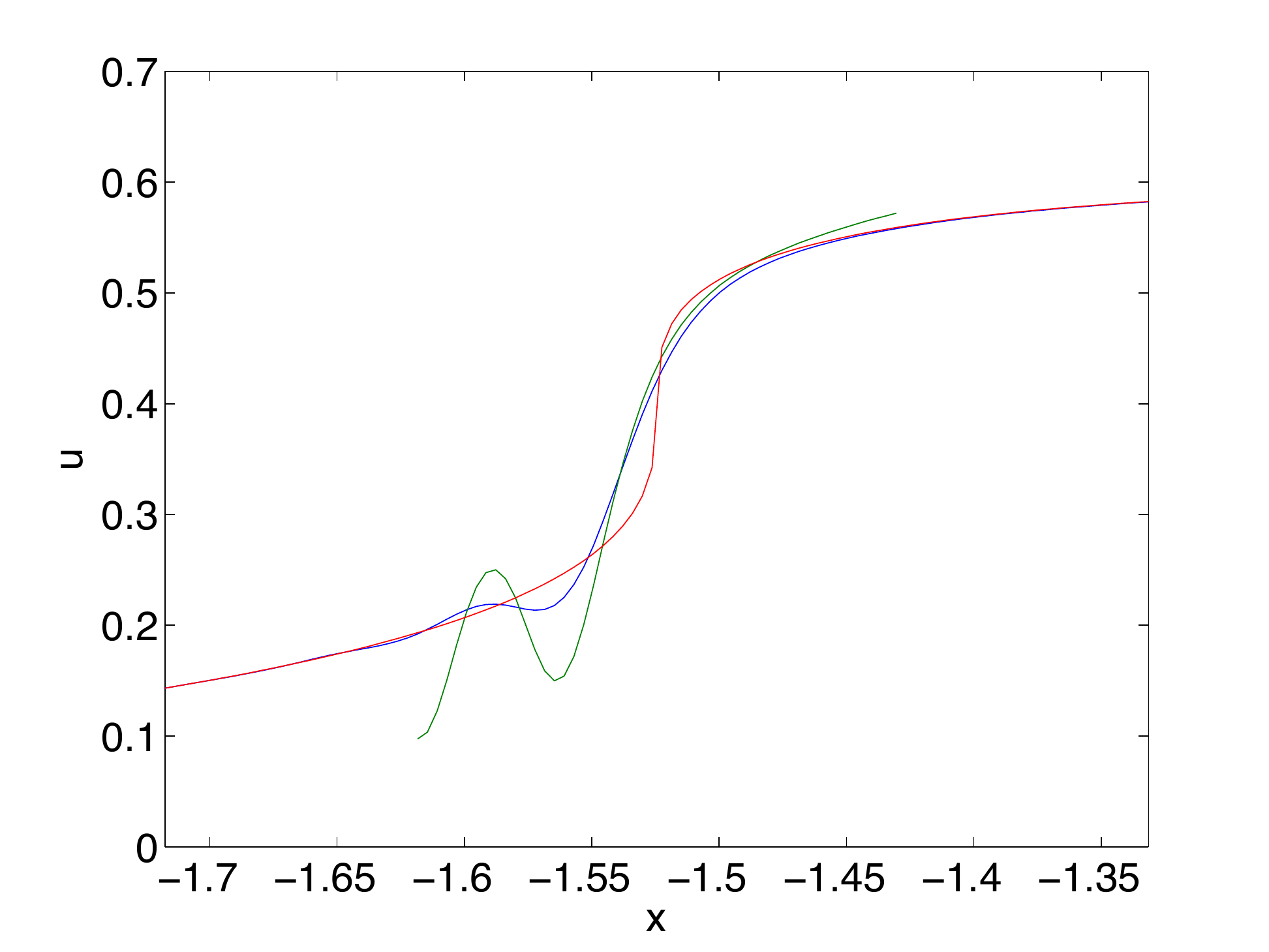}
  \includegraphics[width=0.5\textwidth]{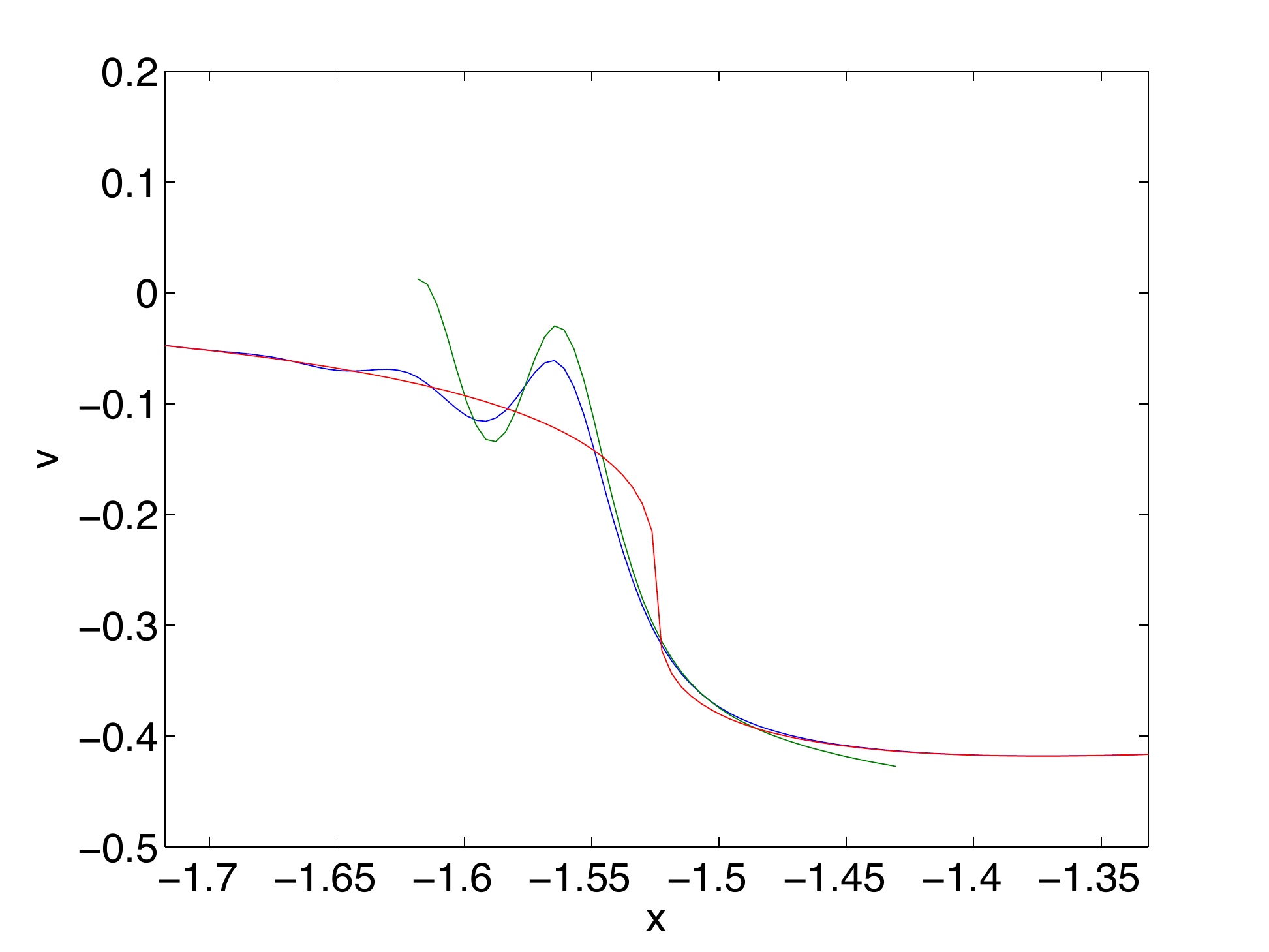}
 \caption{Solution to the defocusing quintic NLS equation for the initial data 
 $\psi_{0}(x)=\mbox{sech }\,x$ and $\epsilon=0.01$ at the critical time 
 $t_0$ in blue,  the corresponding semiclassical solution in red 
 and the P$_{I}^2$ solution (\ref{conj1}) in green; on the left the 
 function $u$, on the right the function $v$.}
 \label{nlsquintdsechc01}
\end{figure}

The corresponding Riemann invariants can be seen in 
Fig.~\ref{nlsquintdsechcr01}.

\begin{figure}[thb!]
\subfigure
{ \includegraphics[width=0.5\textwidth]{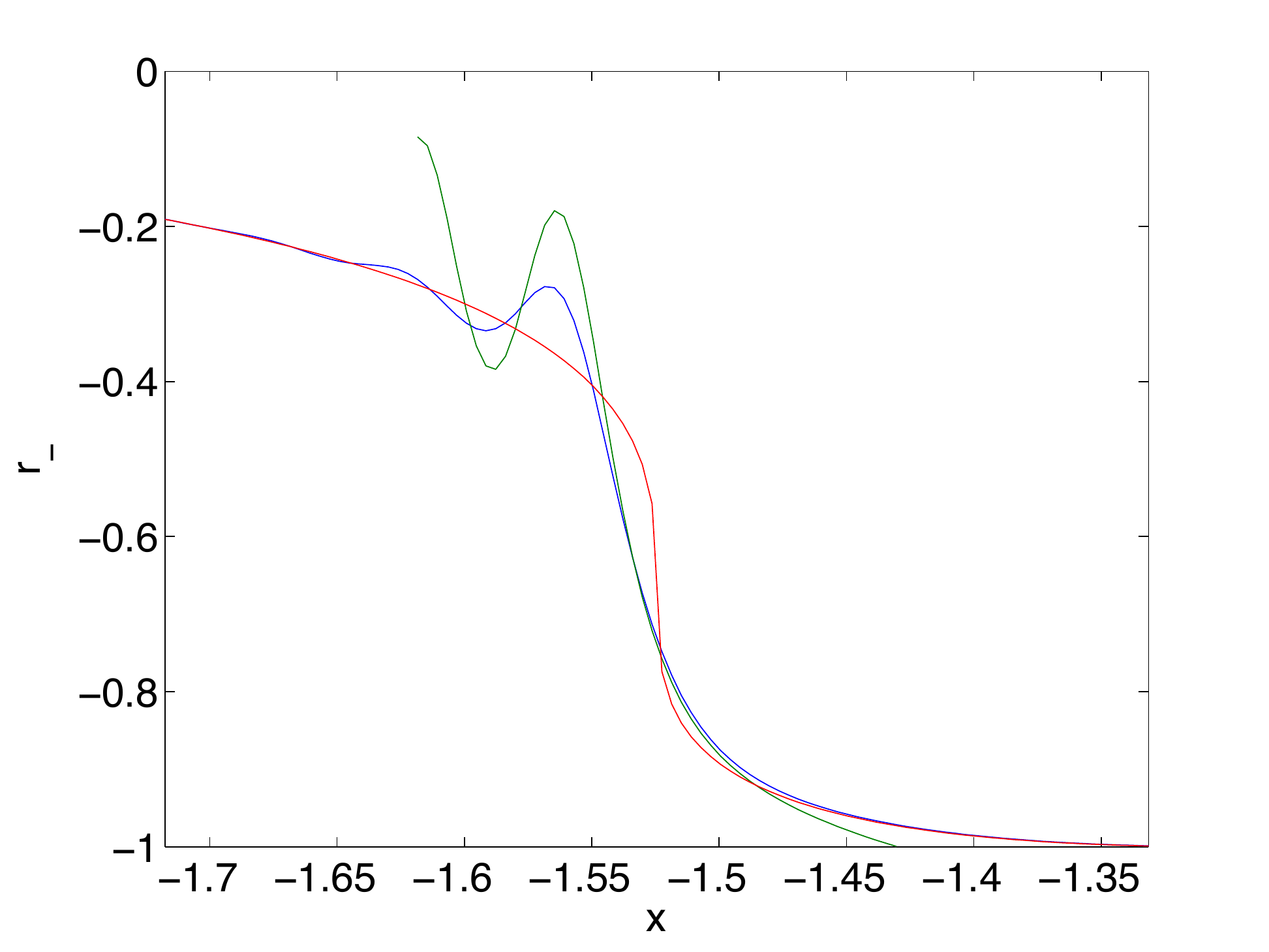}
  \includegraphics[width=0.5\textwidth]{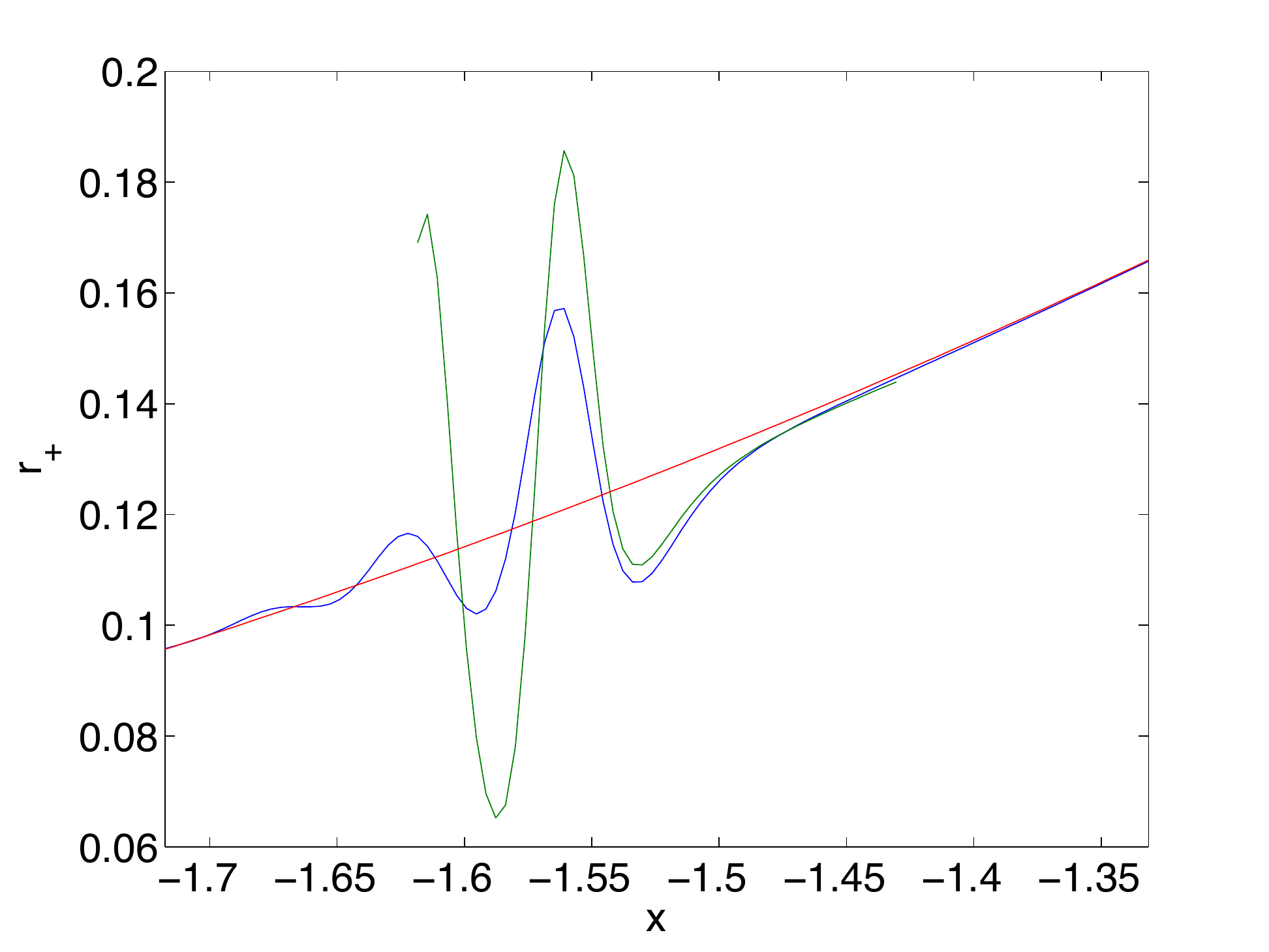}}
\vskip 0.1cm
\subfigure
{\includegraphics[width=0.5\textwidth]{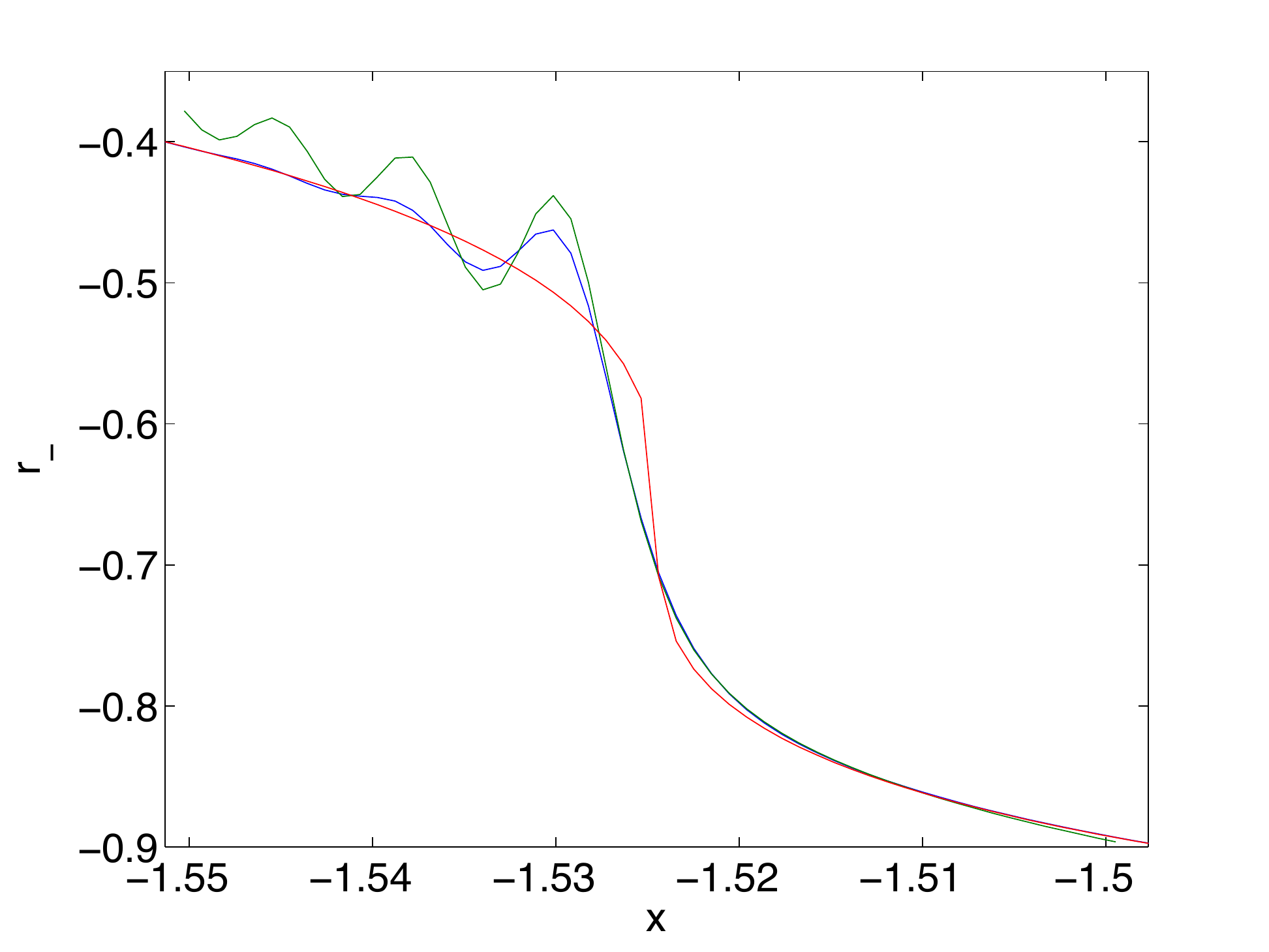}
  \includegraphics[width=0.5\textwidth]{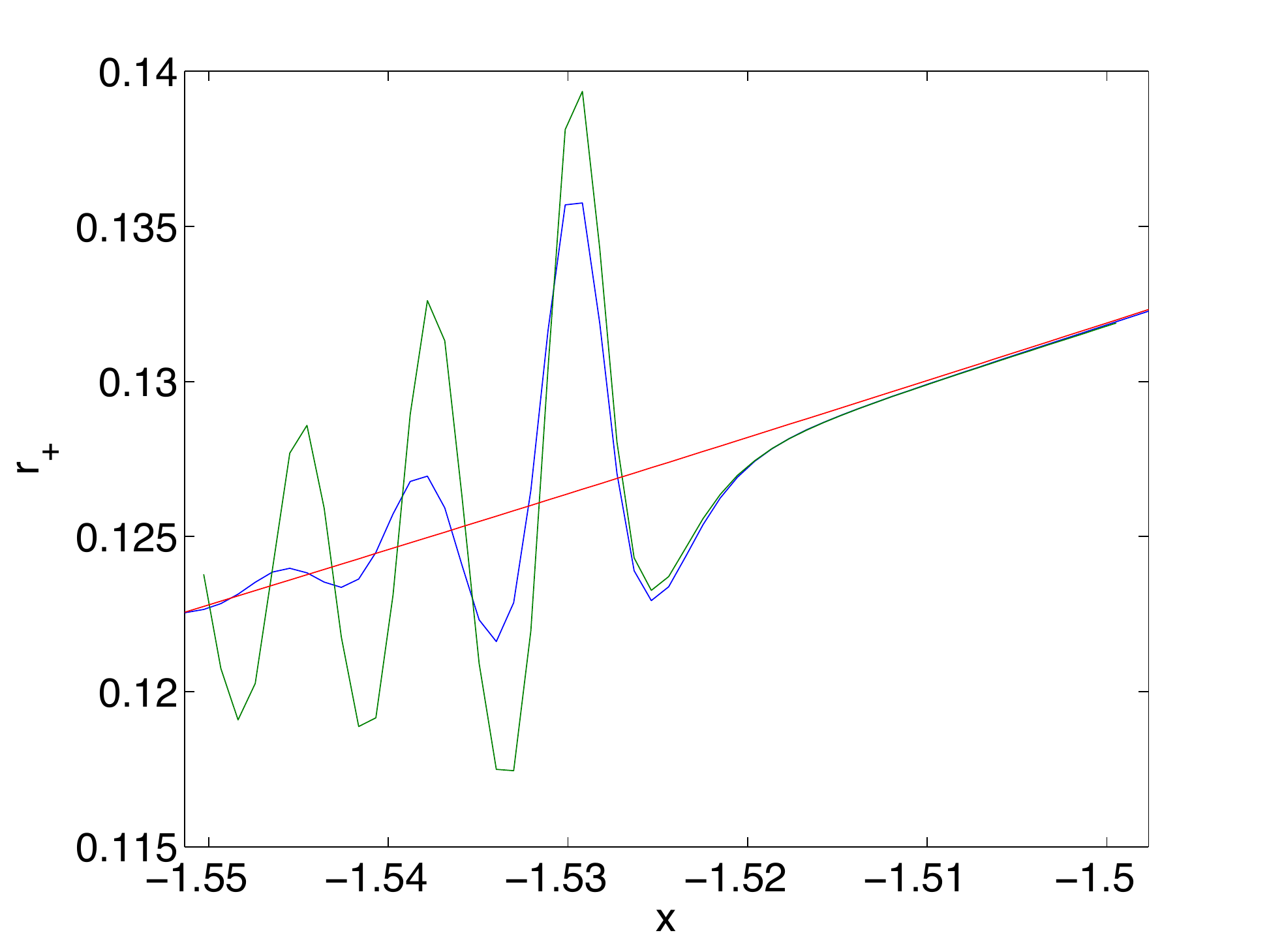}}
 \caption{Solution to the defocusing quintic NLS equation for the initial data 
 $\psi_{0}(x)=\mbox{sech }x$  at the critical time 
 $t_0$ in blue,  the corresponding semiclassical solution in red 
 and the P$_{I}^2$ solution (\ref{conj1}) in green; on the left the 
 Riemann invariant $r_{-}$, on the right the invariant $r_{+}$. The 
 upper figures are for  $\epsilon=0.01$, the lower ones for $\epsilon=0.001$.}
 \label{nlsquintdsechcr01}
\end{figure}

For smaller $\epsilon$ the agreement of NLS and semiclassical 
solution becomes better. We show the Riemann invariants $r_{\pm}$ for 
$\epsilon=10^{-3}$ in  Fig.~\ref{nlsquintdsechcr01}. Note that 
there are also oscillations in the invariant $r_{-}$ which stays smooth at 
this point in the semiclassical limit. 

The P$_{I}^2$ solution (\ref{conj1}) gives a much better agreement 
with the NLS solution close to the critical point as can be seen in 
Fig.~\ref{nlsquintdsechc01} and \ref{nlsquintdsechcr01}. The 
agreement is in fact so good that the difference of the solutions has 
to be studied. The P$_{I}^2$ solution only 
gives locally an asymptotic description, at larger distances from 
the critical point the semiclassical solution provides a better 
description as can be also seen from Fig.~\ref{nlsquintdsechcdelta}.
\begin{figure}[htb!]
  \includegraphics[width=0.5\textwidth]{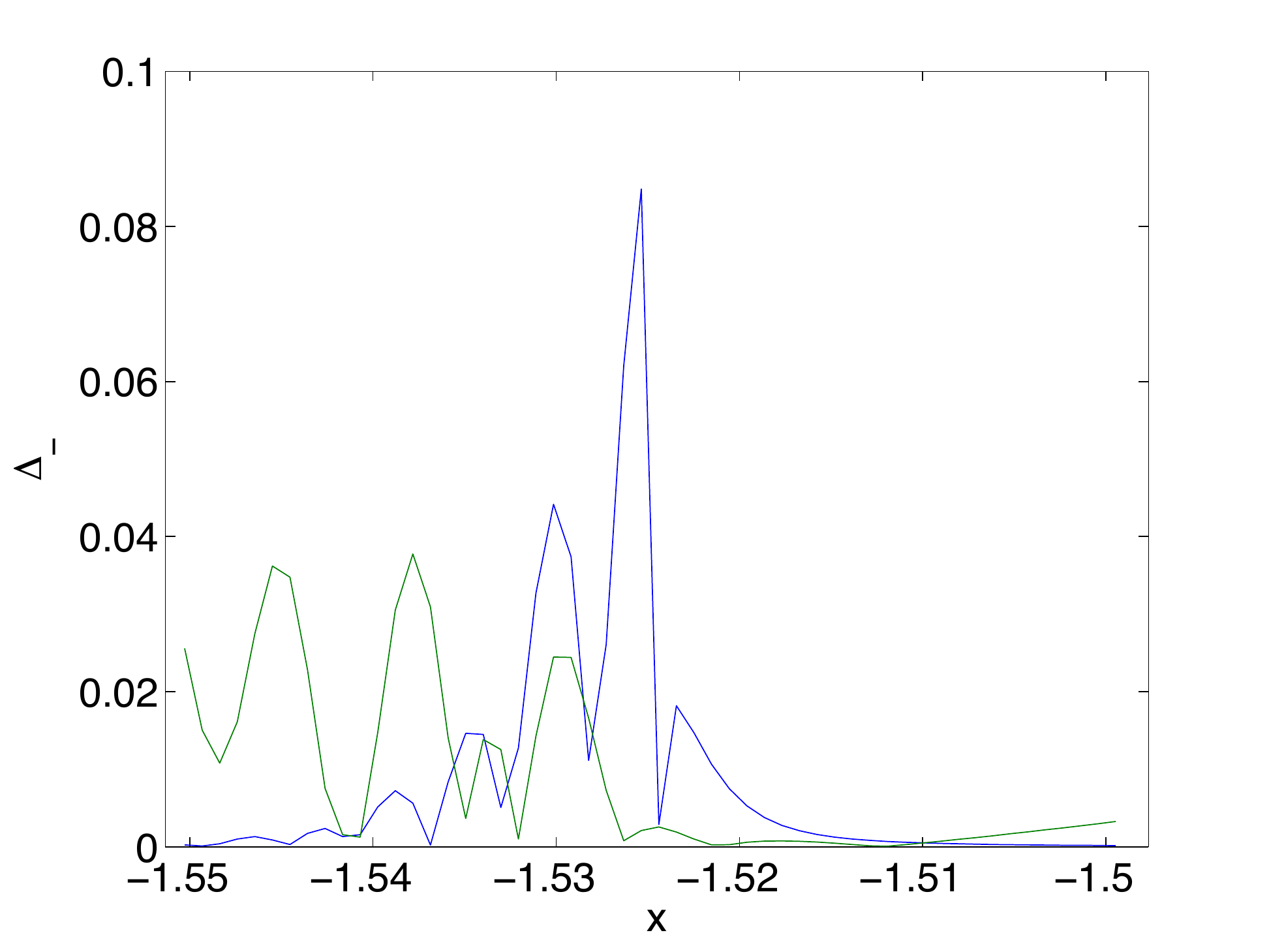}
  \includegraphics[width=0.5\textwidth]{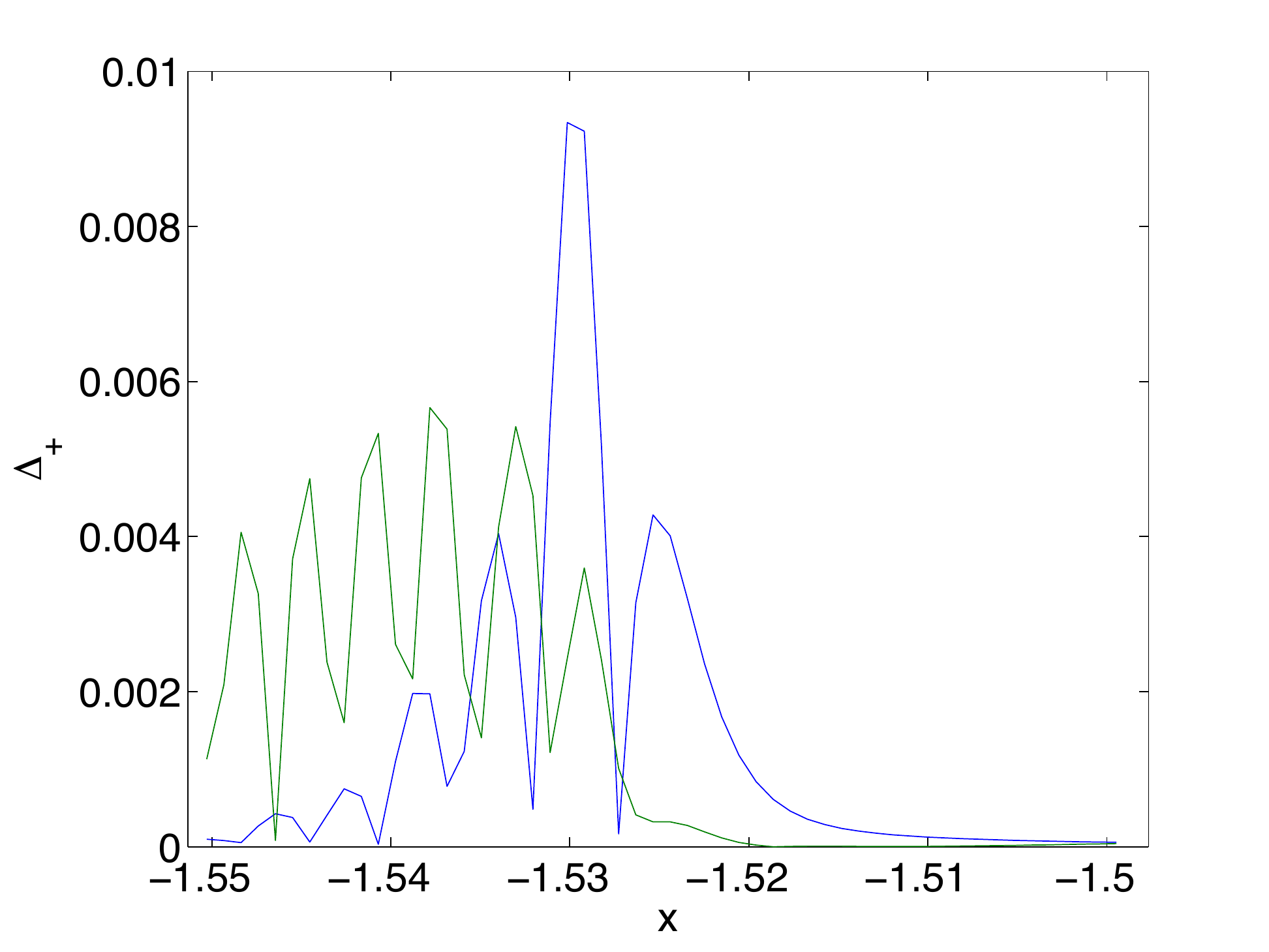}
 \caption{The modulus of the difference af the Riemann invariants for the
 defocusing quintic NLS equation for the initial data 
 $\psi_{0}(x)=\mbox{sech }x$ for $\epsilon=0.01$ at the critical time 
 $t_0$ and the semiclassical solution in blue, and the difference 
 between the corresponding P$_{I}^2$ solution (\ref{conj1}) and the NLS solution 
 in green; on the left the invariant that has a break-up in the 
 semiclassical limit, on the right 
 the invariant that stays smooth.}
   \label{nlsquintdsechcdelta}
\end{figure}

We can identify the regions where each of the asymptotic
solutions gives a better description of NLS than the other by 
identifying the values  $x_{l},x_{r}$  such that for 
all $x_{l}<x<x_{r}$ the P$_{I}^2$ solution provides a better 
asymptotic description than the semiclassical  
solution. Due to the oscillatory character of the NLS 
and the P$_{I}^2$ solution (\ref{conj1}), such a definition leads to 
ambiguities and  oscillations also in the boundaries of these zones 
for $r_{\pm}$. No clear scaling could thus be identified for these 
limits. The oscillatory character of the solution also implies there 
is no obvious scaling of the maximal error in the 
asymptotic description for the values of $\epsilon$ we could treat.  

The matching procedure nonetheless clearly improves the asymptotic description near
the critical point.
In Fig.~\ref{nlsquintdsechcdeltamatch} we see the difference between this matched
asymptotic solution and the NLS solution for two values of
$\epsilon$. Visibly the zone, where the solutions are matched,
decreases with $\epsilon$ (note the rescaling of the $x$-axes with 
a factor
$\epsilon^{6/7}$).
\begin{figure}[htb!]
 \includegraphics[width=0.5\textwidth]{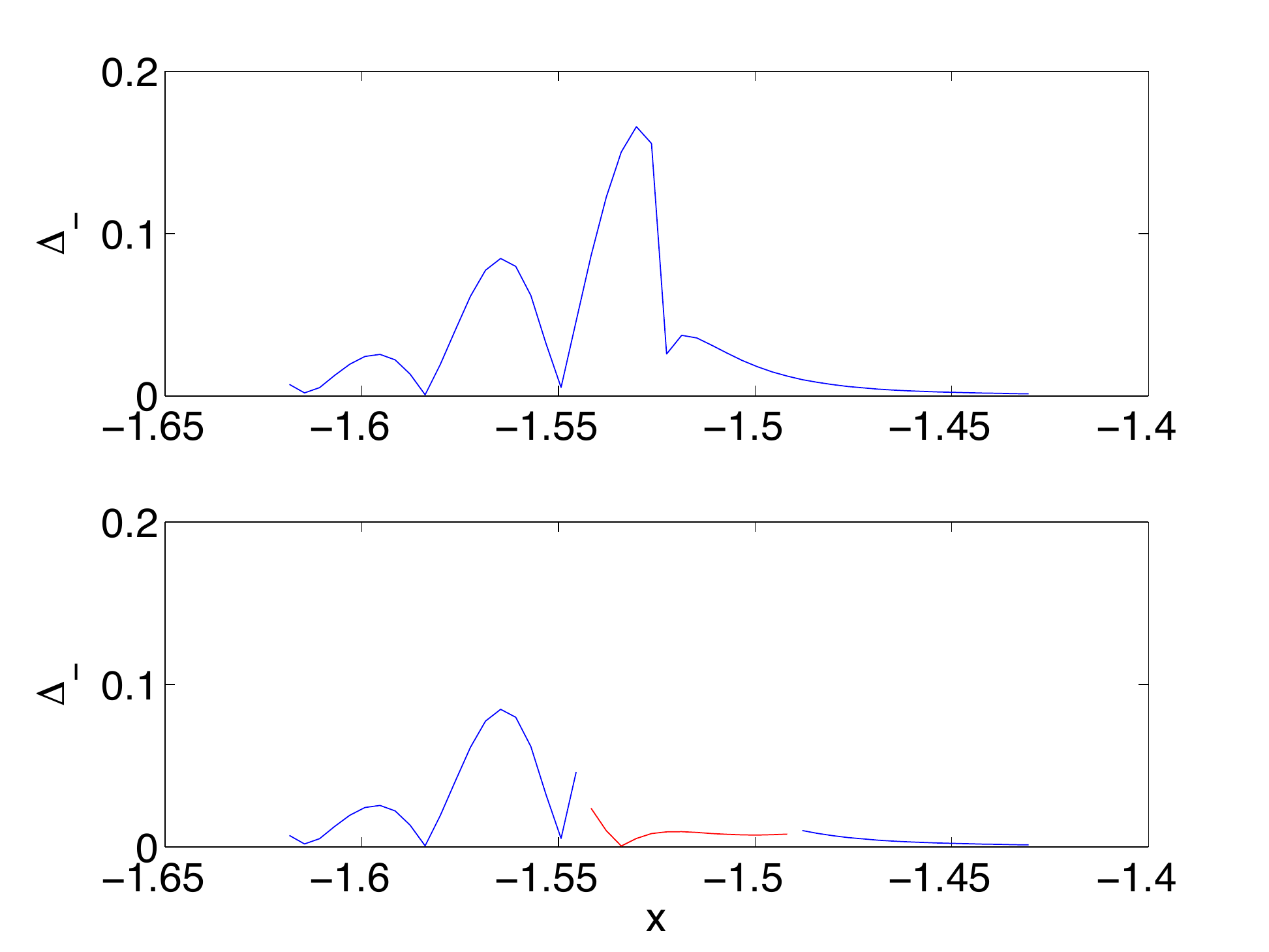}               
 \includegraphics[width=0.5\textwidth]{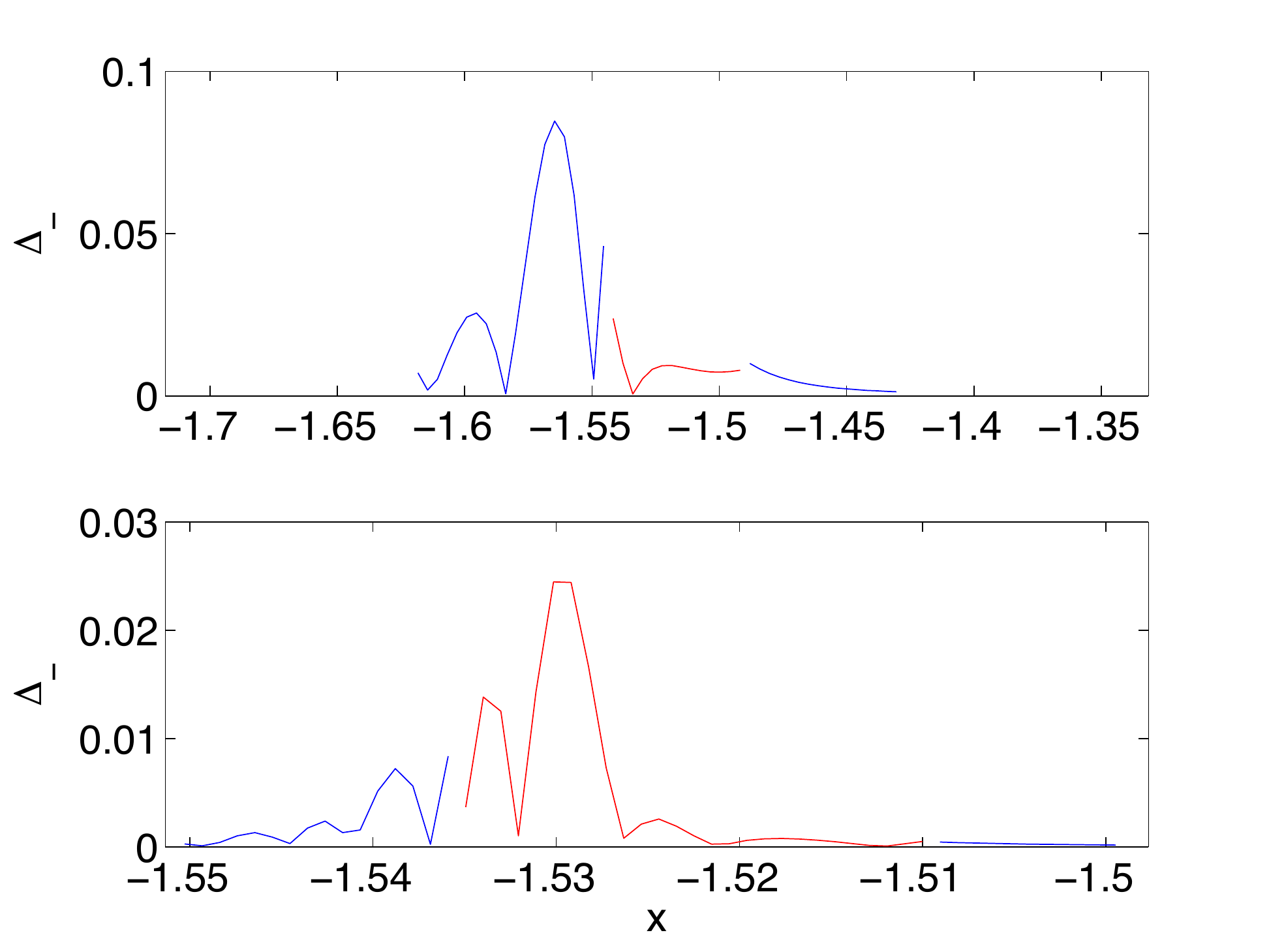}
 \caption{In
the upper part of the left figure  one can see the modulus of the 
difference $\Delta_{-}$ of the Riemann invariant
for the defocusing quintic NLS equation for the initial data 
 $\psi_{0}(x)=\mbox{sech }x$  at the critical time 
 $t_0$ and the semiclassical solution for $\epsilon=0.01$. 
 The lower part shows the same
difference, which is replaced close to the critical point 
by the difference between NLS solution and  the 
P$_{I}^2$ solution (\ref{conj1}) (in red where the error is smaller 
than the one
shown above). The right figure shows the same situation as the lower figure 
on the left  for $\epsilon=0.01$ above and $\epsilon=0.001$ below. The 
$x$-axes are rescaled  by a factor $\epsilon^{6/7}$.}  
\label{nlsquintdsechcdeltamatch}
\end{figure}
The same procedure can be carried out for the invariant $r_{+}$ which 
stays smooth at this point. Obviously the P$_{I}^2$ solution (\ref{conj1})
 provides a description of higher order at this point as 
can be seen in Fig.~\ref{nlsquintdsechcdeltamatch2}. Thus the 
P$_{I}^2$ solution (\ref{conj1}) provides as expected an asymptotic 
description of the oscillations for the Riemann invariant which 
remains smooth in the semiclassical limit.
\begin{figure}[htb!]
 \includegraphics[width=0.5\textwidth]{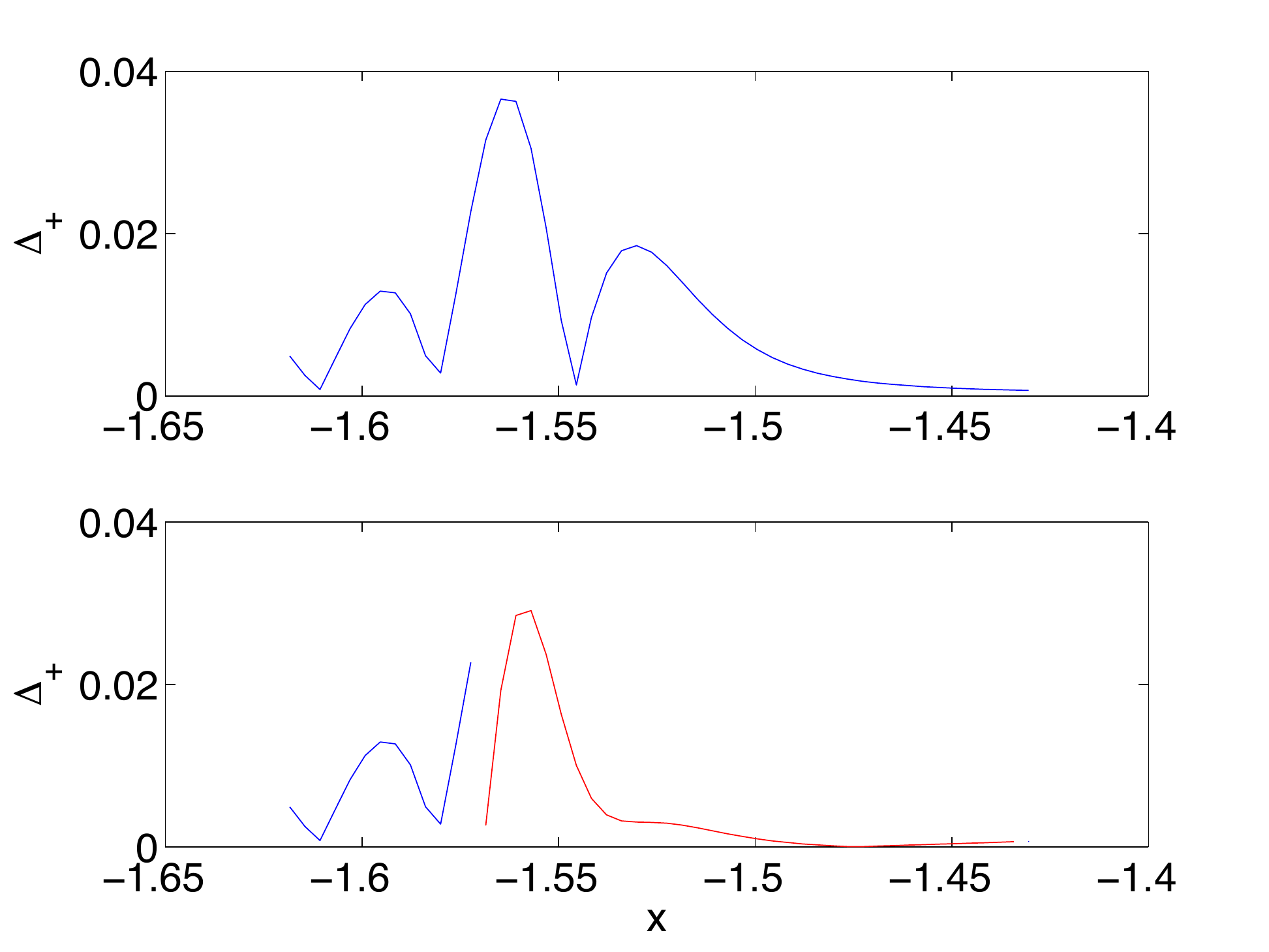}               
 \includegraphics[width=0.5\textwidth]{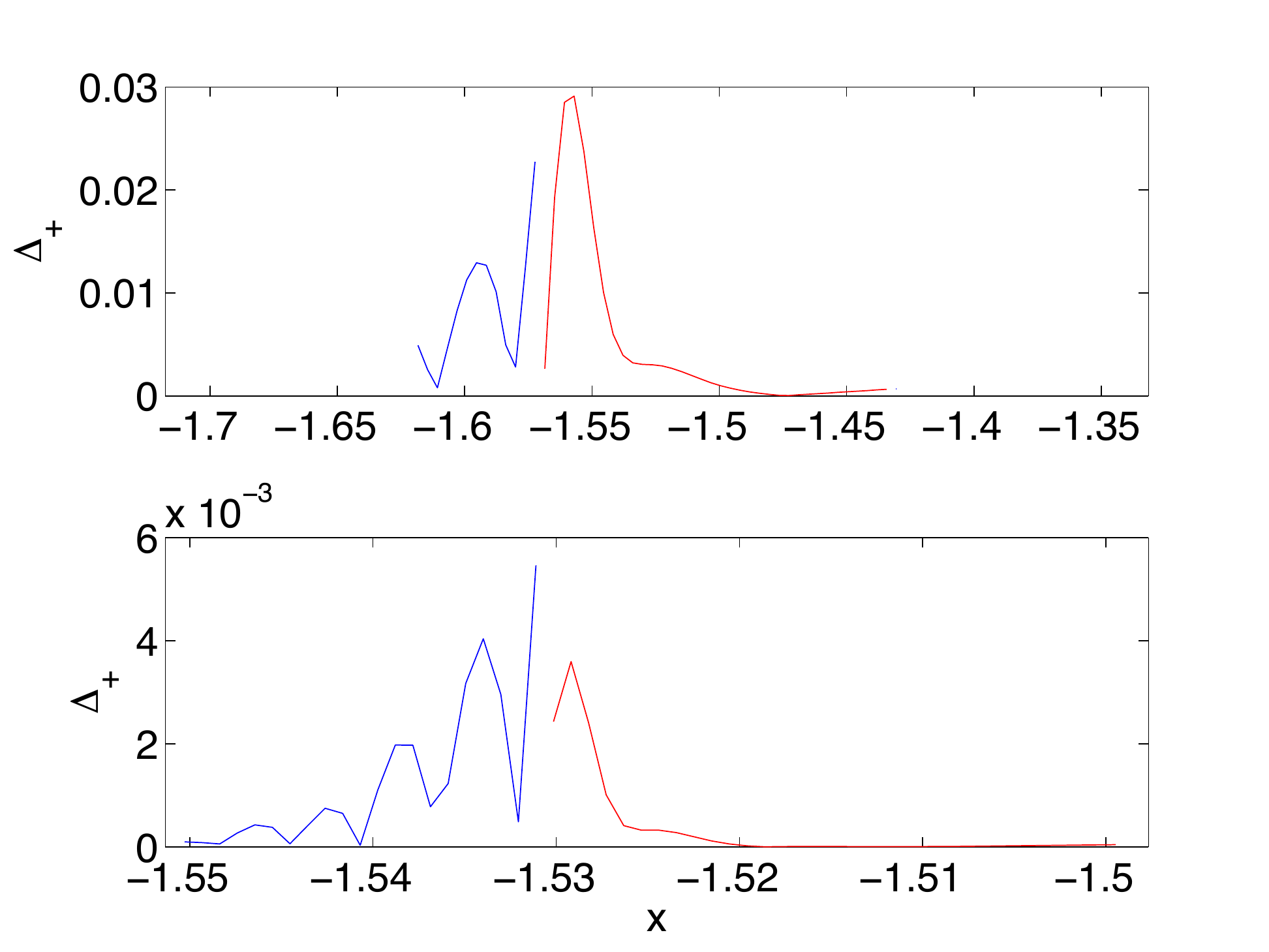}
 \caption{In
the upper part of the left figure  one can see the modulus of the 
difference $\Delta_{+}$ of the Riemann invariant
for the defocusing quintic NLS equation for the initial data 
 $\psi_{0}(x)=\mbox{sech }x$  at the critical time 
 $t_0$ and the semiclassical solution for $\epsilon=0.01$. 
 The lower part shows the same
difference, which is replaced close to the critical point 
by the difference between NLS solution and  the 
P$_{I}^2$ solution (\ref{conj1}) (in red where the error is smaller 
than the one
shown above). The right figure shows the same situation as the lower figure 
on the left  for $\epsilon=0.01$ above and $\epsilon=0.001$ below. The 
$x$-axes are rescaled  by a factor $\epsilon^{6/7}$.}  
\label{nlsquintdsechcdeltamatch2}
\end{figure}

The P$_{I}^2$ solution (\ref{conj1}) holds for small $|x-x_{c}|$ and 
$|t-t_0|$. To illustrate the latter effect, we compare it with the 
NLS solution for the times $t_{\pm}=t_0\pm0.0027$. Note that 
$t-t_0$ appears in the formula (\ref{conj1}) for the P$_{I}^2$ 
solution at several places with different powers of $\epsilon$. Thus 
in contrast to the elliptic case (\ref{F1}), there is no simple 
dependence on $t$ in the hyperbolic case. In 
Fig.~\ref{nlsquintdsechdeltam001} we show the quantities $r_{\pm}$ 
at the time $t_{\pm}$. It can be seen that the P$_{I}^2$ solution 
gives again a clearly better asymptotic description near the break-up 
point than the semiclassical solution. 
\begin{figure}[htb!]
\includegraphics[width=0.5\textwidth]{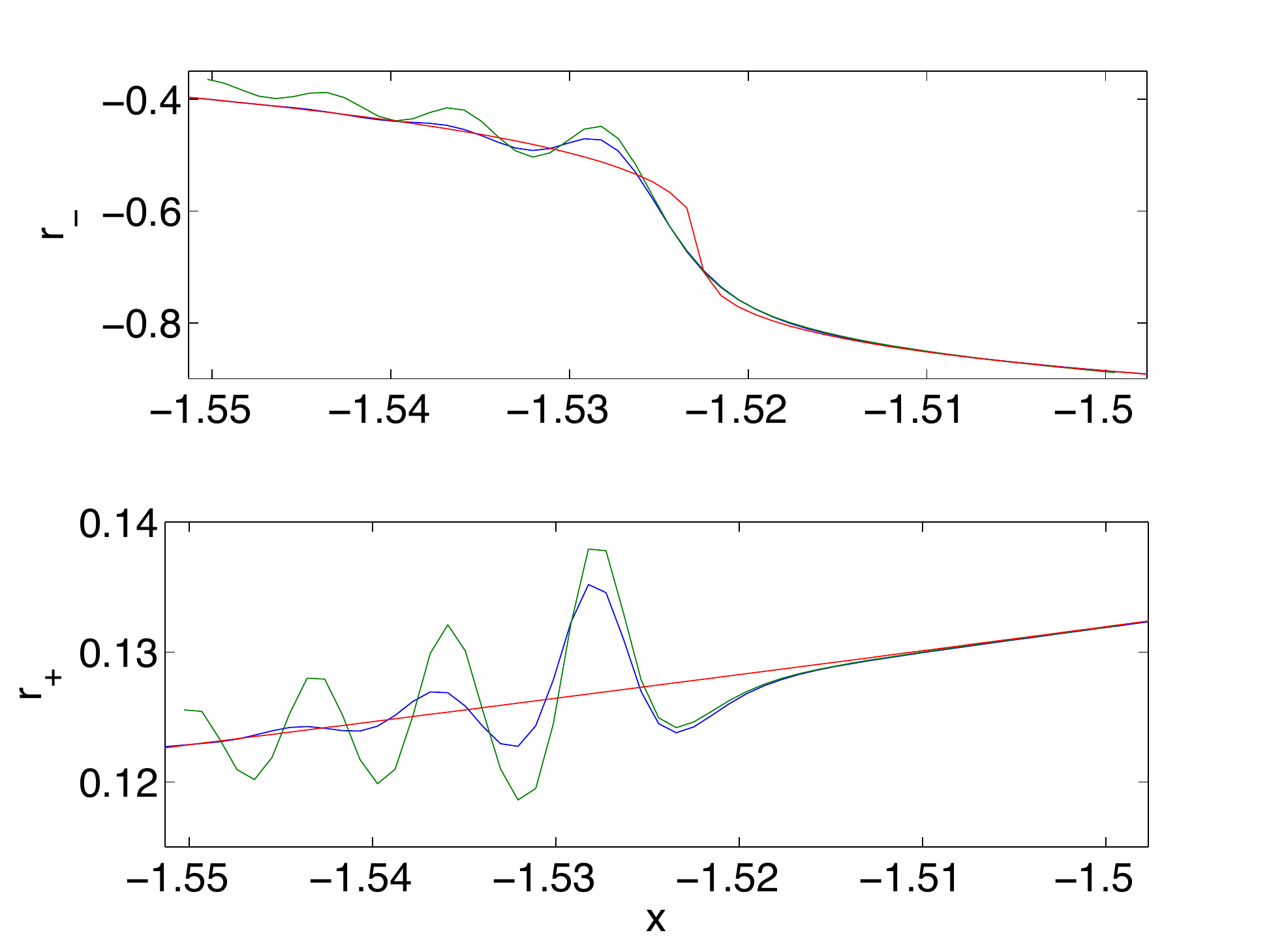}
\includegraphics[width=0.5\textwidth]{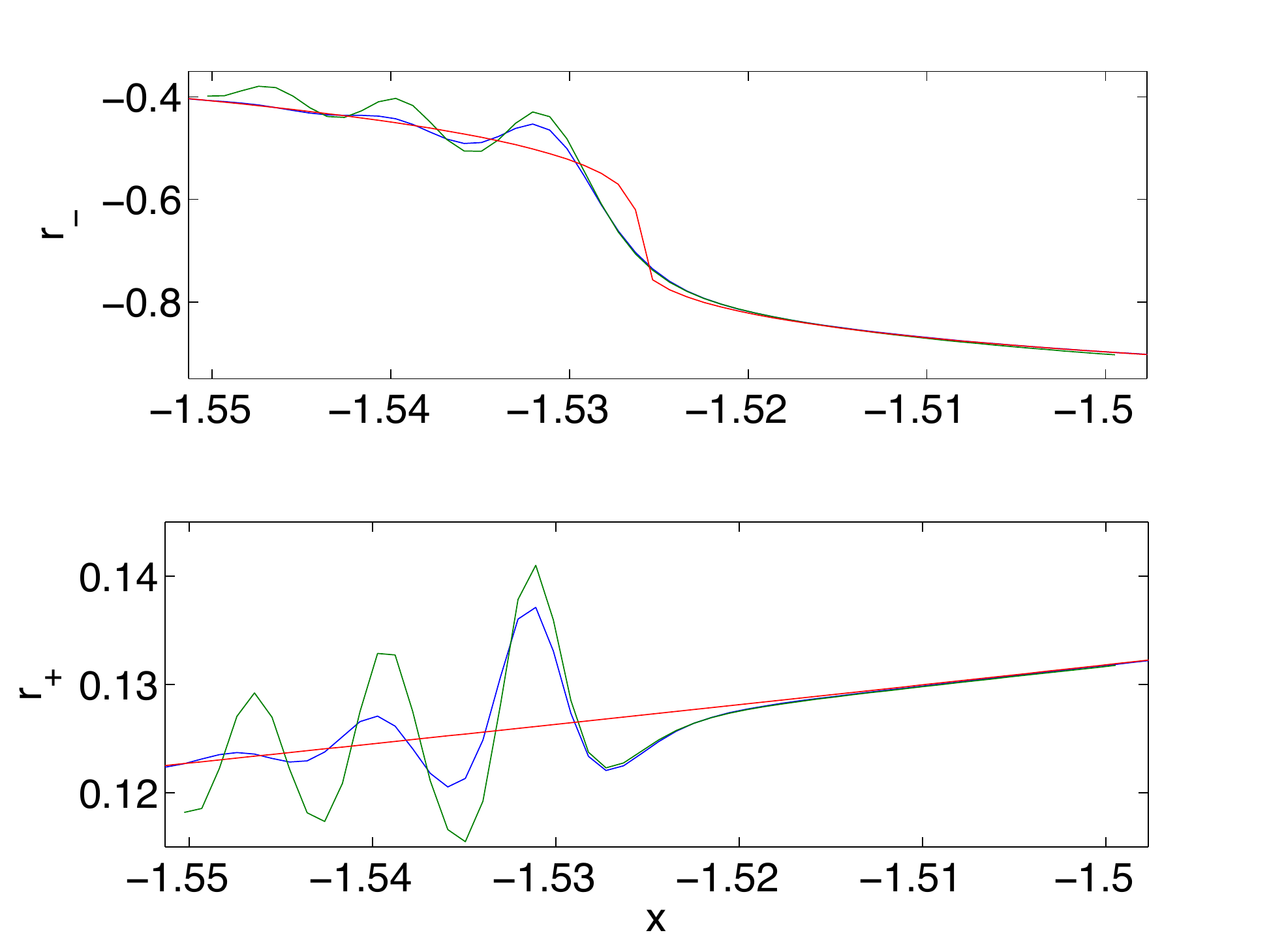}
 \caption{Solution to the defocusing quintic NLS equation for the initial data 
 $\psi_{0}(x)=\mbox{sech }x$ and $\epsilon=0.01$ in blue,  the corresponding semiclassical solution in 
 red and the P$_{I}^2$ solution (\ref{conj1}) in green; above  the 
 function $r_{-}$, below the function $r_{+}$. On the left at the time 
 $t_-=t_0-0.0027$  on the right at the time  $t_+=t_0+0.0027$ }
   \label{nlsquintdsechdeltam001}
\end{figure}


\subsection{`Dark' initial data for the defocusing quintic NLS}
It is well known that the defocusing cubic NLS equation has exact 
solutions called \emph{dark solitons}, i.e., solutions that do not 
tend to zero for $|x|\to\infty$. Such solutions are physically 
problematic since they have infinite energy and are mathematically 
difficult to handle, but they are nonetheless of importance in 
applications. Therefore we will here also study initial data which do 
not decay to zero at spatial infinity. We will consider the example
$\psi_{0}(x)=\tanh^{2}x$ in the following. The time evolution of 
the solution up to the critical time $t_0\sim 1.3448$ can be seen 
in Fig.~\ref{nlsquintddarke01}. The steepening of the two fronts 
of the pulse can be seen as well as the formation of a small 
oscillation on each side.  For times $t\gg t_0$, each of the initial oscillations develops into 
an oscillatory zone which will eventually overlap.
\begin{figure}[htb!]
\includegraphics[width=0.6\textwidth]{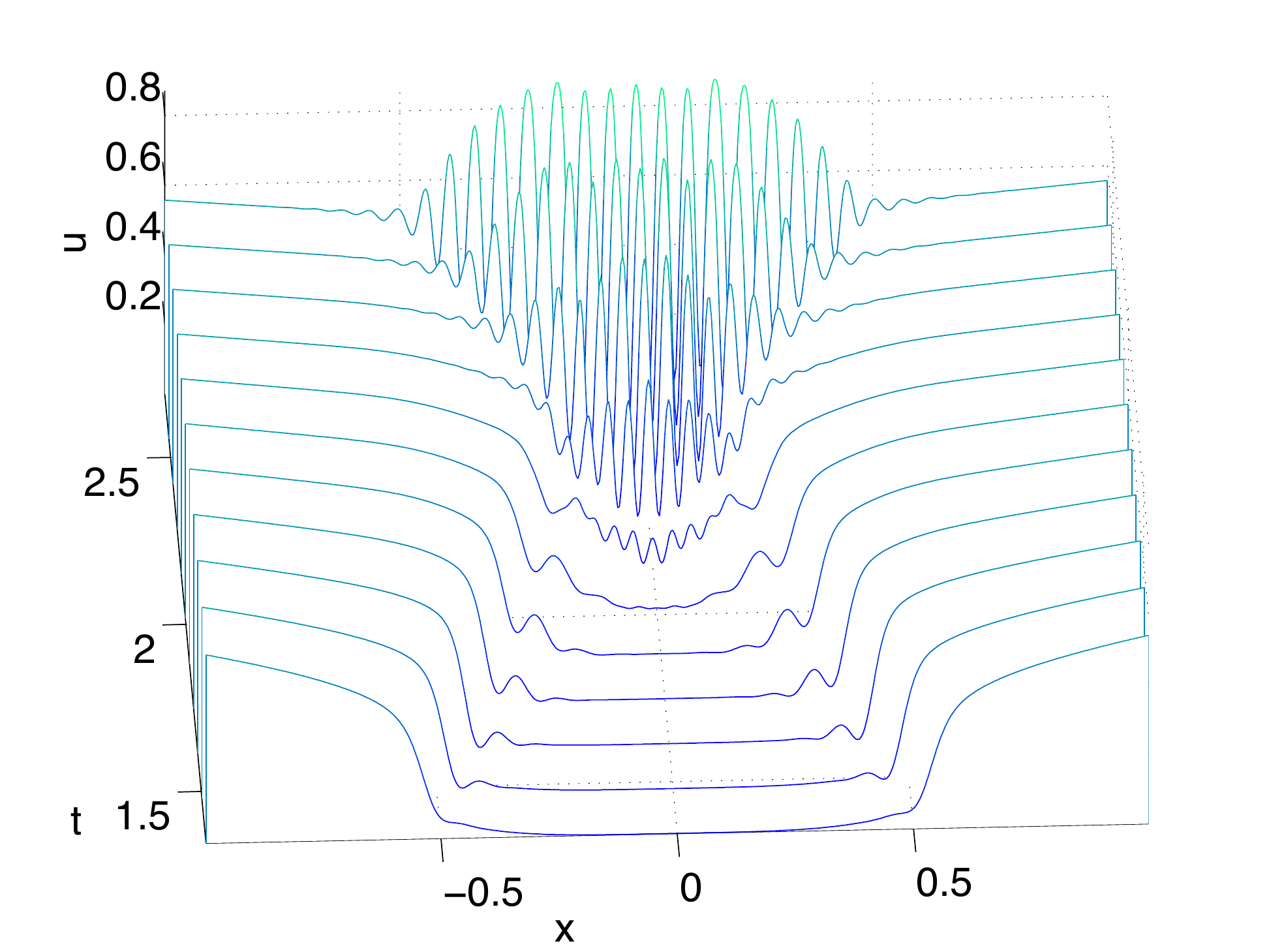}
 \caption{Solution to the defocusing quintic NLS equation for the initial data 
 $\psi_{0}(x)=\tanh^{2}x$ and $\epsilon=0.01$. The critical 
 time is $t_0\sim 1.3448$.}
   \label{nlsquintddarke01}
\end{figure}

Clearly there will be two regions with 
strong gradients symmetric in $x$. We will concentrate on 
positive values of $x$ where the Rieman invariant $r_{-}$ breaks in 
the semiclassical solution.
 In Fig.~\ref{nlsquintddarkce001rm}  the  Riemann invariants for the NLS 
solution, the corresponding semiclassical solution and the P$_{I}^2$ 
asymptotics (\ref{conj1}) can be seen close to $x_{c}\sim0.5476$ for 
$\epsilon=0.001$. 
\begin{figure}[thb!]
  \includegraphics[width=0.5\textwidth]{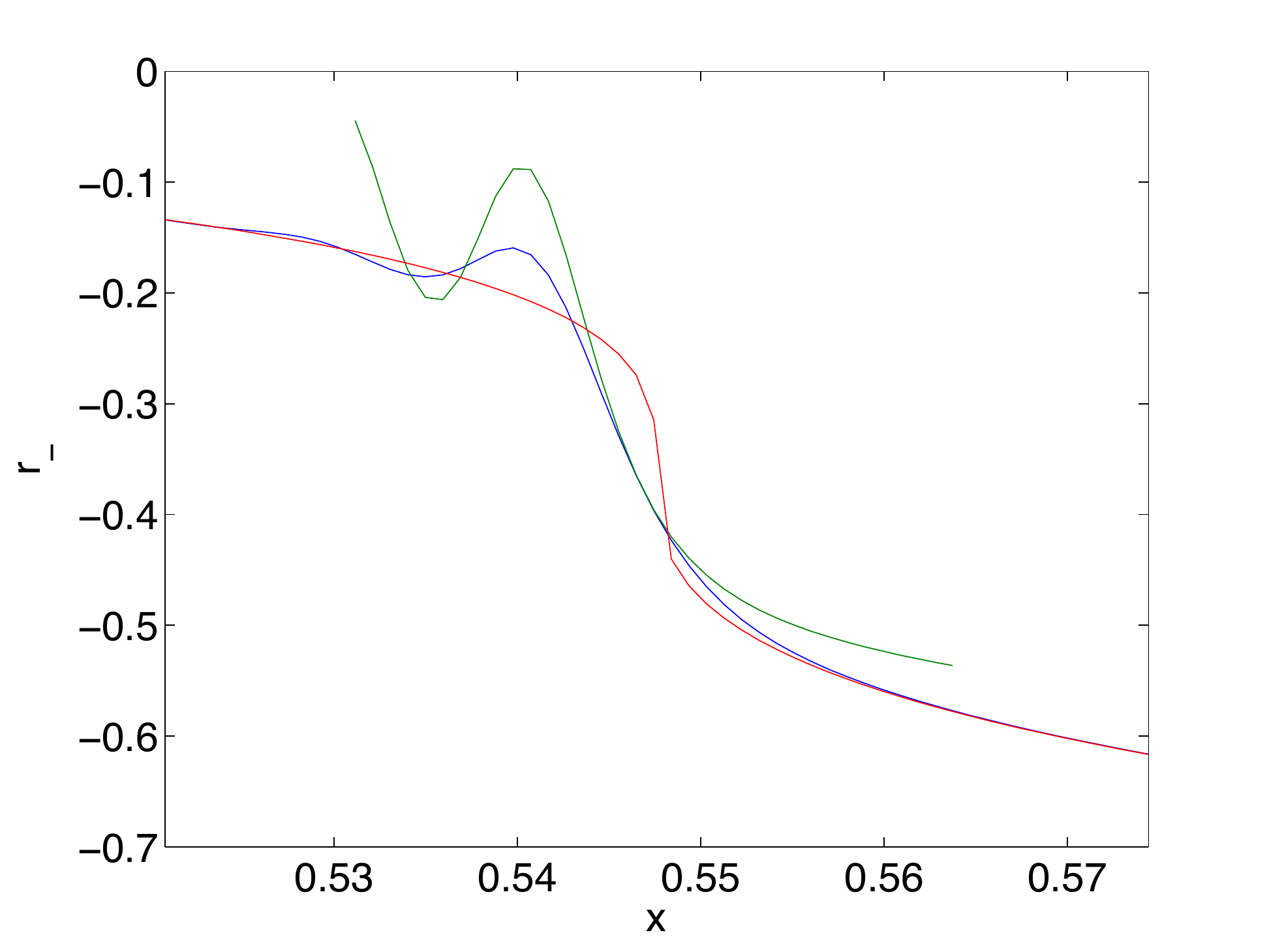}
  \includegraphics[width=0.5\textwidth]{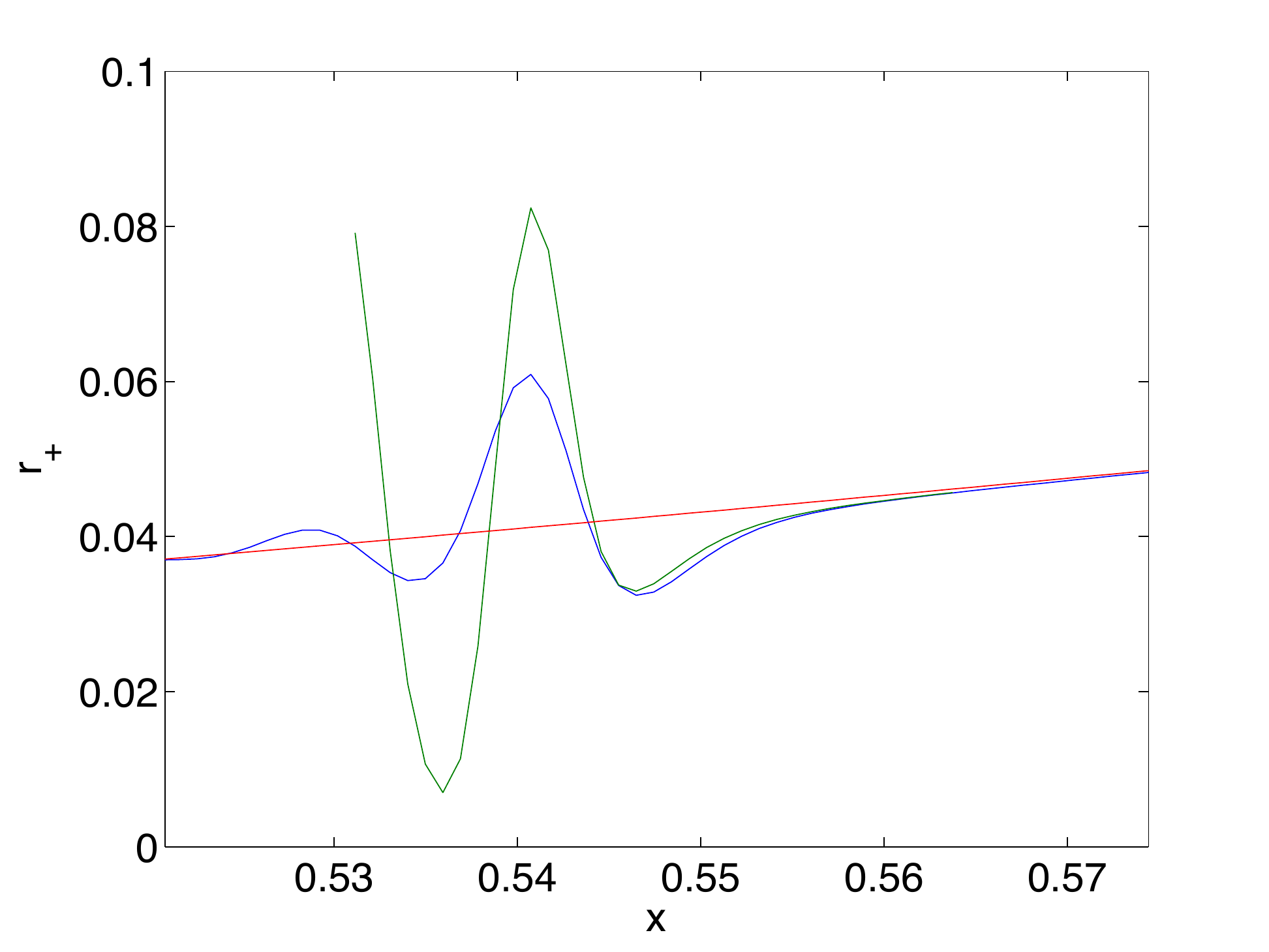}
 \caption{Solution to the defocusing quintic NLS equation for the initial data 
 $\psi_{0}(x)=\tanh^{2}x$ and $\epsilon=0.001$ at the critical time 
 $t_0$ in blue,  the corresponding semiclassical solution in red 
 and the P$_{I}^2$ solution (\ref{conj1}) in green; on the left the 
 Riemann invariant $r_{-}$, on the right the invariant $r_{+}$.}
 \label{nlsquintddarkce001rm}
\end{figure}

\subsection{Defocusing nonlocal NLS}
We will study the small dispersion limit of the 
nonlocal NLS (\ref{full_NNLS_complex}) close to the break-up of the 
corresponding semiclassical solutions. We will concentrate on values 
of $\eta$ such that $\eta\epsilon^{2}\ll1$ for all studied values of 
$\epsilon$. For both cases we will consider the initial data 
$\psi_{0}=\mbox{sech }x$. 
In the defocusing variant of the nonlocal NLS equation 
(\ref{full_NNLS_complex}), the nonlocality has the effect to reduce the defocusing effect of the equation. The dispersion and the steepening of the 
gradient close to the break-up of the corresponding semiclassical 
solution is reduced as can be seen in Fig.~\ref{nlsdnonloc2eta}. 
This also suppresses the formation of dispersive shocks, i.e., the 
oscillations close to the gradient catastrophe of the semiclassical 
solution (see \cite{trillo}). Due to the possible sign change of the quantity $\rho$ in 
(\ref{rho}), an other effect can be observed in 
Fig.~\ref{nlsdnonloc2eta}: for large enough $\eta$, the 
oscillations appear on the other side of the critical point. We again 
consider the initial data $\psi_{0}=\mbox{sech }x$ at the critical 
time $t_0\sim 1.5244$ near the break-up of the Riemann invariant 
$r_{-}$ at $x_{c}\sim-2.2094$ in the semiclassical limit. 
\begin{figure}[htb!]
\includegraphics[width=0.6\textwidth]{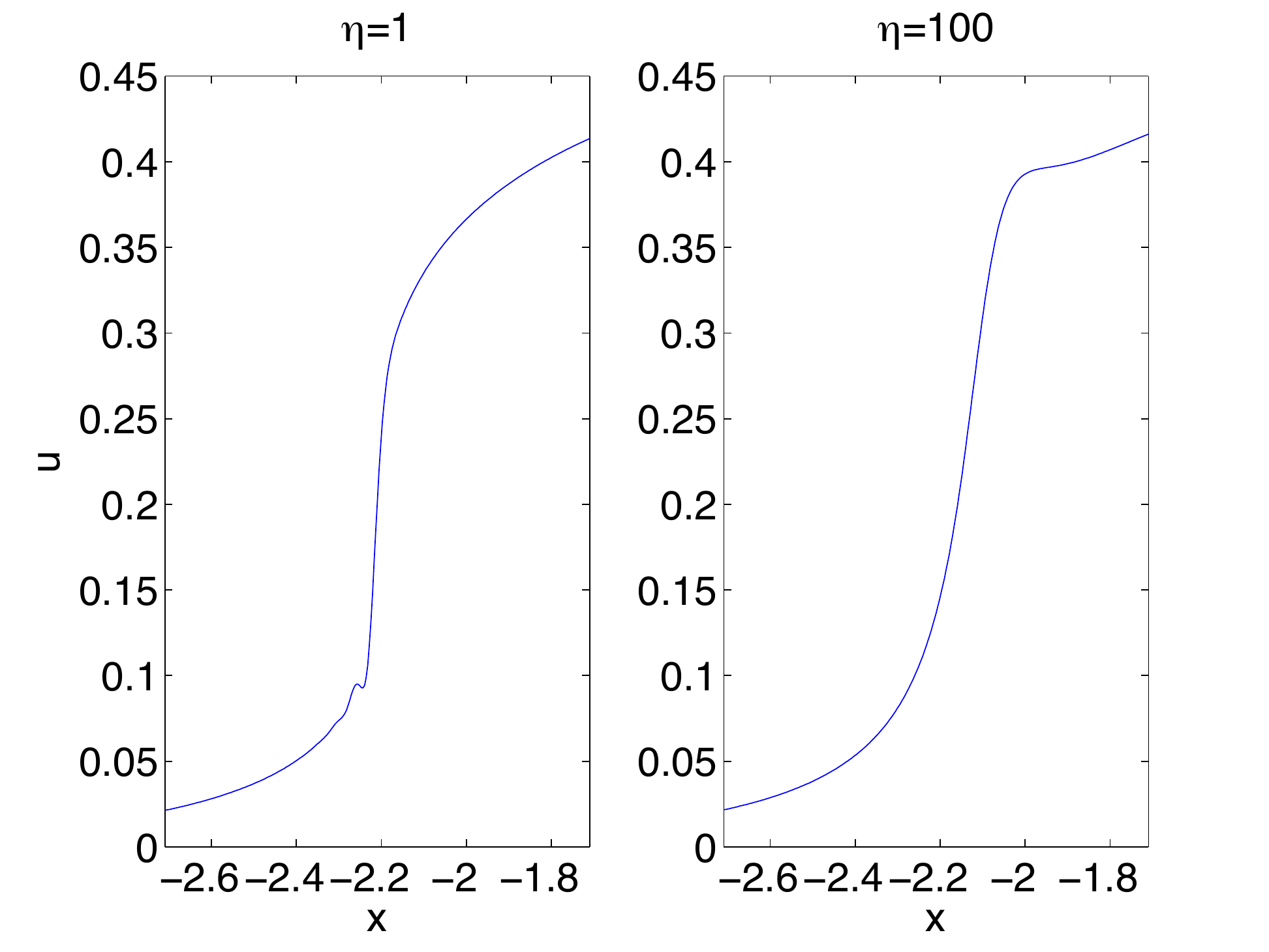}
 \caption{Solution to the defocusing nonlocal NLS equation 
 (\ref{full_NNLS_complex}) for the initial data 
 $\psi_{0}(x)=\mbox{sech }x$ and $\epsilon=0.01$ at the time 
 $t_0=\sim 1.5244$ for two values of $\eta$.}
   \label{nlsdnonloc2eta}
\end{figure}

For larger times this implies for $\rho<0$ that there is just one 
oscillation to the right of $-x_{c}$ as described asymptotically by 
the P$_{I}^2$ solution, and many small oscillations on the other side of 
the critical point as can be seen in Fig.~\ref{nlsnonlocd2tc}. The 
situation is similar to the one of certain Kawahara solutions in the 
small dipsersion limit as discussed in \cite{DGK11}.
\begin{figure}[thb!]
  \includegraphics[width=0.5\textwidth]{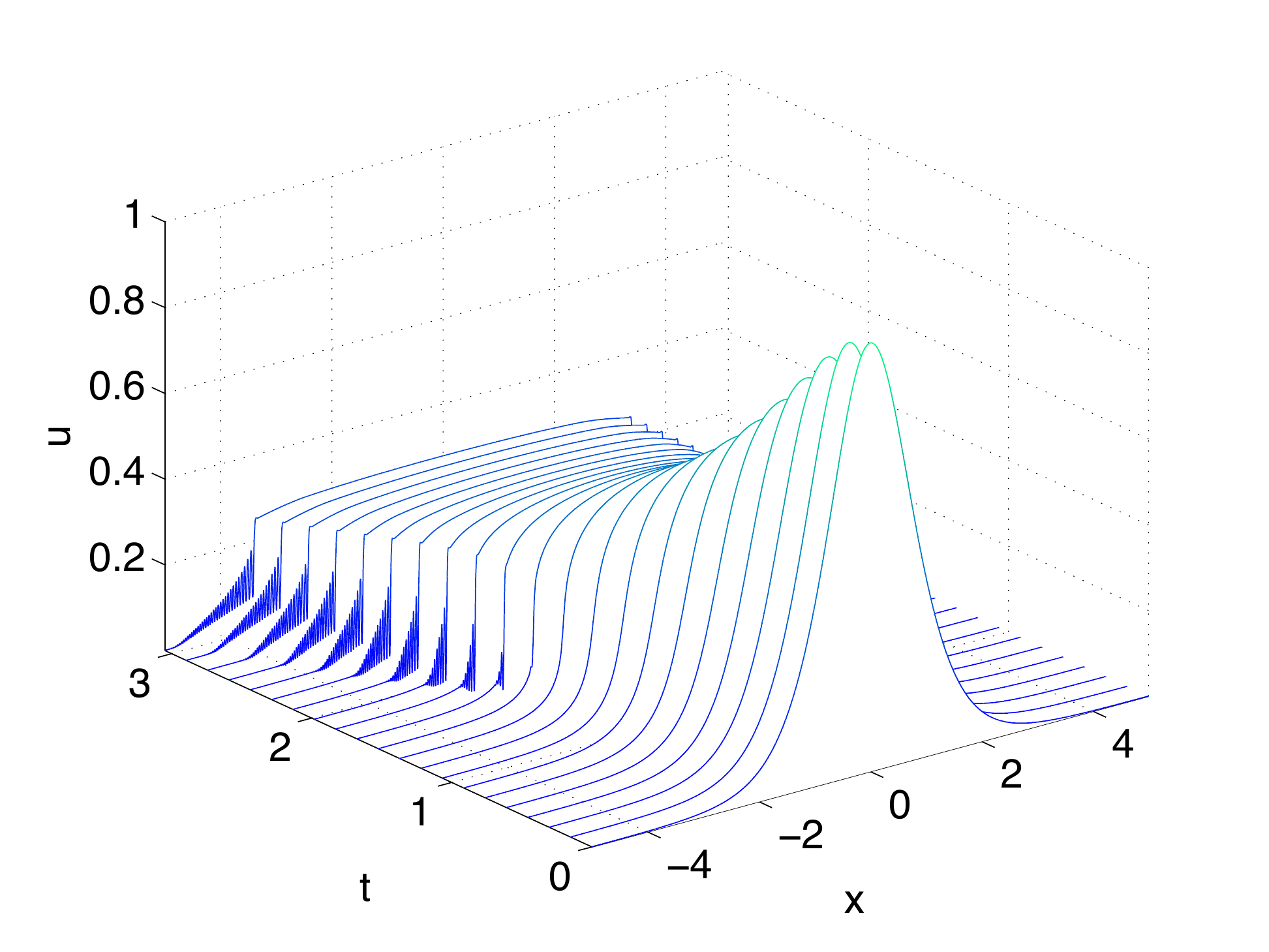}
  \includegraphics[width=0.5\textwidth]{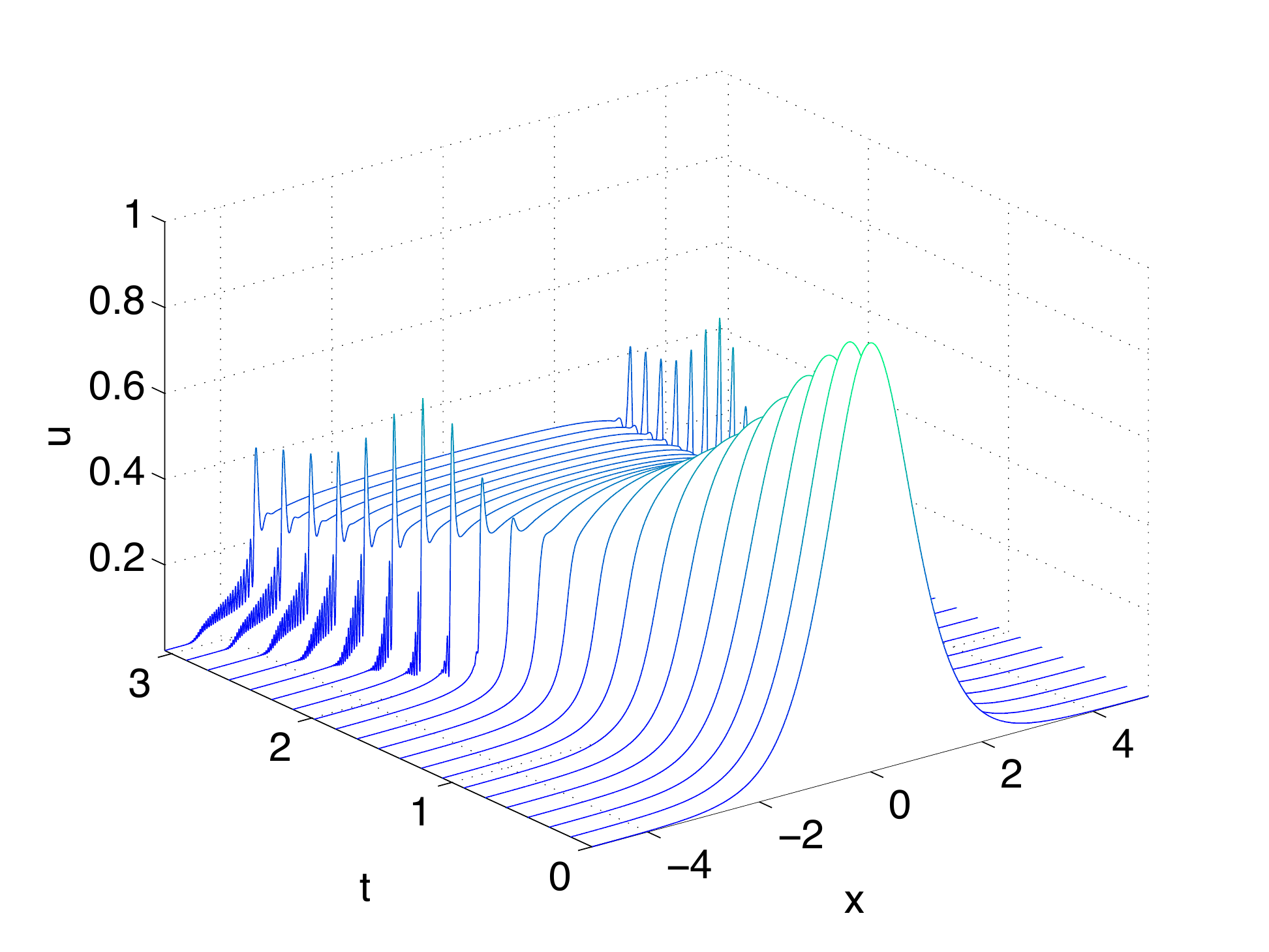}
 \caption{Solution to the defocusing nonlocal NLS equation 
 (\ref{full_NNLS_complex}) for the initial data 
 $\psi_{0}(x)=\mbox{sech }x$ and $\epsilon=0.01$; for $\eta=1$ on the 
 left, for $\eta=100$ on the right. The critical time is
 $t_0\sim 1.5244$.}
 \label{nlsnonlocd2tc}
\end{figure}

In the case $\rho=0$ in (\ref{rho}) the P$_{I}^2$ asymptotics cannot be 
used. In the present example this is the case for $\eta\sim 1.3060$. 
The solution at the critical time for this value of $\eta$ can be 
seen in Fig.~\ref{nlsdnonlocrho0c}.
\begin{figure}[thb!]
  \includegraphics[width=0.7\textwidth]{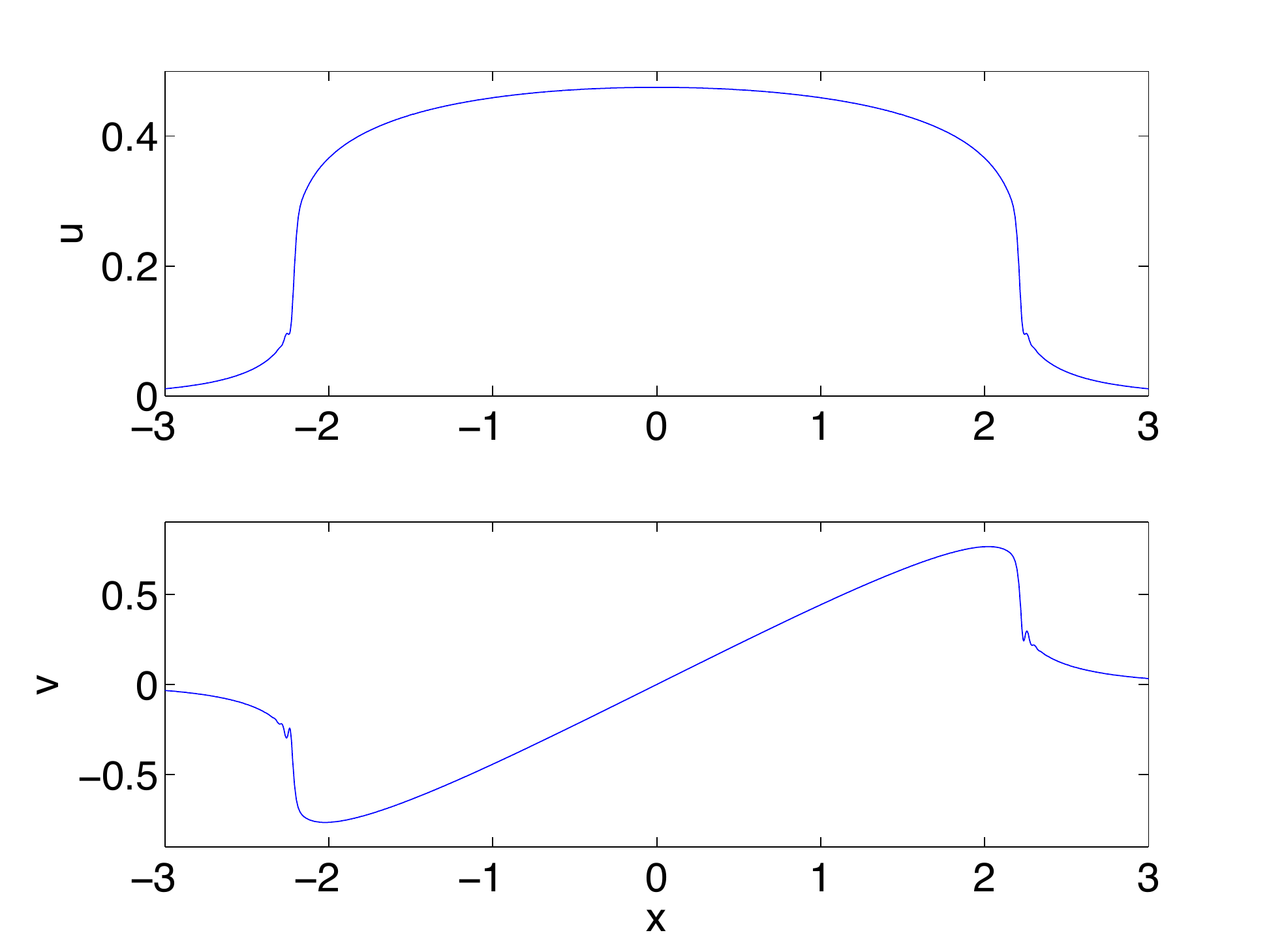}
 \caption{Solution to the defocusing nonlocal NLS equation 
 (\ref{full_NNLS_complex}) for the initial data 
 $\psi_{0}(x)=\mbox{sech }x$, $\epsilon=0.1$ and the non-generic 
 value $\eta\sim 1.3060$ at the critical time
 $t_0\sim 1.5244$.}
 \label{nlsdnonlocrho0c}
\end{figure}

For smaller $\eta$, the nonlocal NLS behaves qualitatively like the 
defocusing cubic NLS close to the critical time as can be seen in 
Fig.~\ref{nlsdnonloceta12erm} for the Riemann invariant breaking in 
the semiclassical limit. For smaller values of $\epsilon$ the same behavior can be seen, but 
on smaller scales. Again there are two different scales in the P$_{I}^2$ 
asymptotics (\ref{conj1}) which means there is no clear scaling in the 
coordinates $x$ 
and $t$. For the representation, we nonetheless rescale $x$ by a 
factor of $\epsilon^{6/7}$ to be able to compare the case 
$\epsilon=0.001$ with $\epsilon=0.01$.  The $y$-axes are rescaled to optimally use 
the space of the figure. The approximation visibly gets better with 
smaller $\epsilon$. 
The Riemann invariant staying smooth in the semiclassical limit can 
be seen for the same situation in the right part of  Fig.~\ref{nlsdnonloceta12erm}. The 
asymptotic description again improves clearly with smaller $\epsilon$.
\begin{figure}[htb!]
\includegraphics[width=0.52\textwidth]{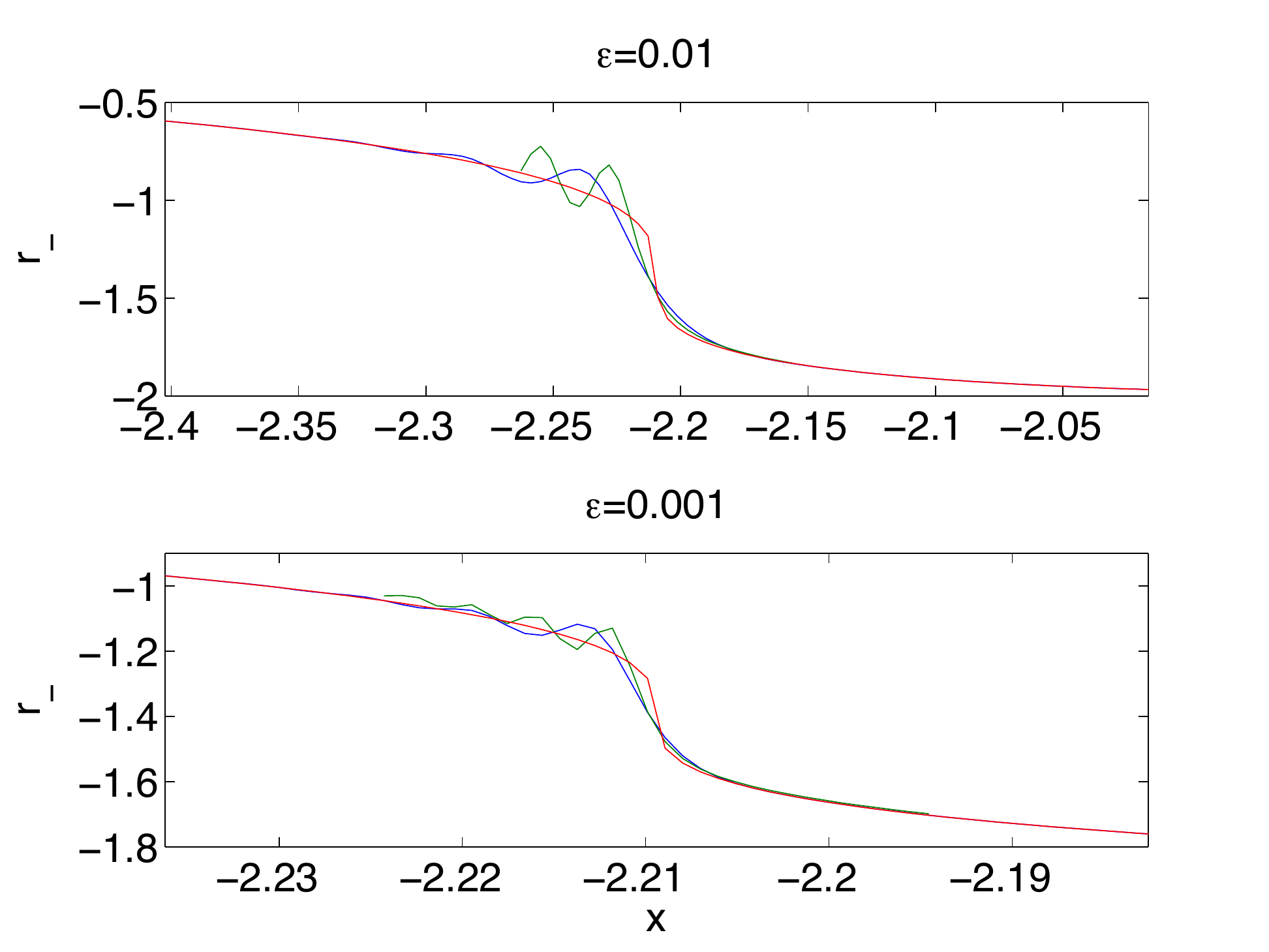}
\includegraphics[width=0.53\textwidth]{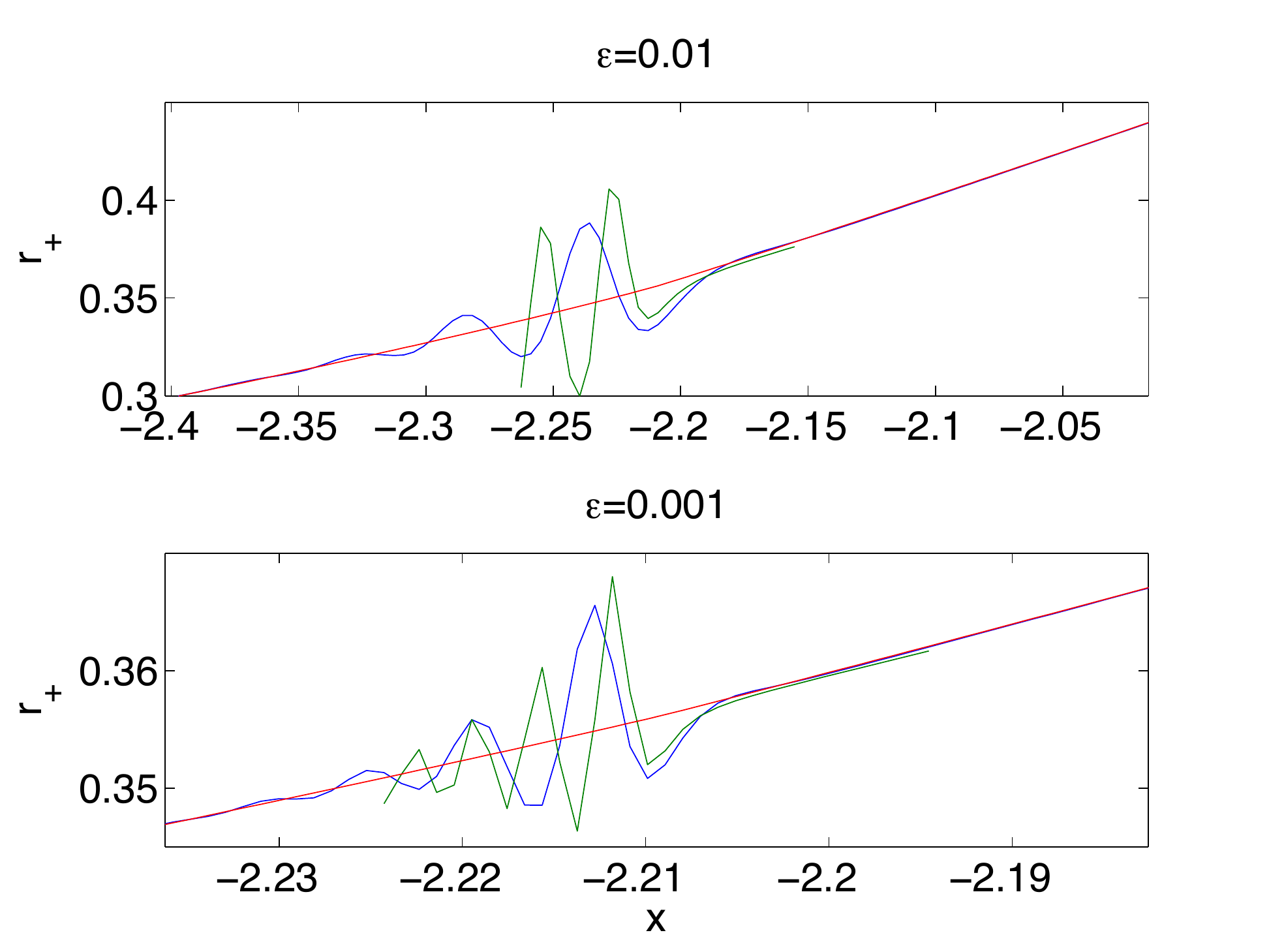}
 \caption{Riemann invariant $r_{-}$ on the left and $r_+$ on the 
 right of the solution to the defocusing nonlocal NLS equation  (\ref{full_NNLS_complex}) for the initial data 
 $\psi_{0}(x)=\mbox{sech }x$ and $\eta=1$ at the time 
 $t_0\sim 1.5244$ for two values of $\epsilon$ in blue,  the corresponding semiclassical solution in 
 red and the P12 solution  (\ref{conj1}) in green.}
 \label{nlsdnonloceta12erm}
\end{figure}


For larger $\eta$ the smoothing out of the gradients near the shock of the 
semiclassical equations implies that the semiclassical solution only 
provides a valid asymptotic description for larger $|x-x_{c}|$ than 
is the case for smaller $\eta$. The P$_{I}^2$ asymptotics (\ref{conj1}) catches this 
behavior as can be seen for $\eta=100$ in Fig.~\ref{nlsdnonloceta1002erm} on the left
 for the 
invariant breaking in the semiclassical limit. There are 
essentially no oscillations in this case. 
\begin{figure}[htb!]
\includegraphics[width=0.5\textwidth]{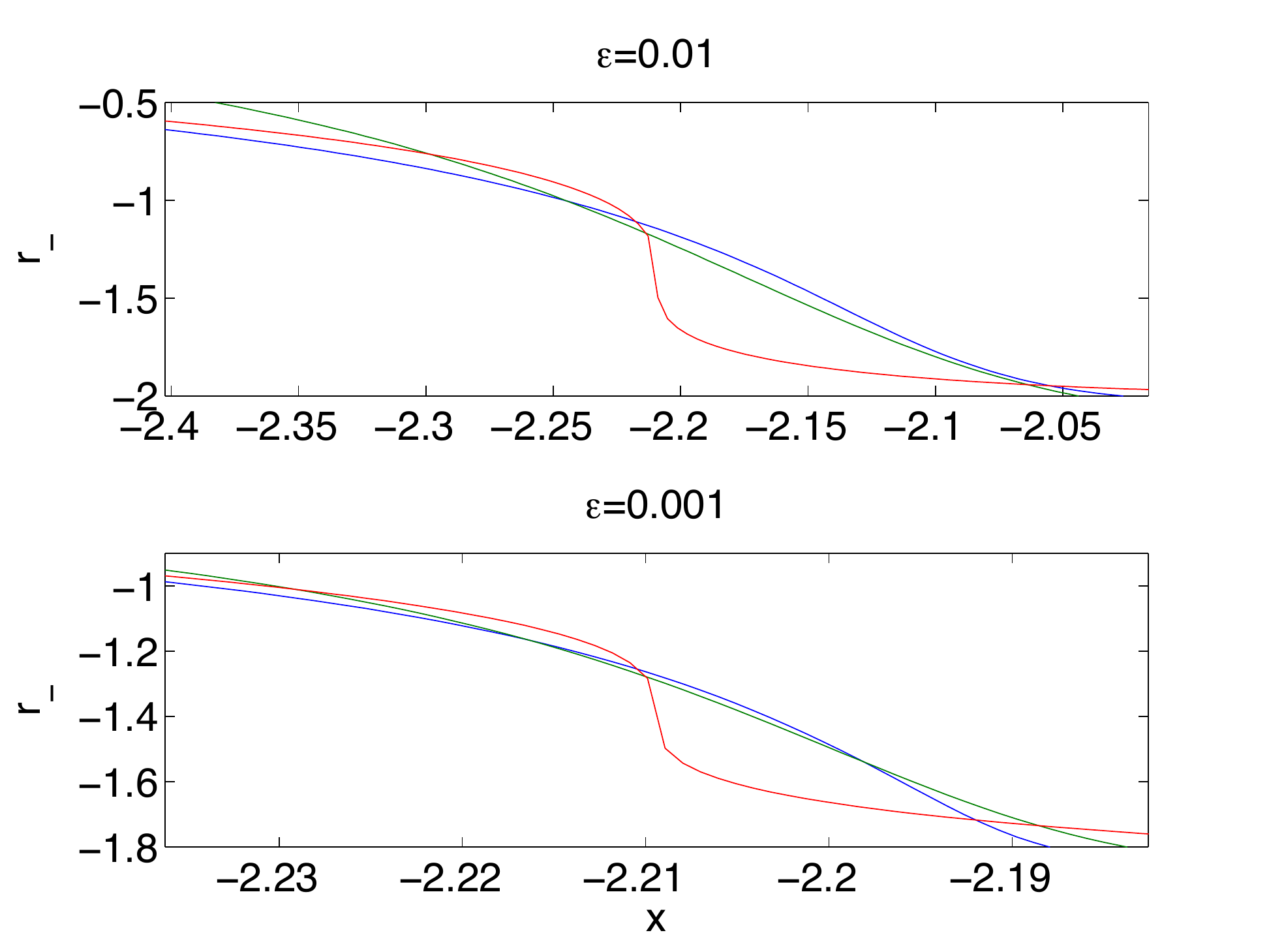}
\includegraphics[width=0.5\textwidth]{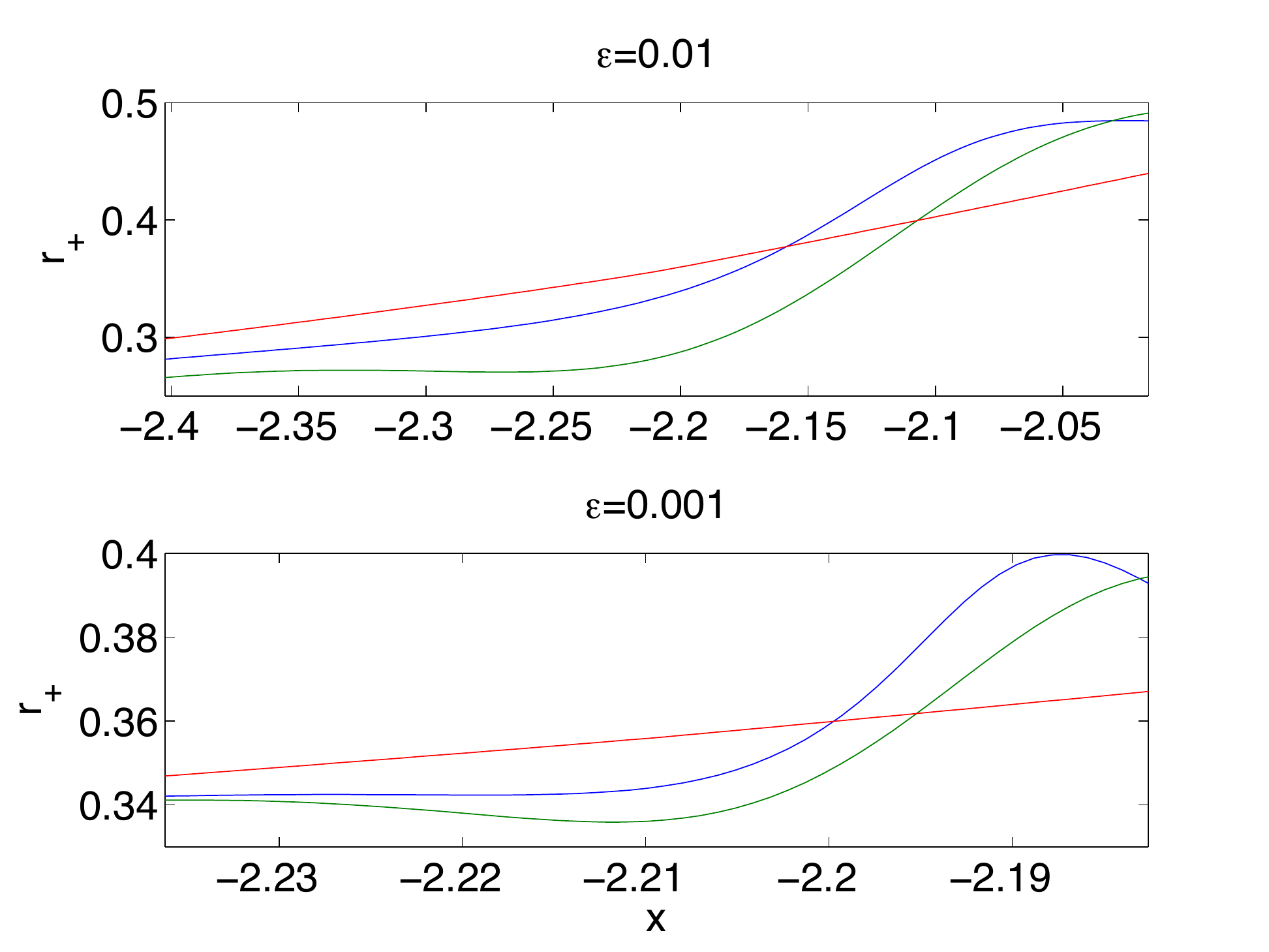}
 \caption{Riemann invariant $r_{-}$ on the left and Riemann invariant $r_+$ 
on the right for  the solution to the defocusing nonlocal NLS equation (\ref{full_NNLS_complex}) for the initial data 
 $\psi_{0}(x)=\mbox{sech }x$ and $\eta=100$ at the time 
 $t_0\sim 1.5244$ for two values of $\epsilon$ in blue, 
 the corresponding semiclassical solution in 
 red and the P$_{I}^2$ solution (\ref{conj1}) in green.}
 \label{nlsdnonloceta1002erm}
\end{figure}

The invariant $r_{+}$ can be seen  on the right part of  Fig.~\ref{nlsdnonloceta1002erm}.
There 
is essentially only one oscillation to the right of the critical 
point in this case. The P$_{I}^2$ asymptotics has an oscillation close to the 
oscillation of the nonlocal NLS and thus catches this behavior in an 
asymptotic sense.  

\section{Numerical study of focusing generalized and nonlocal  NLS equations}\label{section8}
In this section we will study numerically solutions to the focusing 
NLS before and close to the break up of the corresponding 
semiclassical solutions. Since the case of the focusing cubic NLS was 
studied in detail in \cite{DGK}, we concentrate here on the not 
integrable quintic NLS. We compare solutions to NLS and 
semiclassical equations and for $t\sim t_0$ to an asymptotic 
solution in terms of the tritronqu\'ee solution of the Painlev\'e-I 
equation. The same is done for a nonlocal variant of the cubic NLS 
equation. 

\subsection{sech$\, x$ initial data for the focusing quintic NLS}

We will first study the initial data $\psi_{0}(x)=\mbox{sech }x$ for 
several values of $\epsilon$, i.e., $\epsilon=0.1$, $0.09$,\ldots,$0.01$. For this example, 
the break-up occurs for the semiclassical solution at 
$t_0=0.4119\ldots$ at $x_{c}=0$ with the critical values 
$u_{c}=1.5858\ldots$ and $v_{c}=0$. The solution up to the critical 
time can be seen in Fig.~\ref{nlsquintf}. The focusing effect can be 
clearly recognized.
\begin{figure}[thb!]
  \includegraphics[width=0.7\textwidth]{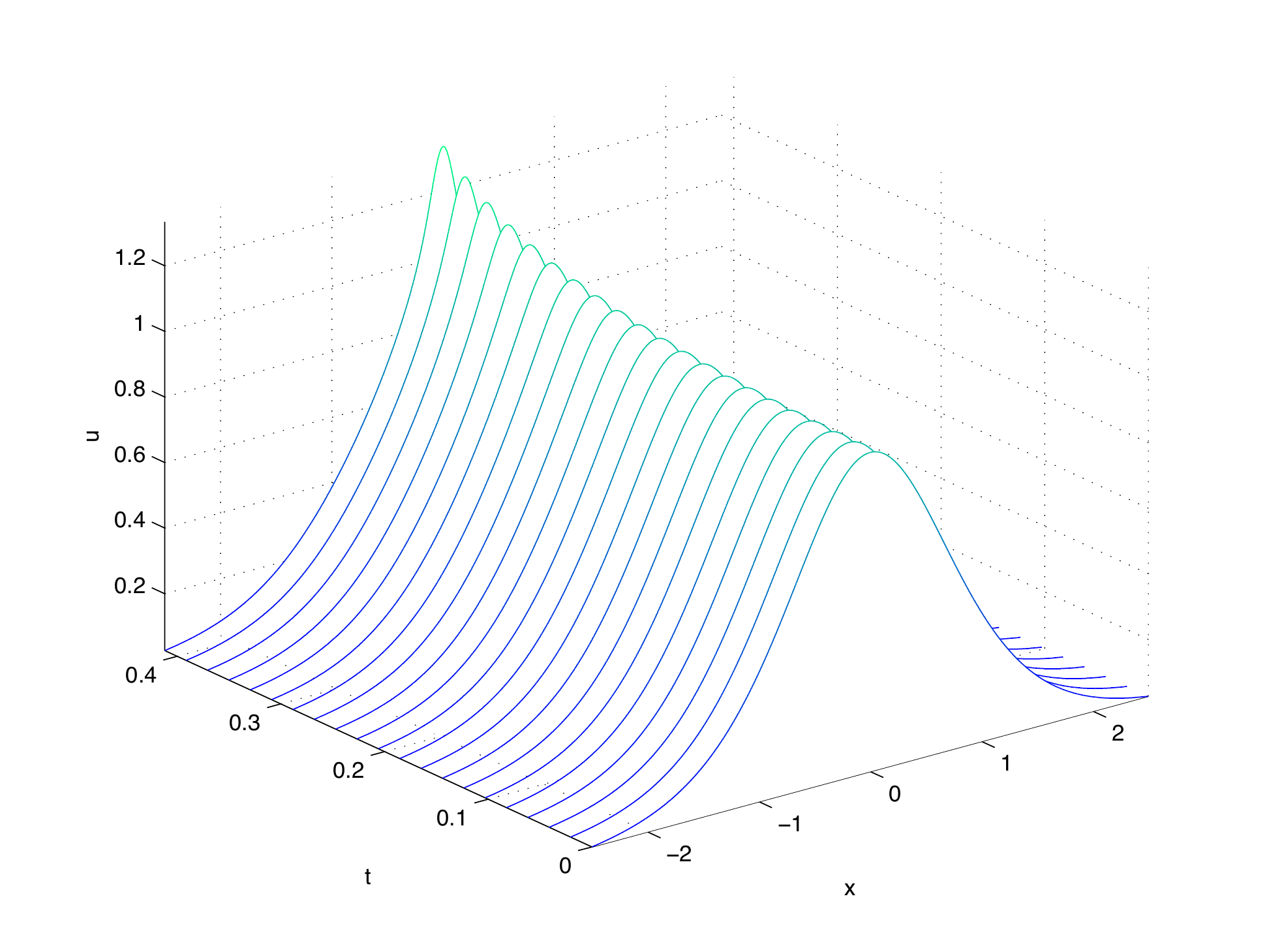}
 \caption{Solution to the focusing quintic NLS equation for the initial data 
 $\psi_{0}(x)=\mbox{sech }x$ and $\epsilon=0.1$ up to the critical time 
 $t_0$ in blue.}
 \label{nlsquintf}
\end{figure}

For times much smaller than the critical time 
one finds that the difference between 
semiclassical and NLS solution scales as $\epsilon^{2}$. For instance 
for $t=t_0/2\ll t_0$ we obtain for $\Delta=|u_{NLS}-u_{sc}|$ via a linear
regression analysis for the logarithm of  $\Delta$
a scaling
of the form $\Delta\propto \epsilon^{a}$ with $a=1.985$ with
standard deviation $\sigma_{a}=0.0018$ and correlation coefficient 
$r=0.999998$.\\
At the critical time the difference between the 
semiclassical solution and the solution to the focusing quintic NLS 
scales roughly as $\epsilon^{2/5}$. More precisely we find via a linear
regression analysis for the logarithm of the difference $\Delta$
between NLS and semiclassical solution  a scaling
of the form $\Delta\propto \epsilon^{a}$ with $a=0.403$ with
standard deviation $\sigma_{a}=0.001$ and correlation coefficient 
$r=0.99998$. As can be seen in Fig.~\ref{nlsquintfsechc3in101}, the 
semiclassical solution has a cusp. Thus the maximal difference 
between semiclassical and NLS solution is always observed for the 
critical point. 
\begin{figure}[thb!]
\subfigure
{  \includegraphics[width=0.5\textwidth]{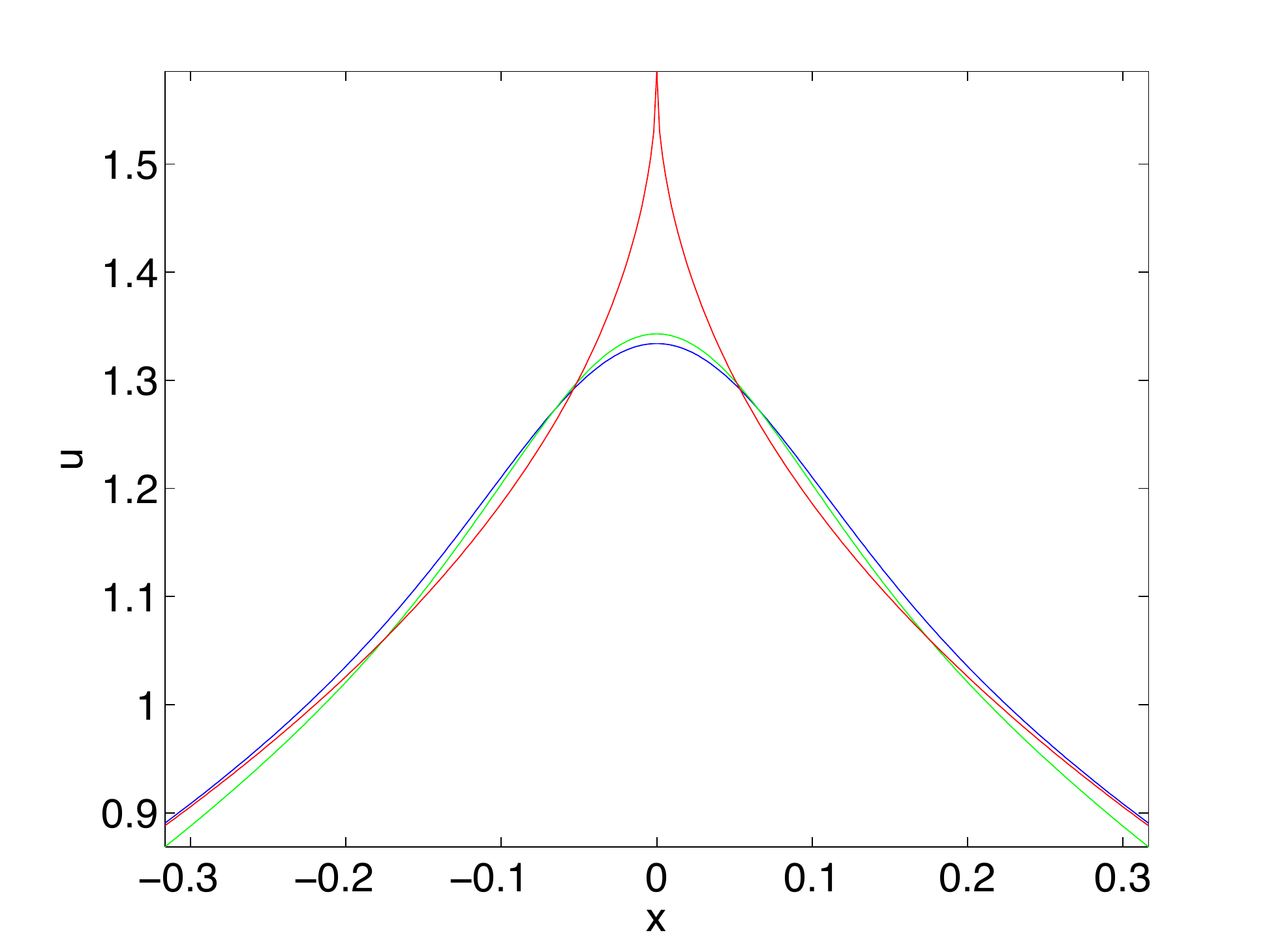}
  \includegraphics[width=0.5\textwidth]{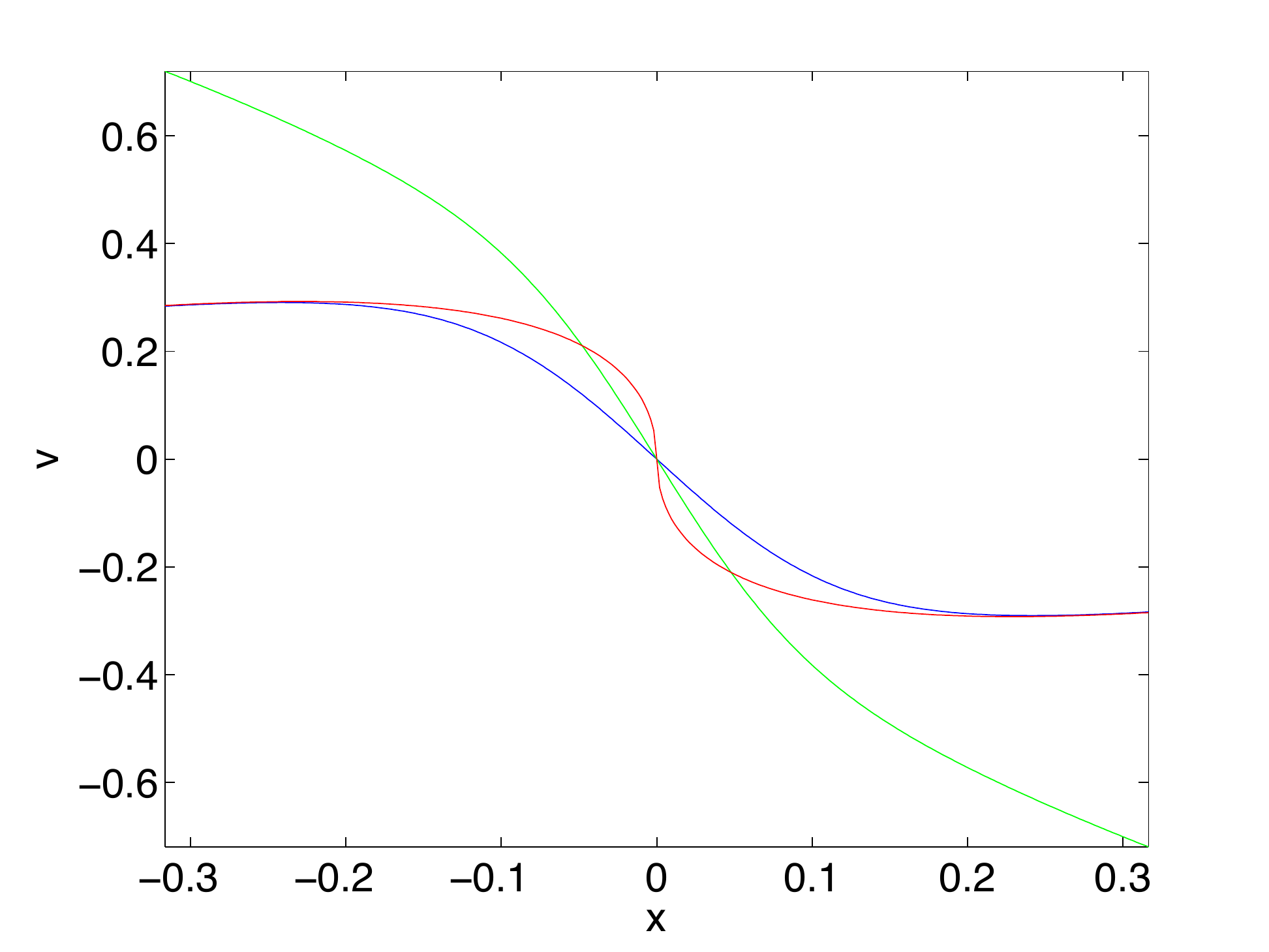}}
\vskip 0.1cm
\subfigure
{\includegraphics[width=0.5\textwidth]{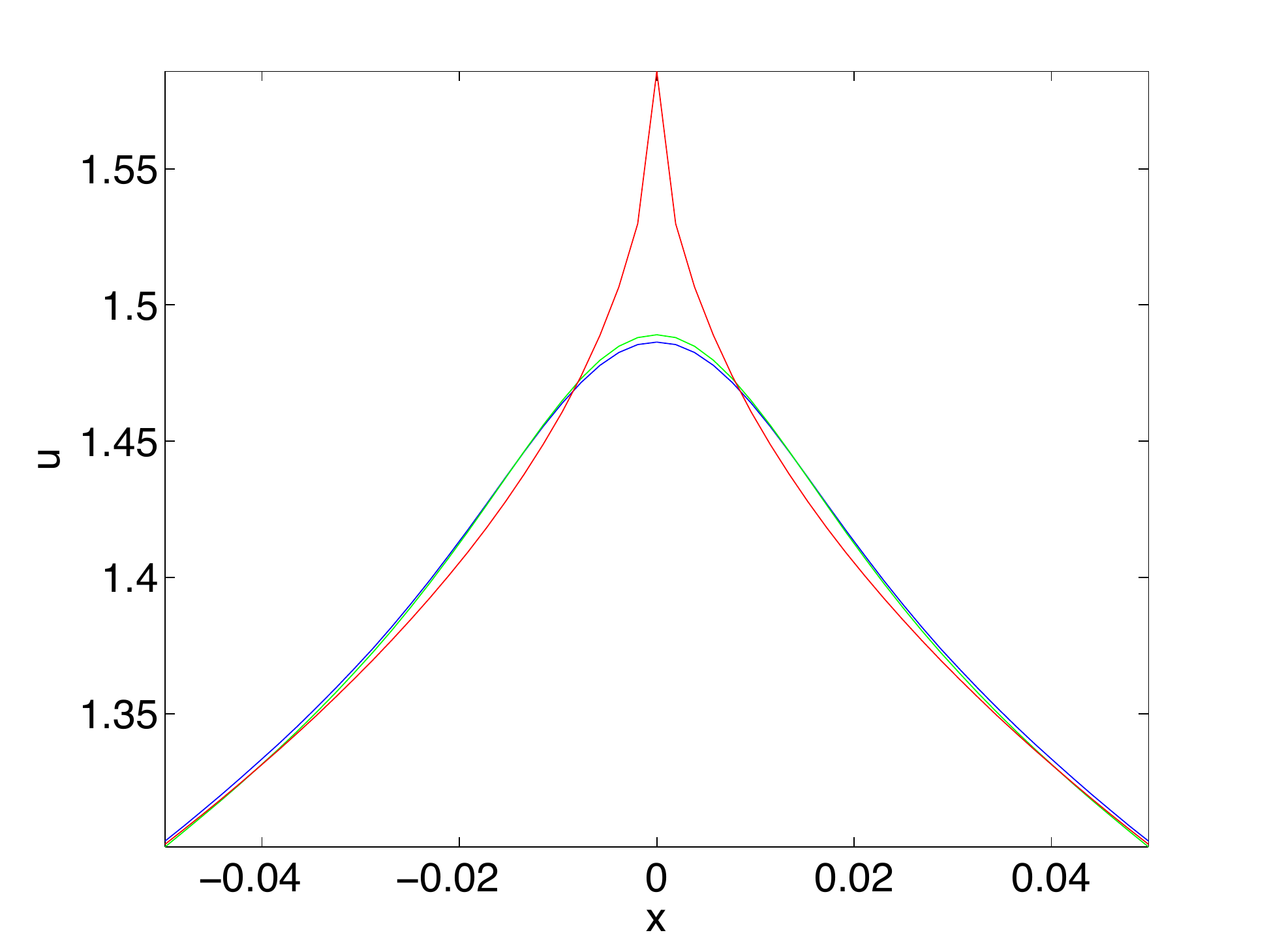}
  \includegraphics[width=0.5\textwidth]{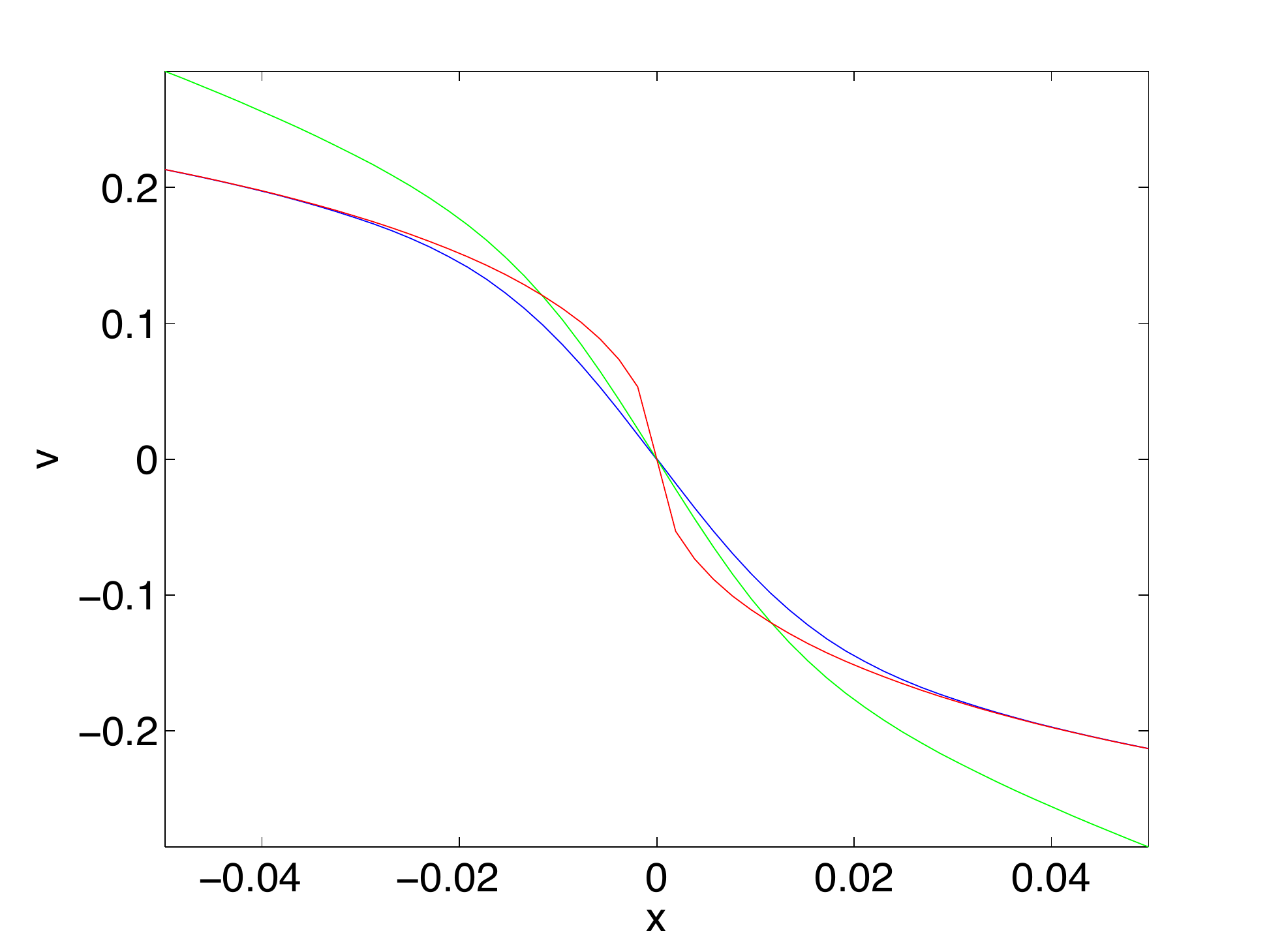}}
 \caption{Solution to the focusing quintic NLS equation for the initial data 
 $\psi_{0}(x)=\mbox{sech }x$  at the critical time 
 $t_0$ in blue,  the corresponding semiclassical solution in red 
 and the P$_{I}$ solution (\ref{F1}) in green; on the left the 
 function $u$, on the right the function $v$. For the upper two 
 figures we have $\epsilon=0.1$, for the lower ones
 $\epsilon=0.01$.}
 The $x$-axis of the figures in the lower row  is 
 rescaled by factor $\epsilon^{4/5}$ with respect to the figures in 
 the upper row.
 \label{nlsquintfsechc3in101}
\end{figure}

For smaller $\epsilon$ the agreement of NLS and semiclassical 
solution becomes better, but the biggest difference is always at the 
critical point as can be seen in the bottom of Fig.~\ref{nlsquintfsechc3in101}.

The P$_{I}$ solution (\ref{F1})  gives a much better agreement 
with the NLS solution close to the critical point as can be seen in 
Fig.~\ref{nlsquintfsechc3in101}.
 The  agreement is in fact so good that the difference of the solutions has 
to be studied. The P$_{I}$ solution only 
gives locally an asymptotic description, at larger distances from 
the critical point the semiclassical solution provides a better 
description as can be also seen from Fig.~\ref{nlsquintfsechdelta}.
\begin{figure}[htb!]
\includegraphics[width=0.5\textwidth]{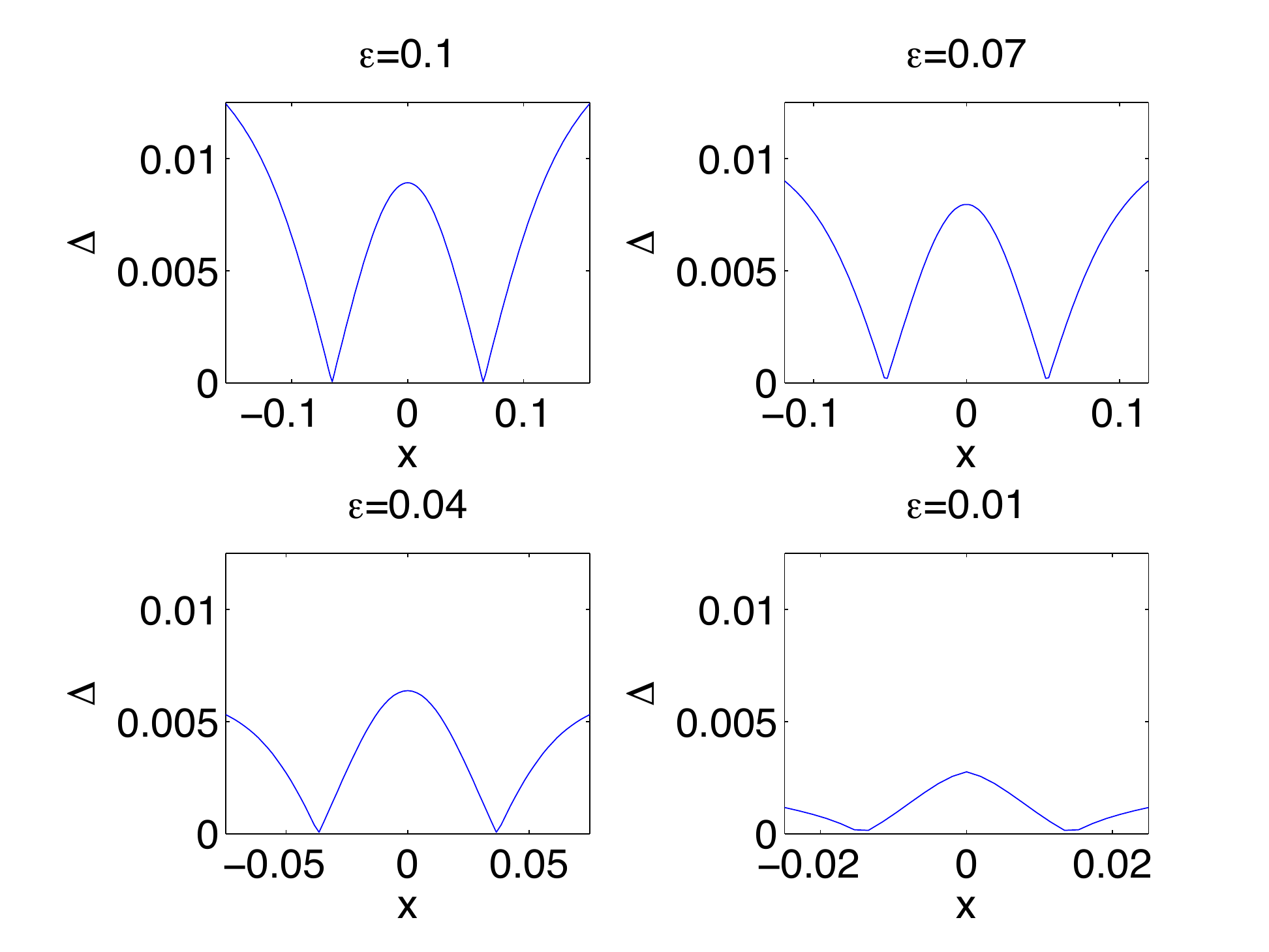}
\includegraphics[width=0.5\textwidth]{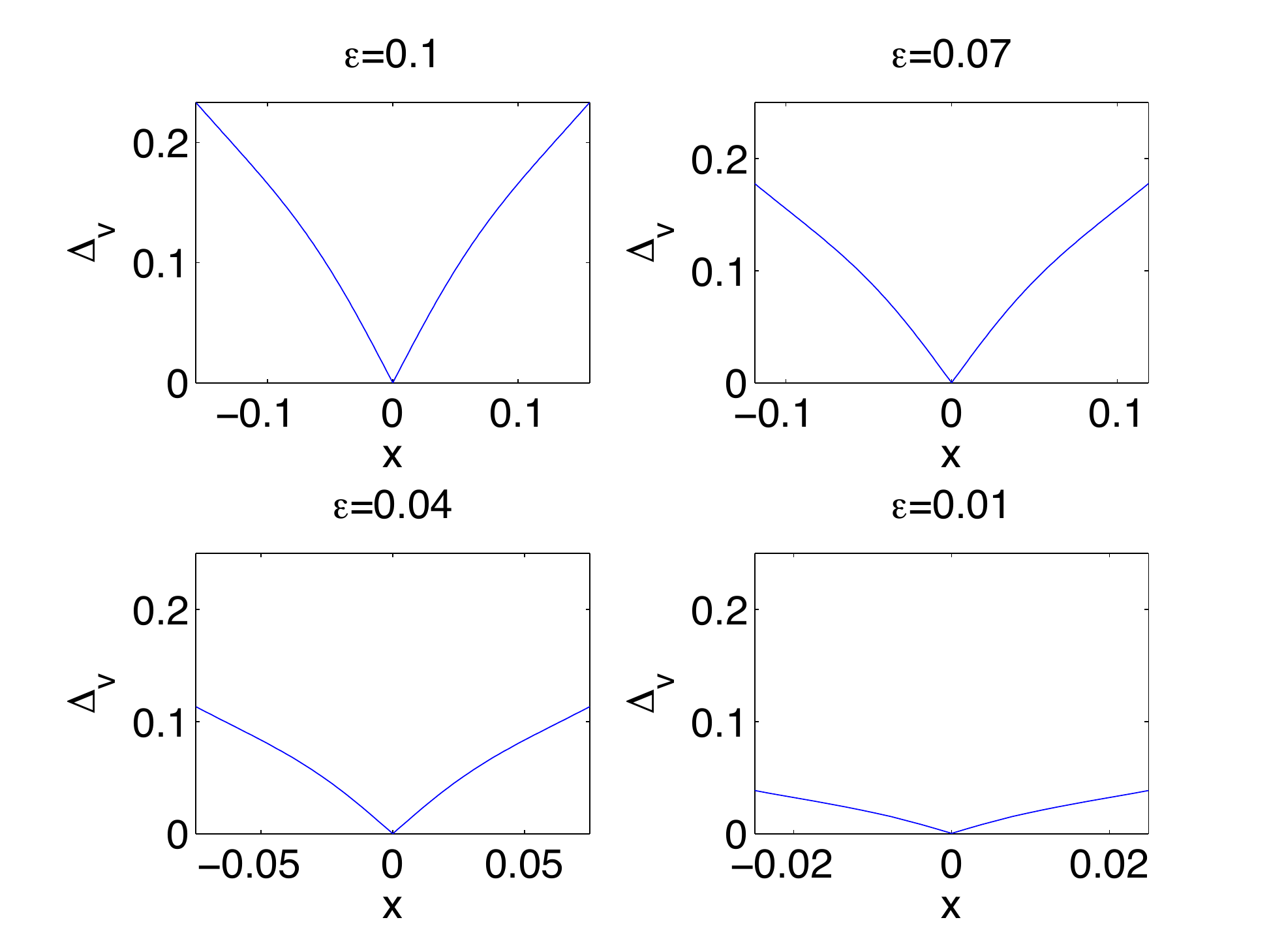}
 \caption{The modulus of the difference of the solution
 to the focusing quintic NLS equation for the initial data 
 $\psi_{0}(x)=\mbox{sech }x$ for $\epsilon=0.1$ at the critical time 
 $t_0$ and the  difference 
 between the corresponding P$_{I}$ solution (\ref{F1}) for 
 several values of $\epsilon$; on the left the difference $\Delta$ 
 for $u$, on the right the difference $\Delta_{v}$ for $v$. The 
 $x$-axes are rescaled with a factor $\epsilon^{4/5}$.}
   \label{nlsquintfsechdelta}
\end{figure}


We can identify the regions where each of the asymptotic
solutions gives a better description of NLS than the other by 
identifying the value of $x_{r}$ such that for all $x>x_{r}$ the 
semiclassical solutions gives a better asymptotic description than 
the multiscales solution (since the solution is symmetric with 
respect to $x$, we only consider positive values of $x$ here). We 
find that  the width of this zone scales 
roughly as $\epsilon^{3/5}$. A linear regression analysis for the 
dependence of $\log_{10}x_{r}$ on $\log_{10}\epsilon$ yields 
 $a=0.634$ with
standard deviation $\sigma_{a}=0.0036$ and correlation coefficient 
$r=0.99993$.

This matching procedure clearly improves the NLS description near
the critical point.
In Fig.~\ref{nlsquintsechcdeltamatch} we see the difference between this matched
asymptotic solution and the NLS solution for two values of
$\epsilon$. Visibly the zone, where the solutions are matched,
decreases with $\epsilon$ (note the rescaling of the $x$-axes by a 
factor $\epsilon^{4/5}$).
\begin{figure}[htb!]
 \includegraphics[width=0.5\textwidth]{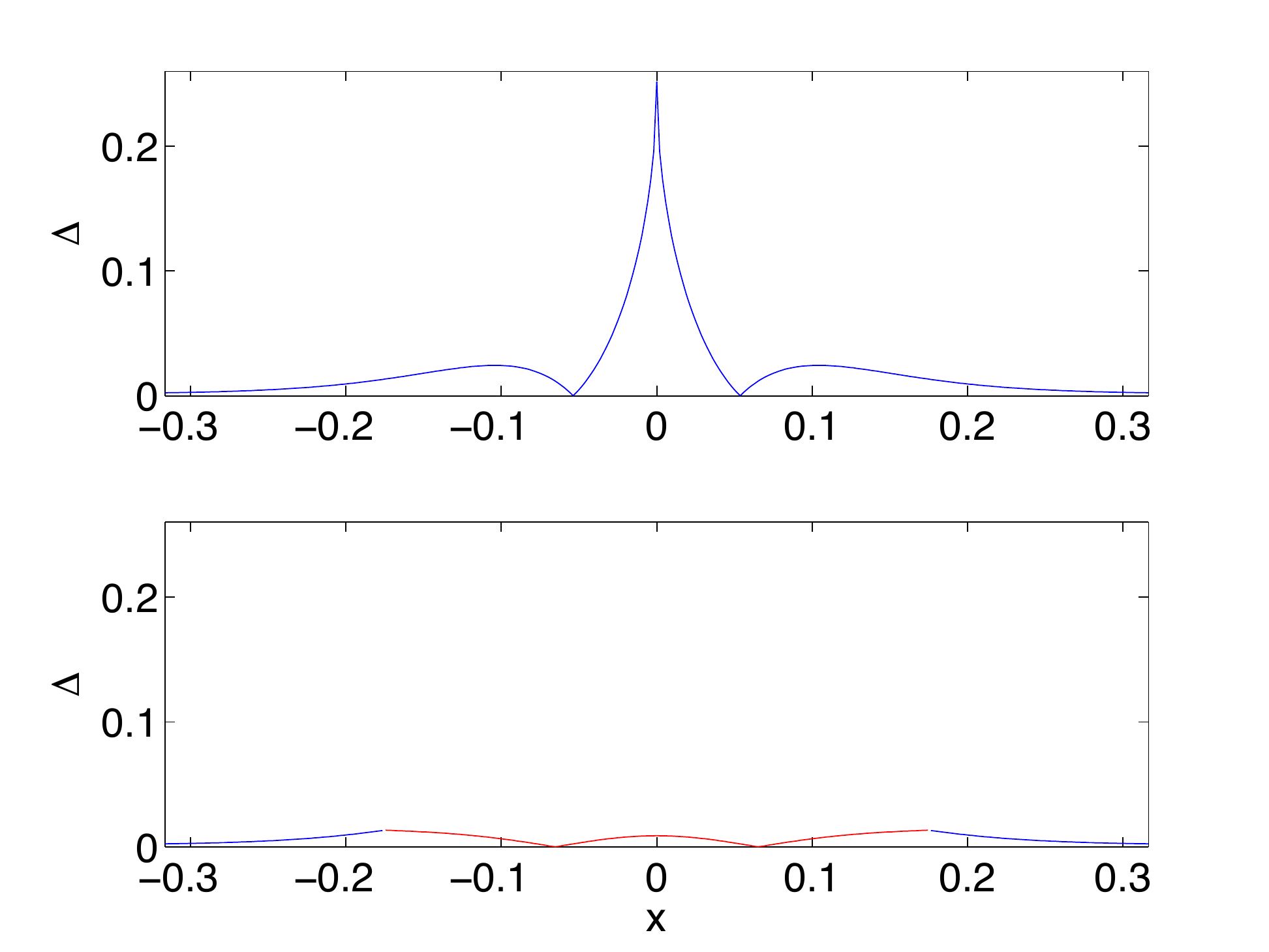}               
 \includegraphics[width=0.5\textwidth]{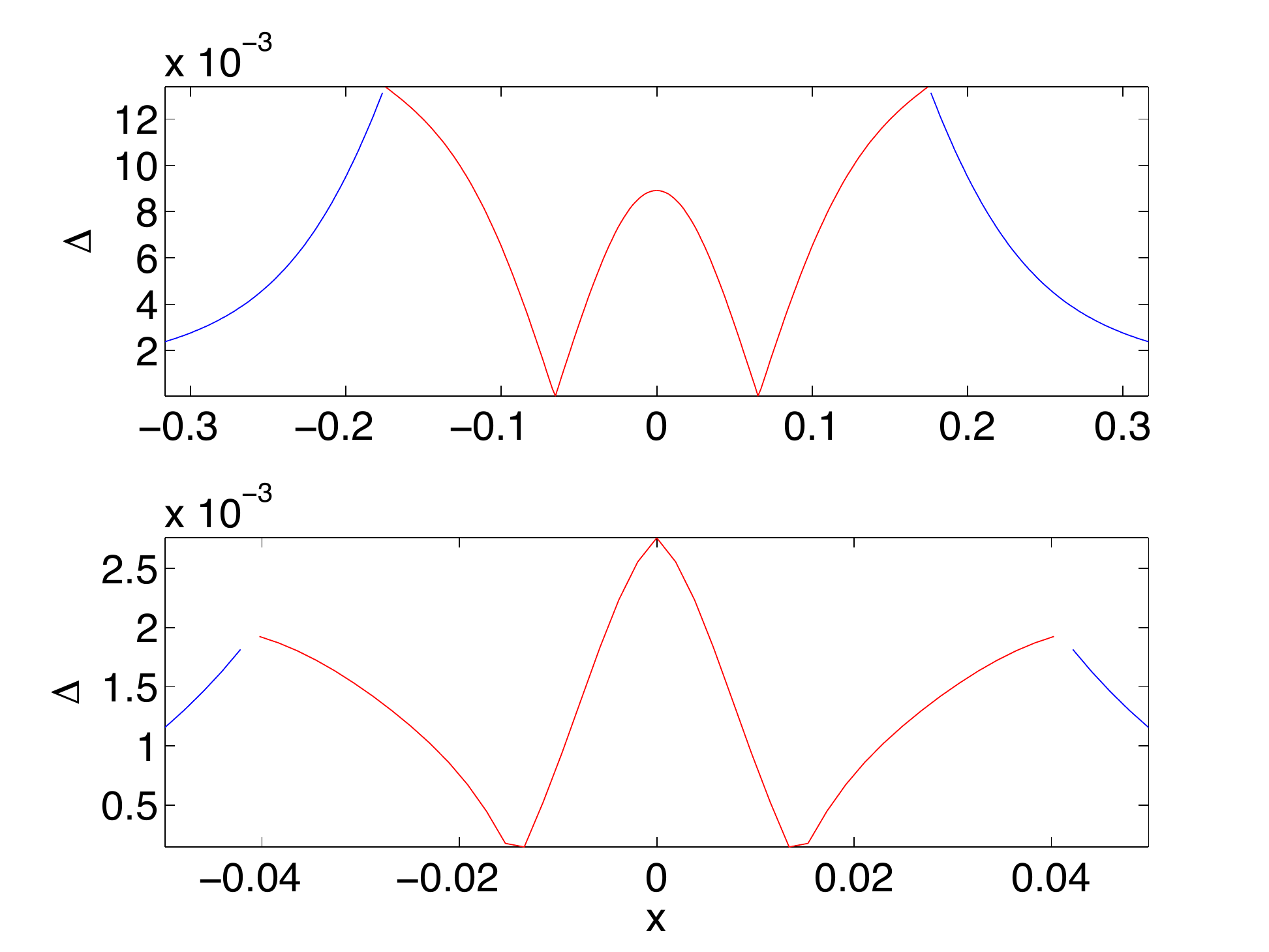}
 \caption{In
the upper part of the left figure  one can see the modulus of the 
difference of the solution $u$ 
 to the focusing quintic NLS equation for the initial data 
 $\psi_{0}(x)=\mbox{sech }x$ at the critical time 
 $t_0$ and the semiclassical solution for $\epsilon=0.1$. 
 The lower part shows the same
difference, which is replaced close to the critical point 
by the difference between NLS solution and  the 
P$_{I}$ solution (\ref{F1}) (in red where the error is smaller 
than the one
shown above). The right figure shows the same situation as the lower figure 
on the left  for $\epsilon=0.1$ above and $\epsilon=0.01$ below. The 
$x$-axes are rescaled  in this figure by a factor $\epsilon^{4/5}$.}  
\label{nlsquintsechcdeltamatch}
\end{figure}

A linear regression analysis for the logarithm of the difference $\Delta$
between NLS and multiscales solution in the matching zone gives a scaling
of the form $\Delta\propto \epsilon^{a}$ with $a=0.6659$ 
with standard deviation $\sigma_{a}=0.032$ and correlation coefficient 
$r=0.995$. The found scaling is thus in the whole interval clearly 
better than the $\epsilon^{2/5}$ 
of the semiclassical solution, but does not reach the expected 
$\epsilon^{4/5}$ scaling in the whole interval. This indicates that 
transition formulae between the multiscales and the semiclassical 
solution have to be established as in \cite{physicad} for KdV, which is, however, beyond the scope 
of the present paper.

The P$_{I}$ solution (\ref{F1}) holds for small $|x-x_{c}|$ and 
$|t-t_0|$. To illustrate the latter effect, we compare it with the 
NLS solution for the times $t_{\pm}(\epsilon)=t_0\pm 0.01\epsilon^{4/5}$ where 
we take care of the scaling of $t$ in (\ref{rasymp}). In 
Fig.~\ref{nlsquintfsechdeltam001} we show the quantity $\Delta$ for 2 
values of $\epsilon$ at the times $t_{-}(\epsilon)$. The $x$-axes are 
rescaled by a factor $\epsilon^{4/5}$. It can be seen that the quality of the 
asymptotic description is slightly lower than at the critical time, but that 
the error is of a similar order. The situation is similar at the time $t_{+}=t_0+0.01\epsilon^{4/5}$ 
as can be seen also in Fig.~\ref{nlsquintfsechdeltam001}. 
\begin{figure}[htb!]
\includegraphics[width=0.5\textwidth]{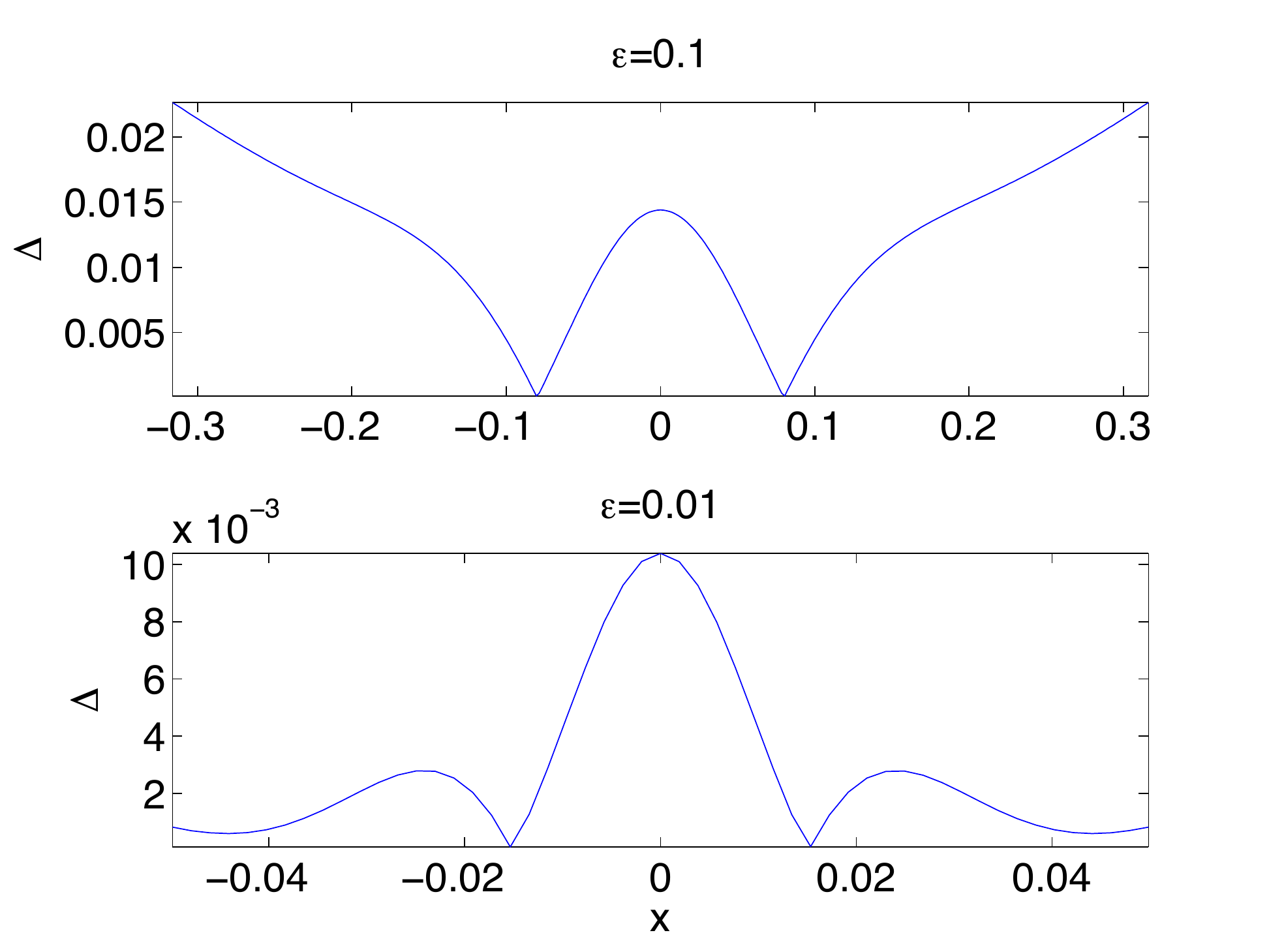}
\includegraphics[width=0.5\textwidth]{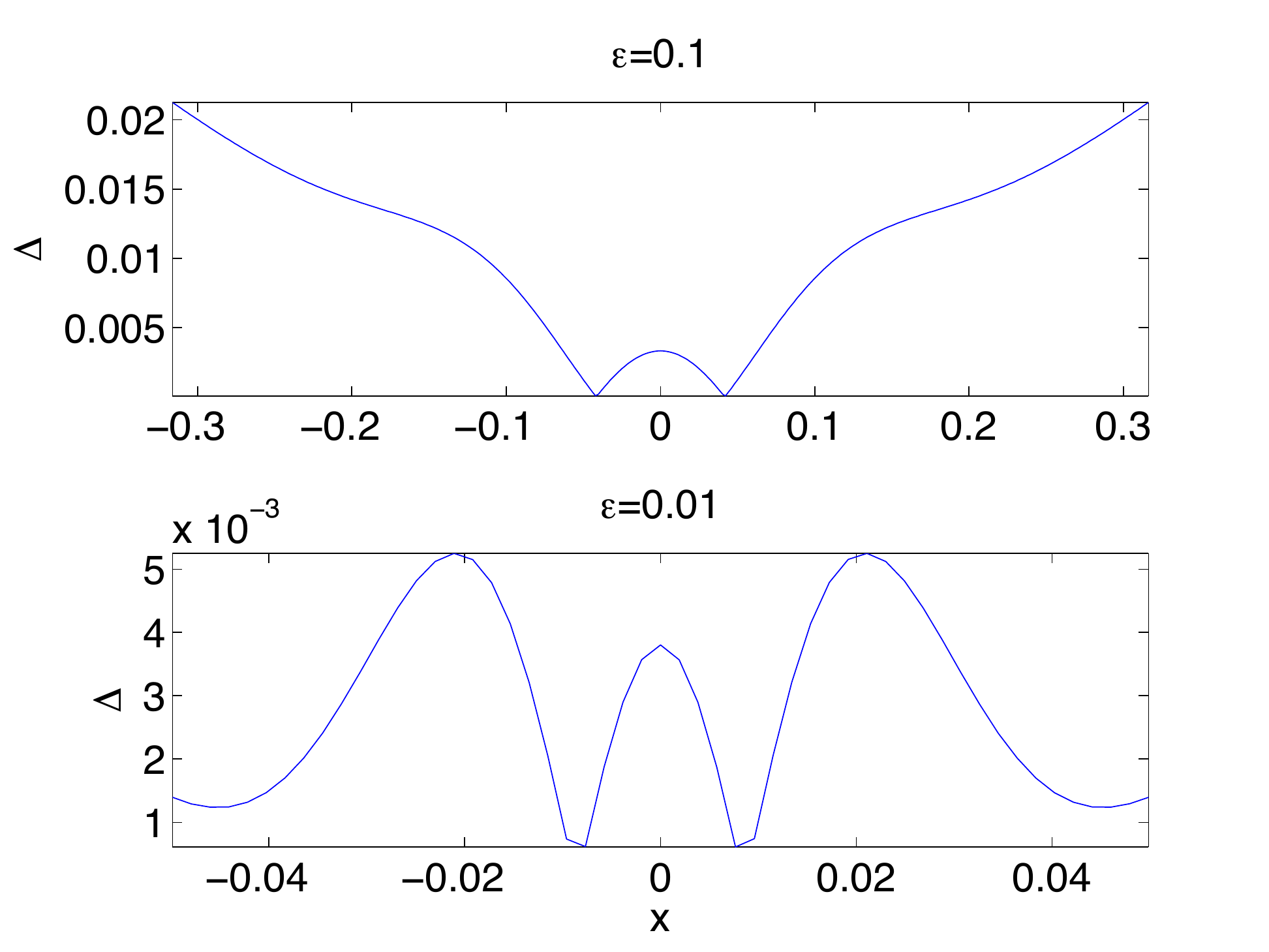}
 \caption{The modulus of the difference of the solution $u$ 
 to the focusing quintic NLS equation for the initial data 
 $\psi_{0}(x)=\mbox{sech }x$ for two values of $\epsilon$ at the time 
 $t_{-}(\epsilon)=t_0-0.01\epsilon^{4/5}$ on the left and at the 
 time  $t_{+}(\epsilon)=t_0+0.01\epsilon^{4/5}$ on the right,  and 
 the corresponding P$_{I}$ solution (\ref{F1}).}
   \label{nlsquintfsechdeltam001}
\end{figure}

\subsection{Non-symmetric initial data for the focusing quintic NLS}
To study solutions to the focusing quintic NLS for  the asymmetric 
initial data (\ref{asym}), we first have to solve equations 
(\ref{asym}) numerically. This is done for values of 
$|x|<15$ in a standard way by solving (\ref{asym}) on some 
Chebyshev collocation points with a Newton iteration. The choice of 
this interval is determined by the fact that the residual of the 
Newton iterate is smaller than $10^{-10}$ on the whole intervall. 
We choose $N_{c}=512$ 
collocation points to ensure that the coefficients of an expansion of 
the solution decrease to machine precision and that the solution is 
thus numerically fully resolved. For values of $|x|>15$, 
we solve (\ref{asym}) asymptotically,
\begin{equation}    
    \label{as1}
    r  
    =-1+(2i)^{1-2\alpha}\exp(-x)+(2i)^{2-4\alpha}\exp(-2x)(-0.5+2\alpha^2\ln(2i)
    +\alpha+\alpha x) +\mathcal{O}(\exp(-3x))
\end{equation}
for $x\to +\infty$ and 
\begin{equation}
    r = 1+i\exp(x)2^{1+2\alpha}+2^{2+4\alpha}\exp(2x)(-0.5+2\alpha^2 
\ln(2)+\alpha x-\alpha)+ \mathcal{O}(\exp(3x))
    \label{as2}
\end{equation}
for $x\to-\infty$.
Machine precision is reached for $|x|>15$ for this asymptotic 
solution. Initial data for $\alpha=0.2$ can be seen in 
Fig.~\ref{nlsfquintasyminitial}.
\begin{figure}[htb!]
\includegraphics[width=0.6\textwidth]{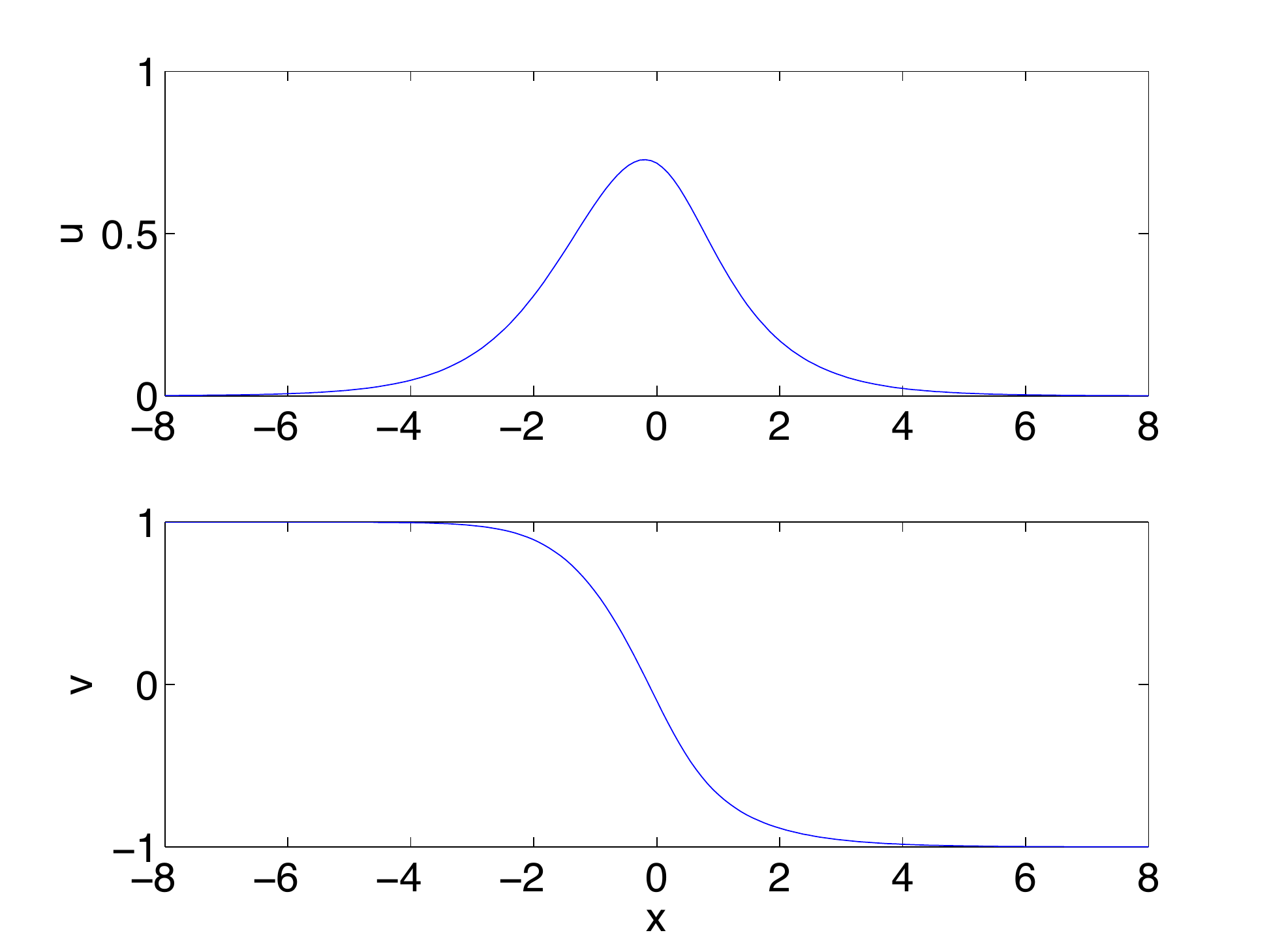}
 \caption{Asymmetric initial data for the focusing quintic NLS 
 equation according to (\ref{Hodasy}) for $\alpha=0.2$.}
   \label{nlsfquintasyminitial}
\end{figure}

To obtain initial data for the NLS equation from $r=v+iu$ in the form 
$\psi=\sqrt{u}\exp(i\int_{x_{0}}^{x}v(x')dx'/\epsilon)$, we have to 
integrate the real part of $r$ with respect to $x$. This is done by 
using an expansion of the solution for $|x|<15$ in terms of 
Chebyshev polynomials via a \emph{discrete cosine transform} (this is the 
reason why the solution was computed on Chebyshev collocation points) 
and applying the well known formula for the integral of Chebyshev 
polynomials. For values of $|x|>15$, the asymptotic formulae 
(\ref{as1}) and (\ref{as2}) are integrated analytically by choosing 
the integration constants to obtain a continuous matching with the 
numerically integrated $v$.  This way we obtain initial data with an 
accuracy of better than $10^{-10}$. We put the Krasny filter to the 
order of this treshold and thus obtain initial data resolved up to 
the level of the Krasny filter.

For $\epsilon=0.1$ the solution to the focusing quintic NLS equation 
for the asymmetric initial data as well as the semiclassical and the 
P$_{I}$ asymptotics (\ref{F1}) can be seen in Fig.~\ref{nlsfquintasyme1}. 
As expected the P$_{I}$ asymptotics gives a much better description of the 
NLS solution close to the critical point of the semiclassical 
solution. The error in the approximation is, however, also not 
symmetric here.
\begin{figure}[htb!]
\subfigure{
 \includegraphics[width=0.5\textwidth]{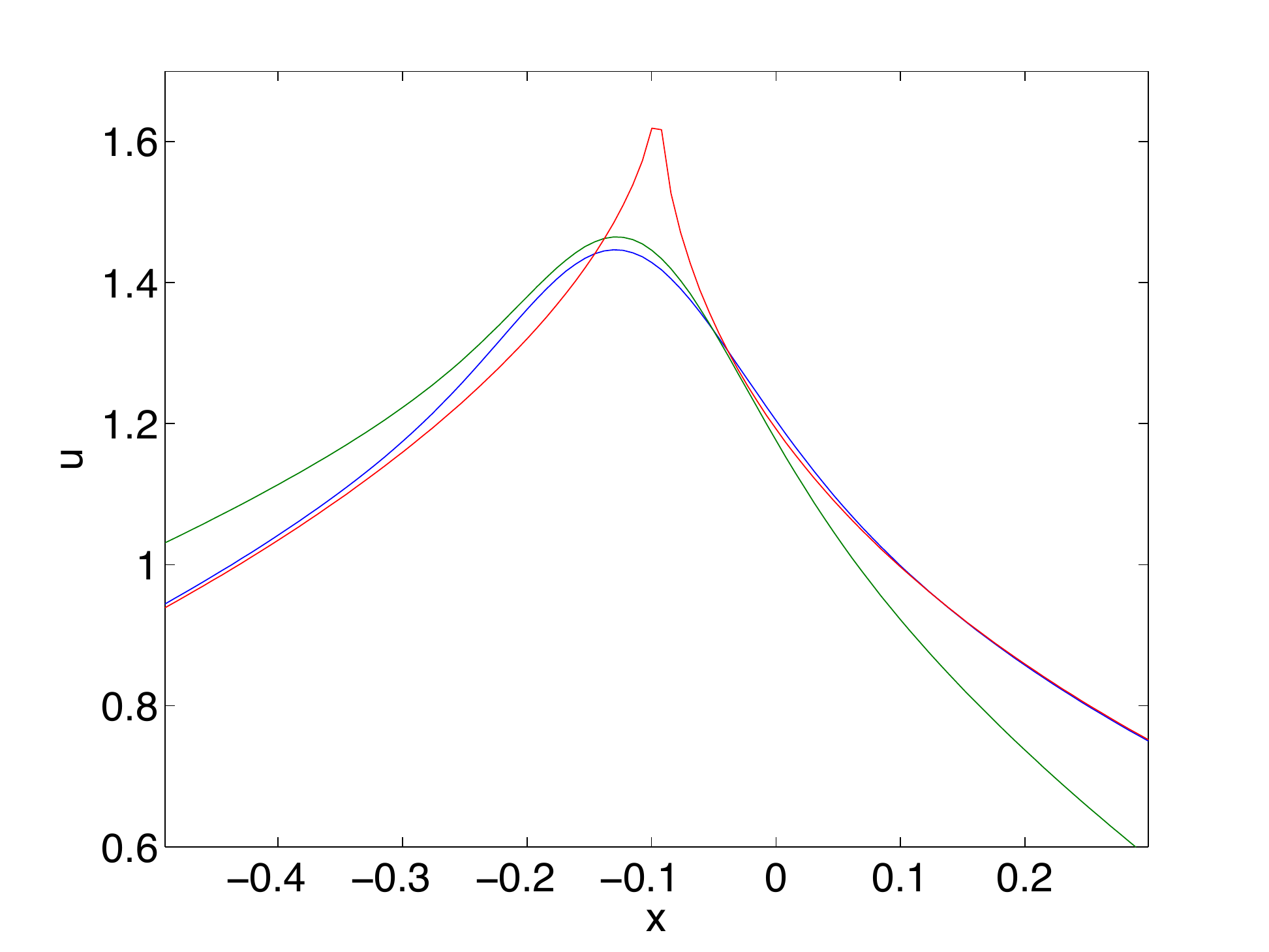}               
 \includegraphics[width=0.5\textwidth]{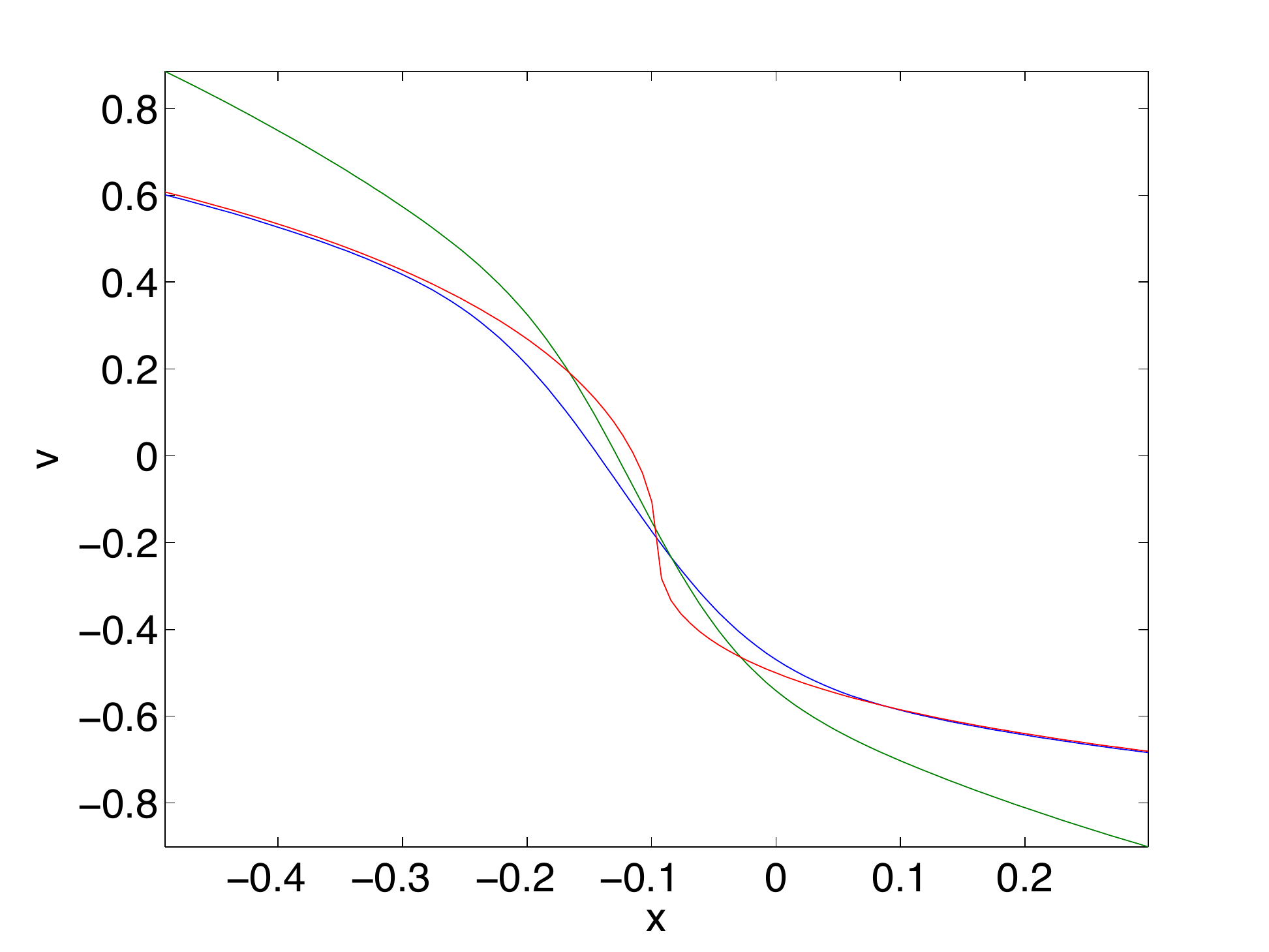}}
\hskip 0.1cm
\subfigure
{ 
 \includegraphics[width=0.5\textwidth]{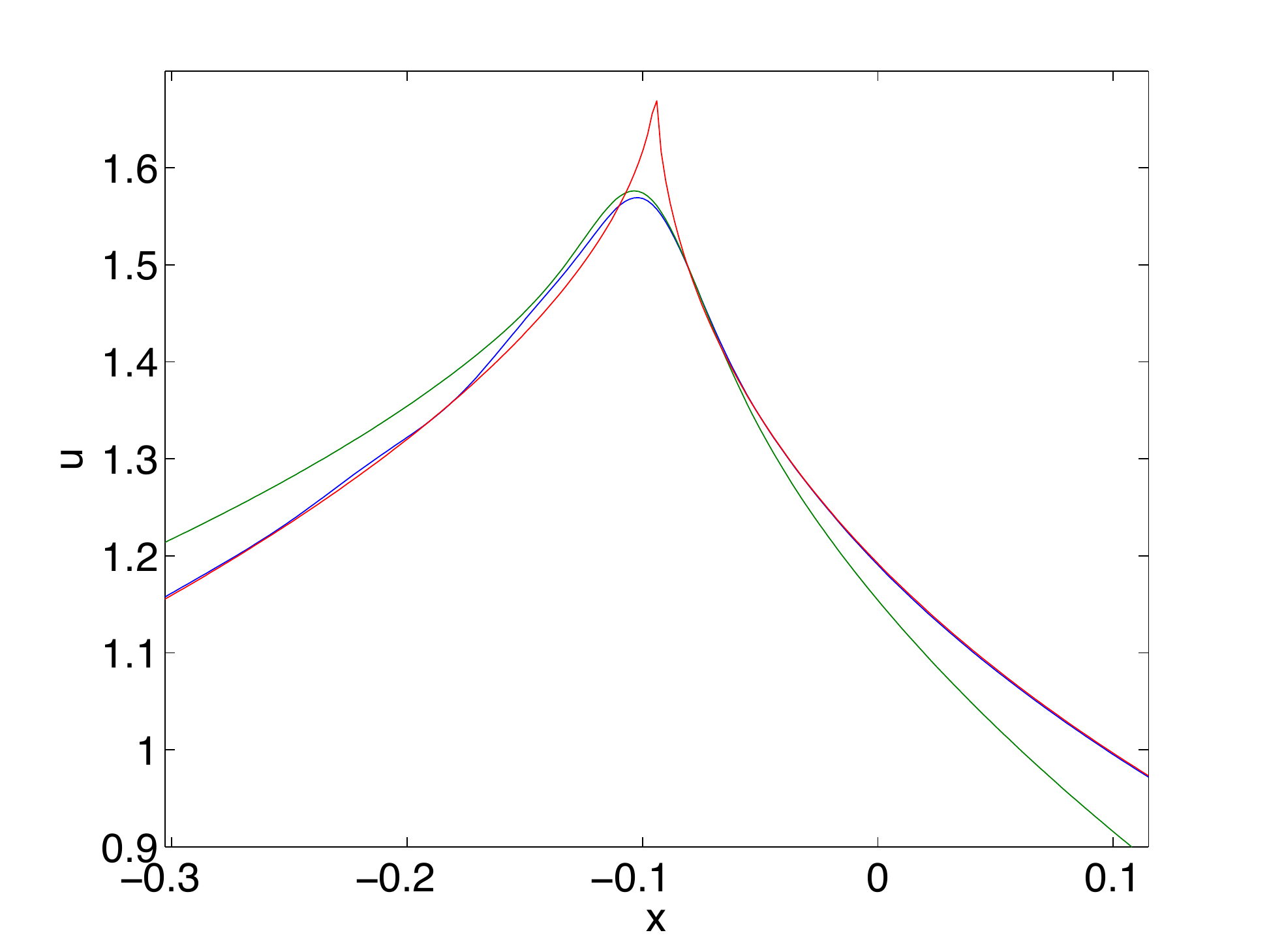}               
 \includegraphics[width=0.5\textwidth]{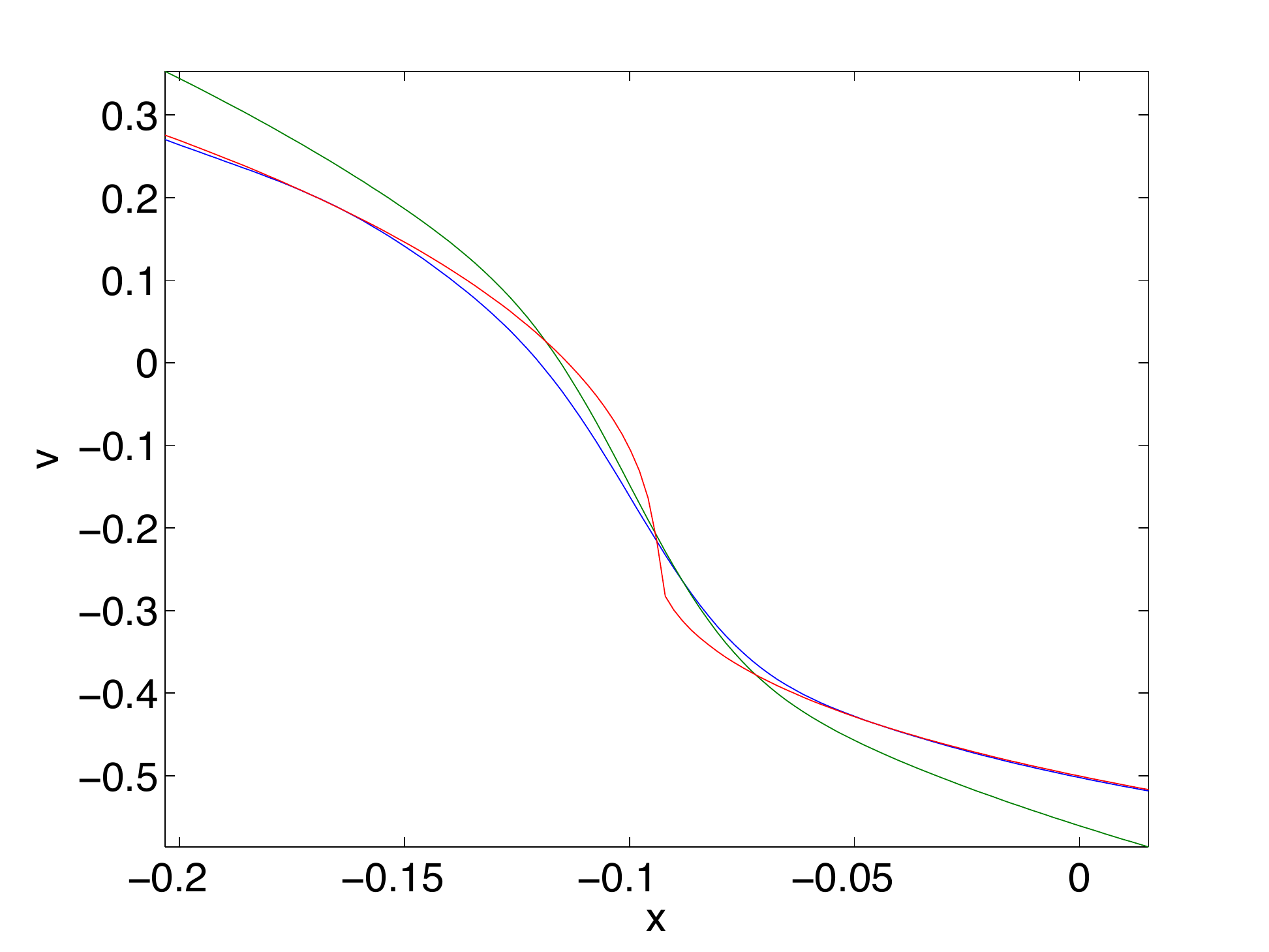}
}
 \caption{Solution to the focusing quintic NLS for the asymmetric 
 initial data as in~(\ref{Hodasy}) for $t =0$ at the critical time in 
 blue, the corresponding semiclassical solution in red and the P$_{I}$ 
 asymptotics (\ref{F1}) in green; on the left the function $u$, on 
 the right the function $v$. The upper figures are for  
 $\epsilon=0.1$, the lower ones for $\epsilon=0.02$.}  
\label{nlsfquintasyme1}
\end{figure}

The agreement gets even better for smaller $\epsilon$. We can reach 
values as low as $\epsilon=0.02$. For smaller $\epsilon$, the blow-up 
singularity of quintic NLS solutions (see below) seems to be too 
close to the critical time of the semiclassical solution which breaks 
the code. The case $\epsilon=0.02$ is, however, numerically fully 
resolved. As can be seen in the lower row of Fig.~\ref{nlsfquintasyme1}, the 
agreement is as expected. Note that also in this case  the $x$-axes of the bottom figures 
 have been rescaled  by a factor $\epsilon^{4/5}$. 

\subsection{`Dark' initial data}
Focusing NLS equations do not have dark solitons as exact solutions, 
i.e., solutions which tend asymptotically to a non-zero constant and 
which vanish for finite values of $x$. But it is mathematically 
interesting to study how initial data of this form lead to a break-up 
of the semiclassical equations, and how the corresponding NLS 
solution behaves in the vicinity of the critical point. We consider 
here initial data of the form $\psi_{0}=\tanh^{2}x$. The solution 
breaks here in the form of two cusps symmetric with respect to $x=0$. 
The critical time is at $t_0=0.9041\ldots$, the cusps form at 
$x_{c}=\pm 1.8723\ldots$. The corresponding solution can be seen in 
Fig.~\ref{nlsquintfdarkt}.
\begin{figure}[thb!]
  \includegraphics[width=0.7\textwidth]{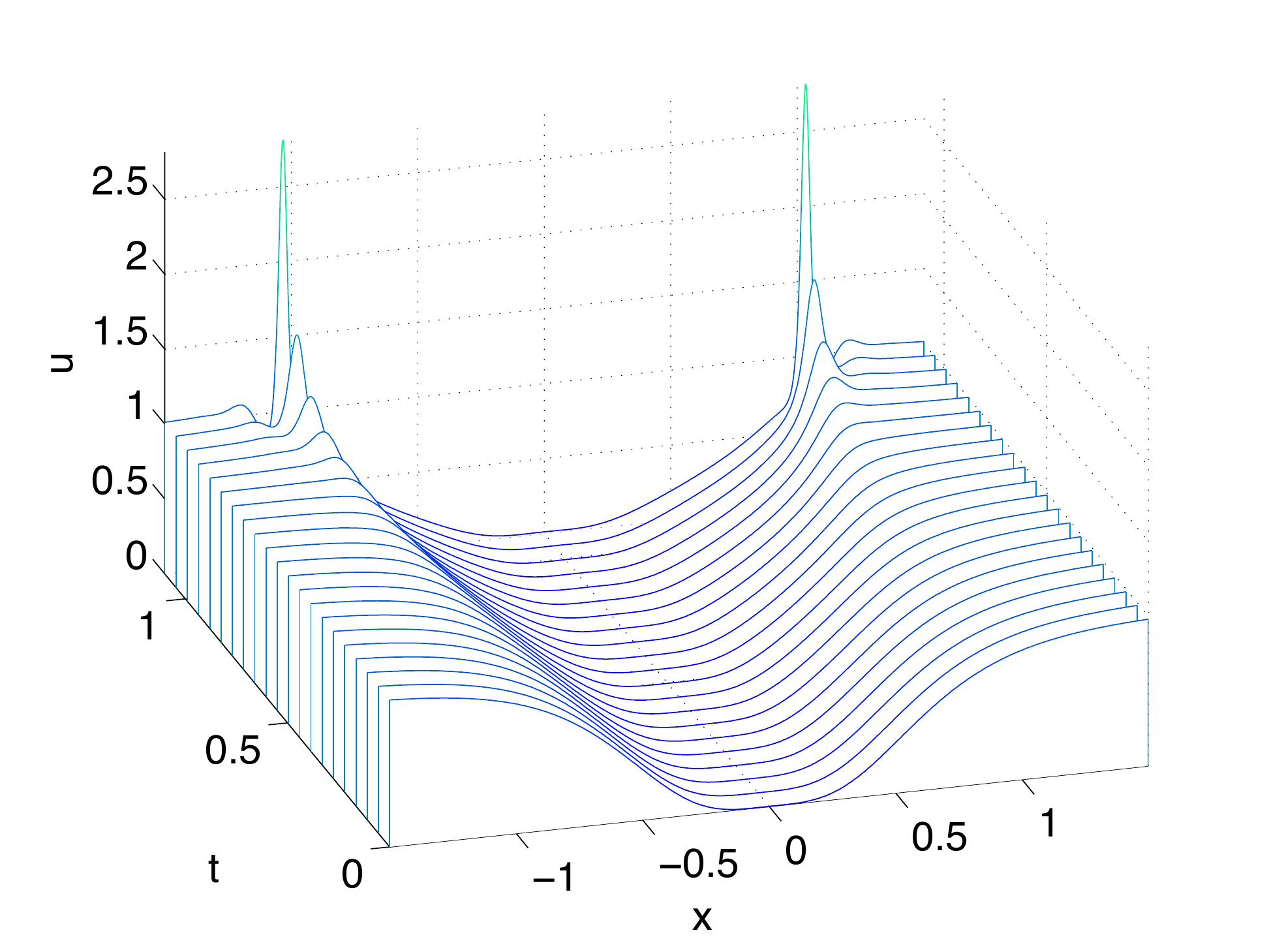}
 \caption{Solution to the focusing quintic NLS equation for the dark initial data 
 $\psi_{0}(x)=\tanh^{2}x$ and $\epsilon=0.1$. The critical time is
 $t_0=0.9041\ldots$}
 \label{nlsquintfdarkt}
\end{figure}
\begin{figure}[ht!]
\subfigure
{
 \includegraphics[width=0.5\textwidth]{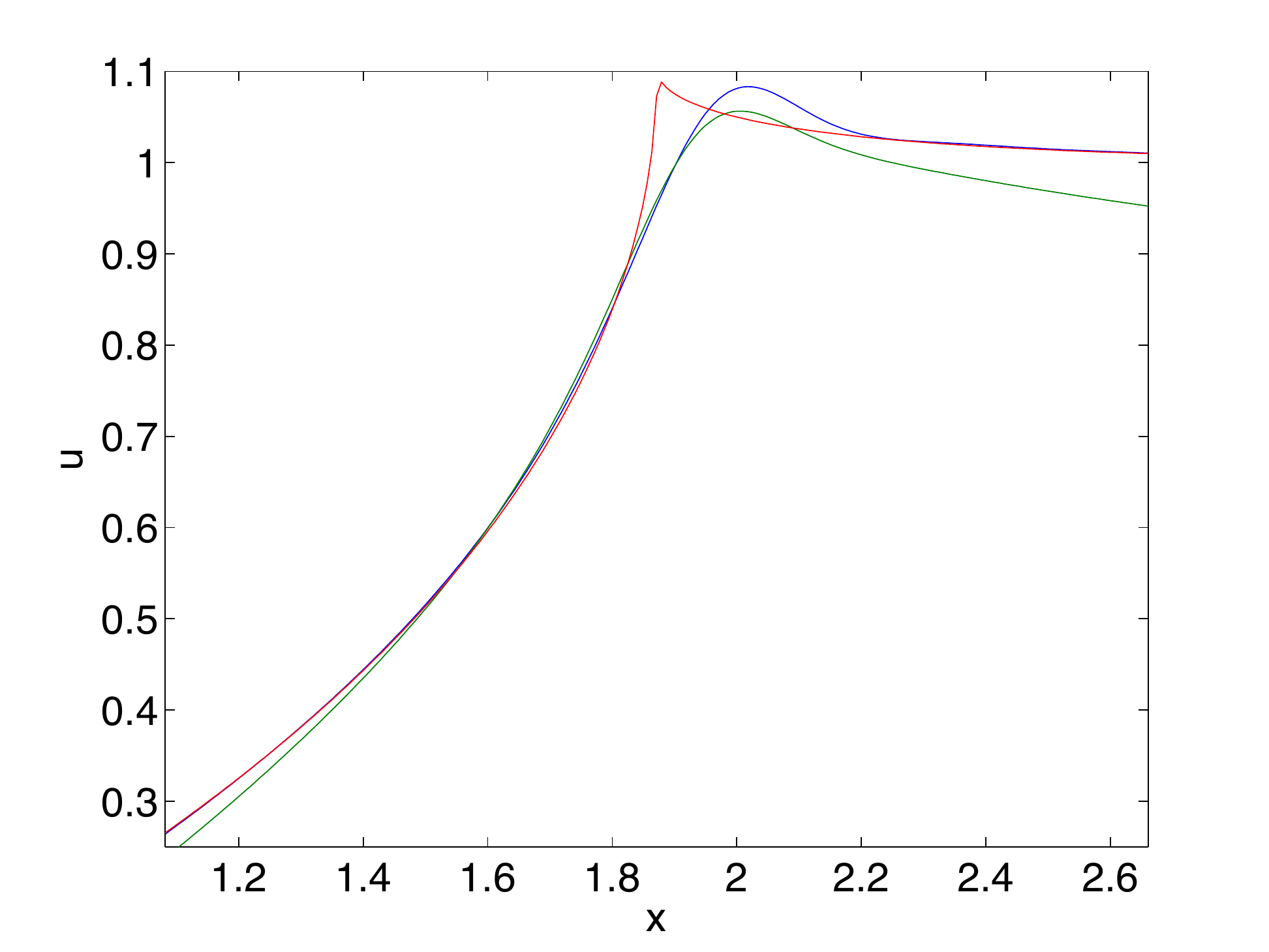}               
 \includegraphics[width=0.5\textwidth]{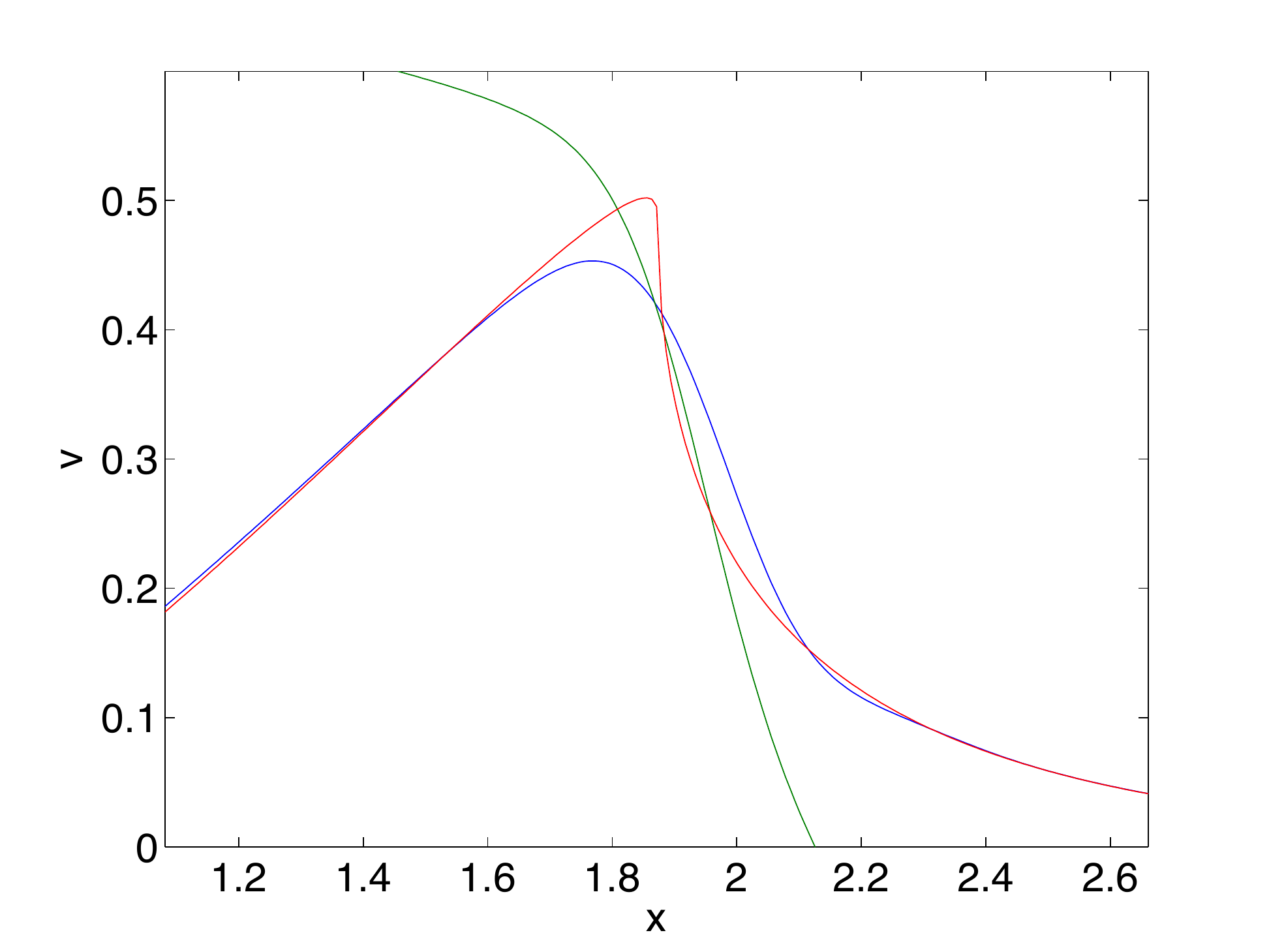}}
\vskip 01.cm
\subfigure
{
 \includegraphics[width=0.5\textwidth]{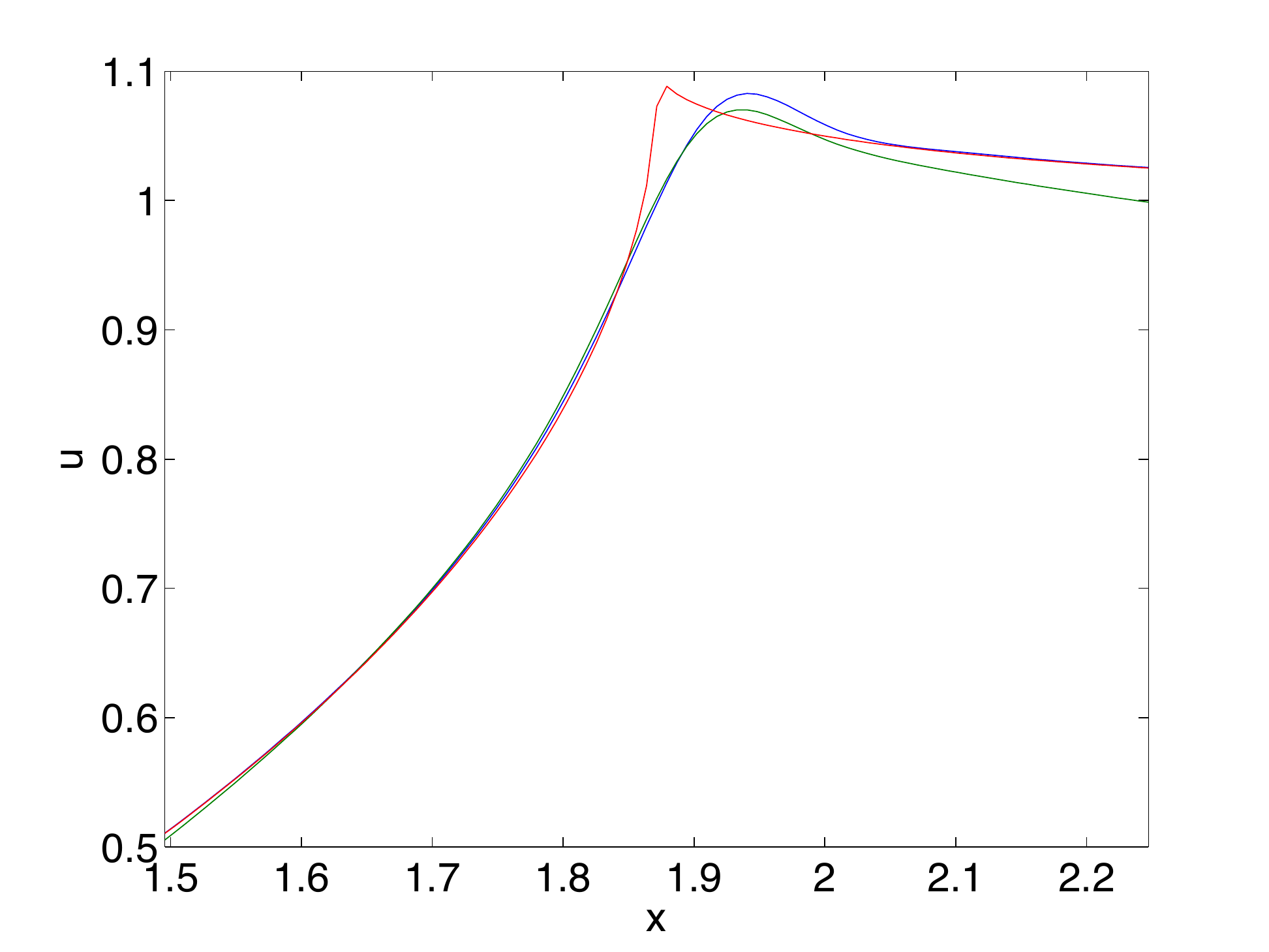}               
 \includegraphics[width=0.5\textwidth]{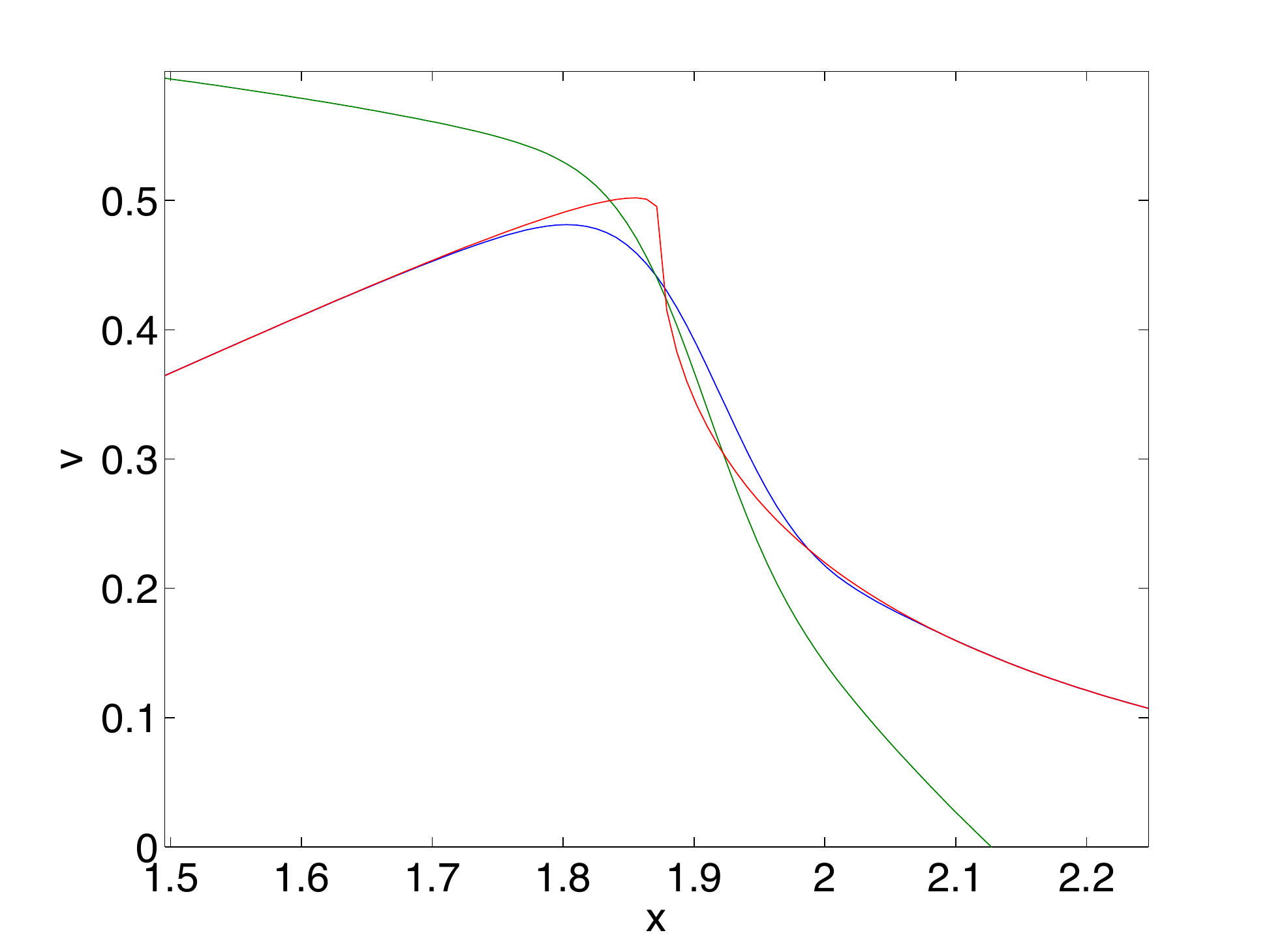}
}
 \caption{Solution to the focusing quintic NLS for the dark 
 initial data $\psi_{0}=\tanh^{2}x$ at the critical time in 
 blue, the corresponding semiclassical solution in red and the P$_{I}$ 
 asymptotics (\ref{F1}) in green; on the left the function $u$, on 
 the right the function $v$. For the figures in the upper row
 $\epsilon=0.1$, for the ones in the lower row $\epsilon=0.04$.}  
\label{nlsfquintdark1}
\end{figure}
For $\epsilon=0.1$ the solution to the focusing quintic NLS equation 
for the dark initial data as well as the semiclassical and the 
P$_{I}$ asymptotics (\ref{F1}) can be seen in Fig.~\ref{nlsfquintdark1}. 
As expected the P$_{I}$ asymptotics gives a much better description of the 
NLS solution close to the critical point of the semiclassical 
solution. 
The agreement gets better for smaller $\epsilon$. We can reach 
values as low as $\epsilon=0.04$, where the modulation instability 
leads to problems for smaller values of $\epsilon$ because of the 
asymptotically non-vanishing solution.  The case $\epsilon=0.04$ is, however, 
numerically accessible. As can be seen in the bottom figures of  Fig.~\ref{nlsfquintdark1}, the 
agreement is as expected.

\subsection{Blow-up}
For the cubic focusing NLS, solutions in the semiclassical limit for 
times $t\gg t_0$ develop a zone of rapid modulated oscillations as 
can be seen for instance in Fig.~\ref{nlsquintfblowup}. The central 
hump close to the critical time splits into several humps of 
smaller amplitude. For the quintic NLS on the other hand it is 
known, see e.g.~\cite{MM}, that initial data with negative energy 
have a blow-up in finite time. For the NLS with the semiclassical 
parameter $\epsilon$ we consider in this paper, this will be always 
the case for sufficiently small $\epsilon$. Thus the  solution of the 
quintic NLS looks for small $\epsilon$ very differently from the 
solution to the cubic NLS for the same initial data and the same 
value of $\epsilon$ as can be seen in Fig.~\ref{nlsquintfblowup}. The 
central hump develops in this case into a blow-up.
\begin{figure}[thb!]
  \includegraphics[width=0.5\textwidth]{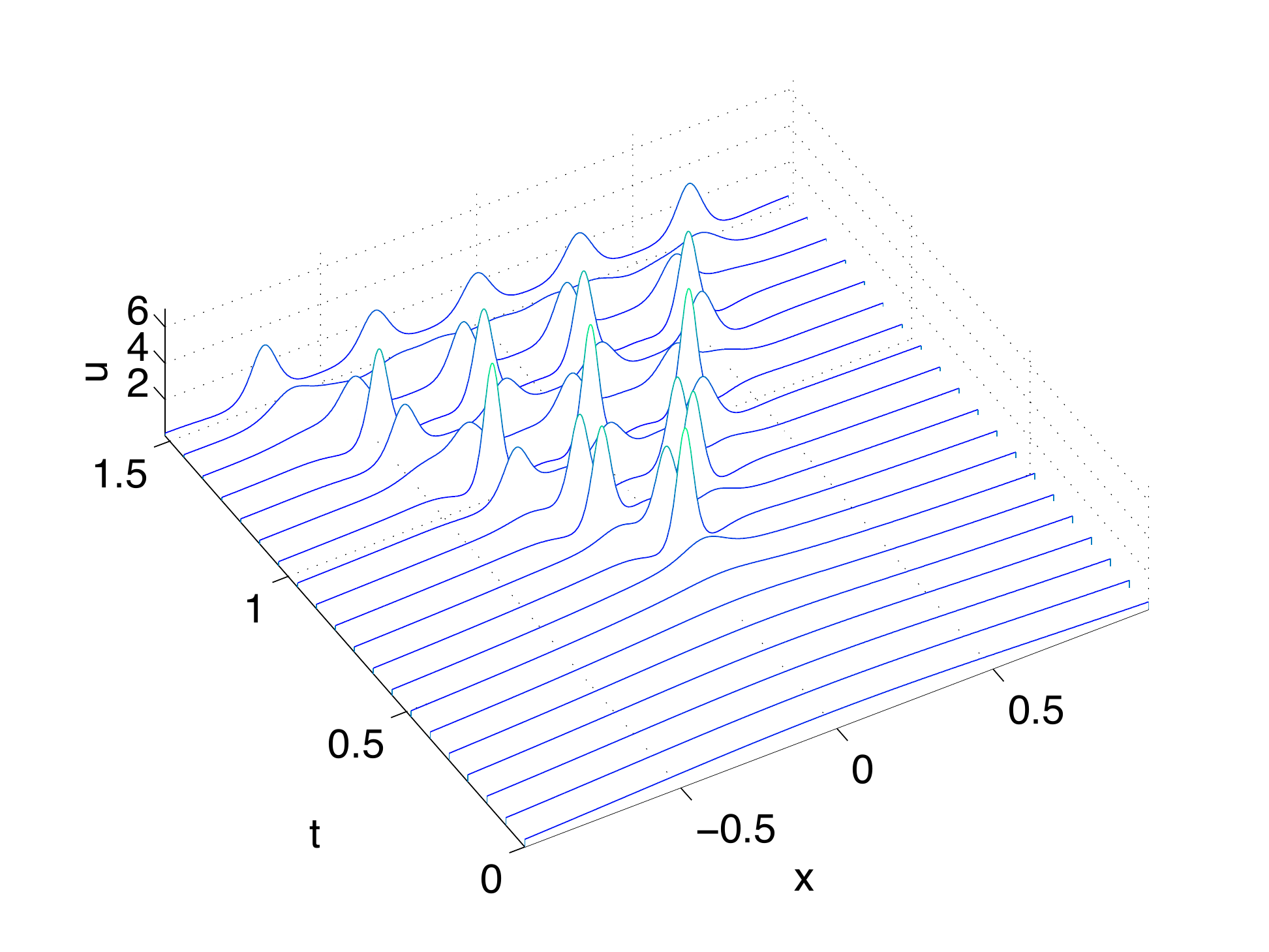}
  \includegraphics[width=0.5\textwidth]{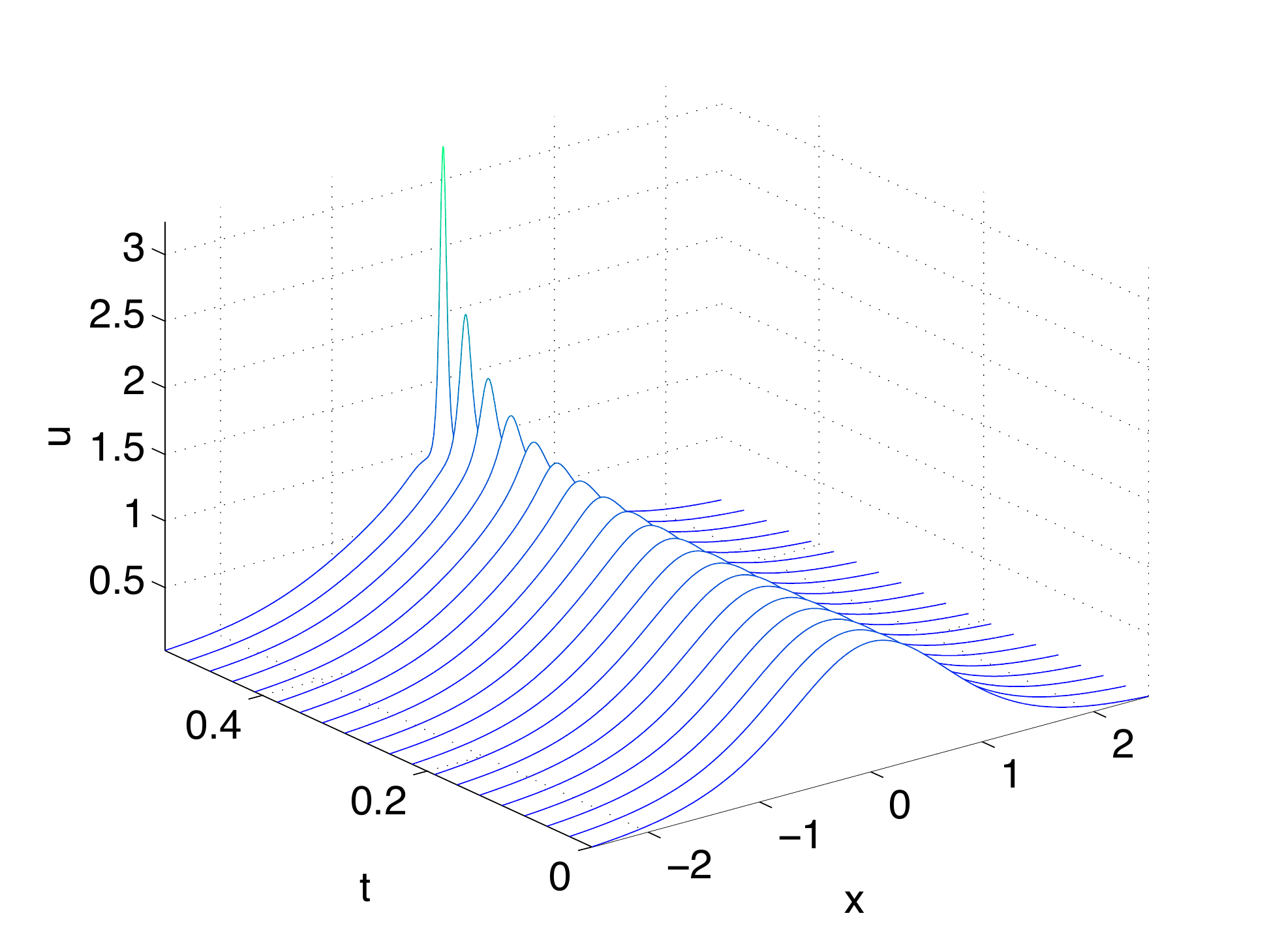}
 \caption{Solution to the focusing  NLS equation for the initial data 
 $\psi_{0}(x)=\mbox{sech }x$ and $\epsilon=0.1$; on the left the 
 solution for the cubic NLS, on the right the solution to the quintic 
 NLS.}
 \label{nlsquintfblowup}
\end{figure}

For obvious reasons it is impossible to treat a blow-up exactly numerically, 
but the numerical solution can get sufficiently close to this case. Driscoll's 
composite Runge-Kutta method produces an overflow 
error close to the $L_{\infty}$ blow-up encountered here because of 
the term $|\psi|^{4}\psi$. We stop the code when this happens and 
note the last time with finite value of $\psi$ as a lower bound 
$t_{B}$ for the blow-up time. The error in the determination of the 
blow-up time with this method is largest for larger $\epsilon$. 
Using linear regression we find for 
$\ln (t_{B}-t_0)=a\ln \epsilon +b$ for values of 
$\epsilon=0.01,0.02,\ldots,0.1$ the value $a=0.83$ close to $4/5$ with
standard deviation $\sigma_{a}=0.0439$, $b=-0.1267$  with standard deviation  $\sigma_b=0.0138$ correlation coefficient 
$r=0.999$, see Fig.~\ref{blowupreg}.
\begin{figure}[thb!]
\hskip 0.5cm
  \includegraphics[width=0.7\textwidth]{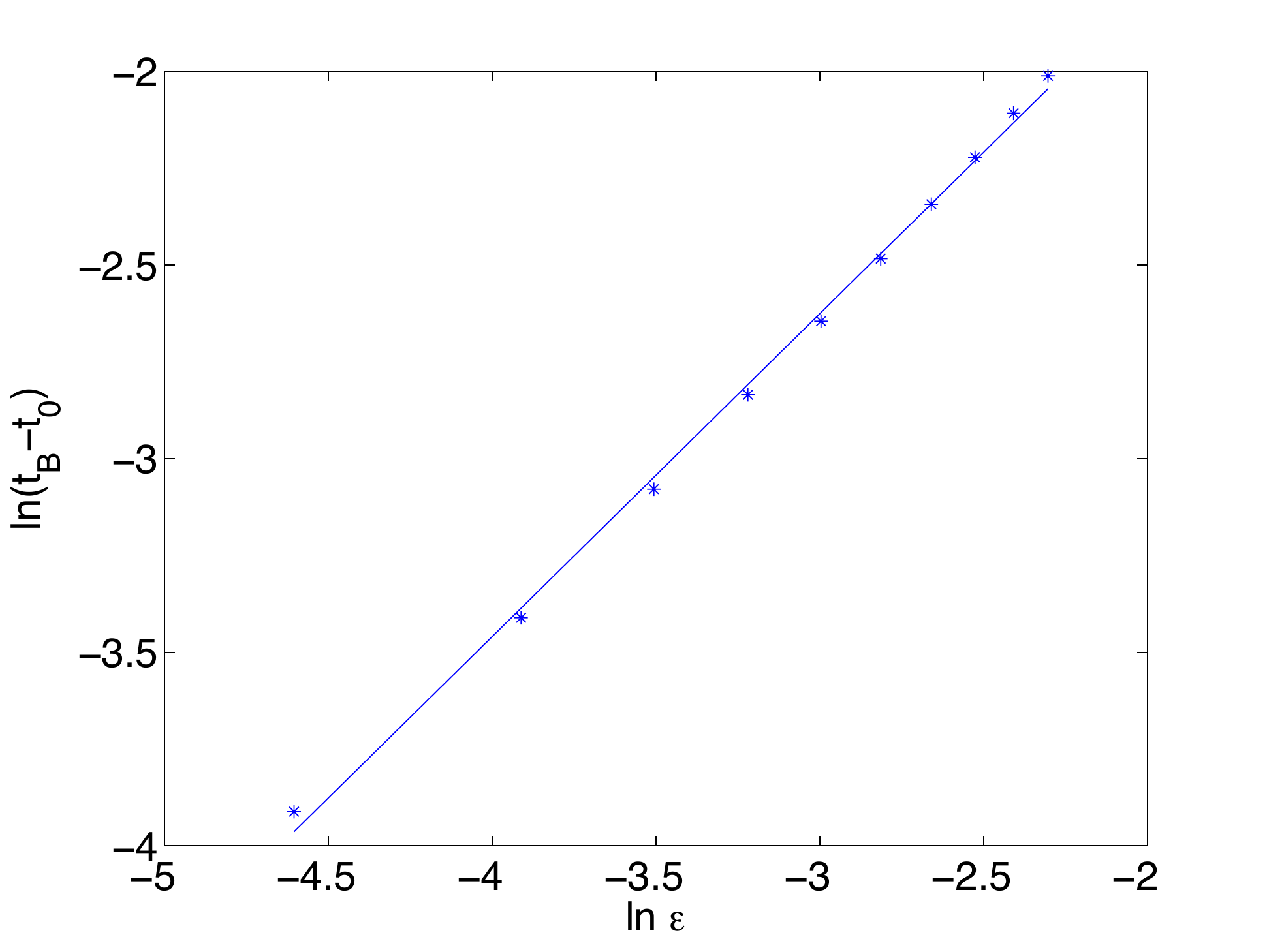}
 \caption{The blow-up time as a function of $\epsilon$ for quintic NLS with $\mbox{sech}x$ initial data.}
 \label{blowupreg}
\end{figure}

As expected from the P$_{I}$ solution (\ref{F1}), 
the time scales with $\epsilon^{4/5}$. Since we expect the error in 
the determination of the blow-up time to decrease with $\epsilon$, a 
slightly stronger decrease with $\epsilon$ of the time $t_{B}$ than predicted is no surprise.

It is an interesting question whether the blow-up time in the limit 
$\epsilon\to 0$ is related to 
the first pole of the tritronqu\'ee solution on the negative real 
axis. In \cite{jk} it was shown that the first pole  is located at 
\begin{equation}
\label{polePI}
\xi_{pole}=-2.3841687\ldots.
\end{equation}
Recalling formula (\ref{xi}) for the argument of the tritronqu\'ee solution in the approximation of the NLS solution  near the point of elliptic umbilic catastrophe
\[
\xi=-i\left(\dfrac{u_0}{V_0'}\left(3V_0'+u_0V''_0\right)^2\dfrac{re^{i\psi}}{3\epsilon^4}\right)^{\frac{1}{5}}(x-x_0-(v_0+i\sqrt{u_0V'_0})(t-t_0)).
\]
one can see that for quintic NLS and  $\mbox{sech}x$ initial data, 
 the point of elliptic umbilic catastrophe is at $x_0=0$, and for 
 symmetry reasons, the blow up  is  at $x_B=0$. 
Using the above formula,  with $V(u)=\frac{u^2}{2}$,   so that $V'_0=u_0$, $V_0''=1$ and 
\[
\;\;a_+=-\dfrac{i}{re^{i\psi}}=-i\dfrac{1}{4(u_0)^2}\dfrac{3u_0-2}{(u_0-1)^{\frac{3}{2}}}.
\] 
with $u_0\simeq=1.5858$ determined in (\ref{shock5nls}) for this specific example, the blowup time $t_B$ is then  
conjectured  to satisfy the equation
\[
\xi_{pole}\simeq-2.3841\simeq -2.0324\dfrac{t_B-t_0}{\epsilon^{\frac{4}{5}}}.
\]
which gives a value of $|b|=\ln(2.3841/2.0324)=0.1596$, in reasonable agreement 
with the numerically found value $|b|\sim 0.1267$.

\subsection{Focusing nonlocal NLS}
We will study the small dispersion limit of the 
nonlocal NLS (\ref{full_NNLS_complex}) close to the break-up of the 
corresponding semiclassical solutions. We will concentrate on values 
of $\eta$ such that $\eta\epsilon^{2}\ll1$ for all studied values of 
$\epsilon$. For both cases we will consider the initial data 
$\psi_{0}=\mbox{sech }x$. 
The effect of the nonlocality in (\ref{full_NNLS_complex}) is to reduce the 
focusing effect of the focusing NLS. This means the larger $\eta$, 
the smaller the value for the maximum at the critical time of the 
corresponding semiclassical solution, and the less pronounced the 
focusing of the maximum, i.e., smaller gradients in the solution. This effect can be clearly seen in 
Fig.~\ref{nlsfnonloc2eta}. 
\begin{figure}[htb!]
\includegraphics[width=0.6\textwidth]{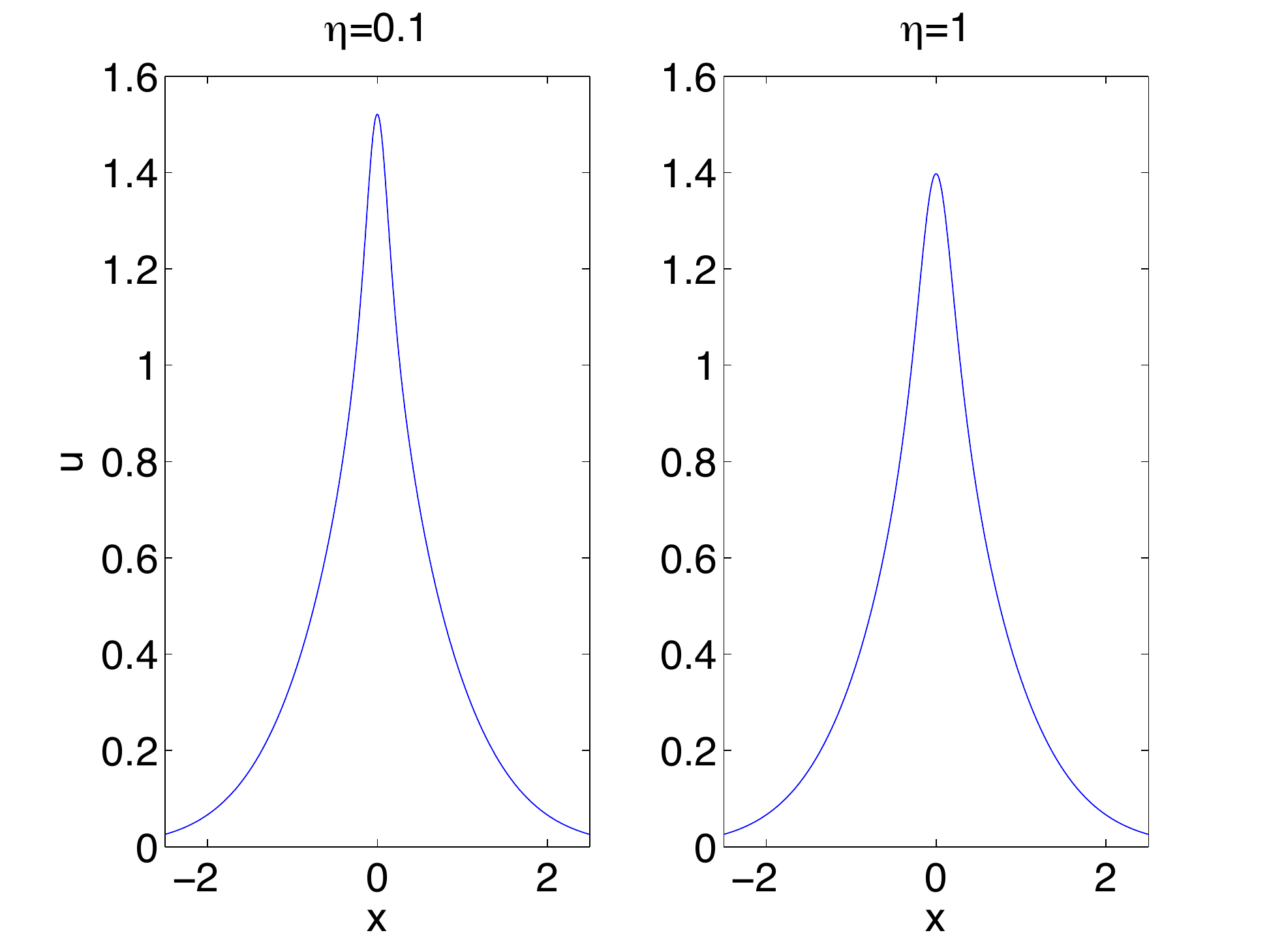}
 \caption{Solution to the focusing nonlocal NLS equation 
 (\ref{full_NNLS_complex}) for the initial data 
 $\psi_{0}(x)=\mbox{sech }x$ and $\epsilon=0.1$ at the time 
 $t_0=0.5$ for two values of $\eta$.}
   \label{nlsfnonloc2eta}
\end{figure}

For larger times the oscillations are suppressed with respect to the 
case $\eta=0$ as can be seen in Fig.~\ref{nlsnonlocf2tc} (compare 
with Fig.~\ref{nlsquintfblowup} on the left).
\begin{figure}[thb!]
  \includegraphics[width=0.5\textwidth]{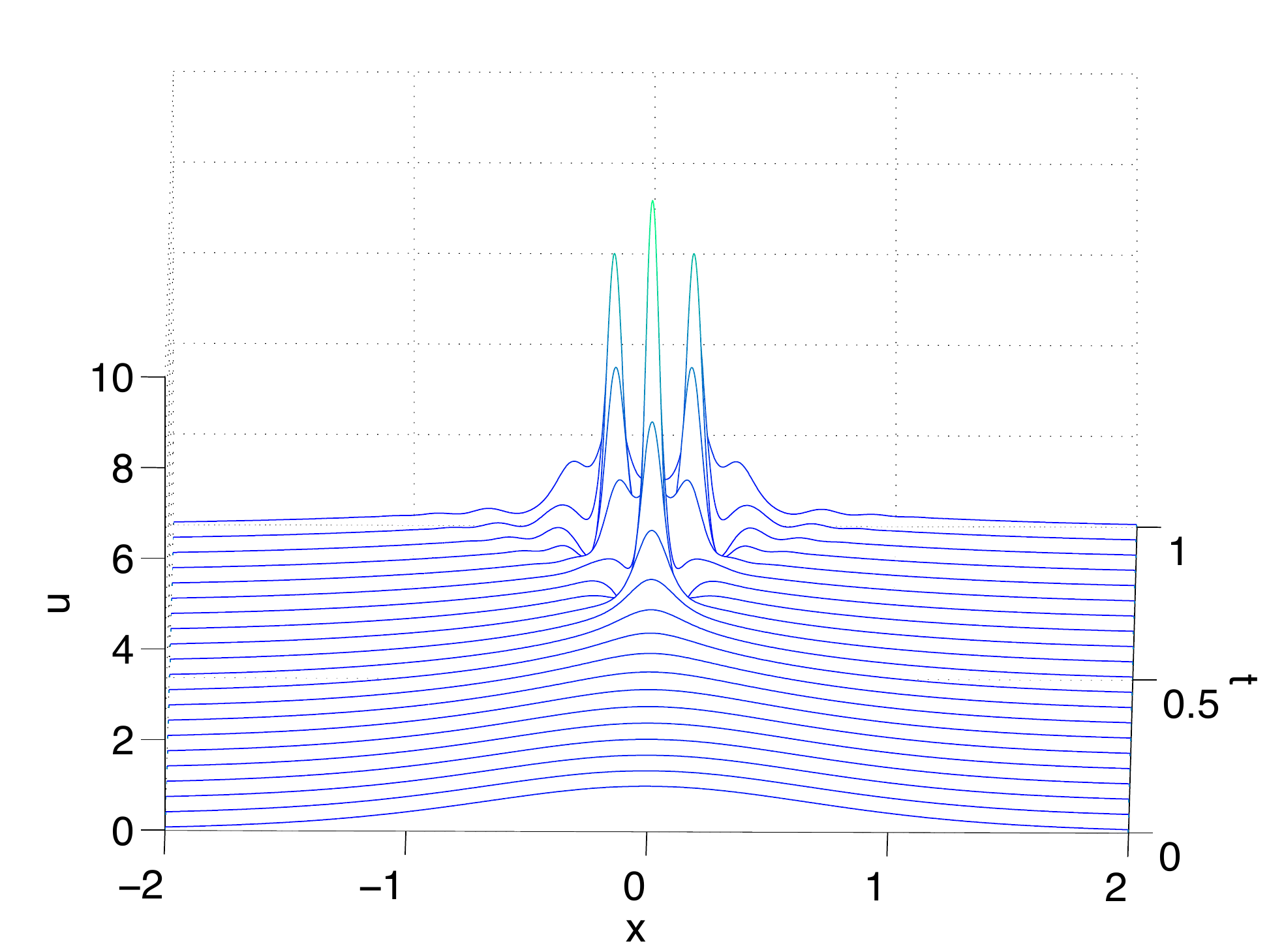}
  \includegraphics[width=0.5\textwidth]{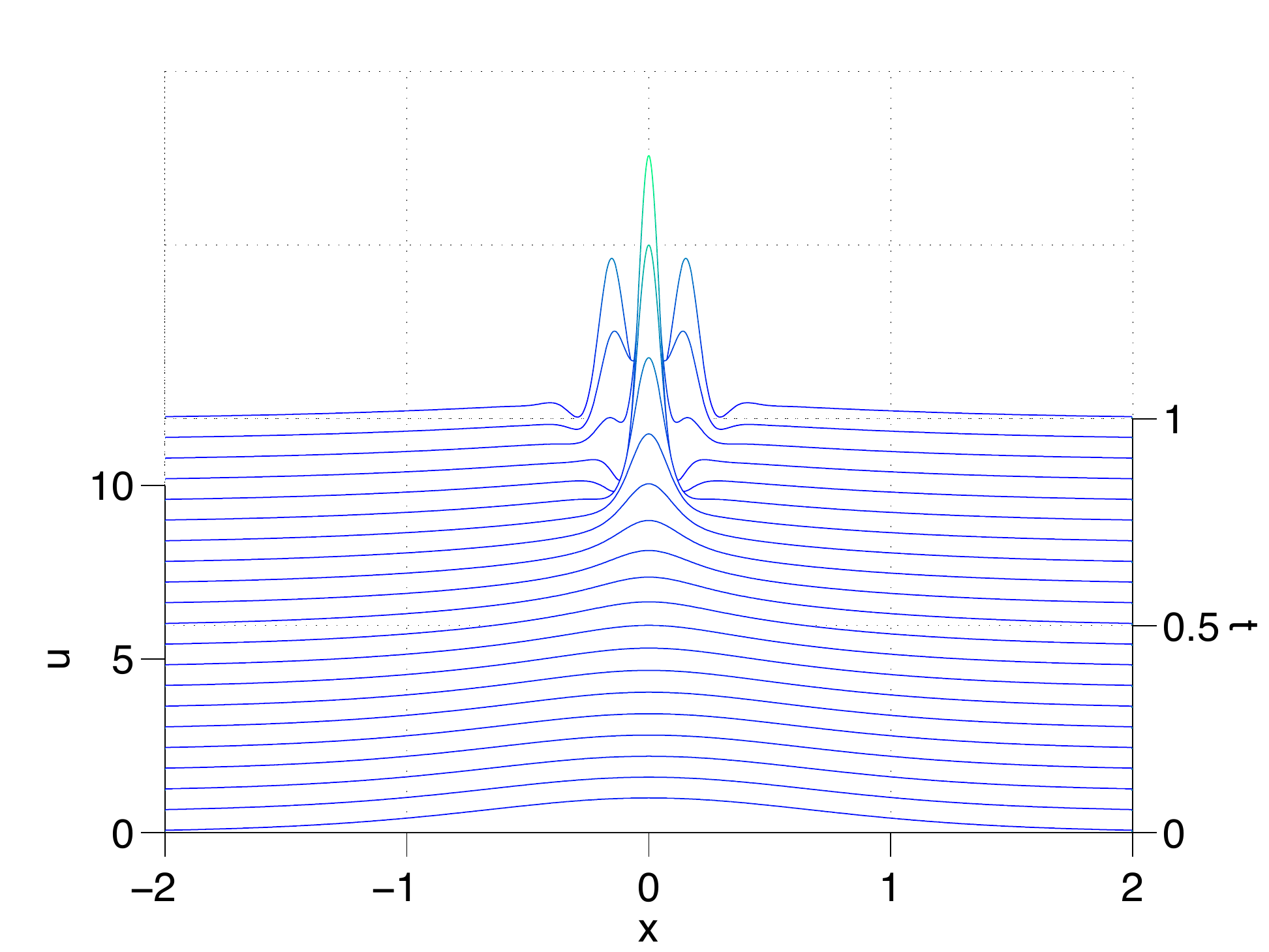}
 \caption{Solution to the focusing nonlocal NLS equation 
 (\ref{full_NNLS_complex}) for the initial data 
 $\psi_{0}(x)=\mbox{sech }x$ and $\epsilon=0.1$; for $\eta=0.1$ 
 on the left, for $\eta=1$ on the right. The critical time is
 $t_0=0.5$.}
 \label{nlsnonlocf2tc}
\end{figure}

At the critical time, the tritronqu\'ee solution to P$_{I}$ gives as 
expected a much better description of the nonlocal NLS solution than 
the semiclassical solution as can be seen for $\eta=0.1$ for $u$ in 
Fig.~\ref{nlsfnonloceta012eu}. The quality of the approximation 
increases visibly for smaller $\epsilon$. Note that the $x$-axes are
rescaled with a factor $\epsilon^{4/5}$.
\begin{figure}[htb!]
\includegraphics[width=0.6\textwidth]{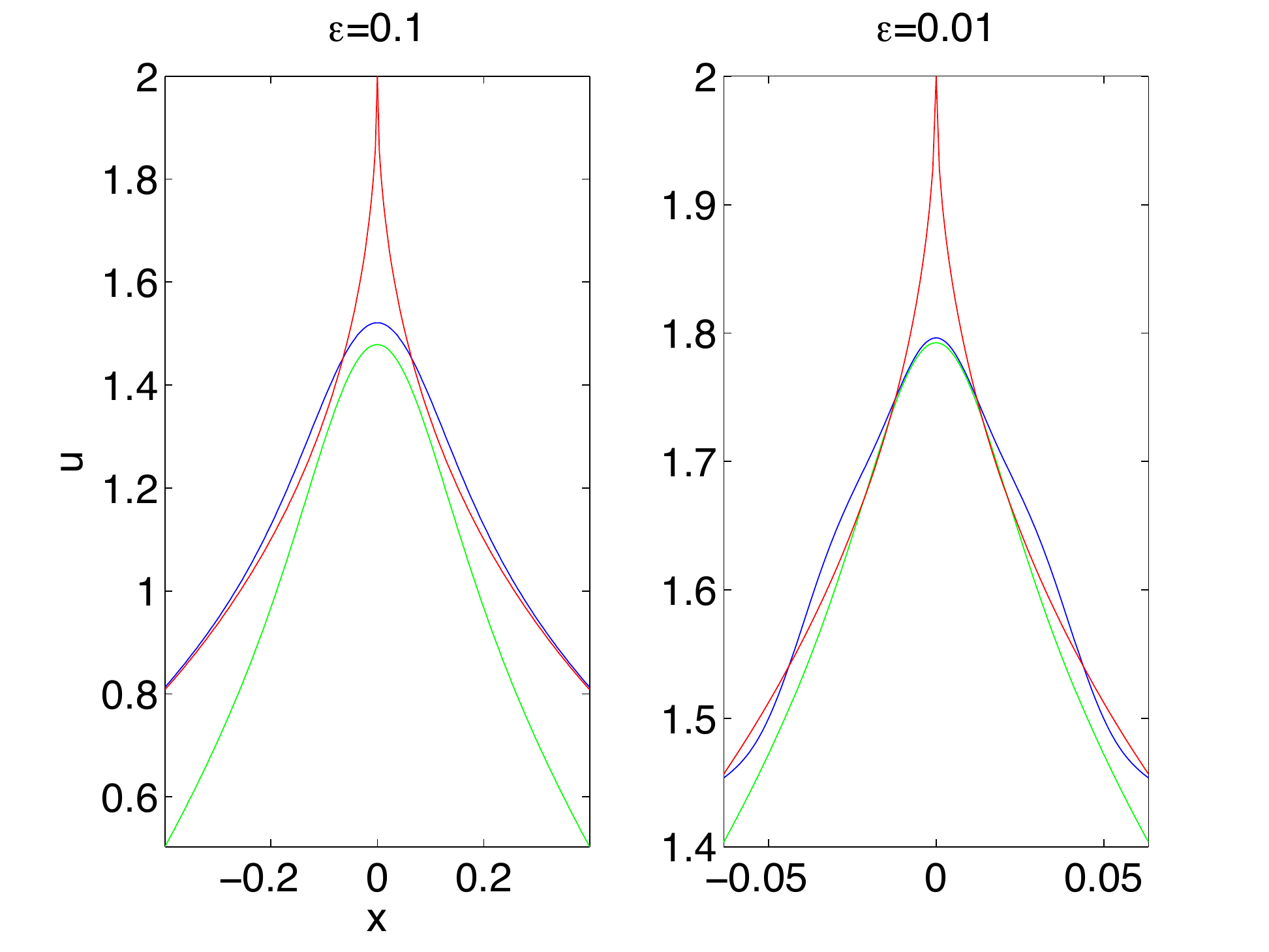}
 \caption{Solution $u$ to the focusing nonlocal NLS equation  (\ref{full_NNLS_complex}) for the initial data 
 $\psi_{0}(x)=\mbox{sech }x$ and $\eta=0.1$ at the time 
 $t_0=0.5$ for two values of $\epsilon$ in blue,  the corresponding semiclassical solution in 
 red and the P$_{I}$ solution (\ref{F1}) in green.}
 \label{nlsfnonloceta012eu}
\end{figure}

The corresponding plots for $v$ can be seen in 
Fig.~\ref{nlsfnonloceta012eu}. The same behavior as for $u$ is visible.
\begin{figure}[htb!]
\includegraphics[width=0.6\textwidth]{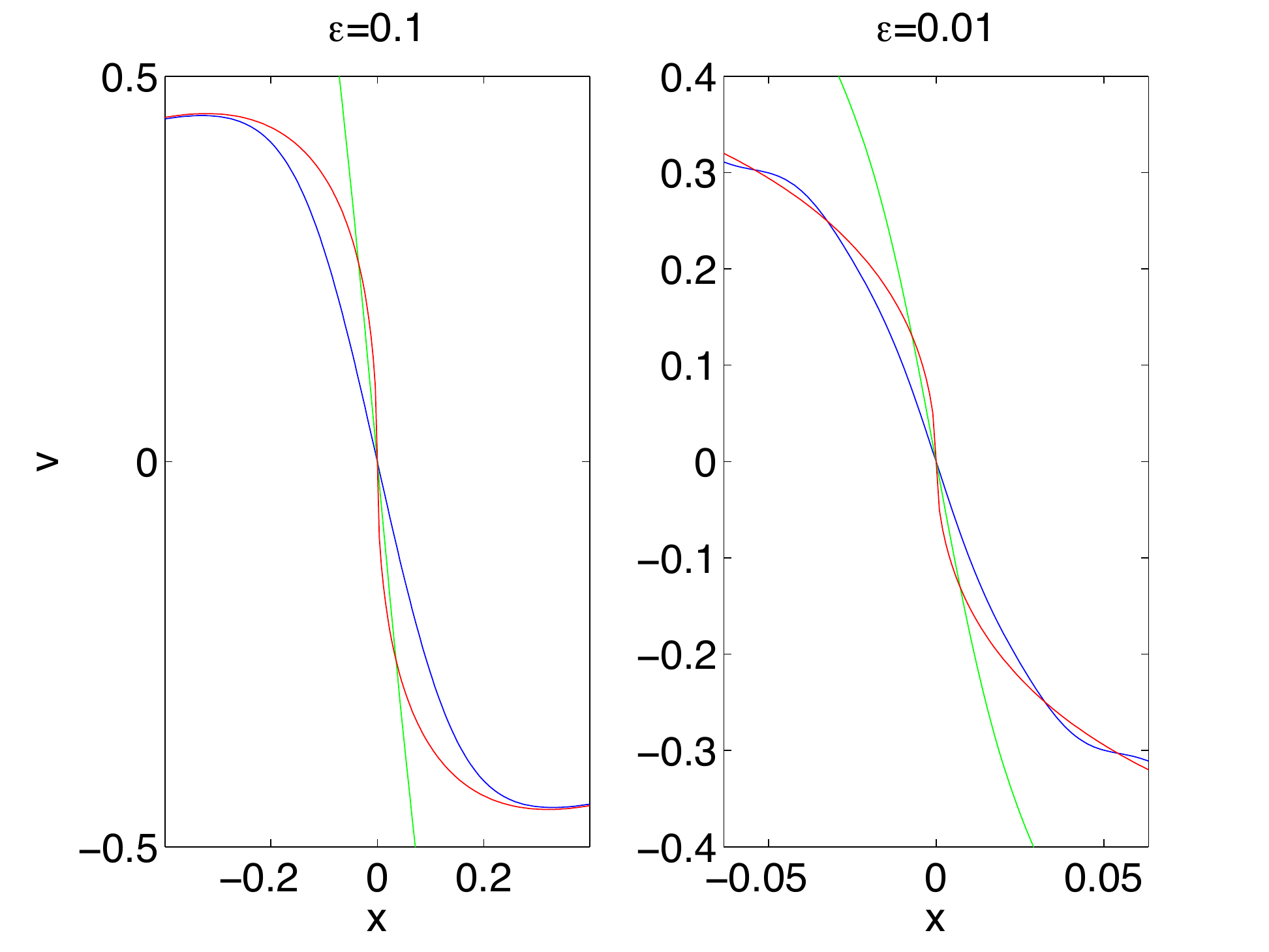}
 \caption{Solution $v$ to the focusing nonlocal NLS equation  (\ref{full_NNLS_complex}) for the initial data 
 $\psi_{0}(x)=\mbox{sech }x$ and $\eta=0.1$ at the time 
 $t_0=0.5$ for two values of $\epsilon$ in blue,  the corresponding semiclassical solution in 
 red and the P$_{I}$ (\ref{F1}) in green.}
 \label{nlsfnonloceta012ev}
\end{figure}

For larger values of $\eta$, the 
agreement is less good for both the semiclassical and the P$_{I}$ 
asymptotics. This is clear for the former since the semiclassical 
solution is independent of $\eta$, and since the focusing effect of 
the nonlocal
NLS is less pronounced for larger values of $\eta$. The P$_{I}$ 
asymptotics takes this into account, the value of its maximum is 
also reduced, but more so than for the nonlocal NLS which implies 
that the agreement between the two solutions is best for $\eta=0$, 
i.e., the cubic NLS. The approximation gets, however, better for 
smaller $\epsilon$ as can be seen for $\eta=1$ in 
Fig.~\ref{nlsfnonloceta12eu}. 
\begin{figure}[htb!]
\includegraphics[width=0.6\textwidth]{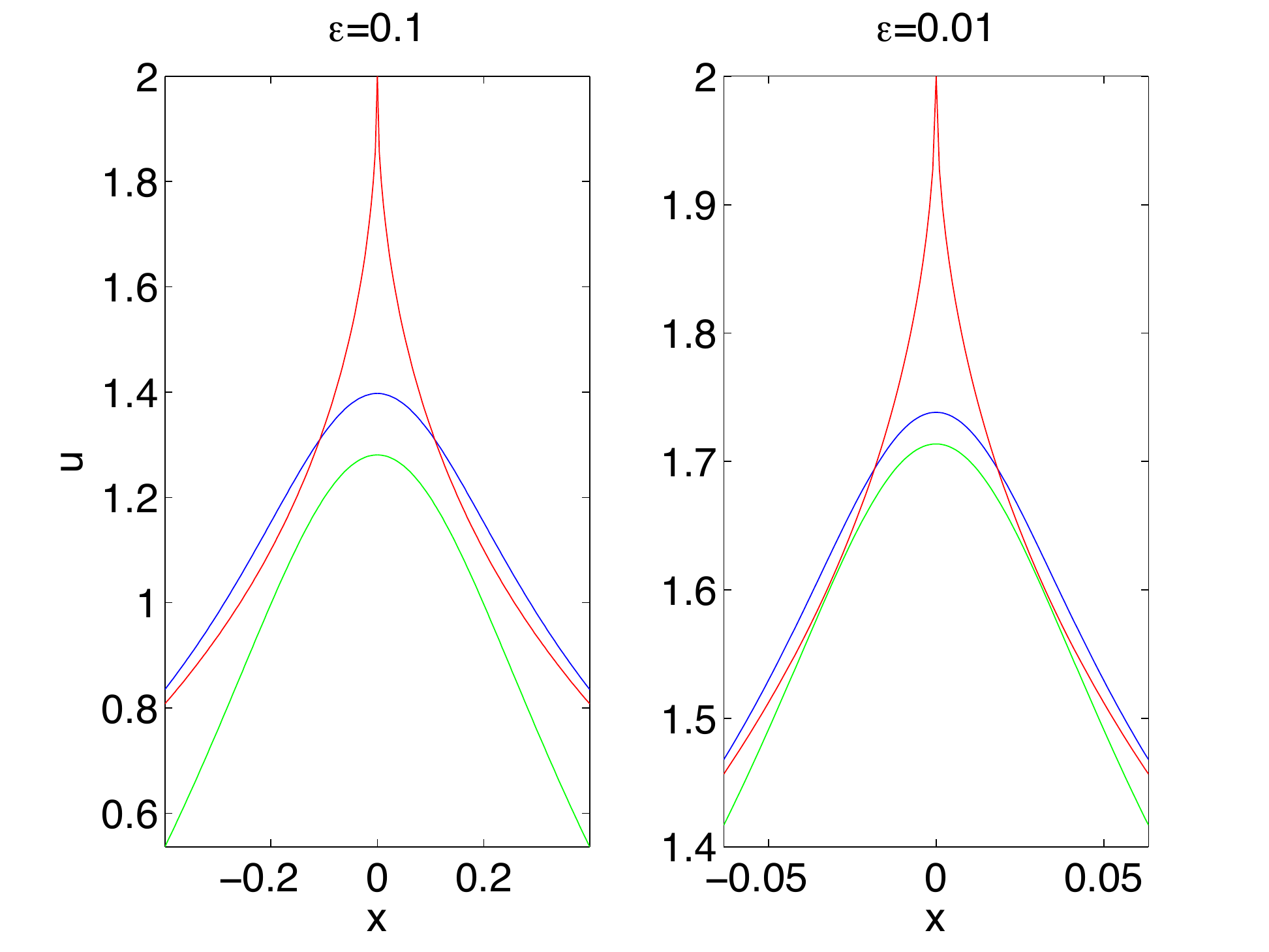}
 \caption{Solution $u$ to the focusing nonlocal NLS equation (\ref{full_NNLS_complex}) for the initial data 
 $\psi_{0}(x)=\mbox{sech }x$ and $\eta=1$ at the time 
 $t_0=0.5$ for two values of $\eta$ in blue,  the corresponding semiclassical solution in 
 red and the P$_{I}$ solution (\ref{F1}) in green.}
 \label{nlsfnonloceta12eu}
\end{figure}

The corresponding plots for $v$ can be seen in 
Fig.~\ref{nlsfnonloceta12ev}.
\begin{figure}[htb!]
\includegraphics[width=0.6\textwidth]{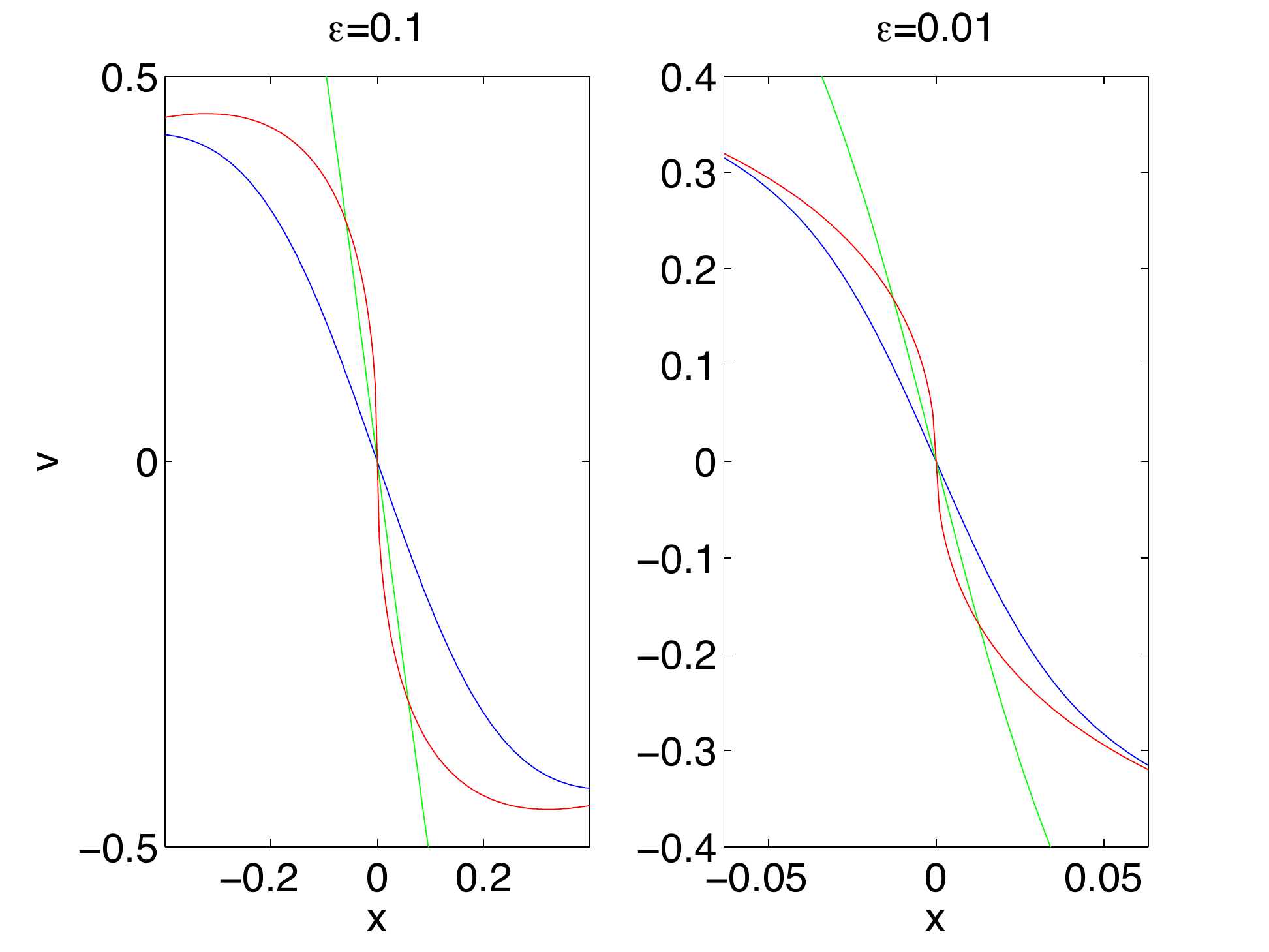}
 \caption{Solution $v$ to the focusing nonlocal NLS equation  (\ref{full_NNLS_complex}) for the initial data 
 $\psi_{0}(x)=\mbox{sech }x$ and $\eta=1$ at the time 
 $t_0=0.5$ for two values of $\epsilon$ in blue,  the corresponding semiclassical solution in 
 red and the P$_{I}$ solution (\ref{F1}) in green.}
 \label{nlsfnonloceta12ev}
\end{figure}

{\bf Acknowledgments}. The work of B.D. and T.G. was partially supported by the European Research Council Advanced Grant FroM-PDE,  by PRIN 2010-11 Grant ``Geometric and analytic theory of Hamiltonian systems in finite and infinite dimensions'' of Italian Ministry of Universities and Researches and by the  FP7 IRSES grant RIMMP  ``Random and Integrable Models in Mathematical Physics".
The work of B.D.  was also partially supported by the Russian Federation Government Grant No. 2010-220-01-077.

 {\large Boris Dubrovin}\\
 {\small SISSA, Via Bonomea 265,    34136 Trieste, Italy\\
  Steklov Math. Institute, Moscow, and N.N.Bogolyubov Laboratory of Geometric Methods in Mathematical Physics\\
Moscow State University, 119899 Moscow, Russia}\\
e-mail: dubrovin@sissa.it\\
\\
{\large Tamara Grava}\\
{\small SISSA, Via Bonomea 265, 34136 Trieste, Italy\\
and School  of Mathematics, University of Bristol, Bristol BS8 1TW, UK}\\
e-mail: grava@sissa.it\\
\\
{\large Christian Klein}\\
 {\small Institut de Math\'ematiques de Bourgogne,
		Universit\'e de Bourgogne,\\
		 9 avenue Alain Savary, 21078 Dijon, Cedex, France}\\
		e-mail: Christian.Klein@u-bourgogne.fr\\
		\\
		{\large Antonio Moro}\\
 {\small Department of Mathematics and Information Sciences,\\
  University of Northumbria at Newcastle upon Tyne, \\
 Pandon building, Camden street, NE2 1XE, Newcastle upon Tyne, UK }\\
e-mail: antonio.moro@northumbria.ac.uk

\begin{thebibliography}{999}














\bibitem{ag} G.P.~Agrawal, {\it Nonlinear Fiber Optics}. Academic Press, San Diego, 2006, 4th edition.

\bibitem{al} S.~Alinhac, {\it Blowup for Nonlinear Hyperbolic Equations.} Progress in Nonlinear Differential Equations and their Applications, {\bf 17}. Birkh\"auser Boston, Inc., Boston, MA, 1995.
\bibitem{ar} V.I.~Arnold,V.V.~Goryunov, O.V.~Lyashko, V.A.~Vasil'ev,  {\it Singularity Theory. I.}  Dynamical systems. VI, Encyclopaedia Math. Sci. {\bf 6}, Springer, Berlin, 1993. 
\bibitem{arise}  A.~Arsie, P.~Lorenzoni, A.~Moro, \textit{Integrable viscous conservation laws}, Preprint: http://xxx.lanl.gov/pdf/1301.0950.pdf.
\bibitem{bambusi} D.~Bambusi, A.~Ponno, \textit{Resonance, metastability and blow up in FPU.   The Fermi-Pasta-Ulam problem}, 191- 205,
{ \it Lecture Notes in Phys.},  {\bf 728}, Springer, Berlin, 2008. 
    \bibitem{baojin}W.~Bao, S.~Jin and P.A.~Markowich,
   On time-splitting spectral approximations for 
the Schr\"odinger equation in the semiclassical regime, 
  \emph{J. Comput. Phys.}  {\bf 175} (2002), pp.~487--524. 

\bibitem{baojin2}W.~Bao, S.~Jin and P.A.~Markowich,
Numerical study of time-splitting spectral 
  discretizations of nonlinear Schr\"odinger equations in the 
semi-classical regimes,    \emph{SIAM J. 
Sci. Comp.}  (2003) pp.~27--64. 
\bibitem{benettin} G.~Benettin,  A.~Ponno,  Time-scales to equipartition in the Fermi-Pasta-Ulam problem: finite-size effects and thermodynamic limit.    \emph{J. Stat. Phys.} {\bf 144},  (2011),  {\bf 4}, 793- 812. 
    
   \bibitem{berska}H.~Berland and B.~Skaflestad, 
  Solving the nonlinear 
    Schr\"odinger equation using exponential integrators, 
    Technical Report 3/05, 
    The Norwegian Institute of Science and Technology, 2005. 
    http://www.math.ntnu.no/preprint/. 
   
\bibitem{BIS}H.~Berland, A.L.~Islas, and C.M.~Schober, Solving the 
nonlinear Schr\"odinger equation using exponential integrators,   \emph{J. Comput. Phys.}, 255 (2007), pp. 284-299.
\bibitem{berry} M.V.~Berry, J.F.~Nye and F.J.~Wright, The Elliptic Umbilic Diffraction Catastrophe,  {\it Philosophical Transactions of the Royal Society of London. Series A}, {\bf 291} (1979),   453-484. 
\bibitem{bertola}    M.~Bertola,  A.~Tovbis,  Asymptotics of orthogonal polynomials with complex varying quartic weight: global structure, critical point behaviour and the first Painlev\'e equation. Preprint  http://xxx.lanl.gov/pdf/1108.0321.pdf.
\bibitem{bt}  M.~Bertola,  A.~Tovbis,  Universality for the focusing nonlinear Schr\"odinger equation at the gradient catastrophe point: rational breathers and poles of the Tritronqu\'ee solution to Painlev\'e-I.  \emph{Comm. Pure Appl. Math.}  {\bf 66} (2013), no. 5, 678-752. 
\bibitem{bleher} P.~Bleher,  A.~Its, 
 Semiclassical asymptotics of orthogonal polynomials, Riemann-Hilbert problem, and universality in the matrix model.  \emph{Ann. of Math.}  {\bf 150} (1999), no. 1, 185-266.

\bibitem{bour} J.~Bourgain,  {\it Global solutions of nonlinear Schr\"odinger equations.} 
American Mathematical Society Colloquium Publications, 46. American Mathematical Society, Providence, RI, 1999. viii+182 pp. ISBN: 0-8218-1919-4 
\bibitem{bo} P.~Boutroux, Recherches sur les transcendants de M. Painlev\'e et l'\'etude asymptotique des \'equations diff\'erentielles du second ordre.  \emph{ Ann. \'Ecole Norm}, {\bf 30} (1913) 265 - 375.
\bibitem{Bressan} A.~Bressan,   \emph{Hyperbolic systems of conservation laws. The one-dimensional Cauchy problem.} Oxford Lecture Series in Mathematics and its Applications, 20. Oxford University Press, Oxford, 2000.

\bibitem{bmp} \'E.~Br\'ezin,  E.~Marinari, G.~Parisi, 
A nonperturbative ambiguity free solution of a string model. 
 \emph{Phys. Lett.}  {\bf B 242}, (1990) 35--38.

\bibitem{bronski} J.C.~Bronski, J.N~Kutz, 
Numerical simulation of the semiclassical limit of the focusing nonlinear Schr\"odinger equation.
{\it Phys. Lett.},  {\bf A 254} (2002) 325 - 336.

\bibitem{bu} R.J.~Buckingham,  P.D.~Miller,
The sine-Gordon equation in the semiclassical limit: Critical behavior near a separatrix. {\it J. Anal. Math.}  {\bf 118}, (2012), no. 2, 397- 492. 

\bibitem{venak} R.~Buckingham, S.~Venakides, Long-time asymptotics of the nonlinear Schr\"odinger equation shock problem. {\it Comm. Pure Appl. Math.}, Published Online 12.03.2007.


%
\bibitem{ca} R.~Carles, WKB analysis for the nonlinear Schr\"odinger equation and instability results.  Preprint: ArXiv:math.AP/0702318.

\bibitem{cern}H.~D.~Ceniceros, A semi-implicit moving mesh 
method for the focusing nonlinear Schr\"odinger equation, 
 \emph{ Comm. Pure Appl. Anal.}  {\bf 1} (2002), pp.~1-18. 
         


\bibitem{ce} H.D.~Ceniceros, F.-R.~Tian, A numerical study of the semi-classical limit of the focusing nonlinear Schr\"odinger equation.  {\it Phys. Lett.} {\bf A 306} (2002) 25--34. 
\bibitem{cg1} T.~Claeys, T.~Grava, Universality of the break-up profile for the KdV equation in the small dispersion limit using the Riemann-Hilbert approach.  \emph{ Comm. Math. Phys. } 
{\bf  286} (2009), no. 3, 979-1009.
\bibitem{cv1} T.~Claeys,  M.~Vanlessen,  Universality of a double scaling limit near singular edge points in random matrix models.  \emph{Comm. Math. Phys.} {  \bf 273}  (2007), no. 2, 499-532. 
\bibitem{cv2} T.~Claeys, M.~Vanlessen, The existence of a real pole-free solution of the fourth order analogue of the Painlev\'e-I equation.\emph{ Nonlinearity} {\bf  20}  (2007), no. 5, 1163-1184.
\bibitem{conti}  C.~Conti, A.~Fratalocchi, M.~Peccianti, G.~Ruocco and S.~Trillo, 
Observation of a gradient catastrophe generating solitons,  {\it Phys. Rev. Letters}, {\bf 102},  (2009), 083902.

\bibitem{co} O.~Costin, Correlation between pole location and asymptotic behavior for Painlev\'e-I solutions. {\it Comm. Pure Appl. Math.} {\bf 52} (1999) 461--478. 

\bibitem{co1} O.~Costin, M.~Huang, S.~Tanveer, Proof of the Dubrovin conjecture and analysis of the tritronqu\'e solutions of PI.  Preprint http://arxiv.org/abs/1209.1009.

\bibitem{ch} M.C.~Cross and P.C.~Hohenberg, Pattern formation outside of equilibrium. {\it Rev. Mod. Phys.} {\bf 65} (1993) 851-1112.


\bibitem{DeBouard}  A.~de Bouard, Analytic solutions to nonelliptic nonlinear Schr\"odinger equations.  \emph{J. Differential Equations}, {\bf 104},  (1993), no. 1, 196- 213. 


\bibitem{dsm} L.~Degiovanni,  F.~Magri, F.V.~Sciacca,  On deformation of Poisson manifolds of hydrodynamic type.  {\it Comm. Math. Phys.} {\bf 253} (2005), no. 1, 1- 24. 

\bibitem{Deift}     P.~Deift,    \emph{Orthogonal Polynomials and Random Matrices: A  Riemann-Hilbert Approach},
    Courant Lecture Notes 3, New York University 1999.
\bibitem{dkmvz1}
    P.~Deift, T.~Kriecherbauer, K.T-R~McLaughlin, S.~Venakides, and 
    X.~Zhou,
    Uniform asymptotics for polynomials orthogonal with respect to
    varying exponential weights and applications to universality
    questions in random matrix theory,
    {\em Comm. Pure Appl. Math.} {\bf 52} (1999), 1335-1425.
\bibitem{dkmvz2}
    P.~Deift, T.~Kriecherbauer, K.T-R~McLaughlin, S.~Venakides,
    and X.~Zhou,
    Strong asymptotics of orthogonal polynomials with respect to
    exponential weights,
    {\em Comm. Pure Appl. Math.} {\bf 52} (1999), 1491-1552.

\bibitem{deiftmc}  P.~Deift,   K.~T-R~McLaughlin,  {\it A continuum limit of the Toda lattice}. Mem. Amer. Math. Soc. 131 (1998), no. 624, x+216 pp.

\bibitem{dvz} P.~Deift, S.~Venakides, X.~Zhou, New results in small dispersion KdV by an extension of the steepest descent method for Riemann-Hilbert problems.  \emph{
Internat. Math. Res. Notices}, {\bf 6},  1997,  286- 299. 

\bibitem{DZ}  P. Deift,  X.Zhou, Perturbation theory for infinite-dimensional integrable systems on the line. A case study. {\it  Acta Math.} {\bf 188} (2002), no. 2, 163Ð262. 

\bibitem{degasperis}  A.~Degasperis, Multiscale expansion and integrability of dispersive wave equations. Integrability,  {\it Lecture Notes in Phys}., {\bf 767}, Springer, Berlin, 2009, pp 215- 244.

\bibitem{difranco} J.~DiFranco, P.D.~Miller,  The semiclassical modified nonlinear Schr\"odinger equation. I. Modulation theory and spectral analysis. {\emph Phys. D} {\bf 237},  (2008), no. 7, 947-997. 

\bibitem{D}T.~Driscoll, A composite Runge-Kutta Method for the spectral Solution of semilinear PDEs,  {\it Journal of Computational Physics}, {\bf 182}, (2002), 357-367.


\bibitem{du2} B.~Dubrovin, On Hamiltonian perturbations of hyperbolic systems of conservation laws, II: universality of critical behaviour, 
{\it Comm. Math. Phys.} {\bf 267} (2006) 117 - 139.

\bibitem{dub3}B.~Dubrovin, On universality of critical behaviour in Hamiltonian PDEs. Geometry, topology, and mathematical physics, pp 59-109, {\it Amer. Math. Soc. Transl. Ser. 2}, {\bf 224}, Amer. Math. Soc., 
Providence, RI, 2008. 

\bibitem{dubel}B.~Dubrovin, M.~Elaeva, On the critical behavior in nonlinear evolutionary PDEs with small viscosity. {\it Russ. J. Math. Phys.} {\bf 19} (2012), no. 4, 449-460.

\bibitem{DGK} B.~Dubrovin, T.~Grava,  C.~Klein, On universality of critical behavior in the focusing nonlinear Schr\"odinger equation, elliptic umbilic catastrophe and the tritronqu\'ee solution to 
the Painlev\'e-I equation. {\it J. Nonlinear Sci.}  {\bf 19} (2009), no. 1, 57-94.

 \bibitem{DGK11} B.~Dubrovin, T.~Grava and C.~Klein, Numerical Study of break-up in generalized Korteweg-de Vries and Kawahara equations, \emph{ SIAM J. Appl. Math.}, {\bf  71},  (2011), 983-1008.

\bibitem{du1} B.~Dubrovin, S.-Q.~Liu, Y.~Zhang, On Hamiltonian perturbations of hyperbolic systems of conservation laws I: quasitriviality of bihamiltonian perturbations. {\it Comm. Pure Appl. Math.} {\bf 59} (2006) 559-615.

\bibitem{DN}   B.~Dubrovin, S.~Novikov,  Hydrodynamics of weakly deformed soliton lattices. Differential geometry and Hamiltonian theory.  {\it Russian Math. Surveys},  {\bf 44}  (1989), no. 6, 35-124.

\bibitem{duits} M.~Duits, A.~Kuijlaars, Painlev\'e-I asymptotics for orthogonal polynomials with respect to a varying quartic weight.  {\it Nonlinearity}  {\bf 19}  (2006), no. 10, 2211-2245. 
\bibitem{el}  G.A.~El, Resolution of a shock in hyperbolic systems modified by weak dispersion. {\it Chaos} {\bf 15 }(2005), no. 3, 037103, 21 pp. 
\bibitem{falqui}  G.~Falqui,  On a Camassa-Holm type equation with two dependent variables. {\it J. Phys. A} {\bf  39} (2006), no. 2, 327-342. 

\bibitem{fokas} A.S.~Fokas, A.R.~Its, and A.V.~Kitaev, Discrete Painlev\'e equations and their appearance in quantum gravity, {\it Comm. Math. Phys.} {\bf 142},  (1991), 313-344.
\bibitem{fo}  M.G.~Forest, J.E.~Lee, Geometry and modulation theory for the periodic nonlinear Schr\"odinger equation. In: {\it Oscillation Theory, Computation, and Methods of Compensated Compactness}
(Minneapolis, Minn., 1985), 35-69. The IMA Volumes in Mathematics and Its Applications, 2. 
Springer, New York, 1986. 

%
\bibitem{ge} P.~G\'erard, 
Remarques sur l'analyse semi-classique de l'\'equation de Schr\"odinger 
non lin\'eaire. {\it S\'eminaire sur les \'equations 
aux D\'eriv\'ees Partielles}, 1992-1993, Exp. No. XIII, 13 pp., \'Ecole Polytech., Palaiseau, 1993. 


\bibitem{getzler} E.~Getzler, A Darboux theorem for Hamiltonian operators in the formal calculus of variations. {\it Duke Math. J.} {\bf 111},  (2002), no. 3, 535-560. 

\bibitem{trillo}  N.~Ghofraniha, C.~Conti, G.~Ruocco, S.~Trillo,  Shocks in Nonlocal Media,  {\it Phys. Rev. Letters}, {\bf 99},  (2007), 043903.


\bibitem{gwp1} J.~Ginibre, G.~Velo, On a class of nonlinear Schr\"odinger equations. I. The Cauchy problem, general case. {\it J. Funct. Anal.} {\bf 32} (1979) 1-32.
    
\bibitem{gradrhy}     
I.S.~Gradshteyn, I.M.~Ryzhik, {\it Table of Integrals, Series, and 
Products.} Translated from the Russian. Sixth edition. Translation 
edited and with a preface by A.~Jeffrey and D.~Zwillinger. Academic Press, Inc., San Diego, CA, 2000.

\bibitem{physicad}T.~Grava and C.~Klein, Numerical study of the small dispersion limit of the Korteweg-de Vries equation and asymptotic solutions, {\it Physica D,} 10.1016/j.physd.2012.04.001 (2012).

 \bibitem{GK2} T.~Grava and C.~Klein, Numerical study of a multiscale expansion of KdV and Camassa-Holm equation, in  \emph{Integrable Systems and Random Matrices}, ed. by J. Baik, T. Kriecherbauer, L.-C. Li, K.D.T-R. McLaughlin and C. Tomei, Contemp. Math. Vol. 458, 81-99 (2008). 

\bibitem{gk1} T.~Grava, C.~Klein, Numerical solution of the small  dispersion limit of Korteweg de Vries and Whitham equations.  {\it 
Comm. Pure Appl. Math.}, \textbf{60}(11), 1623-1664 (2007).


\bibitem{gr} E.~Grenier, Semiclassical limit of the nonlinear Schr\"odinger equation in small time. {\it Proc. Amer. Math. Soc. } {\bf 126}  (1998) 523--530.

\bibitem{GN} P.~Grinevich, S.P.~Novikov,  String equation. II. Physical solution. (Russian)  
{\it Algebra i Analiz} {\bf 6}  (1994),  no. 3, 118--140;  translation in  {\it St. Petersburg Math. J.}  {\bf 6}  (1995),  no. 3, 553--574.

  \bibitem{GP}A.G.~Gurevich,   L.P.~Pitaevskii,  Non stationary structure of a     collisionless shock waves. {\it JEPT Letters } {\bf 17} (1973),
    193-195.


\bibitem{hk2} A.~Henrici, T.~Kappeler, Resonant normal form for even periodic FPU  chains,  {\it J. Eur. Math. Soc.} (JEMS) {\bf 11} (2009), no. 5, 1025-1056.
\bibitem{hoefer} M.A.~Hoefer, B.~Ilan,  Dark solitons, dispersive shock waves, and transverse instabilities.{\it  Multiscale Model. Simul.}  {\bf 10} (2012), no. 2, 306-341.
\bibitem{hl} T.Y.~Hou,  P.D.~Lax,  Dispersive approximations in fluid dynamics. {\it Comm. Pure Appl. Math.} {\bf 44} (1991) 1-40.

\bibitem{ilin}  A.M.~Il'in, {\it Matching of Asymptotic Expansions of Solutions of Boundary Value Problems}. AMS Translations of Mathematical Monographs, Vol. 102,
1992; 281 pp

\bibitem{in} E.L.~Ince, {\it Ordinary Differential Equations}. Dover Publications, New York, 1944.

\bibitem{jin} S.~Jin, C.D.~Levermore, D.W.~McLaughlin, The behavior of solutions of the NLS equation in the semiclassical limit. {\it Singular Limits of Dispersive Waves} (Lyon, 1991), 235--255, NATO Adv. Sci. Inst. Ser. B Phys., {\bf 320}, Plenum, New York, 1994.

\bibitem{ji} S.~Jin, C.D.~Levermore, D.W.~McLaughlin, The semiclassical limit of the defocusing NLS hierarchy. {\it Comm. Pure Appl. Math.} {\bf 52} (1999) 613--654.


\bibitem{jk} N.~Joshi, A.~Kitaev, On Boutroux's tritronqu\'ee solutions of the first Painlev\'e equation. {\it Stud. Appl. Math.} {\bf 107} (2001) 253--291.


\bibitem{kam1} S.~Kamvissis, Long time behavior for the focusing nonlinear Schr\"odinger equation with real spectral singularities.  {\it Comm. Math. Phys.} {\bf 180}  (1996) 325--341.

\bibitem{kam2} S.~Kamvissis, K.D.T.-R.~McLaughlin, P.D.~Miller,  {\it Semiclassical Soliton Ensembles for the Focusing Nonlinear Schr\"odinger Equation.} Annals of Mathematics Studies, 154. Princeton University Press, Princeton, NJ, 2003.

\bibitem{kapaev} A.A.~Kapaev,   Weakly nonlinear solutions of the equation ${\rm P}\sp 2\sb 1$, {\it Zap. Nauchn. Sem. Leningrad. Otdel. Mat. Inst. Steklov. (LOMI)} {\bf 187} (1991), 
Differentsialnaya Geom. Gruppy Li i Mekh. {\bf 12}, 88--109, 172--173, 175; translation in {\it J. Math. Sci.} {\bf 73} (1995), no. 4, 468-481.

\bibitem{ka} A.~Kapaev, Quasi-linear Stokes phenomenon for the Painlev\'e first equation. {\it J. Phys.  A: Math. Gen.} {\bf 37} (2004) 11149-11167.

\bibitem{KKG} A.~Kapaev, C.~Klein and T.~Grava, On the tritronqu\'ee  solutions of P$_I^2$, \emph{preprint} (2013) arXiv:1306.6161

\bibitem{ks} V.~Kudashev, B.~Suleimanov, A soft mechanism for the generation of dissipationless shock waves, {\it Phys. Lett.} {\bf A 221} (1996) 204--208.



\bibitem{KassT}A.-K.~Kassam and L.~Trefethen, Fourth-Order Time-Stepping for stiff PDEs, {\it SIAM J. Sci. Comput.}, {\bf 26} (2005), pp. 1214-1233.

\bibitem{merle}  C.E.~Kenig and F.~Merle,  Global well-posedness, scattering and blow-up for the energy-critical, focusing, non-linear Schr\"odinger equation in the radial case. {\it Invent. Math.} {\bf 166}, 
 (2006), no. 3, 645-675.

\bibitem{kita} A.~Kitaev, The isomonodromy technique and the elliptic asymptotics of the first Painlev\'e transcendent.  {\it Algebra i Analiz} {\bf 5} (1993), no. 3, 179--211; translation in 
{\it St. Petersburg Math. J.} {\bf 5} (1994), no. 3, 577--605.

 \bibitem{etna}C.~Klein, Fourth-Order Time-Stepping for low Dispersion Korteweg-de Vries and nonlinear Schr\"odinger Equation, Electronic {\it Transactions on Numerical Analysis.}, 39 (2008), pp. 116-135. 

\bibitem{km} Y.~Kodama, A.~Mikhailov, Obstacles to asymptotic integrability,  Algebraic aspects of integrable systems, 173--204, {\it Progr. Nonlinear Differential Equations Appl.}, {\bf 26},
Birkh\"auser, Boston, MA, 1997.

\bibitem{kong}  D.~Kong,  Formation and propagation of singularities for $2\times2$  quasilinear hyperbolic systems. 
{\it Trans. Amer. Math. Soc.}  {\bf 354},  (2002), no. 8, 3155�3179. 

\bibitem{krasny} R.~Krasny, A study of singularity formation in a vortex sheet by the point-vortex approximation.  {\it J. Fluid Mech.}  {\bf 167}  (1986) 65--93.

\bibitem{optim}J.C.~Lagarias, J.A.~Reeds, M.H.~Wright,   and P.~E.~Wright, Convergence properties of the Nelder-Mead   simplex method in low dimensions. {\it SIAM Journal of Optimization} 
{\bf 9} (1988) 112-147. 

\bibitem{ll} P.~Lax, D.~Levermore, The small dispersion limit of the Korteweg-de Vries equation. I, II, III.
{\it Comm. Pure Appl. Math.} {\bf 36} (1983) 253--290, 571--593, 809--829.
\bibitem{llv} P.D.~Lax, C.D.~Levermore, S.~Venakides,
The generation and propagation of oscillations in dispersive initial value problems and their limiting behavior. In:  {\it Important Developments in Soliton Theory}, 205--241, 
Springer Ser. Nonlinear Dynam.,  Springer, Berlin, 1993. 


\bibitem{LWZ} S.-Q.~Liu, C.-Z.~Wu, Y.~Zhang, On properties of Hamiltonian structures for a class of evolutionary PDEs,  {\it Lett. Math. Phys.} {\bf 84} (2008), no. 1, 47Ð63. 
\bibitem{LZ2} S.-Q.~Liu, Y.~Zhang, On Quasitriviality and Integrability of a Class of Scalar Evolutionary PDEs, {\it J. Geom. Phys.} {\bf 57} (2006) 101-119.
\bibitem{paleari}  P.~Lorenzoni, S.~Paleari,  Metastability and dispersive shock waves in the Fermi-Pasta-Ulam system. {\it  Phys. D}  {\bf 221} (2006), no. 2, 110-117.

\bibitem{ponce}  F.~Linares,  G.~Ponce, {\it Introduction to nonlinear dispersive equations.} Universitext. Springer, New York, 2009. xii+256 pp. ISBN: 978-0-387-84898-3 
\bibitem{miller} G.D.~Lyng, P.D.~Miller, The $N$-soliton of the focusing nonlinear Schr\"odinger equation for $N$ large. {\it Comm. Pure Appl. Math.} {\bf 60} (2007) 951-1026.

\bibitem{manakov}  S.V.~Manakov, P.M.~Santini,  On the dispersionless Kadomtsev-Petviashvili equation in n+1 dimensions: exact solutions, the Cauchy problem for small initial data and wave breaking. {\it J. Phys. A} {\bf  44} (2011), no. 40, 405203, 15 pp. 

\bibitem{mm1} L.~Mart\'\i nez-Alonso, E.~Medina, Regularization of Hele-Shaw flows , multiscaling expansions and the Painlev\'e-I equation, {\it Chaos Solitons Fractals},  {\bf 41} (2009), no. 3, 1284-1293. 

\bibitem{maoero} D.~Masoero,  A.~Raimondo,  Semiclassical limit for generalized KdV equations before the gradient catastrophe. {\it Lett. Math. Phys.} {\bf 103} (2013), no. 5, 559-583. 
\bibitem{miller0} P.D.~Miller, Z.~Xu, ~ The Benjamin-Ono hierarchy with asymptotically reflectionless initial data in the zero-dispersion limit. 
{\it Commun. Math. Sci.} {\bf 10},  (2012), no. 1, 117-130. 


 \bibitem{MM}F.~Merle, P.~Raphael, On universality of blow-up profile for $L^2$ critical nonlinear Schr\"odinger equation, {\it Invent. math. } {\bf 156},  565-672 (2004). 


\bibitem{me} G.~M\'etivier, Remarks on the well-posedness of the nonlinear Cauchy problem.  Geometric analysis of PDE and several complex variables, 337Ð356,
{\it Contemp. Math.}, {\bf 368}, Amer. Math. Soc., Providence, RI, 2005. 

\bibitem{mi} P.D.~Miller, S.~Kamvissis, On the semiclassical limit of the focusing nonlinear Schr\"odinger equation. {\it  Phys. Lett.} {\bf A  247}  (1998) 75--86.

\bibitem{moore} G.~Moore, Geometry of the string equations, {\it Comm. Math. Phys.} {\bf 133} (1990) 261-304.

\bibitem{new} A.C.~Newell, {\it Solitons in Mathematics and Physics}. CBMS-NSF Regional Conference Series in Applied Mathematics, {\bf 48}. SIAM, Philadelphia, PA, 1985. 

\bibitem{novik} S.P.~Novikov, S.V.~Manakov, L.P.~Pitaevski\u\i, V.E.~Zakharov, {\it Theory of Solitons. The Inverse Scattering Method.} Translated from the Russian. Contemporary Soviet Mathematics. Consultants Bureau [Plenum], New York, 1984.

\bibitem{krol} P.D.~Rasmussen, O.~Bang, W.~Krolikowski, Theory of nonlocal soliton interaction in nematic liquid cristals, {\it Phys. Rev. E}, {\bf 72}, (2005), 066611.
\bibitem{saya}J.~Satsuma and N.~Yajima, Initial value  problems of one-dimensional self-modulation of nonlinear waves in 
dispersive  media, \emph{Supp. Prog. Theo. Phys.}  {\bf 55} (1974), 284-306. 

\bibitem{serre} D.~Serre, {\it Syst\`emes de lois de conservation I : hyperbolicit\'e, entropies, ondes de choc}; {\it Syst\`emes de lois de conservation II: structures g\'eom\'etriques, oscillation et probl\`emes mixtes}, Paris Diderot Editeur 1996.

\bibitem{sh} A.B.~Shabat,  One-dimensional perturbations of a differential operator, and the inverse scattering problem. In: {\it Problems in Mechanics and Mathematical Physics}, 279--296. Nauka, Moscow, 1976. 

\bibitem{bvp4c}  L.F.~Shampine,  M.W.~Reichelt and J.~Kierzenka,  \textit{Solving Boundary Value Problems for Ordinary Differential  Equations in MATLAB with bvp4c}, available at  http://www.mathworks.com/bvp\_tutorial

\bibitem{si} P.~Sikivie, The caustic ring singularity. {\it Phys. Rev.} {\bf D60} (1999) 063501.

\bibitem{sl} M.~Slemrod, Monotone increasing solutions of the Painlev\'e 1 equation $y''=y\sp 2+x$ and their role in the stability of the plasma-sheath transition. {\it European J. Appl. Math.} {\bf 13} (2002) 663--680.

\bibitem{str} I.A.B.~Strachan, Deformations of the Monge/Riemann hierarchy and approximately integrable systems,  {\it J. Math. Phys.}  {\bf 44}  (2003) 251--262.

\bibitem{sulem}  C.~Sulem, P.~Sulem, {\it  The nonlinear Schr\"odinger equation. Self-focusing and wave collapse.} Applied Mathematical Sciences, 139. Springer-Verlag, New York, 1999. 

\bibitem{tao} T. Tao, Why are soliton stable? 
{\it Bull. Amer. Math. Soc. } {\bf 46} (2009), no. 1, 1Ð33.

\bibitem{Tao} T.~Tao,  {\it Nonlinear dispersive equations. Local and global analysis.} CBMS Regional Conference Series in Mathematics, 106. Published for the Conference Board of the Mathematical Sciences, Washington, DC; by the American Mathematical Society, Providence, RI, 2006. 

\bibitem{th} R.~Thom, {\it Structural Stability and Morphogenesis: An Outline of a General Theory of Models.} Reading, MA: Addison-Wesley, 1989.

\bibitem{tian} F.R.~Tian,   The initial value problem for the Whitham averaged system. {\it Comm. Math. Phys.}  {\bf 166} (1994), no. 1, 79-115. 
\bibitem{FRTYE}   F.R.~Tian, J.~Ye,  On the Whitham equations for the semiclassical limit of the defocusing nonlinear Schr\"odinger equation. {\it Comm. Pure Appl. Math.}  {\bf 52} (1999), no. 6, 655-692.


\bibitem{to1} A.~Tovbis, S.~Venakides, X.~Zhou, On semiclassical (zero dispersion limit) solutions of the focusing nonlinear Schr\"odinger equation. {\it Comm. Pure Appl. Math.} {\bf  57}  (2004) 877--985.

\bibitem{to2} A.~Tovbis, S.~Venakides, X.~Zhou,  On the long-time limit of semiclassical (zero dispersion limit) solutions of the focusing nonlinear Sch\"odinger equation: pure radiation case.  {\it Comm. Pure Appl. Math.}  {\bf 59}  (2006) 1379--1432.

\bibitem{trefethen}L.~Trefethen, {\it Spectral Methods in MATLAB}, vol. 10 of {\it Software, Environments, and Tools}, Society for Industrial and Applied Mathematics (SIAM), Philadelphia, PA, 2000.


   \bibitem{tsarev} S.~Tsarev, The geometry of Hamiltonian systems of hydrodynamic type. The generalized hodograph method, {\it Math. USSR Izv.} {\bf 37} (1991), 397--419.
\bibitem{gwp2} Y.~Tsutsumi, $L^2$-solutions for nonlinear Schr\"odinger equations and nonlinear groups. {\it Funkcial. Ekvac.} {\bf 30} (1987) 115--125.

\bibitem{venakides}  S.~Venakides,  The Korteweg-de Vries equation with small dispersion: higher order Lax-Levermore theory. {\it Comm. Pure Appl. Math.} {\bf 43} (1990), no. 3, 335-361.
\bibitem{whi} G.B.~Whitham, {\it Linear and Nonlinear Waves}. Wiley-Intersci. 1974.




\bibitem{w} H.~Whitney, On singularities of mappings of euclidean spaces. I. Mappings of the plane into the plane. 
{\it Ann. of Math.} (2) {\bf 62} (1955), 374--410. 

\bibitem{zk} N.~Zabusky, M.~Kruskal, Interaction of "Solitons" in a Collisionless Plasma and the Recurrence of Initial States {\it Phys. Rev. Lett.} {\bf 15} (1965) 2403.


\bibitem{za} V.E.~Zakharov, A.B.~Shabat, A. B. Exact theory of two-dimensional self-focusing and one-dimensional self-modulation of waves in nonlinear media. {\it Soviet Physics JETP} {\bf 34} (1972),  no. 1, 62-69.; translated from {\it \v{Z}. Eksper. Teoret. Fiz.} {\bf 1} (1971),  118-134.





\end{thebibliography}
\end{document}